\documentclass[lettersize,onecolumn]{IEEEtran}
\usepackage{amsmath,amsfonts}
\usepackage{amsthm}
\usepackage{amsfonts}
\usepackage{xcolor}
\usepackage[normalem]{ulem}
 
\usepackage{mathtools}
\usepackage{algorithmic}
\usepackage{algorithm}
\usepackage{array}
\usepackage[caption=false,font=normalsize,labelfont=sf,textfont=sf]{subfig}
\usepackage{textcomp}
\usepackage{stfloats}
\usepackage{url}
\usepackage{verbatim}
\usepackage{graphicx}
\usepackage{cite}
\usepackage{subcaption}
\usepackage{dsfont}
\usepackage{bbm}
\DeclarePairedDelimiterX{\norm}[1]{\lVert}{\rVert}{#1}
\usepackage{romannum}
\usepackage{enumitem}
\usepackage{mathrsfs}
\usepackage{xcolor}
\DeclareMathOperator*{\argmax}{arg\,max}
\DeclareMathOperator*{\argmin}{arg\,min}
\newtheorem{assumption}{Assumption}
\newtheorem{definition}{Definition}
\newtheorem{theorem}{Theorem}
\newtheorem{proposition}{Proposition}
\newtheorem{corollary}{Corollary}
\newtheorem{lemma}{Lemma}

\usepackage{xcolor}
\begin{document}
\pagenumbering{arabic}
\title{Performance Analysis of Joint Antenna Selection and Precoding Methods in Multi-user Massive MISO}

\author{Xiuxiu~Ma,~\IEEEmembership{Student Member,~IEEE}, Abla~Kammoun,~\IEEEmembership{Member,~IEEE}, Mohamed-Slim Alouini,~\IEEEmembership{Fellow Member,~IEEE}
	and~Tareq~Y.~Al-Naffouri,~\IEEEmembership{Fellow Member,~IEEE}
        % <-this % stops a space
\thanks{X. Ma, A. Kammoun, M. Alouini and T. Y.  Al-Naffouri are with the Division of Computer, Electrical and Mathematical Science \& Engineering, King Abdullah University of Science and Technology (KAUST), Thuwal, KSA. E-mails: (\{xiuxiu.ma;  abla.kammoun; slim.alouini;tareq.alnaffouri\}@kaust.edu.sa)}} 

\maketitle

\begin{abstract}
This paper presents a performance analysis of two distinct techniques for antenna selection and precoding in downlink multi-user massive multiple-input single-output systems with limited dynamic range power amplifiers. Both techniques are derived from the original formulation of the regularized-zero forcing precoder, designed as the solution to minimizing a regularized distortion. Based on this, the first technique, called the $\ell_1$-norm precoder, adopts an $\ell_1$-norm regularization term to encourage sparse solutions, thereby enabling  antenna selection. The second technique, termed the thresholded $\ell_1$-norm precoder, involves post-processing the precoder solution obtained from the first method by applying an entry-wise thresholding operation.
This work conducts a precise performance analysis to compare these two techniques. The analysis leverages the Gaussian min-max theorem which is effective for examining the asymptotic behavior of optimization problems without explicit solutions. While the analysis of the $\ell_1$-norm precoder follows from the conventional convex Gaussian min-max theorem framework, understanding the thresholded $\ell_1$-norm precoder is more complex due to the non-linear behavior introduced by the thresholding operation. To address this complexity, we develop a novel Gaussian min-max theorem tailored to these scenarios. We provide precise asymptotic behavior analysis of the precoders, focusing on metrics such as received signal-to-noise and distortion ratio and bit error rate. Our analysis demonstrates that the thresholded $\ell_1$-norm precoder can offer superior performance when the threshold parameter is carefully selected. Simulations confirm that the asymptotic results are accurate for systems equipped with hundreds of antennas at the base station, serving dozens of user terminals.
\end{abstract}


\section{Introduction}
	\IEEEPARstart{M}{assive} multiple-input single-output (MISO) has emerged as one of the critical enablers for enhancing system spectral efficiency in both current 5G networks and future 6G mobile wireless communication systems \cite{5595728,6736761,8766143}. By utilizing a large number of antennas, massive MISO promises significant improvements in network capacity. However, conventional implementations encounter significant challenges due to the need for multiple radio frequency (RF) chains, which increase costs related to size, power consumption, and hardware complexity \cite{8613274,8353836}.
	To address these issues, various solutions have been proposed to mitigate the associated costs and complexity. Principal approaches include antenna selection \cite{1512132}, load-modulated arrays \cite{7432147}, and hybrid analog-digital precoding \cite{6928432,8284058}. Among these techniques, antenna selection stands out as it can be integrated into conventional MISO structures without imposing a significant burden. While many recent techniques necessitate substantial redesigns of transceivers, antenna selection can be easily implemented by merely adding a switching network to existing structures.

	Antenna selection can be employed in both transmission and reception scenarios, involving the choice of a specific subset from all available antennas. This technique effectively reduces the number of RF chains without significantly compromising performance.
	Identifying the optimal set of antennas poses a challenge, as it is an NP-hard problem that typically requires an exhaustive search through all possible antenna combinations. This complexity becomes particularly high in large-scale MISO systems with numerous antennas. Given the stringent requirements of 5G networks, such as low latency and real-time decision-making \cite{6824752}, achieving optimal antenna selection is not a viable approach.
	Consequently, extensive research has focused on sub-optimal antenna selection procedures that offer manageable complexity. These approaches aim to optimize specific performance criteria tailored to the application at hand. For instance, some studies focus on maximizing capacity, spectral efficiency, or energy efficiency by selecting either transmit or receive antennas under perfect channel state information (CSI) conditions \cite{6824974,6725592,8234671,8519787,8831394}. Other approaches include maximizing the signal-to-noise ratio (SNR) at the receiver \cite{1492684}, minimizing the bit error rate (BER) for coded or uncoded systems \cite{4109369} or minimizing the transmit power \cite{6477161,7881064}.
	
	Despite the advantages offered by the aforementioned techniques, an important consideration is the maximum power at each RF chain, which remains unrestricted. In large-scale MISO systems, this lack of power limitation can lead to increased implementation costs, as high dynamic range power amplifiers are needed to achieve desired performance levels. Alternatively, it may result in significant performance degradation within a constrained implementation budget.
	
	In this paper, we investigate joint antenna selection and precoding design for downlink multi-user MISO transmissions that take into account the constraint of limited power at each RF chain. Specifically, we extend upon the methodologies proposed in \cite{glse} and \cite{arxiv}, which explore precoders with bounded transmit power at each RF chain. These precoders are derived from the regularized zero-forcing (RZF) precoding formulated as a solution to a convex problem by considering constraints imposing limited absolute values for each antenna. To perform antenna selection, we introduce two techniques. The first employs an $\ell_1$-regularization term on the precoder described in \cite{glse,short_glse}, termed the $\ell_1$-norm precoder. This approach allows precise control over solution sparsity by adjusting the $\ell_1$ regularization term. The level of sparsity can be tailored based on the number of RF chains available in the system. However, when the system has few RF chains, achieving the desired sparsity may necessitate a high regularization term, potentially leading to significant performance degradation.
	To address this challenge, we propose a novel technique known as the thresholded $\ell_1$-norm precoder. This method involves applying a thresholding operator to each entry of the $\ell_1$-norm precoder. By doing so, we can effectively use a lower regularization term while aligning with the required number of RF chains by suitable thresholding. 
	
	The main contribution of this work lies in conducting a rigorous analysis of both techniques in the asymptotic regime where the number of transmit antennas and users scale proportionally assuming Rayleigh channels. The $\ell_1$-norm precoder introduced here can be regarded as a specific instance of the generalized least squares estimator (GLSE) studied in \cite{glse}, specifically the variant employing $\ell_1$ regularization. However, our analysis diverges from \cite{glse} in several significant aspects. Firstly, we incorporate BPSK signaling, contrasting with \cite{glse} which focuses on Gaussian symbols. Secondly, we provide precise performance approximations for both tight lower bound of signal-to-noise-and-distortion ratio (SINAD) and BER, which are based on the empirical distribution of distortion error, while the work in \cite{glse} focuses only on the study of the distortion error power, from which the BER can not be obtained.  Lastly, our approach relies on the convex Gaussian min-max theorem (cGMT), contrasting with the replica techniques used in \cite{glse}.
	Additionally, we introduce the thresholded $\ell_1$-norm precoder, which is created by discarding the weak entries of the $\ell_1$-norm precoder to meet the constraint on the number of RF chains. It is worth noting that this approach coincides with the solution that satisfies the required number of RF chains while being the closest in Euclidean norm to the $\ell_1$-norm precoder. Although this idea of adjusting non-feasible solutions to find the nearest feasible one is intuitive and reminiscent of the quantization process used to manage finite resolution in hardware, it has not been previously proposed in the context of linear precoding.
	Furthermore, even when such techniques have been applied in other areas, their analysis is very limited. The challenge lies in the fact that post-processing, such as the thresholding operation, causes the asymptotic distribution of the solution (in this case, the $\ell_1$-norm precoder) to undergo a non-linear transformation. In our case, this post-processing leads to the loss of critical information about the correlation between the channel and the post-processed solution, which is essential for conducting accurate performance analysis.  
	To address this challenge, we introduce a novel GMT based on a redefined formulation of the original $\ell_1$-norm precoder. This theorem not only resolves the loss of correlation issue but also holds potential interest beyond the considered application. Specifically, it can be used to analyze the performance impact of adjustments applied to solutions of optimization problems for practical implementation purposes. 

The content of this work is organized as follows. Section \ref{sec2} introduces the system model and describes the $\ell_1$-norm precoder, along with the metrics of interest. In Section \ref{sec3}, we present our results on the performance of the $\ell_1$-norm precoder. Section \ref{sec4} details our findings on the performance of the thresholded $\ell_1$-norm precoder. Along with the results, we provide numerical results to illustrate the accuracy of our theoretical findings. We conclude by Section \ref{sec_v} which summarizes the key contributions of this work.

{\noindent
{\bf Notations.} For simplicity, we make use of the following notations onwards.   
 For any vector $\mathbf{x}$, we use $x_i$ or $[{\bf x}]_i$ to denote its $i$-th element. For $q\in[1,\infty]\cup\{0\}$ We use the notation $\|\cdot\|_q$ to denote the $\ell_q$ norm and for $\ell_2$ norm, we just denote it by $\|\cdot\|$. The notation $(x)_{+}$ represents $\max(x,0)$.  The notation $\mathcal{N}(\mu,\sigma^2)$ denotes a  Gaussian random variable, which has mean $\mu$ and variance $\sigma^2$.
  We quantify distance between distributions using the Wassertein distance. Recall that the Wassertein $r$-distance between two measures $\mu$ and $\eta$ supported in $\mathbb{R}^{d}$ is defined as:$$
\mathcal{W}_r(\mu,\eta):=\left(\inf_{\gamma\in\mathcal{C}(\mu,\eta)}\int \|{\bf x}-{\bf y}\|_r^r\gamma(dx,dy)\right)^{\frac{1}{r}}
	$$ 
	where the infimum is taken over all couplings of $\mu$ and $\eta$ \cite{villani_optimal_2008}. We denote by $X\xrightarrow[P]{}Y$ that $X$ converges to $Y$ in probability.}

	\section{System model}\label{sec2}
	We consider a downlink transmission system in which a base station equipped with $n$ transmit antennas communicates with $m$ users, each outfitted with a single antenna. We assume single-carrier transmissions over flat-fading channels. Under these assumptions, the received signal at the $i$-th user is given by:
	\begin{equation}
		{y}_i={\bf h}_i^{T}{\bf x}+{ z}_i, \ \ i=1,\cdots m \label{eq:s}
	\end{equation}
	where $z_i$ is the additive noise assumed to follow Gaussian distribution with mean zero and variance $\sigma^2$, ${\bf h}_i$ denotes the channel vector between the base station and user $i$ and ${\bf x}=[x_1,\cdots,x_n]^{T}$ is the precoded vector transmitted by the antennas.
	By concatenating the received signals by each user to a vector ${\bf y}$, we may rewrite \eqref{eq:s} into a compact form as 
	$$
	{\bf y}={\bf H}{\bf x}+{\bf z}
	$$
	where ${\bf H}$ is the $m\times n$ channel matrix whose rows are the vectors ${\bf h}_1^{T},\cdots,{\bf h}_m^{T}$.
	
	The transmitted vector ${\bf x}$ represents the output of the precoder block, designed to transform the information symbol vector denoted as ${\bf s}=[s_1,\cdots,s_m]$. The main objective of this transformation is to mitigate the effect of  multi-user interference  and channel distortion on the system's performance. Generally, designing the precoder involves solving an optimization problem, where the objective may encompass maximizing the channel capacity or minimizing the multi-user interference, all while adhering to specific constraints such as limited power budget or hardware limitations. The careful consideration of these constraints ensures that the designed precoder is both efficient and practical in real-world scenarios.

{In this paper, to focus the reader's attention on our methodology, we consider a simple scenario in which the intended symbols for the users are drawn uniformly from the BPSK constellation, and both the channel matrix and the additive noise are real-valued.  Under this setting, we consider a precoding strategy that involves solving an optimization problem, where the objective penalizes high distortions by incorporating the sum of squared residuals, given by }
	$$
	\mathcal{R}({\bf x})=\frac{1}{n}\|{\bf Hx}-\sqrt{\rho}{\bf s}\|^2,
	$$
	where $\rho$ is a positive control parameter. 
	It is worth noting that two well-known precoders, namely the regularized zero-forcing (RZF) and the zero-forcing (ZF) precoders, belong to this class of precoders. The RZF precoder is formulated as
	\begin{equation}
		\hat{\bf x}_{\rm RZF} := \arg\min_{{\bf x}} \mathcal{R}({\bf x}) +\frac{\lambda_2}{n}\|{\bf x}\|^2, \label{eq:RZF}
	\end{equation}
	here, $\lambda_2$ is a strictly positive regularization parameter that controls the trade-off between reducing distortions (through $\mathcal{R}({\bf x})$) and the norm regularization term $\frac{1}{n}\|{\bf x}\|^2$
	while the ZF precoder is given as 
	$$
	\hat{\bf x}_{\rm ZF} := \arg\min_{{\bf x}} \mathcal{R}({\bf x})
	$$
	and aims to minimize the objective function $\mathcal{R}({\bf x})$ solely, without any additional regularization term.
	
	Although the ZF and RZF precoders are widely used and effective, they lack inherent power control, which can be problematic when utilizing low dynamic range power amplifiers. When such amplifiers are employed, unregulated power levels may lead to suboptimal performance and undesirable distortions.
	To address this issue and ensure the compatibility of the RZF with low dynamic range power amplifiers, a natural and practical approach is to introduce a power constraint on each antenna in the optimization problem \eqref{eq:RZF}. This is equivalent to finding the solution to the following problem 
	\begin{equation}
		\hat{\bf x} \in\arg\min_{\substack{ \|{\bf x}\|_\infty\leq \sqrt{P}}} \mathcal{R}({\bf x}) +\frac{\lambda_2}{n}\|{\bf x}\|^2 \label{eq:nlse}
	\end{equation}
	where $P$ denotes the maximum allowed power at each antenna. 
	This type of precoder belongs to the class of non-linear least squares precoder, which was initially proposed in \cite{8100647}. A more precise performance analysis of the precoder in \eqref{eq:nlse} has been provided in our previous work, as documented in \cite{arxiv}.
	
	Another crucial factor that significantly impacts the power consumption and cost of a massive MISO implementation is the number of RF chains. In typical massive MISO systems, there can be hundreds, if not more, antennas utilized. Traditionally, each antenna is allocated a dedicated RF transmission chain, consisting of a digital-to-analog (D/A) converter, a mixer, and a power amplifier. However, this conventional approach not only consumes a considerable amount of power but also incurs high costs due to the need for multiple RF components.
	
	To address these challenges and achieve further cost and power savings in massive MISO implementations, antenna selection has emerged as a powerful signal processing technique. Instead of using a dedicated RF chain for each antenna, antenna selection dynamically switches the available RF chains to the most appropriate subset of antennas. 
	
	Antenna selection can be performed either before precoding or jointly with the precoding procedure \cite{8937029}. In our work, we focus on joint antenna selection and precoding. Building upon the non-linear least squares precoder \eqref{eq:nlse}, we formulate the joint antenna selection and precoding problem as follows:
	\begin{equation}
		\begin{aligned}
			\hat{\bf x}_{\ell_0} &\in \arg\min_{\substack{ \|{\bf x}\|_\infty\leq \sqrt{P} \\ \|{\bf x}\|_0=k} } \mathcal{R}({\bf x}) +\frac{\lambda_2}{n}\|{\bf x}\|^2 ,
		\end{aligned}
		\label{eq:l0_norm}
	\end{equation}
	here, $k$ represents the number of available RF chains, and we aim to find an optimal solution $\hat{\bf x}_{\ell_0}$ that solves \eqref{eq:l0_norm} under the additional $\ell_0$-norm constraint $\|{\bf x}\|_0=k$.  
	
	However, solving the precoder defined in \eqref{eq:l0_norm} involves a combinatorial problem, which can be computationally expensive. To reduce the complexity, a natural approach is to replace the $\ell_0$ norm with the $\ell_1$ norm. Unlike the $\ell_0$ norm, the $\ell_1$ norm is convex, allowing us to utilize available numerical tools for convex optimization. 
	More specifically, we propose considering the following precoder
	\begin{equation}
		\hat{\bf x}_{\ell_1} \in\arg\min_{\substack{  \|{\bf x}\|_\infty\leq \sqrt{P}}} \mathcal{C}_{\lambda,\rho}({\bf x}):=\mathcal{R}({\bf x}) +\frac{\lambda_2}{n}\|{\bf x}\|^2  +\frac{\lambda_1}{n}\|{\bf x}\|_1.\label{eq:nlse_sparse}
	\end{equation}
	The above precoder belongs also to the class of non-linear least squares precoders which were proposed and analyzed in \cite{8100647}.
	However, it is crucial to highlight that the analysis in \cite{8100647} primarily focuses on the achievable rate and does not extend to the BER, which is equally important in practical communication systems. Furthermore, the aforementioned analysis does not provide insights into determining the appropriate value of $\lambda_1$ that ensures the resulting solution complies with the number of available RF chains, i.e., $\frac{1}{n}\|\hat{\bf x}_{\ell_1}\|_0\sim \frac{k}{n}$.
	
	In our work, we aim to bridge this gap by conducting a comprehensive analysis that goes beyond the achievable rate. We investigate several performance metrics, such as the SINAD, the BER, as well as the per-antenna power of the $\ell_1$-norm  precoder \eqref{eq:nlse_sparse} to obtain a more complete understanding of its behavior. Additionally, our analysis explores the impact of the parameter $\lambda_1$ to find the most suitable value that guarantees the resulting solution meets the constraint $\frac{1}{n}\|\hat{\bf x}_{\ell_1}\|_0\sim \frac{k}{n}$.
	
	To facilitate our investigation, our analysis leverages the statistics of the channel matrix and operates under the asymptotic regime, where the number of antennas $n$ and the number of users $m$ grow simultaneously while maintaining a fixed ratio $\delta:=\frac{m}{n}$. This assumption is commonly used in asymptotic studies and allows for obtaining tight closed-form approximations of the performances providing valuable insights into the precoder's behavior under different network sizes and conditions.
	Furthermore, we assume that the channel matrix ${\bf H}$ has independent and identically distributed Gaussian entries with zero mean and  variance equal to $\frac{1}{n}$. 
   {While this assumption is likely to be violated in practice, it enables a tractable analysis and captures the randomness and variability of real-world channel conditions. The insights gained under this setting are amenable to generalization in broader contexts.}

    More formally, we carry out our analysis  under the following assumptions. 
	\begin{assumption}
		The number of antennas $n$ and the number of users $m$ grow to infinity at a fixed ratio $\delta:=\frac{m}{n}$. 
		\label{ass:regime}
	\end{assumption}
	\begin{assumption}
		The channel matrix ${\bf H}$ has independent and identically distributed Gaussian entries with zero mean  and  variance equal to $\frac{1}{n}$.
		\label{ass:statistic} 
	\end{assumption}

	Before presenting our results, let us define the formal expressions for the metrics that will be studied in our work.
	
	\noindent{\bf Per-antenna power.} When $\hat{\bf x}_{\ell_1}$ is transmitted, the per-antenna consumed power is defined as:
	$$
	P_b(\hat{\bf x}_{\ell_1}):= \frac{\|\hat{\bf x}_{\ell_1}\|^2}{n}.
	$$ 
	This metric quantifies the average power allocated to each antenna in the transmission process.

   {Upon implementing this $\hat{\bf x}_{\ell_1}$, we assume that all users  scale the received signal by a factor $\zeta$\footnote{Since all users experience the same channel statistics, the $\zeta$ determined by some strategy should be the same across all users.}, which will be specified later. The scaled received signal for user $i$ is given by:
	\begin{equation*}
		\zeta y_i= s_i+\zeta{\bf h}_i^{T}\hat{\bf x}_{\ell_1}- s_i+\zeta z_i.  
	\end{equation*}
   This expression decomposes the received signal into three parts: the intended symbol $s_i$, the correlated error $\zeta{\bf h}_i^{T}\hat{\bf x}_{\ell_1} - s_i$ resulting from channel distortion, and the independent noise.
   Based on this, we define the following metrics:}
    
	 \noindent{{\bf A tight lower bound on the average per-user SINAD.}}
   {We first calculate the SINAD experienced by user $i$ as:
	\begin{equation*}
		{\rm SINAD}_i(\hat{\bf x}_{\ell_1})= \frac{1 }{\mathbb{E}_{{\bf s}}\left|\zeta{\bf h}_i^{T}\hat{\bf x}_{\ell_1}- s_i\right|^2 +\zeta^2\sigma^2},
	\end{equation*}
	where the expectation is taken over the distribution of the information symbol vector ${\bf s}$. To assess the average SINAD performance across the network, we define the average per-user SINAD as:
	\begin{equation*}
		\overline{\rm SINAD}(\hat{\bf x}_{\ell_1})=\frac{1}{m}\sum_{i=1}^m {\rm SINAD}_i(\hat{\bf x}_{\ell_1}).
	\end{equation*}
By Jensen's inequality, we obtain a lower bound:
\begin{align*}
	\overline{\rm SINAD}(\hat{\bf x}_{\ell_1})\geq\frac{1 }{\frac{1}{m}\sum_{i=1}^m\mathbb{E}_{{\bf s}}\left|\zeta{\bf h}_i^{T}\hat{\bf x}_{\ell_1}- s_i\right|^2 +\zeta^2\sigma^2}:=\overline{\rm SINAD}_{lb}(\hat{\bf x}). 
\end{align*}
Although our analysis will focus on this lower bound  $\overline{\rm SINAD}_{lb}(\hat{\bf x}_{\ell_1})$, we assert that it provides a tight approximation of the actual per-user SINAD. Previous random matrix analyses of the RZF indicate that the quantity $\mathbb{E}_{{\bf s}}\left|\zeta{\bf h}_i^{T}\hat{\bf x}_{\ell_1}- s_i\right|^2$ 
does not fluctuate significantly and converges to a deterministic value \cite{Wagner2009LargeSA}. We believe the same holds true for the precoder under investigation. Since we assume that the channel matrix has i.i.d. entries, all users experience identical fading. Therefore, the expected distortion for user
$i$, $\mathbb{E}_{{\bf s}}\left|\zeta{\bf h}_i^{T}\hat{\bf x}_{\ell_1}- s_i\right|^2$, should converge to the average distortion across all users: $\frac{1}{m}\sum_{i=1}^m\mathbb{E}_{{\bf s}}\left|\zeta{\bf h}_i^{T}\hat{\bf x}_{\ell_1}-s_i\right|^2$.
After presenting the result on the convergence of the SINAD lower bound in Corollary \ref{cor:SINAD}, we will provide the exact conditions required to establish the tightness of this bound. We will also discuss why the GMT framework is not suitable for rigorously proving it. A rigorous proof would likely require additional developments in variance control, which are beyond the scope of the present work.}

\noindent{\bf BER.} Given the transmitted vector $\hat{\bf x}_{\ell_1}$, 
	a bit error event for user $i$ occurs when ${\rm sign}(\zeta y_i)\neq s_i$. To quantify the BER, we define it as follows:
	\begin{equation*}
		\mathrm{BER}(\hat{\bf x}_{\ell_1}):=\frac{1}{m}\sum_{i=1}^{m}{\bf 1}_{\{\mathrm{sign}(\zeta y_i)\neq s_i\}},
	\end{equation*}
	where ${\bf 1}_{\{\cdot\}}$ denotes the indicator function.

{It follows from the above that all quantities of interest are primarily derived from the transmitted vector $\hat{\bf x}_{\ell_1}$ and the noise-free received vector $\hat{\bf e}_{\ell_1}:={\bf H}\hat{\bf x}_{\ell_1}$. Both quantities will be investigated extensively in later sections. For convenience, we refer to $\hat{\bf e}_{\ell_1}$ as 'distortion', in the sense that it represents the transmitted vector after being distorted by the channel. }

	\section{Joint antenna selection and precoding via $\ell_{1}$-norm induced sparsity}\label{sec3}
	In this section, we explore the asymptotic behavior of the $\ell_1$-norm precoder described in \eqref{eq:nlse_sparse}. By defining $\mathcal{D}:=\left\{\rho,\delta,\lambda_1,\lambda_2,P,\sigma^2\right\}$ as our domain, we reveal that the behavior of the $\ell_1$-norm precoder can be effectively characterized by a max-min scalar optimization problem, involving only the constants in $\mathcal{D}$.
	
	To start, we present our findings regarding the solutions to the max-min problem and demonstrate their connection to a system of fixed-point equations. Subsequently, we delve into a comprehensive convergence analysis of the performance metrics of the $\ell_1$-norm precoder. Through rigorous investigation, we show that these metrics converge to deterministic quantities, which are solely dependent on the solutions of the fixed-point equations studied earlier.
	
	By combining the results from the max-min scalar optimization problem with the convergence analysis of performance metrics, we gain a complete understanding of the performance characteristics of the $\ell_1$-norm precoder.
	
	\subsection{The max-min scalar optimization problem and the equivalent fixed-point problem}
	The following max-min problem plays a major role in our results,
	\begin{equation}
		\max_{\beta\geq0 } \min_{\tau\geq 0} \psi(\tau,\beta):= \frac{\beta\tau\delta}{2}-\frac{\beta^2}{4}+\frac{\beta\rho}{2\tau} + \mathbb{E}\left[\min_{|x|<\sqrt{P}} \frac{\beta}{2\tau }x^2+\lambda_2 x^2-\beta Hx+\lambda_1|x|\right] \label{eq:scalar1}
	\end{equation}
	where ${H}\sim\mathcal{N}(0,1)$ and the expectation is taken over its distribution. 
	
	For $t\geq 0$, we define the following proximal operator function:
	$$
	{\rm prox}(y;t):=\arg\min_{|x|<\sqrt{P}}\frac{1}{2}(x-y)^2+t|x|
	$$
	whose explicit formulation is given in \eqref{eq:prox_explicit}.
	In the following proposition, we characterize the saddle-point of \eqref{eq:scalar1} and provide necessary and sufficient conditions for its uniqueness.
\begin{proposition}
		The max-min in \eqref{eq:scalar1} is achieved at a unique couple $(\beta^\star(\lambda_1,\lambda_2),\tau^\star(\lambda_1,\lambda_2))$ if and only if $\max(\lambda_2,\lambda_1)>0$ or $\lambda_1=\lambda_2=0$ and $\delta \geq 1$. In this  case, $(\beta^\star(\lambda_1,\lambda_2),\tau^\star(\lambda_1,\lambda_2))$ is solution to the following system of equations:
		\begin{equation*}
		\left\{
		\begin{array}{ll}
			\tau^2\delta  & = \rho+\mathbb{E}\left[\left({\rm prox}(\tilde{\tau}H;\frac{\lambda_1\tilde{\tau}}{\beta})\right)^2\right]\\
			\beta&=2\tau\delta -2\mathbb{E}\left[H{\rm prox}(\tilde{\tau}H;\frac{\lambda_1\tilde{\tau}}{\beta})\right]
		\end{array}
		\right. 
		\end{equation*}
		where $\tilde{\tau} = \frac{1}{\frac{1}{\tau}+\frac{2\lambda_2}{\beta}}$.
			\begin{proof}
			See Section \ref{fix_point}.
		\end{proof}
		\label{prop:fixed_points}
	\end{proposition}

	For the sake of simplicity, we will often in the sequel omit the dependency on $\lambda_1$ and $\lambda_2$ and write simply $\tau^\star$ and $\beta^\star$. 

Based on the result of Proposition \ref{prop:fixed_points}, we shall assume the following from now onwards.
\begin{assumption}
		We assume either $\lambda_2=\lambda_1=0$ and $\delta\geq 1$ or $\max(\lambda_2,\lambda_1)>0$ with $\lambda_2$ and $\lambda_1$ non-negative constants. 
		\label{ass:regime_lambda}
	\end{assumption}
	
	\subsection{Asymptotic analysis of the $\ell_1$-norm precoder $\hat{\bf x}_{\ell_1}$}
 {In this part, we reveal the behavior of the $\ell_1$-norm precoder $\hat{\bf x}_{\ell_1}$ via the distributional characterizations of $\hat{\bf x}_{\ell_1}$ and $\hat{\bf e}_{\ell_1}={\bf H}\hat{\bf x}_{\ell_1}$.    
    From now onward, we denote by  $\mathcal{S}^{\otimes m}$ the set of $m\times 1$ vectors with entries from the BPSK constellation; by 'constant', we refer to any quantity that solely depends on $\mathcal{D}$, and unless otherwise specified, the terms $C$ and $c$ represent constants that may vary from one line to another.}
\begin{theorem}[Control of the optimal cost]
		Let $\hat{\bf x}_{\ell_1}$ be a solution to \eqref{eq:nlse_sparse}. Then, under Assumption \ref{ass:regime}, \ref{ass:statistic} and \ref{ass:regime_lambda} there exists positive constants $C$, $c$ and $\gamma$ that depend only on the domain $\mathcal{D}$ such that for any $\epsilon>0$ sufficiently small:
		$$
		\mathbb{P}\Big[|\mathcal{C}_{\lambda,\rho}(\hat{\bf x}_{\ell_1})-\psi(\tau^\star,\beta^\star)|\geq \gamma\epsilon\Big]\leq \frac{C}{\epsilon}\exp(-cn\epsilon^2),
		$$
		where the probability is with respect to the distribution of ${\bf H}$ and ${\bf s}$. 
		\begin{proof}
			The poof is in section \ref{sec:proof_optimal_cost}. 
		\end{proof}
		\label{th:optimal_cost}
	\end{theorem}

    For ${\bf x}\in\mathbb{R}^n$ with entries $x_1,\cdots,x_n$, we define $\hat{\mu}({\bf x})$ the empirical distribution of its entries:
	$$
	\hat{\mu}({\bf x})=\frac{1}{n} \sum_{i=1}^{n}\delta_{x_i}.
	$$In the following Theorem, we prove the empirical distribution $\hat{\mu}(\hat{\bf x}_{\ell_1})$ can be approximated in the asymptotic regime by $\nu^\star$ the distribution of ${\rm prox}(\tilde{\tau}^\star H,\frac{\lambda_1\tilde{\tau}^\star}{\beta^\star})$.
	\begin{theorem}[Distributional characterization of $\hat{\bf x}_{\ell_1}$]Under Assumption \ref{ass:regime}, \ref{ass:statistic} and \ref{ass:regime_lambda},
		for all $r\in[1,\infty)$, there exists constants $C$ and $c$ that depend on $\mathcal{D}$ and $r$ such that for any $\epsilon>0$ sufficiently small,
		$$
		 \mathbb{P}\Big[\left(\mathcal{W}_r(\hat{\mu}(\hat{\bf x}_{\ell_1}),\nu^\star)\right)^r\geq \epsilon  \Big]\leq \frac{C}{\epsilon^2}\exp(-cn\epsilon^4), 
		$$where the probability is with respect to the distribution of ${\bf H}$ and ${\bf s}$.
			\begin{proof}
			See Appendix \ref{app:measure}. 
		\end{proof}
		\label{cor:wassertein}
	\end{theorem}
\begin{corollary}[Per-antenna power] Under assumption \ref{ass:regime}, \ref{ass:statistic} and \ref{ass:regime_lambda}, the per-antenna power $P_b(\hat{\bf x}_{\ell_1})$ satisfies the following result:
		$$
		 P_b(\hat{\bf x}_{\ell_1})\xrightarrow[P]{} ((\tau^\star)^2\delta-\rho). 
		$$
		\label{cor:power_antenna}
	\end{corollary}
\begin{theorem}[The precoder's sparsity level]
		Assume $\lambda_1>0$. Under Assumption \ref{ass:regime}, \ref{ass:statistic}, we can estimate the number of  non-zero elements of the precoder $\hat{\bf x}_{\ell_1}$ as,
		$$
		 \kappa(\hat{\bf x}_{\ell_1}):=\frac{1}{n}\|\hat{\bf x}_{\ell_1}\|_0\xrightarrow[P]{}2Q(\frac{\lambda_1}{\beta^\star}).
		$$
	 \label{th:l0_norm}\begin{proof}See Appendix \ref{sec:sparsity_estimation}.\end{proof}
	\end{theorem}

For ${\bf e}\in\mathbb{R}^{m}$ with entries $e_1,\cdots,e_m$ and ${\bf s}\in \mathcal{S}^{\otimes{m}}$ with entries $s_1,\cdots,s_m$, we define $\hat{\mu}({\bf e},{\bf s})$ the joint empirical distribution as follows:
	$$
	\hat{\mu}({\bf e},{\bf s})=\frac{1}{m}\sum_{i=1}^m \delta_{e_i,s_i}.
	$$

{For $G$ a standard Gaussian random variable, and $S$ a random variable taking $+1$ and $-1$ with equal probability,
define the variable $E(G,S)$ as:
\begin{align*}E(G,S)=\frac{\beta^\star\sqrt{\delta(\tau^\star)^2-\rho}}{2\tau^\star\delta}G+\frac{2\tau^\star\delta-\beta^\star}{2\tau^\star\delta}\sqrt{\rho}S.\end{align*}}
In the following Theorem, we prove that the joint empirical distribution $\hat{\mu}(\hat{\bf e}_{\ell_1}={\bf H}\hat{\bf x}_{\ell_1},{\bf s})$ can be approximated in the asymptotic regime by $\mu^\star$ the joint distribution of the couple $(E,S)$. 
	\begin{theorem}[Distributional characterization of $\hat{\bf e}_{\ell_1}$]\label{Theo3}
		Under Assumption \ref{ass:regime}, \ref{ass:statistic} and \ref{ass:regime_lambda}, there exist positive constants $c$ and $C$  such that for any $\epsilon>0$ sufficiently small,
		$$
		\mathbb{P}\Big[\mathcal{W}_2(\hat{\mu}(\hat{\bf e}_{\ell_1},{\bf s}),\mu^\star)
		\geq \epsilon\Big]\leq \frac{C}{\epsilon^2}\exp(-cn\epsilon^4) 
		$$where the probability is with respect to the distribution of ${\bf H}$ and ${\bf s}$.
		\begin{proof}See Appendix \ref{sec:conv_emp_measure}. \end{proof}
	\end{theorem}
A notable conclusion drawn from Theorem \ref{cor:wassertein} and \ref{Theo3} is that the asymptotic distribution of $\hat{\bf x}_{\ell_1}$ and $\hat{\bf e}_{\ell_1}$ is uniform over permutations of coordinates, which is expected since we considered uncorrelated channels.
	In other words, at each coordinate, the marginal distribution of $\hat{\bf x}_{\ell_1}$ and $(\hat{\bf e}_{\ell_1},{\bf s})$ asymptotically become the same and coincide with $\nu^\star$ and $\mu^\star$, respectively.
    
{Given this result, we obtain the asymptotic characterization of the received signal: Given $Z\sim\mathcal{N}(0,\sigma^2)$ independent of $G$ and $S$, the empirical distribution of the received signals ${\bf h}_i^T\hat{\bf x}_{\ell_1}+{z}_i, i=1,\cdots,m$ across the users converges to  the distribution of the variable \begin{align*}Y=E(G,S)+Z=\frac{\beta^\star\sqrt{\delta(\tau^\star)^2-\rho}}{2\tau^\star\delta}G+\frac{2\tau^\star\delta-\beta^\star}{2\tau^\star\delta}\sqrt{\rho}S+Z.\end{align*} Based on this expression, it is natural to consider scaling the received vector by the scalar $\zeta:=\frac{2\tau^\star\delta}{\sqrt{\rho}(2\tau^\star\delta-\beta^\star)}$. After scaling, the empirical distribution of the scaled received vector $\zeta{\bf y}=\zeta{\bf H}\hat{\bf x}_{\ell_1}+\zeta{\bf z}$ converges to  the distribution of
\begin{align}\zeta Y=\frac{\beta^\star\sqrt{\delta(\tau^\star)^2-\rho}}{\sqrt{\rho}(2\tau^\star\delta-\beta^\star)} G+ S+\frac{2\tau^\star\delta}{\sqrt{\rho}(2\tau^\star\delta-\beta^\star)}Z.\label{ey}\end{align}From these results, we can easily obtain the following corollaries:
\begin{corollary}[Per-user SINAD lower bound] Under Assumption \ref{ass:regime}, \ref{ass:statistic}, and \ref{ass:regime_lambda} the per-user SINAD lower bound of the $\ell_1$-norm  precoder converges to$$
		 \overline{\rm SINAD}_{lb}(\hat{\bf x}_{\ell_1})\xrightarrow[P]{}\frac{ {\rho}(1-\frac{\beta^\star}{2\tau^\star \delta})^2}{ {\sigma^2+\frac{(\beta^\star)^2}{4}\frac{(\tau^\star)^2\delta-\rho}{(\tau^\star)^2\delta^2}}}.
		\label{cor:SINAD}
		$$\end{corollary}}
\noindent{\bf Identifying Key Conditions for Tightness of the SINAD Lower-Bound}

 By expanding the expression of $\overline{\rm SINAD}(\hat{\bf x}_{\ell_1})$, we can write:\begin{align}
		&\overline{\rm SINAD}(\hat{\bf x}_{\ell_1})-\frac{ {\rho}\left(1-\frac{\beta^\star}{2\tau^\star \delta}\right)^2}{ {\sigma^2+\frac{(\beta^\star)^2}{4}\frac{(\tau^\star)^2\delta-\rho}{(\tau^\star)^2\delta^2}}}\nonumber\\
		=&\frac{1}{m}\sum_{i=1}^m \frac{ \left({\sigma^2+\frac{(\beta^\star)^2}{4}\frac{(\tau^\star)^2\delta-\rho}{(\tau^\star)^2\delta^2}}   \right) - \rho\left(\mathbb{E}_{\bf s}\left[\left|\zeta[\hat{\bf e}_{\ell_1}]_i- {s_i} \right|^2\right]+\zeta^2\sigma^2\right)\left(1-\frac{\beta^\star}{2\tau^\star \delta}\right)^2 }{\left({\sigma^2+\frac{(\beta^\star)^2}{4}\frac{(\tau^\star)^2\delta-\rho}{(\tau^\star)^2\delta^2}}\right)\left(\mathbb{E}_{\bf s}\left[\left|\zeta[\hat{\bf e}_{\ell_1}]_i- {s_i} \right|^2\right]+\zeta^2\sigma^2\right)}\nonumber\\
        =&\frac{1}{m}\sum_{i=1}^m \frac{ \left({\sigma^2+\frac{(\beta^\star)^2}{4}\frac{(\tau^\star)^2\delta-\rho}{(\tau^\star)^2\delta^2}}   \right) - \rho\left(\mathbb{E}_{\bf s}\left[\left|\zeta[\hat{\bf e}_{\ell_1}]_i- {s_i} \right|^2\right]+\zeta^2\sigma^2\right)\left(1-\frac{\beta^\star}{2\tau^\star \delta}\right)^2 }{\left({\sigma^2+\frac{(\beta^\star)^2}{4}\frac{(\tau^\star)^2\delta-\rho}{(\tau^\star)^2\delta^2}}\right)\left(\mathbb{E}_{\bf s}\left[\left|\zeta[\hat{\bf e}_{\ell_1}]_i- {s_i} \right|^2\right]+\zeta^2\sigma^2\right)}\nonumber \\
        &-\frac{1}{m}\sum_{i=1}^m\frac{ \left({\sigma^2+\frac{(\beta^\star)^2}{4}\frac{(\tau^\star)^2\delta-\rho}{(\tau^\star)^2\delta^2}}   \right) - \rho\left(\mathbb{E}_{\bf s}\left[\left|\zeta[\hat{\bf e}_{\ell_1}]_i- {s_i} \right|^2\right]+\zeta^2\sigma^2\right)\left(1-\frac{\beta^\star}{2\tau^\star \delta}\right)^2 }{\left({\sigma^2+\frac{(\beta^\star)^2}{4}\frac{(\tau^\star)^2\delta-\rho}{(\tau^\star)^2\delta^2}}\right)^2 }\cdot\rho\big(1-\frac{\beta^\star}{2\tau^\star \delta}\big)^2\nonumber\\&+\frac{1}{m}\sum_{i=1}^m \frac{ \left({\sigma^2+\frac{(\beta^\star)^2}{4}\frac{(\tau^\star)^2\delta-\rho}{(\tau^\star)^2\delta^2}}   \right) - \rho\left(\mathbb{E}_{\bf s}\left[\left|\zeta[\hat{\bf e}_{\ell_1}]_i- {s_i} \right|^2\right]+\zeta^2\sigma^2\right)\left(1-\frac{\beta^\star}{2\tau^\star \delta}\right)^2 }{\left({\sigma^2+\frac{(\beta^\star)^2}{4}\frac{(\tau^\star)^2\delta-\rho}{(\tau^\star)^2\delta^2}}\right)^2 }\cdot\rho\big(1-\frac{\beta^\star}{2\tau^\star \delta}\big)^2\nonumber\\=& \frac{1}{m}\sum_{i=1}^m \frac{ \left(  \left( {\sigma^2+\frac{(\beta^\star)^2}{4}\frac{(\tau^\star)^2\delta-\rho}{(\tau^\star)^2\delta^2}}  \right) - \rho\left(\mathbb{E}_{\bf s}\left[\left|\zeta[\hat{\bf e}_{\ell_1}]_i- {s_i} \right|^2\right]+\zeta^2\sigma^2\right)\left(1-\frac{\beta^\star}{2\tau^\star \delta}\right)^2\right)^2 }{\left({\sigma^2+\frac{(\beta^\star)^2}{4}\frac{(\tau^\star)^2\delta-\rho}{(\tau^\star)^2\delta^2}}\right)^2\left(\mathbb{E}_{\bf s}\left[\left|\zeta[\hat{\bf e}_{\ell_1}]_i- {s_i} \right|^2\right]+\zeta^2\sigma^2\right)}\nonumber\\&+\frac{1}{m}\sum_{i=1}^m \frac{ \left({\sigma^2+\frac{(\beta^\star)^2}{4}\frac{(\tau^\star)^2\delta-\rho}{(\tau^\star)^2\delta^2}}   \right) - \rho\left(\mathbb{E}_{\bf s}\left[\left|\zeta[\hat{\bf e}_{\ell_1}]_i- {s_i} \right|^2\right]+\zeta^2\sigma^2\right)\left(1-\frac{\beta^\star}{2\tau^\star \delta}\right)^2 }{\left({\sigma^2+\frac{(\beta^\star)^2}{4}\frac{(\tau^\star)^2\delta-\rho}{(\tau^\star)^2\delta^2}}\right)^2 }\cdot\rho\big(1-\frac{\beta^\star}{2\tau^\star \delta}\big)^2.\label{e16}
		\end{align}
The second term of \eqref{e16} evidently converges to zero as per Corollary \ref{cor:SINAD}. {Thus, it suffices to establish that:}
		\begin{equation}
		\max_{1\leq i\leq m} \left|\left({\sigma^2+\frac{(\beta^\star)^2}{4}\frac{(\tau^\star)^2\delta-\rho}{(\tau^\star)^2\delta^2}}   \right) - \rho\left(\mathbb{E}_{\bf s}\left[\left|\zeta[\hat{\bf e}_{\ell_1}]_i- {s_i} \right|^2\right]+\zeta^2\sigma^2\right)\left(1-\frac{\beta^\star}{2\tau^\star \delta}\right)^2\right|\xrightarrow[P]{} 0 ,\label{eq:requirement1}
		\end{equation}
		and 
		\begin{equation}
			\max_{1\leq i\leq m} {\rm var}\left(\mathbb{E}_{\bf s}\left[\left|\zeta[\hat{\bf e}_{\ell_1}]_i- {s_i} \right|^2\right]\right) \xrightarrow[P]{} 0. \label{eq:requirement2}
		\end{equation} 
		We believe that the GMT framework may not be sufficient to prove the above convergences. Indeed, this framework establishes an equivalence  between the original empirical distribution of $\hat{\bf e}_{\ell_1}$ and   $E(G,S)$. However, this equivalence does not extend to the distribution of  $\mathbb{E}_{\bf s}|\zeta[\hat{\bf e}_{\ell_1}]_i-s_i|^2, \  i=1,\cdots,n$. Specifically,  the distribution of  $\mathbb{E}_{\bf s}[|[\zeta\hat{\bf e}_{\ell_1}]_i-s_i|^2]$ is not asymptotically equivalent to $\mathbb{E}_S [|\zeta E(G,S)-S|^2]$. To see this, it suffices to note that the latter has a  non-vanishing variance, while we believe that the variance of the former
		 goes to zero.  

		To prove \eqref{eq:requirement1} and \eqref{eq:requirement2}, we believe that other tools for variance control are necessary. These are beyond the scope of this paper, which focuses on the GMT tool.

Consider the receivers determine $\hat{\bf s}$ by ${\rm sign}(\zeta {\bf y})$, according to \eqref{ey}, we can characterize the asymptotil BER as the following corollary:
  \begin{corollary}[BER]
		Under Assumption \ref{ass:regime}, \ref{ass:statistic} and \ref{ass:regime_lambda}, the bit error rate of the $\ell_1$-norm  precoder coverges to:
		$$
		 {\rm BER}(\hat{\bf x}_{\ell_1})\xrightarrow[P]{} Q\left(\frac{\sqrt{\rho}(1-\frac{\beta^\star}{2\tau^\star \delta})}{\sqrt{\sigma^2+\frac{(\beta^\star)^2}{4}\frac{(\tau^\star)^2\delta-\rho}{(\tau^\star)^2\delta^2}}}\right)  .
		$$ 
	\end{corollary}

\subsection{Simulations and results}
{In this section, we carry out experiments to validate the accuracy of our theoretical predictions.  In the following figures, solid lines represent our theoretical predictions, while markers indicate empirical results averaged over $50$ realizations of ${\bf s}$ and noise, for a single realization of ${\bf H}$. As seen, there is a close match between the empirical and theoretical results, which thereby confirms the accuracy of our results. Our focus is on studying the role of the parameter $\lambda_1$ in enabling antenna selection. The effects of other parameters are not presented; for those, we refer the reader to our previous work \cite{arxiv}. In particular, we consider the common MISO setting where $\delta<1$, in which case $\lambda_2$ can be set to a very small value. In this small-value regime, changes in $\lambda_2$ have minimal impact on system performance. The parameter $\rho$, on the other hand, can be adjusted to achieve any desired average power $P_b <P$, as discussed in \cite{arxiv}.
 }

{To enable antenna selection via $\ell_1$ regularization, the result of Theorem \ref{th:l0_norm} plays a key role, }as it relates the asymptotic fraction of non-zero elements in the $\ell_1$-norm precoder to the regularization parameter $\lambda_1$ and the parameter $\beta^\star$, which depends on the parameters in the domain $\mathcal{D}$. While the relationship between the non-zero elements and $\lambda_1$ is not explicit due to the lack of a direct formulation for $\beta^\star$ with respect to $\lambda_1$, the theorem can be applied in practice to determine the regularization parameter that achieves the desired number of non-zero elements in the precoder.
As $\lambda_1$ increases, the sparsity of the $\ell_1$-norm precoder also increases. However, higher values of $\lambda_1$ required to achieve a certain sparsity level may lead to reduced performance. To illustrate this phenomenon, Figure \ref{f1} shows the performance in terms SINAD, BER as well as fraction of non-zero elements $\kappa$ for different values of $\delta$ and $P_b$. For each value of $\lambda_1$, the value of $\rho$ is set to ensure the specified value for the asymptotic per-user power.  As depicted, increasing $\lambda_1$ improves the sparsity level of the $\ell_1$-norm precoder, but this improvement comes at the expense of significant performance degradation. Therefore, if the number of available RF chains is limited, inducing sparsity in the precoder via $\ell_1$-regularization may not be a viable option.
\begin{figure} 
	\centering
	\includegraphics[width=2.8in]{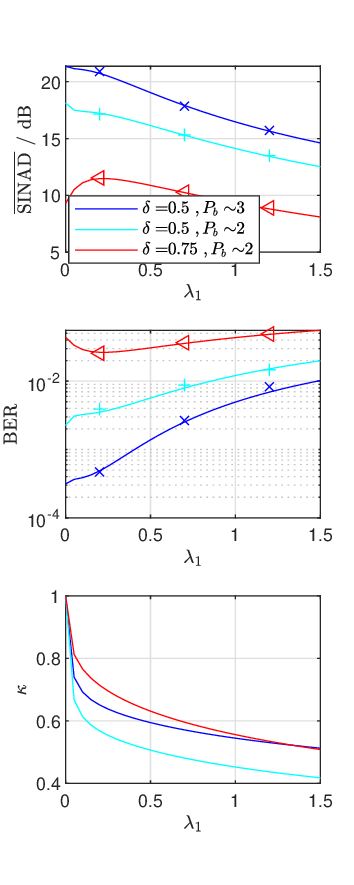}
	\caption{Performance of the $\ell_1$ norm precoder vs. $\lambda_1$: $\rho$ is tuned to ensure asymptotic $P_b$ equals the specified value. Parameters: $P=10$, $\lambda_2=0.005$, $n=256$, $\sigma=0.5$. }
	\label{f1}
\end{figure}

\section{Joint antenna selection and precoding via thresholding}\label{sec4}
	\subsection{Motivation}
	As observed in the previous section, carefully adjusting the regularization parameter $\lambda_1$ allows us to enforce the precoder $\hat{\bf x}_{\ell_1}$ to have a fraction of $\frac{k}{n}$ non-zero entries, with the possibility of a vanishing fraction of additional non-zero entries. However, a significant challenge arises when $k$ is small, as it necessitates a large regularization parameter. This can lead to a substantial loss in performance since the regularization term based on the $\ell_1$-norm does not carry relevant information for precoding.
	To address this issue, a natural strategy involves retaining only the $k$ non-zero entries with the highest magnitudes and setting the other values to zero. This precoder is called the $k$ selection based precoder and is denoted by $\hat{\bf x}_{\ell_1}^k$.  Delving into the intricacies of this approach proves to be complex. This method entails rearranging the entries of the $\ell_1$-norm precoder to single out the $k$ entries with the highest magnitudes, which presents its own set of challenges.
	Instead of direct magnitude-based selection, we opt for a different approach by applying a thresholding operator to the $\ell_1$-norm precoder. This operator, denoted as $\mathcal{T}$, retains only those elements exceeding a certain threshold $t_x$:
	$$
	\mathcal{T}(x):=\left\{\begin{array}{ll} 
		x & \text{if } |x|\geq t_x\\
		0 & \text{otherwise}
	\end{array}.\right.
	$$
	By employing thresholding, we can effectively control the number of non-zero elements while using a small regularization parameter $\lambda_1$, thus avoiding performance degradation. This approach strikes a balance between maintaining a practical number of RF chains and preserving performance, making it a valuable technique in practical applications. 
	Let $\hat{\bf x}_{\ell_1}^{t_x}$ be the precoder obtained from $\hat{\bf x}_{\ell_1}$ after applying the thresholding operation  $\mathcal{T}$.
	We denote this precoder by the thresholded $\ell_1$-norm precoder.
	In the sequel, we start by providing guidelines on how to choose $t_x$  to guarantee that the precoder solution  $\hat{\bf x}_{\ell_1}^{t_x}$ has a fraction of $\kappa=\frac{k}{n}$ non-zero entries, with the possibility of a vanishing fraction of additional non-zero entries, or in other words satisfies: $\frac{1}{n}\|\hat{\bf x}_{\ell_1}^{t_x}\|_0\sim \frac{k}{n}$. 
	We show that with this suitable choice of $t_x$, the performances
	 of the thresholded precoder are asymptotically equivalent to those of the $k$ selection based precoder $\hat{\bf x}_{\ell_1}^{k}$. 
	Then, we go on to developing accurate performance analysis of the precoder  $\hat{\bf x}_{\ell_1}^{t_x}$.  By understanding its behavior, we can assess its suitability for practical communication systems and determine the optimal threshold $t_x$ for achieving the desired trade-off between sparsity and performance.
	
	Before proceeding further, let us shed light on the challenges posed by the study of this precoder, which is more intricate than what is typically found in the literature. Our analysis in Theorem \ref{cor:wassertein} and Theorem \ref{Theo3} characterizes the asymptotic behavior of the empirical distribution of the precoded vector $\hat{\bf x}_{\ell_1}$ and the empirical joint distribution of $(\hat{\bf e}_{\ell_1},{\bf s})$, by connecting them to the distribution of ${\rm prox}(\tilde{\tau}^\star H,\frac{\lambda_1 \tilde{\tau}^\star}{\beta^\star})$ and $E(G,S)$. It is essential to note that both ${\rm prox}(\tilde{\tau}^\star H,\frac{\lambda_1 \tilde{\tau}^\star}{\beta^\star})$ and $E(G,S)$ are independent of ${\bf H}$ and ${\bf s}$ and thus of $\hat{\bf x}_{\ell_1}$ and $\hat{\bf e}_{\ell_1}$. The connection is established purely based on their distributional similarity.
	From this insight, we can infer the asymptotic behavior of any functional of the precoded vector $\hat{\bf x}_{\ell_1}$ or that of the distortion $\hat{\bf e}_{\ell_1}$ and the transmitted vector ${\bf s}$. The reason why we could infer the behavior of the distortion is that this distortion vector can be expressed as the solution to a max problem, as shown in the Appendix.
	However, after applying thresholding, while our results can still be used to characterize the asymptotic behavior of non-random functionals of $\hat{\bf x}_{\ell_1}^{t_x}$, they cannot directly be used to infer that of the distortion, as it is significantly altered by the thresholding operation. Now, the distortion is represented by the quantity $\hat{\bf e}_{\ell_1}^{t_x}={\bf H}\hat{\bf x}_{\ell_1}^{t_x}$. Unlike $\hat{\bf e}_{\ell_1}$, it cannot be directly derived from the solution of a max problem, and the information about the dependence between ${\bf H}$ and $\hat{\bf x}_{\ell_1}$ has been lost. This loss of dependence poses additional challenges in understanding the behavior of $\hat{\bf e}_{\ell_1}^{t_x}$ and its impact on the overall performance of the thresholded precoder.

    \subsection{Choice of the thresholding parameter $t_x$.}
	\begin{theorem}
		Let $t_x\in(0,\sqrt{P})$. Under Assumption \ref{ass:regime}, \ref{ass:statistic} and \ref{ass:regime_lambda}, the entries of the $\ell_1$-based norm precoder satisfies:
		\begin{equation*}
			\mathbb{P}\Big[\left|\frac{\#\{i\in\{1,\cdots,n\}, |[\hat{\bf x}_{\ell_1}]_i|\geq t_x\}}{n}-2Q(\frac{t_x}{\tilde{\tau}^\star}+\frac{\lambda_1}{\beta^\star})\right|\geq \epsilon\Big]\leq \frac{C}{\epsilon^3}\exp(-cn\epsilon^6).
		\end{equation*}
		\label{th:threshold}
			\begin{proof}
			See Section \ref{ix_a}.
		\end{proof}
	\end{theorem}

	The above theorem provides valuable insights into the selection of the threshold value $t_x$. Let's consider a scenario where only $k$ RF chains are available. This situation necessitates the choice of a fraction $\kappa=\frac{k}{n}$ of the available antennas. Let $\epsilon>0$ sufficiently small. By applying thresholding to the $\ell_1$-norm precoder at $\overline{t}_x:=\left(Q^{-1}\left(\frac{\kappa}{2}\right)-\frac{\lambda_1}{\beta^\star}\right)\tilde{\tau}^\star$, Theorem \ref{th:threshold} informs us that the $\ell_1$-norm- precoder $\hat{\bf x}_{\ell_1}^{\overline{t}_x}$ and the $k$-selection based precoder $\hat{\bf x}_{\ell_1}^{k}$ differ in at most $n\epsilon$ entries with a probability at least $1-\frac{C}{\epsilon^3}\exp(-cn\epsilon^6)$. 

 More formally,  the following result is a direct by-product  of Theorem \ref{th:threshold}:
	\begin{equation}
		\mathbb{P}\Big[\|\hat{\bf x}_{\ell_1}^{\overline{t}_x}-\hat{\bf x}_{\ell_1}^{k}\|_0 \geq n\epsilon \Big] \leq \frac{C}{\epsilon^3}\exp(-cn\epsilon^6).
		\label{eq:l_0_norm_control}
	\end{equation}
   {By noting that \(\frac{1}{n}\|\hat{\bf x}_{\ell_1}^{t} - \hat{\bf x}_{\ell_1}^{k}\|^2 \leq \frac{4P}{n} \|\hat{\bf x}_{\ell_1}^{t} - \hat{\bf x}_{\ell_1}^{k}\|_0\), we conclude that, with probability at least \(1 - \frac{C}{\epsilon^3} \exp(-cn\epsilon^2)\), we have \(\frac{1}{n} \|\hat{\bf x}_{\ell_1}^{t} - \hat{\bf x}_{\ell_1}^{k}\|^2 \leq \epsilon\). This result implies that any Lipschitz-continuous function of the thresholded  \(\ell_1\)-norm  precoder \(\hat{\bf x}_{\ell_1}^{\overline{t}_x}\) and  the $k$ selection based precoder \(\hat{\bf x}_{\ell_1}^{k}\) converges to the same asymptotic limit. In particular, the per-antenna power of both precoders converges to the same asymptotic value.  }

	Define $\hat{\bf e}_{\ell_1}^k={\bf H}\hat{\bf x}_{\ell_1}^{k}$, the distortion  of the precoding vector $\hat{\bf x}_{\ell_1}^{k}$. Then, {as an important implication of Theorem \ref{th:threshold},} we obtain the following result:
	\begin{corollary}
		Let $\epsilon>0$.	Under Assumption \ref{ass:regime}, \ref{ass:statistic} and \ref{ass:regime_lambda}, there exists constants $C$, $c$ and $\tilde{c}$ such that with probability $1-\frac{C}{\epsilon^3}\exp(-cn\epsilon^6)$
		$$
		\mathcal{W}_2(\hat{\mu}(\hat{\bf e}_{\ell_1}^k,{\bf s}),\hat{\mu}(\hat{\bf e}_{\ell_1}^{\overline{t}_x},{\bf s}))\leq \tilde{c}\sqrt{\epsilon},
		$$
		or in other words:
		$$
		\mathbb{P}\left[\mathcal{W}_2(\hat{\mu}(\hat{\bf e}_{\ell_1}^k,{\bf s}),\hat{\mu}(\hat{\bf e}_{\ell_1}^{\overline{t}_x},{\bf s}))\geq \tilde{c}{\sqrt{\epsilon}}\right]\leq \frac{C}{\epsilon^3}\exp(-cn\epsilon^6).
		$$
		\begin{proof}
			To begin with, we write the following equality:
			$$
			\frac{1}{\sqrt{m}}\|\hat{\bf e}_{\ell_1}^k-\hat{\bf e}_{\ell_1}^{\overline{t}_x}\|_2=\frac{1}{\sqrt{m}}\|{\bf H}(\hat{\bf x}_{\ell_1}^{\overline{t}_x}-\hat{\bf x}_{\ell_1}^{k})\|_2\leq \frac{1}{\sqrt{m}}\|{\bf H}\|2\sqrt{P}\sqrt{\|\hat{\bf x}_{\ell_1}^{\overline{t}_x}-\hat{\bf x}_{\ell_1}^{k}\|_0}
			$$
			where the last inequality follows from noting that:
			$$
			\|\hat{\bf x}_{\ell_1}^{\overline{t}_x}-\hat{\bf x}_{\ell_1}^{k}\|\leq 2\sqrt{P}\sqrt{\|\hat{\bf x}_{\ell_1}^{\overline{t}_x}-\hat{\bf x}_{\ell_1}^{k}\|_0}.
			$$
			It follows from Lemma \ref{lem:spectral_norm} that with probability at least $1-2\exp(-n/2)$, 
			$$
			\|{\bf H}\|\leq 3{\max(1,\sqrt{\delta})},
			$$
			hence,  with probability at least $1-\frac{C}{\epsilon^3}\exp(-cn\epsilon^6)$,
			$$
			\frac{1}{\sqrt{m}}\|\hat{\bf e}_{\ell_1}^k-\hat{\bf e}_{\ell_1}^{\overline{t}_x}\|_2\leq 6\sqrt{P}\frac{\sqrt{n}}{\sqrt{m}}\sqrt{\epsilon}{\max(1,\sqrt{\delta})}.
			$$
			or equivalently,
			$$
			\frac{1}{\sqrt{m}}\|\hat{\bf e}_{\ell_1}^k-\hat{\bf e}_{\ell_1}^{\overline{t}_x}\|_2\leq  6{\sqrt{P}}\sqrt{\epsilon}{\max(\frac{1}{\sqrt{\delta}},1)} .
			$$
			The result follows by noting that:
			$$
			\mathcal{W}_2(\hat{\mu}(\hat{\bf e}_{\ell_1}^k,{\bf s}),\hat{\mu}(\hat{\bf e}_{\ell_1}^{\overline{t}_x},{\bf s}))\leq \frac{1}{\sqrt{m}}\|\hat{\bf e}_{\ell_1}^k-\hat{\bf e}_{\ell_1}^{\overline{t}_x}\|_2.
			$$
		\end{proof}
		\label{cor:wassertein_d}
	\end{corollary}

		Corollary \ref{cor:wassertein_d} extends the result of Theorem \ref{th:threshold} by establishing the equivalence in terms of Wassertein distance between the empirical measures $\hat{\mu}(\hat{\bf e}_{\ell_1}^k,{\bf s})$ and $\hat{\mu}(\hat{\bf e}_{\ell_1}^{\overline{t}_x},{\bf s})$. { This observation directly implies that, just as previously established for the per-antenna power, both the BER and the SINAD lower bound of the two precoders  converge to the same asymptotic value.}  
        In the sequel, we will henceforth  focus  on the performance of the thresholded $\ell_1$-norm precoder.

	\subsection{Asymptotic performance characterization of the thresholded precoder}
   { Assume  $\hat{\bf x}_{\ell_1}^{t_x}=\mathcal{T}(\hat{\bf x}_{\ell_1})$ is transmitted.  We define the per-antenna transmit power of the thresholded $\ell_1$ norm precoder as $$P_b(\hat{\bf x}_{\ell_1}^{t_x})=\frac{\|\hat{\bf x}_{\ell_1}^{t_x}\|^2}{n}.$$ As discussed earlier in section \ref{sec2}, at the receiver side, we assume that all users scale their received signals by an appropriate factor $\omega$, which results in $$ \omega y^{t_x}_i=s_i+\omega {\bf h}_i^T\hat{\bf x}_{\ell_1}^{t_x}-s_i+\omega z_i.$$ The  SINAD experienced by user $i$ is thus defined  as $${\rm SINAD}_i(\hat{\bf x}_{\ell_1}^{t_x})=\frac{1}{\mathbb{E}_{\bf s}|\omega{\bf h}_i^T\hat{\bf x}_{\ell_1}^{t_x}-s_i|^2+\omega^2\sigma^2}$$ We focus on studying a lower bound of the average per-user SINAD  defined as$$\overline{\rm SINAD}_{lb}(\hat{\bf x}_{\ell_1}^{t_x}):=\frac{1}{\frac{1}{m}\sum_{i=1}^m\mathbb{E}_{\bf s}|\omega{\bf h}_i^T\hat{\bf x}_{\ell_1}^{t_x}-s_i|^2+\omega^2\sigma^2}.$$ and the BER expressed as: $${\rm BER}(\hat{\bf x}_{\ell_1}^{t_x}):=\frac{1}{m}\sum_{i=1}^m\boldsymbol{1}_{\{{\rm sign}(\omega y^{t_x}_i)\neq s_i\}}.$$}
    
\begin{theorem}[Per-antenna power] Under Assumption \ref{ass:regime}, \ref{ass:statistic} and \ref{ass:regime_lambda}, the per-antenna power of the thresholded $\ell_1$-norm precoder satisfies
		$$
		P_b(\hat{\bf x}_{\ell_1}^{t_x})\xrightarrow[P]{} \mathbb{E}\Big[|({\rm prox}(\tilde{\tau}^\star H,\frac{\lambda_1\tilde{\tau}^\star}{\beta^\star}))^2 {\bf 1}_{\{|{\rm prox}(\tilde{\tau}^\star H,\frac{\lambda_1\tilde{\tau}^\star}{\beta^\star})|\geq t_x\}}\Big].
		$$
        \begin{proof}
		The Theorem is a direct by-product of Theorem \ref{th:function} in Appendix \ref{app:th_function}. 
		\end{proof}
	\end{theorem}

	While the performance of the thresholded precoder in terms of per-antenna power does not require knowledge of the distribution of the distortion vector $\hat{\bf e}_{\ell_1}^{t_x}$, its performance in terms of SINAD and BER does. The following theorem describes the asymptotic distribution of the distortion vector $\hat{\bf e}_{\ell_1}^{t_x}$, which is essential for characterizing the performance in terms of SINAD and BER. Before stating this result, we shall define the following quantities:
	\begin{align*}
		\theta^\star&:=-\frac{\mathbb{E}\big[H{\rm prox}(\tilde{\tau}^\star H;\frac{\lambda_1\tilde{\tau}^\star}{\beta^\star}){\bf 1}_{\{|{\rm prox}(\tilde{\tau}^\star H;\frac{\lambda_1\tilde{\tau}^\star}{\beta^\star})|\geq t_x\}}\big]}{\tau^\star\delta},\\ 
		\tilde{\alpha}^\star&:=	\sqrt{\mathbb{E}\big[\big|{\rm prox}(\tilde{\tau}^\star H;\frac{\lambda_1\tilde{\tau}^\star}{\beta^\star})\big|^2{\bf 1}_{\{|{\rm prox}(\tilde{\tau}^\star H;\frac{\lambda_1\tilde{\tau}^\star}{\beta^\star})|\geq t_x\}}\big]}.			
	\end{align*}
	\begin{theorem}[Distributional characterization of $\hat{\bf e}_{\ell_1}^{t_x}$]
		Let $\overline{G}$ be a standard normal variable. Let ${S}$ be drawn from the Rademacher distribution such that $S=1,-1$ with equal probabilities. Define ${E}_t$ as: 
        $$
E_t(\overline{G},S)= \sqrt{(\tilde{\alpha}^\star)^2+(\theta^\star)^2(\delta (\tau^\star)^2-\rho)+2(\tilde{\alpha}^\star)^2\theta^\star} \overline{G}-{\theta}^\star\sqrt{\rho}{S}
        $$
		and let $\mu_t^\star$ be the joint distribution of the couple $(E_t,{S})$. 
		
		Consider the set:
		\begin{align*}
		\mathcal{S}_e^\circ=\left\{\mathcal{W}_2^2(\hat{\mu}({\bf e},{\bf s}),\mu_t^\star)\geq c_e\sqrt{\epsilon}\right\}, 
		\end{align*}
		then under Assumption \ref{ass:regime}, \ref{ass:statistic} and \ref{ass:regime_lambda},
		$$
		\mathbb{P}\Big[\hat{\bf e}_{\ell_1}^{t_x}\in \mathcal{S}_e^\circ\Big]\xrightarrow[P]{} 0. 
		$$
		\label{th:main_theorem}
	\end{theorem}
	\begin{proof}
		See Section \ref{app_core}.
	\end{proof}
{Similar to the $\ell_1$-norm precoder, Theorem \ref{th:main_theorem} allows us to obtain an asymptotic characterization of the received signal. Let $\overline{G}$ be a standard Gaussian vector, $Z\sim\mathcal{N}(0,\sigma^2)$, independent of $\overline{G}$ and ${S}$, then the empirical distribution of the received vector converges asymptotically to the distribution of the random variable: 
$$Y_t=E_t(\overline{G},S)+Z=\sqrt{(\tilde{\alpha}^\star)^2+(\theta^\star)^2(\delta (\tau^\star)^2-\rho)+2(\tilde{\alpha}^\star)^2\theta^\star} \overline{G}-{\theta}^\star\sqrt{\rho}{S}+Z.$$ It is thus reasonable to consider scaling the received signal by $\omega=-\frac{1}{\sqrt{\rho}\theta^\star}$. Prior to scaling, the empirical distribution of the received vector $\omega {\bf y}$ converges to that  of\begin{align*}\omega Y=-\frac{\sqrt{(\tilde{\alpha}^\star)^2+(\theta^\star)^2(\delta (\tau^\star)^2-\rho)+2(\tilde{\alpha}^\star)^2\theta^\star}  }{\sqrt{\rho}\theta^\star} \overline{G}  +{S}-\frac{1}{\sqrt{\rho}\theta^\star}Z. \end{align*} From this asymptotic characterization, we can readily derive the following performance metrics:
}
	\begin{corollary}[Per-user {\rm SINAD} lower bound]
		Under Assumption \ref{ass:regime}, \ref{ass:statistic} and \ref{ass:regime_lambda}, the per-user SINAD lower bound of the thresholded $\ell_1$-norm precoder $\hat{\bf x}_{\ell_1}^{t_x}$ satisfies:
		$$
		\overline{\rm SINAD}_{lb}(\hat{\bf x}_{\ell_1}^{t_x})\xrightarrow[P]{}\frac{\rho(\theta^\star)^2}{(\tilde{\alpha}^\star)^2+(\theta^\star)^2(\delta (\tau^\star)^2-\rho)+2(\tilde{\alpha}^\star)^2\theta^\star+\sigma^2}.
		$$
		\label{cor:SINAD_th}
	\end{corollary}
\begin{corollary}[{\rm BER}]
	Under Assumption \ref{ass:regime}, \ref{ass:statistic} and \ref{ass:regime_lambda}, the BER of the thresholded $\ell_1$-norm precoder $\hat{\bf x}_{\ell_1}^{t_x}$ satisfies:
	$$
	{\rm BER}(\hat{\bf x}_{\ell_1}^{t_x})\xrightarrow[P]{}Q\left(\frac{-\sqrt{\rho}\theta^\star}{\sqrt{(\tilde{\alpha}^\star)^2+(\theta^\star)^2(\delta (\tau^\star)^2-\rho)+2(\tilde{\alpha}^\star)^2\theta^\star+\sigma^2}}\right).
	$$
	\label{cor:BER_th}
\end{corollary}	
	
	The thresholded $\ell_1$-norm precoder utilizes $\ell_1$ regularization and a thresholding operation to achieve the required sparsity level. In practice, the asymptotic expressions for SINAD lower bound and BER in Corollary \ref{cor:SINAD_th} and Corollary \ref{cor:BER_th} can be used to find the optimal combination of the threshold parameter and regularization parameter, optimizing performance under RF chain constraints. As will be demonstrated in the simulation section, this method achieves better performance than the $\ell_1$ norm precoder based solely on tuning the regularization parameter $\lambda_1$.

 \subsection{Simulations and results}
\begin{figure} 
	\centering
	\includegraphics[width=2.8in]{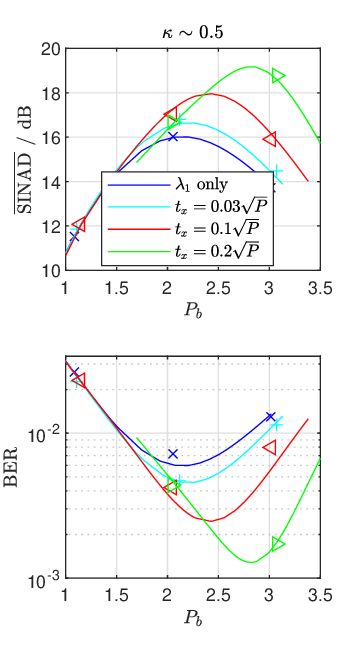}
	\caption{Performance of the thresholded $\ell_1$ norm precoder versus $P_b$ for different threshold values: Here $\rho$ and $\lambda_1$ are tuned to satisfy an asymptotic value for the sparsity level of the precoder $\kappa=\frac{k}{n}=0.5$ and the required value for the asymptotic $P_b$. Other parameters are set to: $P=10$, $\lambda_2=0.005$, $m=128$, $n=256$ and $\sigma=0.5$. }\label{f2}
\end{figure}
\begin{figure} 
	\centering
	\includegraphics[width=2.8in]{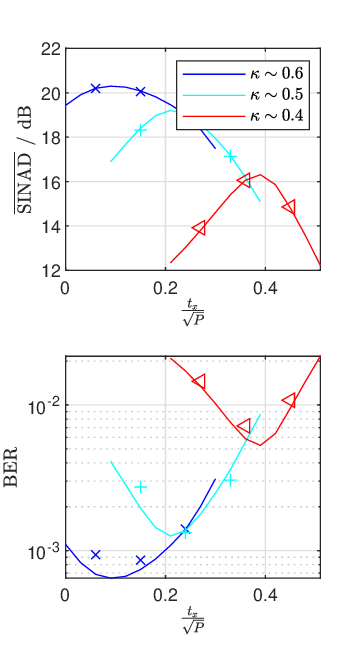}
	\caption{Performance of the thresholded $\ell_1$ norm precoder versus threshold.  For each threshold value $t_x$, $\rho$ and $\lambda_1$ are tuned to satisfy an asymptotic value for the sparsity level of the precoder $\kappa$ and a fixed value for the asymptotic $P_b\sim 2.8$. Other parameters are set to: $P=10$, $\lambda_2=0.005$, $m=256$, $n=512$ and $\sigma=0.5$. }\label{f3}
\end{figure}

As a first experiment, we consider the situation in which there only are  $\kappa=k/n=0.5$ available RF chains, and compare in Fig \ref{f2} the performance of the $\ell_1$-norm precoder with that of the thresholded $\ell_1$-norm precoder for different values of regularization parameters $\lambda_1$. The curve referred to as '$\lambda_1$ only' represents the performance of the $\ell_1$-norm precoder in which the regularization parameter $\lambda_1$ is set to achieve $\kappa=0.5$. The other curves represent the performance of the thresholded $\ell_1$-norm precoder for different choices of the threshold value. For each point, the value of $\lambda_1$ and $\rho$ are carefully tuned in order to ensure the same value of $P_b$ for all precoders. As a first observation, we note that for all types of precoders, performance does not always improve with increasing transmit power. This is due to the constraint on the transmit antenna power, which prevents the precoder from finding a solution with low distortion. This result is consistent with our recent work in \cite{arxiv}, which addressed the performance of the GLSE with a limited peak-to-average power ratio (PAPR). We observed that similar to the $\ell_1$-norm precoder, the thresholded $\ell_1$-norm precoder exhibits the same behavior. This underscores the importance of carefully selecting the transmit power to optimize performance. To investigate this further, we present in Figure \ref{f3} the performance of the thresholded $\ell_1$-norm precoder in terms of the threshold parameter $t_x$ for different values of $\kappa$. The parameters $\lambda_1$ and $\rho$ are tuned to asymptotically satisfy the required value of $\kappa$ and an asymptotic value of $P_b$ equal to 2.8. As shown, performance varies significantly with the threshold value. For a limited fraction of RF chains  ($\kappa=0.4$), the optimal threshold value is relatively high, indicating the need for thresholding rather than $\ell_1$ regularization.  In contrast, for higher fractions of RF chains ($\kappa=0.6$) a smaller threshold value is sufficient to optimize performance, meaning the precoder can rely more on $\ell_1$ regularization with minimal thresholding.
It is important to note that in all cases, the thresholding operation consistently offers better performance than without thresholding.

Finally, we investigate in Figure.\ref{f4} the behavior of the thresholded $\ell_1$ norm precoder with respect to the transmit SNR defined as \begin{align*}{\rm tSNR}=\frac{P_b}{\sigma^2},\end{align*} when the parameters $\lambda_1$ and $\rho$ are tuned to satisfy the required $\kappa$ and the transmit power $P_b$, while the threshold value is set optimally to maximize the $\overline{\rm SINAD}_{lb}$. As shown, when the parameters are optimally selected, the configuration with a higher number of RF chains achieves the best performance. However, for low SNR values, the performance gap between the thresholded $\ell_1$ norm precoder with fewer RF chains and that with a higher fraction of RF chains is small. This gap increases as the transmit SNR increases.
Consistent with the results in Figure \ref{f2}, we observe that while performance initially improves with increasing transmit power, it deteriorates at very high transmit powers. However, compared to the scenario in Figure \ref{f2}, the optimal transmit power in this case is much higher, and in most cases, performance improves with increasing transmit power. This is because, in the scenario represented by Figure \ref{f4}, the parameters are optimized, allowing for better management of distortion power.

\begin{figure} 
	\centering
	\includegraphics[width=2.8in]{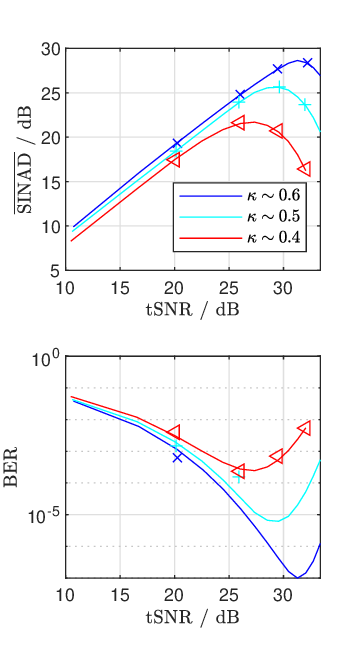}
	\caption{Performance of the thresholded $\ell_1$ norm precoder versus the transmit SNR.  For each value of $P_b$, $\rho$ and $\lambda_1$ are tuned to satisfy an asymptotic value for the sparsity level of the precoder $\kappa$ and the required value for the asymptotic $P_b$. There are combinations of the $(\rho,\lambda_1)$ dependent of the $t_x$, we choose the one that is optimal in terms of $\overline{\rm SINAD}_{lb}$. Other parameters are set to: $P=10$, $\lambda_2=0.005$, $m=256$, $n=512$ and $\sigma=0.3$.  }\label{f4}
\end{figure}
\section{Conclusion}
\label{sec_v}
This paper presented a comprehensive performance analysis of two advanced techniques for antenna selection and precoding in downlink multi-user massive MISO  systems with limited dynamic range power amplifiers and a constrained number of RF chains. By deriving these techniques from the RZF precoder and incorporating $\ell_1$-norm regularization and thresholding operations, we addressed the challenges posed by the limited dynamic range of amplifiers and the restricted availability of RF chains.

Our focus centered on the analysis of two techniques: the $\ell_1$-norm precoder and the thresholded $\ell_1$-norm precoder. The $\ell_1$-norm precoder utilized $\ell_1$-norm regularization to encourage sparse solutions, facilitating effective antenna selection. The second technique, the thresholded $\ell_1$-norm precoder, applied an entry-wise thresholding operation to the solutions obtained from the $\ell_1$ norm precoder, aiming to further refine the precoding performance.

By leveraging the standard cGMT and developing a novel GMT framework tailored to the complexities introduced by the thresholding operation, we provided a precise asymptotic behavior analysis for both precoders. Our analysis focused on critical performance metrics such as received SINAD and BER.

The results demonstrated that the thresholded $\ell_1$-norm precoder achieves superior performance when the threshold parameter is optimally selected. This was validated through empirical simulations, confirming that our asymptotic findings hold true for systems with hundreds of antennas at the base station, serving dozens of user terminals.

In conclusion, the proposed thresholded $\ell_1$-norm precoder offers a promising solution for enhancing the performance of downlink multi-user massive MISO systems, particularly in environments with limited dynamic range power amplifiers and constrained RF chains. The insights gained from this study not only advance the understanding of antenna selection and precoding strategies but also provide a foundation for future research in optimizing massive MISO systems under practical hardware limitations. 

\appendices
\renewcommand\thesubsection{\thesection.\Roman{subsection}}
\renewcommand\thesubsectiondis{\thesection.\Roman{subsection}}

	\section{Applying cGMT to the asymptotic performance analysis of the $\ell_1$-norm precoder}
	\label{sec5}{In this section, we briefly review the cGMT framework and reformulate the precoder in \eqref{eq:nlse_sparse} into forms that are manageable within the framework. These formulations facilitate our proofs in later sections.}
	\subsection{Brief review of the standard cGMT framework}
The convex Gaussian min-max theorem (cGMT) builds upon Gordon's min-max inequality \cite{gor}. In the realm of high-dimensional inference, the Gordon's min-max inequality has proven useful in analyzing sharp phase transitions in noiseless compressive sensing \cite{sto09,tropp,Chandrasekaran2012,Tropp2}. Stojnic highlighted in \cite{Sto2} that Gordon's min-max inequality achieves tightness with additional convexity conditions. This concept was  revisited  in \cite{mesti}, where the cGMT was applied to rigorously study general high-dimensional regression problems.

{We will use cGMT to analyze the $\ell_1$-norm precoder in \eqref{eq:nlse_sparse}, which does not admit a closed-form solution due to the presence of a power constraint  on the precoder entries and the non-smoothness introduced by the $\ell_1$-norm penalty. cGMT has been successfully applied to similar formulations, as shown in \cite{ayed1,ayed2}. }

	To illustrate the main concepts, consider the following two Gaussian processes:, 
	\begin{align*}
		X({\bf x},{\bf u},{\bf b},{\bf t})&={\bf u}^{T}{\bf G}{\bf x}+\psi({\bf x},{\bf u},{\bf b},{\bf t}),\\
		Y({\bf x},{\bf u}, {\bf b},{\bf t})&=\|{\bf x}\|{\bf g}^{T}{\bf u}-\|{\bf u}\|{\bf h}^{T}{\bf x}+\psi({\bf x},{\bf u},{\bf b},{\bf t}),
	\end{align*}
	where ${\bf G}\in \mathbb{R}^{m\times n}$, ${\bf g}\in \mathbb{R}^{m}$, and ${\bf h}\in \mathbb{R}^n$ have all independent entries with standard Gaussian distribution, and function ${\psi}:\mathbb{R}^{n}\times \mathbb{R}^m \times \mathbb{R}^{m_1}\times \mathbb{R}^{n_1}\to\mathbb{R}$ is continuous. Associated with these processes, we define the following min-max optimization problems, termed the primary optimization (PO) problem and the auxiliary optimization (AO) problem, respectively:
	\begin{align*}
		\Phi({\bf G})&:=\min_{\substack{{\bf x}\in \mathcal{S}_{{\bf x}}\\{\bf b}\in \mathcal{S}_{{\bf b}}}}\max_{\substack{{\bf u}\in \mathcal{S}_{{\bf u}}\\{\bf t}\in \mathcal{S}_{{\bf t}}}}X({\bf x},{\bf u},{\bf b},{\bf t}), \\
		\phi({\bf g},{\bf h})&:= \min_{\substack{{\bf x}\in \mathcal{S}_{{\bf x}}\\{\bf b}\in \mathcal{S}_{{\bf b}}}}\max_{\substack{{\bf u}\in \mathcal{S}_{{\bf u}}\\{\bf t}\in \mathcal{S}_{{\bf t}}}}Y({\bf x},{\bf u},{\bf b},{\bf t}),
	\end{align*}
	where $\mathcal{S}_{{\bf x}}\subset \mathbb{R}^n$,  $\mathcal{S}_{{\bf u}}\subset \mathbb{R}^m$, $\mathcal{S}_{{\bf b}}\subset \mathbb{R}^{m_1}$ and $\mathcal{S}_{{\bf t}}\subset \mathbb{R}^{n_1}$ are compact sets. 
	Then for any $t\in \mathbb{R}$,
	\begin{align}
		\mathbb{P}\Big[\Phi({\bf G})\leq t\Big]\leq 2\mathbb{P}\Big[\phi({\bf g},{\bf h})\leq t\Big]. \label{eq:prim_2}
	\end{align}
	
	Define also the following random processes obtained by switching the order of  max-min:
	\begin{align*}
			\stackrel{\circ}{\Phi}({\bf G})&:=\max_{\substack{{\bf u}\in \mathcal{S}_{{\bf u}}\\{\bf t}\in \mathcal{S}_{{\bf t}}}} \min_{\substack{{\bf x}\in \mathcal{S}_{{\bf x}}\\{\bf b}\in \mathcal{S}_{{\bf b}}}}X({\bf x},{\bf u},{\bf b},{\bf t}), \\
		\stackrel{\circ}{\phi}({\bf g},{\bf h})&:= \max_{\substack{{\bf u}\in \mathcal{S}_{{\bf u}}\\{\bf t}\in \mathcal{S}_{{\bf t}}}} \min_{\substack{{\bf x}\in \mathcal{S}_{{\bf x}}\\{\bf b}\in \mathcal{S}_{{\bf b}}}}Y({\bf x},{\bf u},{\bf b},{\bf t}).
		\end{align*}
		By noting that $-\max_{x}\min_y f(x,y)=\min_x\max_y -f(x,y)$, we can also obtain
	\begin{align}
		\mathbb{P}\Big[\stackrel{o}{\Phi}({\bf G})\geq t\Big]\leq 2\mathbb{P}\Big[\stackrel{o}{\phi}({\bf g},{\bf h})\geq t\Big]\leq 2\mathbb{P}\Big[{\phi}({\bf g},{\bf h})\geq t\Big], \label{eq:prim_1}
	\end{align}
	where in the last inequality, we used the fact that ${\phi}({\bf g},{\bf h})\geq \stackrel{o}{\phi}({\bf g},{\bf h})$. 
	If $\psi$ is jointly convex in $({\bf x},{\bf b})$ and jointly concave in ${\bf u}$ and ${\bf t}$, and all sets $\mathcal{S}_{\bf x}, \mathcal{S}_{\bf u},$ $\mathcal{S}_{{\bf b}}$ and $\mathcal{S}_{{\bf t}}$ are convex and compact, then it follows from Sion's min-max theorem \cite{Sion58} that:
	$$
	\Phi({\bf G})=\stackrel{o}{\Phi}({\bf G}).
	$$
	Assume there exists $\psi^\star$ such that $\phi({\bf g},{\bf h})-\psi^\star$ converges to zero in probability in some chosen asymptotic regime. This can be stated equivalently as: 
	$$
	\mathbb{P}[\phi({\bf g},{\bf h})\geq \psi^\star+\epsilon]\to 0
	$$ 
	and 
	$$
	\mathbb{P}[\phi({\bf g},{\bf h})\leq \psi^\star-\epsilon]\to 0.
	$$ 
	Then, based on \eqref{eq:prim_2} and \eqref{eq:prim_1} we also have
	$$
	\mathbb{P}[\Phi({\bf G})\geq \psi^\star+\epsilon]\to 0
	$$
	and
	$$
	\mathbb{P}[\Phi({\bf G})\leq \psi^\star-\epsilon]\to 0,
	$$
which is  equivalent to saying that $\Phi({\bf G})$ converges to $\psi^\star$ in probability. 
	\subsection{Application of the cGMT framework to the $\ell_1$-norm precoder}\label{a2}

 \noindent{\bf Formulation of the POs.}
	For $\mathbb{X}=[-\sqrt{P},\sqrt{P}]$, recall the solution of the regularized least squares problem \eqref{eq:nlse_sparse}:
	\begin{equation}
		\hat{\bf x}\in\arg\min_{\|{\bf x}\|_{\infty}\leq \sqrt{P}} \frac{1}{n}\|{\bf Hx}-\sqrt{\rho}{\bf s}\|_2^2+\frac{\lambda_2}{n} \|{\bf x}\|_2^2+\frac{\lambda_1 \|\mathbf{x}\|_1}{n}. \label{eq:op}
	\end{equation}
	Using the following identities:
	\begin{equation*}
		\|{\bf z}\|^2= \max_{{\bf u}\in \mathbb{R}^{m}} {\bf u}^{T}{\bf z}-\frac{\|{\bf u}\|^2}{4},
	\end{equation*}
and
	\begin{equation*}
		\|{\bf x}\|_1=\max_{\|{\bf t}\|_{\infty}\leq 1} {\bf t}^{T}{\bf x},
	\end{equation*}
we can write the optimization problem in \eqref{eq:op} as
	\begin{equation}
		\underset{\|{\bf x}\|_{\infty}\leq \sqrt{P}}{\min} \ \ \underset{\mathbf{u}}{\max} \max_{\|{\bf t}\|_{\infty}\leq 1}\frac{\sqrt{n}\mathbf{u}^T\mathbf{Hx}}{n}-\frac{\sqrt{\rho}\mathbf{u}^T\mathbf{s}}{\sqrt{n}}-\frac{\|\mathbf{u}\|_2^2}{4}+\frac{\lambda_2\|\mathbf{x}\|_2^2}{n}+\frac{\lambda_1 {\bf t}^{T}{\bf x}}{n}.
		\label{afterC0}
	\end{equation}
	The above problem is in the form of the PO, except that the constraint set over ${\bf u}$ is not bounded. From first-order optimality conditions, we can easily check that the optimal ${\bf u}$ is given by
	\begin{equation*}
		{\bf u}^\star=2\left(\frac{1}{\sqrt{n}}{\bf Hx}-\frac{\sqrt{\rho}{\bf s}}{\sqrt{n}}\right),
	\end{equation*}
	hence, 
	\begin{equation*}
		\|{\bf u}^\star\|\leq 2\sqrt{P}\left\|{\bf H}\right\| +\frac{2\sqrt{\rho}\sqrt{m}}{\sqrt{n}}.
	\end{equation*}
	Also using Lemma \ref{lem:spectral_norm}, $\|{\bf H}\|\leq 3\sqrt{\delta}$ with a probability greater than $1-C\exp(-cn)$,
	$
	\|{\bf u}^\star\|
	$ is bounded by $(6\sqrt{P}+2\sqrt{\rho})\sqrt{\delta}$ with probability $1-C\exp(-cn)$. Thus, with probability $1-C\exp(-cn)$, the set of solutions and optimal cost of \eqref{afterC0} coincide also with those of the following problem:
	\begin{equation}
		\underset{\|{\bf x}\|_{\infty}\leq \sqrt{P}}{\min} \ \ \underset{\mathbf{u}\in\mathcal{S}_{\mathbf{u}}}{\max} \max_{\|{\bf t}\|_{\infty}\leq 1}\frac{\sqrt{n}\mathbf{u}^T\mathbf{Hx}}{n}-\frac{\sqrt{\rho}\mathbf{u}^T\mathbf{s}}{\sqrt{n}}-\frac{\|\mathbf{u}\|_2^2}{4}+\frac{\lambda_2\|\mathbf{x}\|_2^2}{n}+\frac{\lambda_1 {\bf t}^{T}{\bf x}}{n},
		\label{afterC1}
	\end{equation}
	where $\mathcal{S}_{\bf u}:=\left\{{\bf u}\in\mathbb{R}^m, \ \ \|{\bf u}\|\leq (6\sqrt{P}+2\sqrt{\rho})\sqrt{\delta}:=C_\beta\right\}$. 
	Our interest is to characterize the asymptotic behavior of the solutions ${\bf x}$, ${\bf u}$ and ${\bf t}$ to \eqref{afterC0}. For that, 
	we introduce the following cost functions:
	\begin{align*}
		\mathcal{C}_{\lambda,\rho}({\bf x})&= \max_{{\bf u}\in\mathcal{S}_{\bf u}} \frac{\sqrt{n}{\bf u}^{T}{\bf Hx}}{n}  
		-\frac{\sqrt{\rho}{\bf u}^{T}{\bf s}}{\sqrt{n}}- \frac{\|\mathbf{u}\|_2^2}{4}+\frac{\lambda_2\|\mathbf{x}\|_2^2}{n}+\frac{\lambda_1 \|{\bf x}\|_1}{n},\\
		\mathcal{V}_{\lambda,\rho}({\bf u})&= \min_{\|{\bf x}\|_\infty\leq \sqrt{P}} \frac{\sqrt{n}{\bf u}^{T}{\bf Hx}}{n}  
		-\frac{\sqrt{\rho}{\bf u}^{T}{\bf s}}{\sqrt{n}}- \frac{\|\mathbf{u}\|_2^2}{4}+\frac{\lambda_2\|\mathbf{x}\|_2^2}{n}+\frac{\lambda_1 \|{\bf x}\|_1}{n},\\
		\mathcal{T}_{\lambda,\rho}({\bf t},{\bf u})&=  \min_{\|{\bf x}\|_\infty\leq \sqrt{P}} \frac{\sqrt{n}{\bf u}^{T}{\bf Hx}}{n}  
		-\frac{\sqrt{\rho}{\bf u}^{T}{\bf s}}{\sqrt{n}}- \frac{\|\mathbf{u}\|_2^2}{4}+\frac{\lambda_2\|\mathbf{x}\|_2^2}{n}+\frac{\lambda_1}{n}{\bf t}^{T}{\bf x},\end{align*}
	and consider the following primary problems:
	\begin{align}
		\Phi_{\lambda,\rho}({\bf H})&:=\min_{\|{\bf x}\|_{\infty}\leq \sqrt{P}} \mathcal{C}_{\lambda,\rho}({\bf x}) \label{eq:PO1},\\
		\tilde{\Phi}_{\lambda,\rho}({\bf H})&:= \max_{{\bf u}\in\mathcal{S}_{\bf u}} \mathcal{V}_{\lambda,\rho}({\bf u}), \label{eq:PO2}\\
		\hat{\Phi}_{\lambda,\rho}({\bf H})&:=\max_{\|{\bf t}\|_\infty\leq 1}\max_{{\bf u}\in\mathcal{S}_{\bf u}} \mathcal{T}_{\lambda,\rho}({\bf t},{\bf u}).\label{eq:PO3}
	\end{align}
	Since the objective function in \eqref{afterC1} is convex in ${\bf x}$, concave in ${\bf u}$ and ${\bf t}$, then
	\begin{equation}
		\Phi_{\lambda,\rho}({\bf H})=\tilde{\Phi}_{\lambda,\rho}({\bf H})=\hat{\Phi}_{\lambda,\rho}({\bf H}). \label{eq:equality}
	\end{equation}

	We denote by  $\hat{\bf x}^{\rm PO}$,  $\hat{\bf u}^{\rm PO}$, and $\hat{\bf t}^{\rm PO}$, solutions to   \eqref{eq:PO1}-\eqref{eq:PO3}
	\begin{align*}
		\hat{\bf x}^{\rm PO}&\in\arg\min_{x_i^2\leq P} \mathcal{C}_{\lambda,\rho}({\bf x}), \\
		\hat{\bf u}^{\rm PO}&\in\arg\max_{{\bf u}\in\mathcal{S}_{\bf u}} \mathcal{V}_{\lambda,\rho}({\bf u}),\\
		\hat{\bf t}^{\rm PO}&\in\arg\max_{\|{\bf t}\|_{\infty}\leq 1} (\max_{{\bf u}\in \mathcal{S}_{{\bf u}}}\mathcal{T}_{\lambda,\rho}({\bf t},{\bf u})).  
	\end{align*}

	\noindent{\bf Formulation of the AOs.}
	With the PO problems in \eqref{eq:PO1}-\eqref{eq:PO3}, we associate the following AO problems:
	\begin{align}
		\phi_{\lambda,\rho}({\bf g},{\bf h})&:=\min_{\|{\bf x}\|_\infty\leq \sqrt{P}} \mathcal{L}_{\lambda,\rho}({\bf x}) \label{eq:AO1},\\
		\tilde{\phi}_{\lambda,\rho} ({\bf g},{\bf h})&=\max_{{\bf u}\in\mathcal{S}_{\bf u}} \mathcal{F}_{\lambda,\rho}({\bf u}), \label{eq:AO2}\\
		\hat{\phi}_{\lambda,\rho}({\bf g},{\bf h})&=\max_{\|{\bf t}\|_{\infty}\leq 1}\max_{{\bf u}\in\mathcal{S}_{\bf u}} \mathcal{S}_{\lambda,\rho}({\bf t},{\bf u}) ,\label{eq:AO3}\end{align}
	where $\mathcal{L}_{\lambda,\rho}({\bf x})$ and $ \mathcal{F}_{\lambda,\rho}({\bf u})$, $\mathcal{S}_{\lambda,\rho}({\bf t})$ are given by
	\begin{align}
		\mathcal{L}_{\lambda,\rho}({\bf x})&:=\underset{\mathbf{u}\in\mathcal{S}_{\bf u}}{\max}\frac{1}{n}\|\mathbf{x}\|_2\mathbf{g}^T\mathbf{u}-\frac{1}{n}\|\mathbf{u}\|_2\mathbf{h}^T\mathbf{x}
		-\frac{\sqrt{\rho}\mathbf{u}^T\mathbf{s}}{\sqrt{n}}-\frac{\|\mathbf{u}\|_2^2}{4}+\frac{\lambda_2\|\mathbf{x}\|_2^2}{n}+\frac{\lambda_1 \|{\bf x}\|_1}{n} \label{eq:L},\\
		\mathcal{F}_{\lambda,\rho}({\bf u})&:=\min_{\|{\bf x}\|_{\infty}\leq \sqrt{P}} \frac{1}{n}\|\mathbf{x}\|_2\mathbf{g}^T\mathbf{u}-\frac{1}{n}\|\mathbf{u}\|_2\mathbf{h}^T\mathbf{x} 
		-\frac{\sqrt{\rho}\mathbf{u}^T\mathbf{s}}{\sqrt{n}}-\frac{\|\mathbf{u}\|_2^2}{4}+\frac{\lambda_2\|\mathbf{x}\|_2^2}{n}+\frac{\lambda_1 \|{\bf x}\|_1}{n}, \nonumber\\
		\mathcal{S}_{\lambda,\rho}({\bf t},{\bf u})&:=\min_{\|{\bf x}\|_{\infty}\leq \sqrt{P}}\frac{1}{n}\|{\bf x}\|_2{\bf g}^T\mathbf{u}-\frac{1}{n}\|\mathbf{u}\|_2\mathbf{h}^T\mathbf{x} 
		-\frac{\sqrt{\rho}\mathbf{u}^T\mathbf{s}}{\sqrt{n}}-\frac{\|\mathbf{u}\|_2^2}{4}+\frac{\lambda_2\|\mathbf{x}\|_2^2}{n}+\frac{\lambda_1 {\bf t}^{T}{\bf x}}{n}.\nonumber
	\end{align}
	Similarly, we denote by $\hat{\bf x}^{\rm AO}$, $\hat{\bf u}^{\rm AO}$ and  $\hat{\bf t}^{\rm AO}$ solutions to \eqref{eq:AO1}-\eqref{eq:AO3}: 
	\begin{align*}
		\hat{\bf x}^{\rm AO}&\in \arg\min_{\|{\bf x}\|_{\infty}\leq \sqrt{P}} \mathcal{L}_{\lambda,\rho}({\bf x}), \\
		\hat{\bf u}^{\rm AO}&\in \arg\max_{{\bf u}\in\mathcal{S}_{\mathbf{u}}} \mathcal{F}_{\lambda,\rho} ({\bf u}),\\
		\hat{\bf t}^{\rm AO}&\in\arg\max_{\|{\bf t}\|_{\infty}\leq 1}(\max_{{\bf u}\in\mathcal{S}_{\mathbf{u}}}	\mathcal{S}_{\lambda,\rho}({\bf t},{\bf u})).
	\end{align*}
	By applying the probability inequalities in \eqref{eq:prim_2} and \eqref{eq:prim_1} to the POs in \eqref{eq:PO1}-\eqref{eq:PO3}, we obtain the following probability inequalities:
		\begin{itemize}
		\item For all $t\in\mathbb{R}$ and any compact sets $\tilde{\mathcal{S}}_{{\bf x}}\subset \mathbb{R}^n$ , $\tilde{\mathcal{S}}_{\bf u}\subset \mathbb{R}^m$, and $\tilde{\mathcal{S}}_{{\bf t}}\subset \mathbb{R}^{n}$, 
		\begin{align}
			&\mathbb{P}\left[\min_{{\bf x}\in\tilde{\mathcal{S}}_{{\bf x}}} \mathcal{C}_{\lambda,\rho}({\bf x})\leq t\right]\leq  2\mathbb{P}\left[\min_{{\bf x}\in\tilde{\mathcal{S}}_{{\bf x}}} \mathcal{L}_{\lambda,\rho}({\bf x})\leq t\right]+C\exp(-cn), \label{eq:1}\\
			& \mathbb{P}\left[\max_{{\bf u}\in\tilde{\mathcal{S}}_{{\bf u}}} \mathcal{V}_{\lambda,\rho}({\bf u})\geq t\right]\leq  2\mathbb{P}\left[\max_{{\bf u}\in\tilde{\mathcal{S}}_{{\bf u}}} \mathcal{F}_{\lambda,\rho}({\bf u})\geq t\right], \label{eq:2}\\
			& \mathbb{P}\left[\max_{{\bf t}\in\tilde{\mathcal{S}}_{{\bf t}}} \max_{{\bf u}\in\tilde{\mathcal{S}}_{{\bf u}}}\mathcal{T}_{\lambda,\rho}({\bf t},{\bf u})\geq t\right]\leq  2\mathbb{P}\left[\max_{{\bf t}\in \tilde{\mathcal{S}}_{{\bf t}}}\max_{{\bf u}\in\tilde{\mathcal{S}}_{{\bf u}}} \mathcal{S}_{\lambda,\rho}({\bf t},{\bf u})\geq t\right]+C\exp(-cn).\label{eq:3t}
		\end{align}
		where in \eqref{eq:1} and \eqref{eq:3t} the additional cost $C\exp(-cn)$ comes from the fact that the optimum in ${\bf u}$ belong to $\tilde{\mathcal{S}}_{{\bf u}}$ with probability $1-C\exp(-cn)$. 
		\item If  the sets $\tilde{\mathcal{S}}_{{\bf x}}$, $\tilde{\mathcal{S}}_{\bf u}$, $\tilde{\mathcal{S}}_{\bf t}$ are additionally convex, we have for all $t\in\mathbb{R}$, 
		\begin{align}
			&\mathbb{P}\left[\min_{{\bf x}\in\tilde{\mathcal{S}}_{{\bf x}}} \mathcal{C}_{\lambda,\rho}({\bf x})\geq t\right]\leq  2\mathbb{P}\left[\min_{{\bf x}\in\tilde{\mathcal{S}}_{{\bf x}}} \mathcal{L}_{\lambda,\rho}({\bf x})\geq t\right]+C\exp(-cn), \label{eq:31}\\
			& \mathbb{P}\left[\max_{{\bf u}\in\tilde{\mathcal{S}}_{{\bf u}}} \mathcal{V}_{\lambda,\rho}({\bf u})\leq t\right]\leq  2\mathbb{P}\left[\max_{{\bf u}\in\tilde{\mathcal{S}}_{{\bf u}}} \mathcal{F}_{\lambda,\rho}({\bf u})\leq t\right], \label{eq:4}\\
			&\mathbb{P}\left[\max_{{\bf t}\in \tilde{\mathcal{S}}_{{\bf t}}}\max_{{\bf u}\in\tilde{\mathcal{S}}_{{\bf u}}} \mathcal{T}_{\lambda,\rho}({\bf t},{\bf u})\leq t\right]\leq  2\mathbb{P}\left[\max_{{\bf t}\in \tilde{\mathcal{S}}_{{\bf t}}}\max_{{\bf u}\in\tilde{\mathcal{S}}_{{\bf u}}} \mathcal{S}_{\lambda,\rho}({\bf t},{\bf u})\leq t\right]+C\exp(-cn) .\label{eq:5}
		\end{align}
	\end{itemize}
	The inequalities presented above are fundamental to our proofs and will be utilized extensively. They not only help demonstrate the concentration of the optimal cost for the PO but also facilitate the transfer of concentration results from the optimal solution of the AO to that of the PO, as will be detailed in our proofs.

	\section{Preliminaries: Study of the asymptotic scalar optimization problem}
	\label{fix_point}
	The objective of this section is to study the scalar optimization problem: 
	\begin{equation}
		\max_{\beta\geq0 } \min_{\tau\geq 0}  \frac{\beta\tau\delta}{2}-\frac{\beta^2}{4}+\frac{\beta\rho}{2\tau} + \mathbb{E}\left[\min_{|x|<\sqrt{P}} \frac{\beta}{2\tau }x^2+\lambda_2 x^2-\beta Hx+\lambda_1|x|\right] \label{eq:scalar}
	\end{equation}
	where ${H}\sim\mathcal{N}(0,1)$ and the expectation is taken over its distribution. 
	
	For $t\geq 0$, we define the following Moreau envelope and proximal operator functions:
	\begin{align}
		e(y;t)&:= \min_{|x|\leq \sqrt{P}}\frac{1}{2}(x-y)^2+t|x|,\nonumber\\
		{\rm prox}(y;t)&:=\arg\min_{|x|<\sqrt{P}}\frac{1}{2}(x-y)^2+t|x|\nonumber\\
		&= \left\{
		\begin{array}{ll}
			\sqrt{P}& \text{if} \ \ y\geq t+\sqrt{P}\\
			y-t & \text{if} \ \ t \leq y\leq t+\sqrt{P}\\
			0 & \text{if} \ \ -t\leq y\leq t\\
			y+t & \text{if} \ \ -t-\sqrt{P}\leq y \leq -t\\
			-\sqrt{P} & \text{if} \ \ y\leq -t-\sqrt{P}
		\end{array}.
		\right. \label{eq:prox_explicit}
	\end{align}
	With these notations, the optimization problem in \eqref{eq:scalar} writes as:
	$$
	\max_{\beta\geq0 } \min_{\tau\geq 0}  \psi(\tau,\beta):=\frac{\beta\tau\delta}{2}-\frac{\beta^2}{4}+\frac{\beta\rho}{2\tau} -\frac{1}{2}\frac{\beta^2}{\frac{\beta}{\tau}+2\lambda_2} + \left(\frac{\beta}{\tau}+2\lambda_2\right) \mathbb{E}\left[e\left(\frac{H}{\frac{1}{\tau}+\frac{2\lambda_2}{\beta}};\frac{\lambda_1}{\frac{\beta}{\tau}+2\lambda_2}\right)\right].
	$$

 Let $({\tau}^\star,\beta^\star)$ be the saddle point of the above optimization problem. Then, starting from \eqref{eq:scalar}, we can easily see that $(\tau^\star,\beta^\star)$ is solution to the following system of equations:
	$$
	\left\{
	\begin{array}{ll}
		\tau^2\delta  & = \rho+\mathbb{E}\left[\left({\rm prox}(\tilde{\tau}H;\frac{\lambda_1\tilde{\tau}}{\beta})\right)^2\right]\\
		\beta&=2\tau\delta -2\mathbb{E}\left[H{\rm prox}(\tilde{\tau}H;\frac{\lambda_1\tilde{\tau}}{\beta})\right]
	\end{array}
	\right.
	$$
	where $\tilde{\tau} = \frac{1}{\frac{1}{\tau}+\frac{2\lambda_2}{\beta}}$. 
	\begin{lemma}
		For any $\beta> 0$, the following equation 
		\begin{equation}
			\tau^2\delta   = \rho+\mathbb{E}\left[\left({\rm prox}(\tilde{\tau}H;\frac{\lambda_1\tilde{\tau}}{\beta})\right)^2\right] \label{eq:sol_t}
		\end{equation}
		admits a unique solution $\tau^\star(\beta)$. Moreover $\tau^\star(\beta)$ is the unique minimizer of $\min_{\tau\geq 0} \psi(\tau,\beta)$ and the map $\beta\to \tau^\star(\beta)$ is continuously differentiable on $\mathbb{R}_{+}$.  
		\label{lem:tau_bound}
	\end{lemma}
	\begin{proof}
		Let $\beta$ be a strictly positive scalar. We will first prove that the following equation:
		\begin{equation}
			\tau^2\delta =\rho+ \mathbb{E}\left[\left({\rm prox}(\tilde{\tau}H;\frac{\lambda_1\tilde{\tau}}{\beta})\right)^2\right] \label{eq:sol}
		\end{equation}
		admits a unique solution $\tau^\star(\beta)$. To begin with, we denote by $j(\tau)$ the following function:
		$$
		j(\tau):=\tau^2\delta -\rho- 2\mathbb{E}\left[\tilde{\tau}^2(H-\frac{\lambda_1}{\beta})^2{1}_{\{\frac{\lambda_1}{\beta}\leq H\leq \frac{\lambda_1}{\beta}+\frac{\sqrt{P}}{\tilde{\tau}} \}}\right]-2P\mathbb{P}\left[H\geq \frac{\lambda_1}{\beta}+\frac{\sqrt{P}}{\tilde{\tau}}\right].
		$$
		The problem thus reduces to showing that function $j$ admits a unique zero. We note that function $j$ is a continuous map with $\lim_{\tau\to 0} j(\tau)=-\rho$, and $\lim_{\tau\to\infty}j(\tau)=\infty$ since $|{\rm prox}(\tilde{\tau}H;\frac{\lambda_1 \tilde{\tau}}{\beta})|\leq\sqrt{P}$. It thus admits at least a zero on $[0,\infty)$. We will prove by contradiction that such a zero must be unique.  Assume that there exists $0<\tau_1<\tau_2<\infty$ such that $j(\tau_1)=j(\tau_2)=0$. Computing the derivative of $j$ with respect to the variable $\tau$, we obtain:
		\begin{align*}
			j'(\tau)&=2\tau\delta-4\tilde{\tau}^3\tau^{-2} \mathbb{E}\left[(H-\frac{\lambda_1}{\beta})^2 {1}_{\{\frac{\lambda_1}{\beta}\leq H\leq \frac{\lambda_1}{\beta}+\frac{\sqrt{P}}{\tilde{\tau}}\}}\right]\\
			&=\frac{2}{\tau}\left(\tau^2\delta-2\frac{\tilde{\tau}}{\tau} \mathbb{E}\left[\tilde{\tau}^2(H-\frac{\lambda_1}{\beta})^2{1}_{\{\frac{\lambda_1}{\beta}\leq H\leq \frac{\lambda_1}{\beta}+\frac{\sqrt{P}}{\tilde{\tau}}\}}\right]\right).
		\end{align*}
		Since $\frac{\tilde{\tau}}{\tau}\leq 1$, we obtain:
		\begin{equation*}j'(\tau)\geq \frac{2}{\tau}\left(\tau^2\delta-2 \mathbb{E}\left[\tilde{\tau}^2(H-\frac{\lambda_1}{\beta})^2{1}_{\{\frac{\lambda_1}{\beta}\leq H\leq \frac{\lambda_1}{\beta}+\frac{\sqrt{P}}{\tilde{\tau}}\}}\right]\right)>0.
		\end{equation*}
		Hence, at the values of $\tau_1$ and $\tau_2$ where $j(\tau_1)=j(\tau_2)=0$, we have: $j'(\tau_1)> 0$ and $j'(\tau_2)>0$. 
		If for all $\tau\in[\tau_1,\tau_2]$, $j'(\tau)> 0$, then function $j$ would be strictly increasing, which leads to a contradiction. Hence, there exists $\tau'$ such that $j'(\tau')=0$, or equivalently $\tau'$ is solution to the following equation:
		$$
		2\delta- 4(\tilde{\tau}')^{-3}\tau^3 \mathbb{E}\left[(H-\frac{\lambda_1}{\beta})^2 {1}_{\{\frac{\lambda_1}{\beta}\leq H\leq \frac{\lambda_1}{\beta}+\frac{\sqrt{P}}{\tilde{\tau}'}\}}\right]=0.
		$$
		Function $\tau\mapsto 2\delta- 4(\tilde{\tau}')^{-3}\tau^3 \mathbb{E}\left[(H-\frac{\lambda_1}{\beta})^2 {1}_{\{\frac{\lambda_1}{\beta}\leq H\leq \frac{\lambda_1}{\beta}+\frac{\sqrt{P}}{\tilde{\tau}'}\}}\right]$ is an increasing function taking strictly positive values at $\tau_1$ and $\tau_2$. Hence, it cannot be zero at any value between $\tau_1$ and $\tau_2$. This leads to a contradiction with the result obtained before, and hence, there exists a unique solution for the equation $j(\tau)=0$. In other words, for any $\beta>0$, the solution to the equation in \eqref{eq:sol} admits a unique solution $\tau^\star(\beta)$. We can also easily check that the minimizer of function $\tau\mapsto \psi(\tau,\beta)$ is in the interior of $(0,\infty)$ and hence should verify the first order optimality condition. Since this condition is satisfied by a unique $\tau^\star(\beta)$, we deduce that $\tau^\star(\beta)$ is the unique minimizer of $\tau\mapsto \psi(\tau,\beta)$. 
		Finally, to prove that $\beta\mapsto \tau^\star(\beta)$ is continuously differentiable, it suffices to note that for $\tau^\star(\beta)$, $j'(\tau^\star(\beta))>0$. The result is thus obtained by invoking the implicit function theorem. 
	\end{proof}
	\begin{lemma} 
		The following system of equations in $(\tau,\beta)$ admits a unique strictly positive solution 
		\begin{equation}
			\left\{
			\begin{array}{ll}
				\tau^2\delta  & = \rho+\mathbb{E}\left[\left({\rm prox}(\tilde{\tau}H;\frac{\lambda_1\tilde{\tau}}{\beta})\right)^2\right]\\
				\beta&=2\tau\delta -2\mathbb{E}\left[H{\rm prox}(\tilde{\tau}H;\frac{\lambda_1\tilde{\tau}}{\beta})\right]
			\end{array}
			\right. \label{eq:sys}
		\end{equation}
		if and only if $\max(\lambda_2,\lambda_1)>0$ or $\lambda_1=\lambda_2=0$ and $\delta\geq 1$. When $\delta<1$ and $\lambda_2=\lambda_1=0$, then $\beta^\star=0$ is solution to the optimization problem in \eqref{eq:scalar}. 
		\label{lem:fixed_point}
	\end{lemma}
	To prove Lemma \ref{lem:fixed_point}, we need to show the following result:
	\begin{lemma}
		Define function $\Psi$ as:
		$$
		\Psi:\beta\mapsto \min_{\tau\geq 0}\psi(\tau,\beta).
		$$
		Then function $\Psi$ is continuously differentiable with derivative:
		$$
		\Psi'(\beta)=\tau^\star(\beta)\delta - \mathbb{E}\Big[H{\rm prox}(\tilde{\tau}^\star(\beta)H;\frac{\lambda_1\tilde{\tau}^\star}{\beta})\Big]-\frac{\beta}{2}.
		$$
		Moreover, assume that either $\max(\lambda_2,\lambda_1)>0$ or $\lambda_1=\lambda_2=0$ and $\delta\geq 1$. Then, there exists a strictly positive constant $\alpha$ that depends only on $\delta,P$ and $\rho$ such that:
		$$
		\lim_{\beta\to 0^{+}}	\Psi'(\beta)>\alpha. 
		$$
		\label{lem:deriv_Psi}
	\end{lemma}
	\begin{proof}
		Define function $\Psi$ as:
		$$
		\Psi:\beta\mapsto \min_{\tau\geq 0}\psi(\tau,\beta).
		$$
		For $\beta>0$, using the result of Lemma \ref{lem:tau_bound}, the minimizer $\tau$ of $\psi(\tau,\beta)$ is uniquely defined as the unique solution to the following equation:
		\begin{equation}
	\tau^2\delta=\rho+\mathbb{E}\big[\big({\rm prox}(\tilde{\tau}H;\frac{\lambda_1\tilde{\tau}}{\beta})\big)^2\big]. \label{eq:res}
		\end{equation}
		Moreover, the map $\beta\mapsto\tau^\star(\beta)$ is continuously differentiable on $\mathbb{R}_{+}$. 
		Taking the derivative of $\Psi$ with respect to $\beta$ and using \eqref{eq:res}, we obtain:
		$$
		\Psi'(\beta)=\tau^\star(\beta)\delta - \mathbb{E}\Big[H{\rm prox}(\tilde{\tau}^\star(\beta)H;\frac{\lambda_1\tilde{\tau}^\star}{\beta})\Big]-\frac{\beta}{2}.
		$$
		To prove that $$
		\lim_{\beta\to 0^{+}}	\Psi'(\beta)>\alpha,
		$$
		it suffices to check that 
		\begin{equation}
			\lim_{\beta\to 0} 2\tau^\star(\beta)\delta-2\mathbb{E}\left[H{\rm prox}(\tilde{\tau}^\star(\beta)H,\frac{\lambda_1\tilde{\tau}^\star(\beta)}{\beta})\right]>0, \label{eq:cond}
		\end{equation}
		where $\tilde{\tau}^\star(\beta):=\frac{1}{\frac{1}{\tau^\star(\beta)}+\frac{2\lambda_2}{\beta}}$. Towards this goal,  we need to consider separately the following three cases: $1)$ $\lambda_1> 0$ and $\lambda_2=0$, $2)$ $\lambda_2>0$ and $3) \lambda_1=\lambda_2=0$ and $\delta\geq 1$. In each of the above three cases, we will check that \eqref{eq:cond} holds true. 
		
		\noindent{ \underline{Case $\lambda_1>0$ and $\lambda_2=0$}.} In this case, $\tilde{\tau}^\star(\beta)=\tau^\star(\beta)$. Since $\tau^\star(\beta)$ is bounded above by $\sqrt{\frac{\rho+P}{\delta}}$ since $|{\rm prox}(\tilde{\tau}H;\frac{\lambda_1 \tilde{\tau}}{\beta})^2|\leq P$ and below by $\sqrt{\frac{\rho}{\delta}}$, we thus have $\lim_{\beta\to 0} \frac{\lambda_1\tau^\star(\beta)}{\beta}=\infty$. 
		It takes no much effort to check that for any scalar $y$, $\lim_{t\to\infty} {\rm prox}(y;t)=0$. Hence, 
		using the dominated convergence theorem, we thus have:
		$$
		\lim_{\beta \to 0} \mathbb{E}\left[H{\rm prox}(\tilde{\tau}^\star(\beta)H;\frac{\lambda_1 \tilde{\tau}^\star(\beta)}{\beta})\right]=0.
		$$
		and as a consequence:
		$$
		\lim_{\beta\to 0} 2\tau^\star(\beta)\delta-2\mathbb{E}\left[H{\rm prox}(\tilde{\tau}^\star(\beta)H,\frac{\lambda_1\tilde{\tau}^\star(\beta)}{\beta})\right] =2\sqrt{\delta \rho}>0.
		$$
		\noindent{ \underline{Case  $\lambda_2>0$}.} Since $|\tilde{\tau}^\star(\beta)|\leq \frac{\beta}{2\lambda_2}$, $\lim_{\beta\to 0} \tilde{\tau}^\star(\beta)=0$. Hence, it follows from the dominated convergence theorem that:
		$$
		\lim_{\beta\to 0} 2\tau^\star(\beta)\delta-2\mathbb{E}\left[H{\rm prox}(\tilde{\tau}^\star(\beta)H,\frac{\lambda_1\tilde{\tau}^\star(\beta)}{\beta})\right] =\lim_{\beta\to 0} 2\delta{\tau}^\star(\beta) >2\sqrt{\delta\rho},
		$$
		where to find the first equality we used the fact that $|{\rm prox}(y,t)|\leq |y|$. 
		
		\noindent{ \underline{Case  $\lambda_1=\lambda_2=0$ and $\delta \geq 1$.}} In this case, we can easily check that $\tilde{\tau}^\star=\tau^\star$ and that
		\begin{equation}
			{\rm prox}(\tilde{\tau}^\star H,0)=\tau H 1_{\{|\tau^\star H|\leq \sqrt{P}\}}+ \sqrt{P}{\rm sign}(H)1_{\{|\tau^\star H|\geq \sqrt{P}\}}. \label{eq:prox}
		\end{equation}
		Hence, 
		\begin{align}
			&\lim_{\beta\to 0} 2\tau^\star(\beta)\delta -2\mathbb{E}\left[H{\rm prox}(\tilde{\tau}^\star(\beta)H,\frac{\lambda_1\tilde{\tau}^\star(\beta)}{\beta})\right]\nonumber\\=& \lim_{\beta\to 0} 2\tau^\star(\beta)\delta -2\mathbb{E}\left[\tau^\star(\beta) H^2 1_{\{|\tau^\star(\beta) H|\leq \sqrt{P}\}}\right]
			-2\mathbb{E}\left[\sqrt{P}|H|1_{\{|\tau^\star(\beta) H|\geq \sqrt{P}\}}\right]\label{eq:3}\\
			=&\lim_{\beta\to 0}\frac{2}{\tau^\star(\beta)}\left[\tau^\star(\beta)^2\delta -  \mathbb{E}\left[(\tau^\star(\beta))^2 H^2 1_{\{|\tau^\star(\beta) H|\leq \sqrt{P}\}}\right]-\mathbb{E}\left[\sqrt{P}\tau^\star(\beta)|H|1_{\{|\tau^\star(\beta) H|\geq \sqrt{P}\}}\right]\right]\nonumber\\
			=&\lim_{\beta\to 0}\frac{2}{\tau^\star(\beta)}\left[\rho+P\mathbb{E}\left[1_{\{|\tau^\star(\beta )H|\geq \sqrt{P}\}}\right] -\mathbb{E}\left[\sqrt{P}\tau^\star(\beta)|H|1_{\{|\tau^\star(\beta) H|\geq \sqrt{P}\}}\right]\right]\nonumber \\
			=& \lim_{\beta\to 0}\frac{2}{\tau^\star(\beta)}\mathbb{E}\left[(\rho+P-\sqrt{P}\tau^\star(\beta)|H|)1_{\{|\tau^\star(\beta)H|\geq \sqrt{P}\}}\right]+\frac{2}{\tau^\star(\beta)}\rho\mathbb{P}\left[|\tau^\star(\beta) H|\leq \sqrt{P}\right]\nonumber\\
			\geq &\lim_{\beta \to 0}\frac{2\rho}{\tau^\star(\beta)} \mathbb{P}\left[|\tau^\star(\beta) H|\leq \sqrt{P}\right]\label{eq:f}\\
			\geq& \frac{2\sqrt{\delta}}{\sqrt{\rho+P}}\mathbb{P}\left[|H|\leq \frac{\sqrt{P\delta}}{\sqrt{\rho+P}}\right]>0, \label{eq:r}
		\end{align}
		where \eqref{eq:3} is obtained by replacing the proximal operator by its expression in \eqref{eq:prox}. In equation \eqref{eq:f}, we used the fact that $\rho +P \geq \sqrt{P}\tau^\star(\beta)$ since $\tau^\star(\beta)\leq \frac{\sqrt{\rho+P}}{\sqrt{\delta}}$ and $\delta\geq 1$ while in equation \eqref{eq:r}, we used the facts that $ \tau^\star(\beta)\leq \frac{\sqrt{\rho+P}}{\sqrt{\delta}}$. This completes the proof of Lemma \ref{lem:deriv_Psi}.
	\end{proof}
	\begin{proof}[Proof of Lemma \ref{lem:fixed_point}]
		Function $\Psi$ is $\frac{1}{2}-$ strongly concave. Moreover, $\lim_{\beta\to\infty}\Psi(\beta) =-\infty$. Hence, $\Psi$ admits a unique finite maximizer in $[0,\infty)$ which we denote by $\beta^\star$. If $\beta^\star> 0$, using Lemma \ref{lem:tau_bound}, there exists a unique $\tau^\star$ that is solution to  \eqref{eq:sol_t}. 
		Taking the following first derivatives of function $\psi$ with respect to $\tau$ and $\beta$, we obtain:
		\begin{equation*}
			\begin{split}
				\frac{\partial \psi}{\partial \tau}&= \frac{\beta \delta}{2} -\frac{\beta\rho}{2\tau^2} -\frac{\beta }{2\tau^2} \mathbb{E}\left[{\rm prox}(\tilde{\tau}H;\frac{\lambda_1 \tilde{\tau}}{\beta})\right] ,\\
				\frac{\partial \psi}{\partial \beta}&= \frac{\tau\delta}{2}-\frac{\beta}{2} +\frac{\rho}{2\tau} + \frac{1}{2\tau} \mathbb{E}[{\rm prox}(\tilde{\tau}H;\frac{\lambda_1 \tilde{\tau}}{\beta})^2 ]-  \mathbb{E}\left[H{\rm prox}(\tilde{\tau}H;\frac{\lambda_1 \tilde{\tau}}{\beta})\right],
			\end{split}
		\end{equation*}
		 then $(\beta^\star,\tau^\star)$ satisfy:
		\begin{align*}
			\frac{\partial \psi}{\partial \tau}&=0,\\
			\frac{\partial \psi}{\partial \beta}&=0,  
		\end{align*}
		or equivalently $(\beta^\star,\tau^\star)$ is the unique solution to the system of equations in \eqref{eq:sys}, in which case it coincides with the saddle point of the optimization problem in \eqref{eq:scalar}. 
  To complete the proof, we need thus to check that $\beta^\star> 0$ if $\max(\lambda_2,\lambda_1)>0$ or $\lambda_2=\lambda_1=0$ and $\delta\geq 1$. For that, because of the concavity of function $\Psi$, this follows because
		$$
		\lim_{\beta\to 0} \Psi'(\beta)>0
		$$
		as shown in Lemma \ref{lem:deriv_Psi}.
		
		With this, it remains to check that when $\lambda_1=\lambda_2=0$, and $\delta<1$, the solution in $\beta$ to \eqref{eq:scalar} is $\beta^\star=0$. For that,  it suffices to check that there exists a feasible $\overline{\tau}>0$ for which
		\begin{equation}	\max_{\beta}\psi(\overline{\tau},\beta)=\psi(\overline{\tau},0)=0. \label{eq:req22}
		\end{equation}
		Indeed, if the above equation holds true, then, 
		$$
		\max_{\beta\geq 0}\min_{\tau\geq 0}\psi(\tau,\beta)\leq \max_{\beta\geq 0} \psi(\overline{\tau},\beta)=0.
		$$
		Since for any $\tau\geq 0$, $\psi(\tau,0)=0$, we conclude due to the strong concavity of $\beta\mapsto \min_{\tau\geq 0}\psi(\tau,\beta)$ that $\beta^\star=0$ is the unique solution in $\beta$ to the optimization problem in \eqref{eq:scalar}. 
		To show \eqref{eq:req22}, it suffices to prove that there exists $\overline{\tau}$ for which 
		\begin{equation}
			2\overline{\tau}\delta-2\mathbb{E}\left[H{\rm prox}(\overline{\tau}H,0)\right]<0, \label{eq:req1}
		\end{equation}
		or equivalently:
		$$
		2\delta-2\mathbb{E}\left[\frac{1}{\overline{\tau}}H{\rm prox}(\overline{\tau}H,0)\right]<0.
		$$
		Replacing the proximal operator by its expression in \eqref{eq:prox}, we get:
		\begin{align*}
			2\delta-2\mathbb{E}\left[\frac{1}{\overline{\tau}}H{\rm prox}(\overline{\tau}H,0)\right]&=2\delta-2\mathbb{E}[H^21_{\{|\overline{\tau}H|\leq \sqrt{P}\}}]-2\mathbb{E}\left[\sqrt{P}|H|1_{\{\overline{\tau}|H|\geq \sqrt{P}\}}\right],
		\end{align*}
		taking the limit of the right-hand side of the above equation when $\overline{\tau}\downarrow 0$, we obtain:
		$$
		\lim_{\overline{\tau}\to 0 }2\delta-2\mathbb{E}[H^21_{\{|\overline{\tau}H|\leq \sqrt{P}\}}]-2\mathbb{E}\left[\sqrt{P}|H|1_{\{\overline{\tau}|H|\geq \sqrt{P}\}}\right]=2(\delta-1)<0. 
		$$
		Hence, there exists $\overline{\tau}$ strictly positive such that \eqref{eq:req1} holds true. This shows that in case $\delta<1$, $\beta^\star=0$ is the unique maximizer of $\beta\mapsto \min_{\tau\geq 0}\psi(\tau,\beta)$. 
	\end{proof}
	For ${H}\sim\mathcal{N}(0,1)$, define $x(\tilde{\tau}^\star,\beta^\star)$ as:
	$$
	x(\tilde{\tau}^\star,\beta^\star):={\rm prox}(\tilde{\tau}^\star H,\frac{\lambda_1\tilde{\tau}^\star}{\beta^\star})
	$$ 
	where $\tau^\star$ and $\beta^\star$ are the solutions to \eqref{eq:scalar}.
	Then, the following lemma holds true:
	\begin{lemma}
		Assume that $\max(\lambda_2,\lambda_1)>0$ or $\lambda_1=\lambda_2=0$ and $\delta\geq 1$.
		At optimum, we have:	
		$$
		\psi(\tau^\star,\beta^\star)=\frac{(\beta^\star)^2}{4}+\lambda_2 \mathbb{E}\left[|x(\tilde{\tau}^\star,\beta^\star)|^2\right]+\lambda_1\mathbb{E}\left[|x(\tilde{\tau}^\star,\beta^\star)|\right].
		$$
		\begin{proof}
			From the previous lemma, the solution to the optimization problem in \eqref{eq:scalar} is the one that solves the system of equation in \eqref{eq:sys}, which we denote by $(\tau^\star,\beta^\star)$. From \eqref{eq:sys}, we thus have:
			$$
			\mathbb{E}\left[H{\rm prox}(\tilde{\tau}^\star H,\frac{\lambda_1\tilde{\tau}^\star}{\beta^\star})\right]=\tau^\star\delta -\frac{\beta^\star}{2}.
			$$
			Using this, we may simplify $\psi(\tau^\star,\beta^\star)$ as follows:
			\begin{align*}
				\psi(\tau^\star,\beta^\star)&=\frac{\beta^\star}{2\tau^\star}\Big((\tau^\star)^2\delta+\mathbb{E}\Big[{\rm prox}(\tilde{\tau}^\star H;\frac{\lambda_1\tilde{\tau}^\star}{\beta^\star})^2\Big]\Big)-\frac{(\beta^\star)^2}{4}+\frac{\beta^\star\rho}{2\tau^\star}+\lambda_2\mathbb{E}\left[{\rm prox}(\tilde{\tau}^\star H;\frac{\lambda_1\tilde{\tau}^\star}{\beta^\star})^2\right] \\
				&-\beta^\star\tau^\star\delta+\frac{(\beta^\star)^2}{2}+\lambda_1\mathbb{E}\Big[\left|{\rm prox}(\tilde{\tau}^\star H;\frac{\lambda_1\tilde{\tau}^\star}{\beta^\star})\right|\Big]\\
				&=\frac{\beta^\star}{2\tau^\star}(2(\tau^\star)^2\delta-\rho)+\frac{\beta^\star\rho}{2\tau^\star}+\frac{(\beta^\star)^2}{2}+\frac{\beta^\star\rho}{2\tau^\star}-\beta^\star\tau^\star \delta+\lambda_2 \mathbb{E}\Big[{\rm prox}(\tilde{\tau}^\star H;\frac{\lambda_1\tilde{\tau}^\star}{\beta^\star})^2\Big] \\
				&+\lambda_1\mathbb{E}\Big[\left|{\rm prox}(\tilde{\tau}^\star H;\frac{\lambda_1\tilde{\tau}^\star}{\beta^\star})\right|\Big]-\frac{(\beta^\star)^2}{4}\\
				&=\frac{(\beta^\star)^2}{4}+\lambda_2 \mathbb{E}\Big[{\rm prox}(\tilde{\tau}^\star H;\frac{\lambda_1\tilde{\tau}^\star}{\beta^\star})^2\Big]+\lambda_1\mathbb{E}\Big[\left|{\rm prox}(\tilde{\tau}^\star H;\frac{\lambda_1\tilde{\tau}^\star}{\beta^\star})\right|\Big].
			\end{align*}
		\end{proof}
		\label{lem:psi}
	\end{lemma}
	
	\begin{theorem}
		There exists strictly positive constants $\tau_{\rm min}$, $\tau_{\rm max}$, $\beta_{\rm min}$ and $\beta_{\rm max}$ that depend only on $P$, $\delta$, $\rho$, $\lambda_1$ and $\lambda_2$ such that:
		\begin{align*}
			&\tau_{\rm min}\leq (\tau^\star)^2\delta-\rho\leq \tau_{\rm max},\\
			& \beta_{\rm min}\leq \beta^\star\leq \beta_{\rm max}.
		\end{align*}
		\label{th:control_bounds}
	\end{theorem}
	\begin{proof}
		Recall the definition of function 
		$$
		\Psi:\beta\mapsto \min_{\tau\geq 0}\psi(\tau,\beta),
		$$
		which is concave. Moreover, one can easily check that it is continuously differentiable on $\mathbb{R}_{+}$ with derivative:
		$$
		\Psi'(\beta)=\tau^\star(\beta)\delta-\mathbb{E}[H{\rm prox}(\tilde{\tau}^\star(\beta);\frac{\lambda_1\tilde{\tau}^\star(\beta)}{\beta})] -\frac{\beta}{2}.
		$$
		It follows from Lemma \ref{lem:deriv_Psi} that $\lim_{\beta\to 0^{+}} \Psi'(\beta)>\alpha$ where $\alpha$ is a strictly positive constant depending only on $\delta$, $P$ and $\rho$. Since $\Psi'$ is non-increasing and continuous and tends to $-\infty$ as $\beta \to \infty$, there is $\beta_{\rm min}>0$ that depends only on $\lambda_1,\lambda_2, P,\delta$ and $\rho$ such that:
		$\forall \beta\in(0,\beta_{\rm min})$, $\Psi'(\beta)\geq \frac{\alpha}{2}$. Hence, necessarily $\beta^\star\geq \beta_{\rm min}$. 
		Now using the fact that $\tau^\star\leq \sqrt{\frac{P+\rho}{\delta}}$, we have
		$$
		\beta^\star\leq 2\tau^\star\delta+2\tau^\star\leq \beta_{\rm max}:=(\delta+1)\sqrt{\frac{P+\rho}{\delta}}.
		$$
		
		Recall that $\tau^\star$ satisfies:
		$$
		(\tau^\star)^2\delta-\rho=\mathbb{E}[\left({\rm prox}(\tilde{\tau}^\star H;\frac{\lambda_1\tilde{\tau}^\star}{\beta})\right)^2]
		$$
		Since $|{\rm prox}(\tilde{\tau}^\star H;\frac{\lambda_1\tilde{\tau}^\star}{\beta})|\leq \sqrt{P}$, we thus have
		$$
		(\tau^\star)^2\delta-\rho\leq P.
		$$ 
		Next, using the expression of the proximal operator given in \eqref{eq:prox_explicit}, we note that:
		$$
		(\tau^\star)^2\delta-\rho\geq \mathbb{E}[P{\bf 1}_{\{|H|\geq \frac{\lambda_1}{\beta}+\frac{\sqrt{P}}{\tilde{\tau}^\star}\}}].
		$$
		From the above inequalities, it can be readily seen that $\tilde{\tau}^\star\geq \frac{1}{\frac{\sqrt{\delta}}{\sqrt{\rho}}+\frac{2\lambda_2}{\beta_{\rm min}}}$. Hence,
		$$
		(\tau^\star)^2\delta-\rho\geq 2PQ(\frac{\lambda_1}{\beta_{\rm min}}+\sqrt{P}\frac{\sqrt{\delta}}{\sqrt{\rho}}+\sqrt{P}\frac{2\lambda_2}{\beta_{\rm min}}).
		$$
	\end{proof}
	\section{Study of the PO solution $\hat{\bf x}^{\rm PO}$ via the AO problem $\min_{\|{\bf x}\|_{\infty}\leq \sqrt{P}} \mathcal{L}_{\lambda,\rho}({\bf x})$}
	\subsection{Study of the AO problem} 
	
	\label{app:Lrho}
	\noindent{\bf Organization of the proof.} The proof consists of three major steps. In the first step, we prove that the AO cost is with high probability lower-bounded by the optimal cost of the asymptotic scalar optimization problem studied in the previous section. More formally, we show that there exists $\gamma_l>0$ such that  for all $\epsilon\in(0,1]$ ,
	\begin{equation}
		\inf_{{\bf s}}\mathbb{P}\left[\min_{x_i^2\leq P} \mathcal{L}_{\lambda,\rho}({\bf x})\geq \psi(\tau^\star,\beta^\star)-\gamma_l\epsilon\right] \geq 1-\frac{C}{\epsilon}e^{-cn\epsilon^2}\label{eq:re}
	\end{equation}
	for some $C$ and $c$ positive constants. 
	
	In the second step,  we define $\overline{\bf x}^{\rm AO}$ as:
	\begin{align}
	[\overline{\bf x}^{\rm AO}]_i:={\rm prox}(\tilde{\tau}^\star [{\bf h}]_i, \frac{\lambda_1\tilde{\tau}^\star}{\beta^\star})\label{eq:xoverlineAO}
	\end{align}
	where ${\bf h}$ is the Gaussian random vector appearing in the AO cost, and observe that there exists $\gamma_u$ such that for any $\epsilon\in(0,1)$
	\begin{equation}
		\inf_{{\bf s}} \mathbb{P}\Big[\mathcal{L}_{\lambda,\rho}(\overline{\bf x}^{\rm AO})\leq \psi(\tau^\star,\beta^\star)+ \gamma_u\epsilon\Big]\geq 1-C\exp(-cn\epsilon^2). \label{eq:dd}
	\end{equation}
	Using the fact that $$
	\min_{\|{\bf x}\|_{\infty}\leq \sqrt{P}} \mathcal{L}_{\lambda,\rho}({\bf x})\leq \mathcal{L}_{\lambda,\rho}(\overline{\bf x}^{\rm AO}),
	$$  we thus obtain
	\begin{equation}
	\inf_{{\bf s}} \mathbb{P}\Big[\psi(\tau^\star,\beta^\star)-\max(\gamma_l,\gamma_u)\epsilon\leq \mathcal{L}_{\lambda,\rho}(\overline{\bf x}^{\rm AO})\leq \psi(\tau^\star,\beta^\star)+\max(\gamma_l,\gamma_u)\epsilon\Big] \geq 1-\frac{C}{\epsilon}\exp(-cn\epsilon^2) \label{eq:AO_r}
	\end{equation}
	and 
	\begin{equation}
	\inf_{{\bf s}} \mathbb{P}\big[\psi(\tau^\star,\beta^\star)-\max(\gamma_u,\gamma_l)\epsilon\leq\min_{\|{\bf x}\|_{\infty}\leq \sqrt{P}} \mathcal{L}_{\lambda,\rho}({\bf x})\leq \psi(\tau^\star,\beta^\star)+\max(\gamma_u,\gamma_l)\epsilon\big]\geq 1-C\exp(-cn\epsilon^2).\label{eq:optimal_cost}
	\end{equation}
As a direct consequence of the inequalities in \eqref{eq:1} and \eqref{eq:31}, we can transfer \eqref{eq:optimal_cost}
 to the optimal cost of the PO in \eqref{eq:PO1}, yielding:
	\begin{equation}
		\inf_{{\bf s}} \mathbb{P}\big[\psi(\tau^\star,\beta^\star)-\max(\gamma_u,\gamma_l)\epsilon\leq\min_{\|{\bf x}\|_{\infty}\leq \sqrt{P}} \mathcal{C}_{\lambda,\rho}({\bf x})\leq \psi(\tau^\star,\beta^\star)+\max(\gamma_u,\gamma_l)\epsilon\big]\geq 1-C\exp(-cn\epsilon^2).\label{eq:optimal_cost_po}
	\end{equation}
	Next, by combining \eqref{eq:optimal_cost} and \eqref{eq:AO_r}, we obtain:
	\begin{equation}
		\inf_{{\bf s}}\mathbb{P}\Big[\big|\min_{\substack{{\bf x}\\ x_i^2\leq P}} \mathcal{L}_{\lambda,\rho}({\bf x}) -  \mathcal{L}_{\lambda,\rho}(\overline{\bf x}^{\rm AO})\big|\leq 2\max(\gamma_u,\gamma_l)\epsilon  \Big] \geq 1-\frac{C}{\epsilon}e^{-cn\epsilon^2} .\label{eq:min_dif}
	\end{equation}
	The above $\overline{\bf x}^{\rm AO}$ will be termed in the sequel as the AO equivalent solution, since as it will be made clear later, the optimal solution of the AO lies within the ball centered at $\overline{\bf x}^{\rm AO}$ with probability approaching one. 
	By exploiting this inequality together with the local strong-convexity of $\mathcal{L}_{\lambda,\rho}({\bf x})$ around $\overline{\bf x}^{\rm AO}$, we finally prove in the third step that $\mathcal{L}_{\lambda,\rho}({\bf x})$ is with high-probability  uniformly sub-optimum  outside any ball containing  $\overline{\bf x}^{\rm AO}$. More specifically, we prove that, there exists a positive constant $\gamma$ such that for any $\tilde{\lambda}\in(0,1]$,
	\begin{equation}
	\inf_{{\bf s}} \mathbb{P}\Big[\forall \ \  \text{feasible }{\bf x}, \ \ \frac{1}{n}\|{\bf x}-\overline{\bf x}^{\rm AO}\|^2\geq \tilde{\lambda}{\epsilon}^2  \ \text{and} \ \mathcal{L}_{\lambda,\rho}({\bf x})\geq  \psi(\tau^\star,\beta^\star)+\gamma \tilde{\lambda}{\epsilon}^2  \Big]\geq 1-\frac{C}{\epsilon^2}\exp(-cn\epsilon^4). \label{eq:resultat2}
	\end{equation}
	Since the constants in the right-hand side inequality are independent of ${\bf s}$, the below inequality also holds 
	\begin{equation}
		\mathbb{P}\Big[\forall \ \  \text{feasible }{\bf x}, \ \ \frac{1}{n}\|{\bf x}-\overline{\bf x}^{\rm AO}\|^2\geq \tilde{\lambda}\epsilon^2  \ \text{and} \  \mathcal{L}_{\lambda,\rho}({\bf x})\geq \psi(\tau^\star,\beta^\star)+\gamma \tilde{\lambda}\epsilon^2  \Big]\geq 1-\frac{C}{\epsilon^2}\exp(-cn\epsilon^4),\label{eq:sub_opt}
	\end{equation}
	where the probability is with respect to the distribution of ${\bf g}$, ${\bf h}$ and ${\bf s}$. 
	This property, typically established by leveraging the local strong-convexity of the AO cost around its optimal solution, turns out to be instrumental for deducing the  properties of the PO's solution through a deviation argument.
	
	The procedural approach is as follows: consider a deterministic set $\mathcal{S}_{x}$ in $\mathbb{R}^{n}$ for which, with a probability approaching one, the following condition holds:
	$$
	{\bf x}\in \mathcal{S}_{x}^\epsilon \Longrightarrow \frac{1}{n}\|{\bf x}-\overline{\bf x}^{\rm AO}\|^2\geq \tilde{\lambda}\epsilon^2
	$$ 
	for some $\tilde{\lambda}\in(0,1)$. 
	From this, we can use \eqref{eq:1} to write:
	$$
	\mathbb{P}\Big[\min_{{\bf x}\in \mathcal{S}_x^\epsilon} \mathcal{C}_{\lambda,\rho}({\bf x}) \leq \psi(\tau^\star,\beta^\star)+\gamma\tilde{\lambda}\epsilon^2\Big]\leq 2\mathbb{P}\Big[\min_{{\bf x}\in \mathcal{S}_x^\epsilon} \mathcal{L}_{\lambda,\rho}({\bf x}) \leq \psi(\tau^\star,\beta^\star)+\tilde{\lambda}\gamma\epsilon^2\Big]+C\exp(-cn)
	$$
	and hence, based on \eqref{eq:sub_opt}, we obtain that with probability approaching one:
	$$
	\min_{{\bf x}\in \mathcal{S}_x^\epsilon}\mathcal{C}_{\lambda,\rho}({\bf x})\geq \psi(\tau^\star,\beta^\star)+\gamma \tilde{\lambda}\epsilon^2.
	$$
	Since from \eqref{eq:optimal_cost_po}, with probability greater than $1-\frac{C}{\epsilon}\exp(-cn\epsilon^2)$, 
	$$
	\mathcal{C}_{\lambda,\rho}(\hat{\bf x}^{\rm PO})\leq \psi(\tau^\star,\beta^\star)+\gamma\tilde{\lambda} \epsilon^2,
	$$
	we immediately deduce that $\hat{\bf x}^{\rm PO}\notin \mathcal{S}_x^\epsilon$ with probability approaching one. 
 \subsubsection{High probability lower-bound of the AO (Proof of \eqref{eq:re})}
	We note that the variable ${\bf u}$ appears in the objective of \eqref{eq:L} through a linear term and through its magnitude, which suggests that one can first optimize over its direction for fixed amplitude.
	Obviously, the direction of ${\bf u}$ that optimizes the above expression is the one that aligns with $\left(\frac{1}{n}\|{\bf x}\|_2{\bf g}-\sqrt{\frac{\rho}{n}}{\bf s}\right)$. Hence, \eqref{eq:L} simplifies as:
	\begin{align}
		\mathcal{L}_{\lambda,\rho}({\bf x})&=\max_{C_\beta\geq\beta\geq 0}\ell_{\lambda,\rho}({\bf x},\beta) \label{eq:LL}
	\end{align}  
	with $\ell_{\lambda,\rho}({\bf x},\beta)$ defined as:
	\begin{align}
		\ell_{\lambda,\rho}({\bf x},\beta)&=	\frac{\beta \sqrt{m}}{\sqrt{n}}  \left\|\frac{\|{\bf x}\|_2{\bf g}}{\sqrt{n}\sqrt{m}}-\frac{\sqrt{\rho}{\bf s}}{\sqrt{m}}\right\|_2 -\frac{\beta}{n}{\bf h}^{T}{\bf x}-\frac{\beta^2}{4}+\frac{\lambda_2\|{\bf x}\|_2^2}{n}+\frac{\lambda_1\|{\bf x}\|_1}{n}. \label{eq:ll}
	\end{align}
	Define $$\ell_{\lambda,\rho}^{\circ}({\bf x},\beta)=\beta\sqrt{\delta}\sqrt{\frac{\|{\bf x}\|^2}{n}+\rho}-\frac{\beta}{n}{\bf h}^{T}{\bf x} -\frac{\beta^2}{4} +\frac{\lambda_2\|{\bf x}\|^2}{n}+\frac{\lambda_1\|{\bf x}\|_1}{n}.$$
	For all ${\bf s}\in \mathcal{S}^{\otimes m}$, the event:
	$$
	\mathcal{C}:=\left\{\left|\frac{\|{\bf g}\|_2^2}{m}-1\right|\leq \epsilon\right\}\cap \left\{\left|\frac{1}{m}{\bf g}^{T}{\bf s}\right|\leq \frac{\epsilon\sqrt{\rho}}{4\sqrt{P}}\right\}
	$$
	occurs with at least probability $1-Ce^{-cn\epsilon^2}$ where $C$ and ${c}$ does not depend on ${\bf s}$ and $\epsilon\in(0,1]$. In other words:
	\begin{equation*}
		\inf_{{\bf s}} \mathbb{P}\left[\left\{\left|\frac{\|{\bf g}\|_2}{\sqrt{m}}-1\right|\leq \epsilon\right\}\cap \left\{\left|\frac{1}{m}{\bf g}^{T}{\bf s}\right|\leq \frac{\epsilon\sqrt{\rho}}{4\sqrt{P}}\right\}\right]\geq 1-Ce^{-cn\epsilon^2}. 
	\end{equation*}
	Using the relation $|\sqrt{a}-\sqrt{b}|=\frac{|a-b|}{\sqrt{a}+\sqrt{b}}$ for any positive $a$ and $b$, we have:
	\begin{align*}
		\left|\ell_{\lambda,\rho}({\bf x},\beta)-\ell_{\lambda,\rho}^{\circ}({\bf x},\beta)\right|&\leq \beta\sqrt{\delta}\frac{\left|\frac{\|{\bf x}\|^2}{n}\big(\frac{{\bf g}^{T}{\bf g}}{m}-1\big)-2\sqrt{\rho}\frac{{\bf s}^{T}{\bf g}}{m}\frac{\|{\bf x}\|}{\sqrt{n}}\right|}{\sqrt{\frac{\|\bf x\|^2}{n}+\rho}\|\frac{\|{\bf x}\|{\bf g}}{\sqrt{n m}}-\frac{\sqrt{\rho}{\bf s}}{\sqrt{m}}\|}.
	\end{align*}
	On the event $\mathcal{C}$, it is easy to see that:
	$$
	\|\frac{\|{\bf x}\|{\bf g}}{\sqrt{n m}}-\frac{\sqrt{\rho}{\bf s}}{\sqrt{m}}\|\geq \sqrt{\rho-2\sqrt{P}\sqrt{\rho}\frac{|{\bf s}^{T}{\bf g}|}{m}}\geq \sqrt{\rho}\sqrt{1-\frac{1}{2}\epsilon}\geq\sqrt{\rho} \sqrt{\frac{1}{2}}
	$$
	where in the last inequality we used the fact that $0\leq\epsilon\leq 1$. 
	Hence, for all ${\bf x}$ satisfying $\|{\bf x}\|^2\leq Pn$ and $0\leq \beta\leq C_\beta$, we have:
	$$
	\sup_{x_i^2\leq P} \sup_{{0\leq \beta\leq C_\beta}} \left|\ell_{\lambda,\rho}({\bf x},\beta)-\ell_{\lambda,\rho}^{\circ}({\bf x},\beta)\right|\leq C_\beta\sqrt{2\delta}\frac{P\epsilon+\epsilon\frac{\rho}{2}}{\rho}.
	$$
	Set $K:=C_\beta\sqrt{2\delta}\frac{P+ \frac{\rho}{2}}{\rho}$, we thus have 
	\begin{equation}
		\inf_{{\bf s}} \mathbb{P}\Big[\sup_{\substack{x_i^2\leq P\\ 0\leq \beta\leq C_\beta}}\left|	\ell_{\lambda,\rho}({\bf x},\beta)-\ell_{\lambda,\rho}^\circ({\bf x},\beta)\right|\leq K\epsilon\Big]\geq 1-Ce^{-cn\epsilon^2}. \label{eq:uniform2}
	\end{equation}
	Now, starting from \eqref{eq:LL}, we get:
	\begin{equation*}
		\min_{\substack{{\bf x}\\ x_i^2\leq P}}\mathcal{L}_{\lambda,\rho}({\bf x})\geq \min_{\substack{{\bf x}\\ x_i^2\leq P}} \ell_{\lambda,\rho}({\bf x},\beta^\star) .\end{equation*}
	Using \eqref{eq:uniform2}, this implies that:
	\begin{equation}
		\inf_{{\bf s}}\mathbb{P}\Big[\min_{\substack{{\bf x}\\ x_i^2\leq P}}\mathcal{L}_{\lambda,\rho}({\bf x})\geq\min_{\substack{{\bf x}\\ x_i^2\leq P}} \ell_{\lambda,\rho}^\circ({\bf x},\beta^\star)-K\epsilon\Big]\geq 1-C\exp(-cn\epsilon^2). \label{eq:or}
	\end{equation}
	Using the relation
	\begin{equation*} 
		\sqrt{\chi}= \min_{\tau\geq 0} \frac{\tau}{2}+\frac{\chi}{2\tau},  
	\end{equation*}
	we may write $\ell_{\lambda,\rho}^\circ({\bf x},\beta^\star)$ as:
	$$
	\ell_{\lambda,\rho}^\circ({\bf x},\beta^\star) = \inf_{\tau \geq 0} F(\tau,{\bf h})
	$$
	with
	$$
	F(\tau,{\bf h}):=\frac{\beta^*\tau}{2}\sqrt{\delta}-\frac{(\beta^*)^2}{4}+\frac{\beta^*\sqrt{\delta}\rho}{2\tau} + \min_{\substack{{\bf x}\\ x_i^2\leq P}} \beta^* \sqrt{\delta} \frac{\|{\bf x}\|^2}{2n\tau}-\frac{\beta^*}{n}{\bf h}^{T}{\bf x} + \frac{\lambda_2\|{\bf x}\|^2}{n}+\frac{\lambda_1\|{\bf x}\|_1}{n}.
	$$
	Since the optimum $\tau^\star$ minimizing $F(\tau,{\bf h})$ is given by:
	$$
	\tau^\star=\sqrt{\frac{\|{\bf x}\|^2}{n}+\rho},
	$$
	we thus have:
	$$
	\sqrt{\rho}\leq \tau^\star\leq \sqrt{\rho+P}.
	$$
	Hence, we obtain:
	$$\ell_{\lambda,\rho}^\circ({\bf x},\beta^\star) = \min_{\sqrt{\rho}\leq\tau \leq \sqrt{\rho+P}} F(\tau,{\bf h}).
	$$
	To continue, we note that function:
	$$
	{\bf h}:\mapsto \min_{\substack{{\bf x}\\ x_i^2\leq P}} \beta^* \sqrt{\frac{m}{n}} \frac{\|{\bf x}\|^2}{2n\tau}-\frac{\beta^*}{n}{\bf h}^{T}{\bf x} + \frac{\lambda_2\|{\bf x}\|^2}{n}+\frac{\lambda_1\|{\bf x}\|_1}{n}
	$$
	is $\frac{C_\beta\sqrt{P}}{\sqrt{n}}$-Lipschitz. Using standard Gaussian concentrations   inequalities for Lipschitz functions, we obtain for any $t> 0$
	\begin{equation}
		\mathbb{P}\left[\left|F(\tau,{\bf h})-\mathbb{E}[F(\tau,{\bf h})]\right|\geq t\right]\leq 2\exp(-cnt^2),\label{eq:111}
	\end{equation}
	where $c=\frac{1}{2 C_\beta^2 P}$. Let $\tilde{\epsilon}>0$ and denote by $\mathcal{N}_{\tilde{\epsilon}}$ be an $\tilde{\epsilon}$-net of the interval $[\sqrt{\rho},\sqrt{\rho+P}]$. Then, for any $\tau\in[\sqrt{\rho},\sqrt{\rho+P}]$, there exists $\tau'\in\mathcal{N}_\epsilon$ such that $|\tau-\tau'|\leq \tilde{\epsilon}$. Since $\tau\mapsto F(\tau,{\bf h})$ and $\tau\mapsto \mathbb{E}[F(\tau,{\bf h})]$ are both  $L$-Lipschitz on $[\sqrt{\rho},\sqrt{\rho+P}]$ with $L:=\frac{C_\beta}{2}\sqrt{\frac{m}{n}}+\frac{C_\beta\sqrt{\frac{m}{n}}}{2}+C_\beta\sqrt{\frac{m}{n}}\frac{P}{2\rho}$, we thus obtain:
	\begin{equation*}
		\sup_{\sqrt{\rho}\leq \tau\leq \sqrt{\rho+P}} |{F}(\tau,{\bf h})-\mathbb{E}[F(\tau,{\bf h})]|\leq \max_{\tau\in\mathcal{N}_{\tilde{\epsilon}}} |{F}(\tau,{\bf h})-\mathbb{E}[F(\tau,{\bf h})]|+2L\tilde{\epsilon}. 
	\end{equation*}
	Now, since $\mathcal{N}_{\tilde{\epsilon}}$ is composed of $\frac{\sqrt{\rho+P}-\sqrt{\rho}}{\epsilon}$ discrete points, by using the union bound, we obtain
	\begin{equation}
\mathbb{P}\left[\max_{\tau\in\mathcal{N}_{\tilde{\epsilon}}} |{F}(\tau,{\bf h})-\mathbb{E}[F(\tau,{\bf h})]|\geq \tilde{\epsilon}\right]\leq 2\frac{(\sqrt{\rho+P}-\sqrt{\rho})}{\tilde{\epsilon}}\exp(-cn\tilde{\epsilon}^2). \label{eq:212}
	\end{equation}
	With this, combining \eqref{eq:111} and \eqref{eq:212}, we show that there exists positive constants $C$ and $c'$ such that for any $\epsilon >0$
	$$
	\mathbb{P}\left[\sup_{\sqrt{\rho}\leq \tau\leq \sqrt{\rho+P}}\left|F(\tau,{\bf h})-\mathbb{E}[F(\tau,{\bf h})]\right|\geq \epsilon\right]\leq \frac{C}{\epsilon}\exp(-c'n\epsilon^2).\label{eq:11}
	$$
	Hence, with probability $1-\frac{C}{\epsilon}\exp(-c'n\epsilon^2)$, we obtain:
	$$
	\min_{\sqrt{\rho}\leq \tau\leq \sqrt{\rho+P}} F(\tau,{\bf h})\geq \min_{\sqrt{\rho}\leq \tau\leq \sqrt{\rho+P}} \mathbb{E}[F(\tau,{\bf h})]-\epsilon. 
	$$
	Consequently, using \eqref{eq:or}, we thus obtain that with probability $1-\frac{C'}{\epsilon}e^{-c'n\epsilon^2}$,
	$$
	\min_{\substack{{\bf x}\\ x_i^2\leq P}}\mathcal{L}_{\lambda,\rho}({\bf x})\geq \min_{c_1\leq \tau\leq c_2} F(\tau,{\bf h}) -K\epsilon\geq \min_{\sqrt{\rho}\leq \tau\leq \sqrt{\rho+P}} \mathbb{E}[F(\tau,{\bf h})]-(K+1)\epsilon\geq \min_{\tau\geq 0} \mathbb{E}[F(\tau,{\bf h})]-(K+1)\epsilon.
	$$
	By performing the change of variable $\tau\leftrightarrow \frac{\tau}{\sqrt{\delta}}$, we can write $\min_{\tau\geq 0} \mathbb{E}[F(\tau,{\bf h})]$ as:
	$$
	\min_{\tau\geq 0} \mathbb{E}[F(\tau,{\bf h})]=\min_{\tau\geq 0} \frac{\beta^\star\tau \delta}{2}-\frac{(\beta^\star)^2}{4}+\frac{\beta^\star\rho}{2\tau} +\mathbb{E}\left[\min_{|x|<\sqrt{P}}\beta^\star \frac{x^2}{2\tau}- \beta^\star Hx +\lambda_2 x^2+\lambda_1 |x|\right]
	$$
	where ${H}$ is a standard Gaussian variable and the expectation is taken with respect to its distribution. It follows from the previous section that $\min_{\tau\geq 0} \mathbb{E}[F(\tau,{\bf h})]=\psi(\tau^\star,\beta^\star)$, thus proving \eqref{eq:re}.
	
	\subsubsection{ Asymptotic properties of the  AO equivalent solution (Proof of \eqref{eq:dd})}
	
	As can be seen from \eqref{eq:ll}, the expression of the optimization cost $ \mathcal{L}_{\lambda,\rho}({\bf x})$ depends on ${\bf x}$ through $\frac{1}{\sqrt{n}}\|{\bf x}\|$, $\frac{1}{n}{\bf h}^{T}{\bf x}$ and $\frac{1}{n}\|{\bf x}\|_1$. To prove that $\mathcal{L}_{\lambda,\rho}(\overline{\bf x}^{\rm AO})$ concentrates around $\psi(\tau^\star,\beta^\star)$, we shall thus study the concentration of $\frac{1}{\sqrt{n}}\|\overline{\bf x}^{\rm AO}\|$, $\frac{1}{n}{\bf h}^{T}\overline{\bf x}^{\rm AO}$ and $\frac{1}{n}\|\overline{\bf x}^{\rm AO}\|_1$ around their expectations and then plug their approximations back into $ \mathcal{L}_{\lambda,\rho}(\overline{\bf x}^{\rm AO})$. Towards this end, 
	recalling the expression of the AO equivalent solution:
	$$
	[\overline{\bf x}^{\rm AO}]_i:={\rm prox}(\tilde{\tau}^\star [{\bf h}]_i, \frac{\lambda_1\tilde{\tau}^\star}{\beta^\star})
	$$
	where $\tilde{\tau}^\star=\frac{1}{\frac{1}{\tau^\star}+\frac{2\lambda_2}{\beta^\star}}$. 
	We start by noting  that $\overline{\bf x}^{\rm AO}$ satisfies:
	$$
	\overline{\bf x}^{\rm AO}=\arg\min_{\substack{{\bf x}\\ \|{\bf x}\|_\infty\leq \sqrt{P}}} \frac{\beta^\star}{2\lambda_1\tilde{\tau}^\star}\|\tilde{\tau}^\star{\bf h}-{\bf x}\|^2+\|{\bf x}\|_1.
	$$
	Hence, function $ \overline{\bf x}^{\rm AO}$ as a function of $\tilde{\tau}^\star{\bf h}$ is a proximal operator.
	This  implies that  $ \overline{\bf x}^{\rm AO}$ is  a $\tilde{\tau}^\star$ Lipschitz function of ${\bf h}$, and thus ${\bf h}\mapsto \frac{1}{\sqrt{n}}\|\overline{\bf x}^{\rm AO}\|$ is a $\frac{\tilde{\tau}^\star}{\sqrt{n}}$-Lipschitz function.  Using Gaussian concentrations for Lipschitz functions, we obtain:
	\begin{equation*}
		\mathbb{P}\left[\left|\frac{1}{\sqrt{n}}\|\overline{\bf x}^{\rm AO}\|- \mathbb{E}[\frac{1}{\sqrt{n}}\|\overline{\bf x}^{\rm AO}\|]\right|>t\right]\leq 2\exp(-\frac{nt^2}{2(\tilde{\tau}^\star)^2}).  
	\end{equation*}
	Moreover, it follows from \cite[Proposition G.5]{Miolane} that $\frac{1}{n}\|\overline{\bf x}^{\rm AO}\|^2$ is a $(C/n,C/n)$sub-Gamma random variable. Hence, for some constant $c$,
	\begin{equation}
		\mathbb{P}\left[\left|\frac{1}{n}\|\overline{\bf x}^{\rm AO}\|^2- \mathbb{E}[\frac{1}{n}\|\overline{\bf x}^{\rm AO}\|^2]\right|>t\right]\leq 2\exp(-{cnt^2})+2\exp(-{cnt}). \label{eq:concen1}
	\end{equation}
	Similarly, to handle the variable $\frac{1}{n}{\bf h}^{T}\overline{\bf x}^{\rm AO}$, we write it as:
	$$
	\frac{1}{n}{\bf h}^{T}\overline{\bf x}^{\rm AO}=\frac{1}{n}\|{\bf h}-\overline{\bf x}^{\rm AO}\|^2-\frac{1}{n}\|{\bf h}\|_2^2-\frac{\|\overline{\bf x}^{\rm AO}\|^2}{n}
	$$
	and exploit the fact that every quantity in the above decomposition is a $(C/n,C/n)$ sub-Gamma random variable. There exists thus a constant $c$ depending on $\tilde{\tau}^\star$ such that:
	\begin{equation}
		\mathbb{P}\left[\left|\frac{1}{n}{\bf h}^{T}\overline{\bf x}^{\rm AO}- \mathbb{E}[\frac{1}{n}{\bf h}^{T}\overline{\bf x}^{\rm AO}]\right|>t\right]\leq 2\exp(-{cnt^2})+2\exp(-{cnt}) \label{eq:concen2}
	\end{equation}
	for some constant $c$. 
	Finally, it is clear that ${\bf h}\mapsto \frac{1}{n}\|\overline{\bf x}^{\rm AO}\|_1$ is a $\frac{\tilde{\tau}^\star}{\sqrt{n}}$ Lipschitz function. Hence,
	\begin{equation}
		\mathbb{P}\left[\left|\frac{1}{n} \|\overline{\bf x}^{\rm AO}\|_1-\mathbb{E}[\frac{1}{n} \|\overline{\bf x}^{\rm AO}\|_1]\right|\geq t\right]\leq 2\exp(-\frac{nt^2}{2(\tilde{\tau}^\star)^2}). \label{eq:concen3}
	\end{equation}
	By combining \eqref{eq:concen1}, \eqref{eq:concen2} and \eqref{eq:concen3}, we conclude that the event:
	\begin{equation}
	\mathcal{A}:=	\left\{\max\left\{ \left|\frac{1}{n}\|\overline{\bf x}^{\rm AO}\|^2- \mathbb{E}[\frac{1}{n}\|\overline{\bf x}^{\rm AO}\|^2]\right|,\left|\frac{1}{n}{\bf h}^{T}\overline{\bf x}^{\rm AO}- \mathbb{E}[\frac{1}{n}{\bf h}^{T}\overline{\bf x}^{\rm AO}]\right|, \left|\frac{1}{n} \|\overline{\bf x}^{\rm AO}\|_1-\mathbb{E}[\frac{1}{n} \|\overline{\bf x}^{\rm AO}\|_1]\right|\right\}\leq \epsilon  \right\} \label{eq:A}
	\end{equation}
	occurs with probability at least $1-C\exp(-cn\epsilon^2)-C \exp(-cn\epsilon)$, for some constants $C$ and $c$.
	
	Next, to continue, we shall simplify the expressions for $\frac{1}{n}\mathbb{E}\left[{\bf h}^{T}\overline{\bf x}^{\rm AO}\right]$, $\mathbb{E}\left[\frac{1}{n}\|\overline{\bf x}^{\rm AO}\|^2\right]$ and $\mathbb{E}\left[\frac{1}{n}\|\overline{\bf x}^{\rm AO}\|_1\right]$. Indeed, it takes no much effort to check that:
	\begin{align*}
		\frac{1}{n}\mathbb{E}\left[{\bf h}^{T}\overline{\bf x}^{\rm AO}\right]&= \mathbb{E}\left[H{\rm prox}(\tilde{\tau}^\star H;\frac{\lambda_1\tilde{\tau}^\star}{\beta^\star})\right],\\
		\mathbb{E}\left[\frac{1}{n}\left\|\overline{\bf x}^{\rm AO}\right\|^2\right]&= \mathbb{E}\left[({\rm prox}(\tilde{\tau}^\star H;\frac{\lambda_1\tilde{\tau}^\star}{\beta^\star}))^2\right],\\
		\mathbb{E}\left[\frac{1}{n}\left\|\overline{\bf x}^{\rm AO}\right\|_1\right]&= \mathbb{E}\left[\left|{\rm prox}(\tilde{\tau}^\star H;\frac{\lambda_1\tilde{\tau}^\star}{\beta^\star})\right|\right].
	\end{align*}
	With this, we are now ready to prove \eqref{eq:dd}. First, we note that for all ${\bf s}\in\mathcal{S}^{\otimes m}$ and  all $\epsilon\in(0,1)$ the event:
	$$
	\mathcal{B}:=\left\{\left|\frac{\|{\bf g}\|^2}{m}-1\right|\leq \epsilon\right\}\cap \left\{\frac{{\bf g}^{T}{\bf s}} {m}\leq \epsilon\right\}
	$$
	occurs with probability $1-C\exp(-cn\epsilon^2)$ where $C$ and $c$ does not depend on ${\bf s}$. On this event, the following inequality holds true:
	$$
	\frac{\|{\bf g}\|^2}{n}\leq \delta(1+\epsilon).
	$$
	Hence, we can easily check that the  following upper-bound holds true on the event $\mathcal{A}\cap\mathcal{B}$
	\begin{align*}
		\left\|\frac{\|\overline{\bf x}^{\rm AO}\|{\bf g}}{n}-\frac{\sqrt{\rho}{\bf s}}{\sqrt{n}}\right\|_2&\leq
		\sqrt{\delta}\sqrt{\frac{1}{n}\mathbb{E}\|\overline{\bf x}^{\rm AO}\|^2+\rho+P\epsilon+\epsilon+\epsilon^2+2\sqrt{\rho}\epsilon \sqrt{P}}\\
		&\leq \sqrt{\delta}\sqrt{\frac{1}{n}\mathbb{E}[\|{\overline{\bf x}^{\rm AO}}\|^2]+\rho}\sqrt{1+\frac{P\epsilon+\epsilon+\epsilon^2}{\rho}+2\epsilon\sqrt{\frac{P}{\rho}}}\\
		&\leq \sqrt{\delta}\sqrt{\frac{1}{n}\mathbb{E}[\|{\overline{\bf x}^{\rm AO}}\|^2]+\rho} \left(1+\epsilon\left( \frac{P+1+\epsilon}{2\rho}+\sqrt{\frac{P}{\rho}}\right)\right)\end{align*}
	where the last inequality follows from using the fact that for $0\leq \epsilon\leq 1$ and $\alpha>0$, $\sqrt{1+\alpha \epsilon}\leq 1+\frac{\alpha}{2}\epsilon$ and $\epsilon^2\leq \epsilon$. 
	Hence, setting $K=\sqrt{\delta}\sqrt{P+\rho}(\frac{P+1+\epsilon}{2\rho}+\sqrt{\frac{P}{\rho}})$, on the event $\mathcal{A}\cap \mathcal{B}$
	the following inequality holds true:
	$$
	\left\|\frac{\|\overline{\bf x}^{\rm AO}\|{\bf g}}{n}-\frac{\sqrt{\rho}{\bf s}}{\sqrt{n}}\right\|_2\leq \sqrt{\delta}\sqrt{\frac{1}{n}\mathbb{E}[\|{\overline{\bf x}^{\rm AO}}\|^2]+\rho}+K\epsilon.
	$$
	As a result, we may upper-bound $\mathcal{L}_{\lambda,\rho}(\overline{\bf x}^{\rm AO})$ on the event $\mathcal{A}\cap\mathcal{B}$ by:
	\begin{align*}
		\mathcal{L}_{\lambda,\rho}(\overline{\bf x}^{\rm AO})\leq & \beta^\star \sqrt{\delta}\sqrt{\mathbb{E}\left[({\rm prox}(\tilde{\tau}^\star H;\frac{\lambda_1\tilde{\tau}^\star}{\beta^\star}))^2\right]+\rho} -\beta^\star \mathbb{E}\left[H{\rm prox}(\tilde{\tau}^\star H;\frac{\lambda_1\tilde{\tau}^\star}{\beta^\star})\right] -\frac{(\beta^\star)^2}{4} \\
		&+\lambda_2 \mathbb{E}\left[({\rm prox}(\tilde{\tau}^\star H;\frac{\lambda_1\tilde{\tau}^\star}{\beta^\star}))^2\right]  + \lambda_1\mathbb{E}\left[\left|{\rm prox}(\tilde{\tau}^\star H;\frac{\lambda_1\tilde{\tau}^\star}{\beta^\star})\right|\right] +(\lambda_2+\lambda_1+(K+1)\beta^\star)\epsilon.
	\end{align*}
	To continue, we use the fixed-point of equations verified by $\tau^\star$ and $\beta^\star$ in Lemma \ref{lem:fixed_point} to simplify the upper-bound as:
	\begin{align*}
		\mathcal{L}_{\lambda,\rho}(\overline{\bf x}^{\rm AO})\leq &\frac{\beta^\star}{2}\left(2 \delta \tau^\star -2 \mathbb{E}\left[H{\rm prox}(\tilde{\tau}^\star H;\frac{\lambda_1\tilde{\tau}^\star}{\beta^\star})\right]\right)- \frac{(\beta^\star)^2}{4}\nonumber \\
		&+\lambda_2 \mathbb{E}\left[({\rm prox}(\tilde{\tau}^\star H;\frac{\lambda_1\tilde{\tau}^\star}{\beta^\star}))^2\right]  + \lambda_1\mathbb{E}\left[\left|{\rm prox}(\tilde{\tau}^\star H;\frac{\lambda_1\tilde{\tau}^\star}{\beta^\star})\right|\right] +(\lambda_2+\lambda_1+(K+1)\beta^\star)\epsilon\\
		=&\frac{(\beta^\star)^2}{4} +\lambda_2 \mathbb{E}\left[({\rm prox}(\tilde{\tau}^\star H;\frac{\lambda_1\tilde{\tau}^\star}{\beta^\star}))^2\right]  + \lambda_1\mathbb{E}\left[\left|{\rm prox}(\tilde{\tau}^\star H;\frac{\lambda_1\tilde{\tau}^\star}{\beta^\star})\right|\right] +(\lambda_2+\lambda_1+(K+1)\beta^\star)\epsilon\\
		=&\psi(\tau^\star,\beta^\star)+(\lambda_2+\lambda_1+(K+1)\beta^\star)\epsilon.
	\end{align*}
	By setting $\gamma_u=(\lambda_2+\lambda_1+(K+1)\beta^\star)$, we thus prove that:
	$$
	\mathcal{L}_{\lambda,\rho}(\overline{\bf x}^{\rm AO})\leq \psi(\tau^\star,\beta^\star)+\gamma_u\epsilon.
	$$
	\subsubsection{Asymptotic convergence of the optimal cost. (Proof of \eqref{eq:optimal_cost_po}.)}
	\label{sec:proof_optimal_cost}
	Applying \eqref{eq:1}, we obtain:   
	\begin{equation*}
		\begin{aligned}
	&\underset{{\bf s}\in \mathcal{S}^{\otimes m}}{\sup}\mathbb{P}\Big[\min_{\|{\bf x}\|_{\infty}\leq \sqrt{P}} \mathcal{C}_{\lambda,\rho}({\bf x})\leq \psi(\tau^\star,\beta^\star)-\max(\gamma_l,\gamma_u)\epsilon\Big]\\\leq& 2\underset{{\bf s}\in \mathcal{S}^{\otimes m}}{\sup}\mathbb{P}\Big[\min_{\|{\bf x}\|_{\infty}\leq \sqrt{P}} \mathcal{L}_{\lambda,\rho}({\bf x})\leq \psi(\tau^\star,\beta^\star)-\max(\gamma_l,\gamma_u)\epsilon\Big]+C\exp(-cn)\\
	\leq&\frac{C}{\epsilon}\exp(-cn\epsilon^2) 
	\end{aligned}
	\end{equation*}
	where in the first inequality, we used \eqref{eq:1} and in the second inequality we used \eqref{eq:re}.
	Similarly, applying \eqref{eq:31}, we obtain:
	\begin{equation*}
		\begin{aligned}
			&\underset{{\bf s}\in \mathcal{S}^{\otimes m}}{\sup}\mathbb{P}\Big[\min_{\|{\bf x}\|_{\infty}\leq \sqrt{P}} \mathcal{C}_{\lambda,\rho}({\bf x})\geq \psi(\tau^\star,\beta^\star)+\max(\gamma_l,\gamma_u)\epsilon\Big]\\\leq &\underset{{\bf s}\in \mathcal{S}^{\otimes m}}{\sup}2\mathbb{P}\Big[\min_{\|{\bf x}\|_{\infty}\leq \sqrt{P}} \mathcal{L}_{\lambda,\rho}({\bf x})\geq \psi(\tau^\star,\beta^\star)+\max(\gamma_l,\gamma_u)\epsilon\Big]\\
			\leq& \underset{{\bf s}\in \mathcal{S}^{\otimes m}}{\sup}2\mathbb{P}\Big[ \mathcal{L}_{\lambda,\rho}(\overline{\bf x}^{\rm AO})\geq \psi(\tau^\star,\beta^\star)+\max(\gamma_l,\gamma_u)\epsilon\Big]\\
			\leq& \frac{C}{\epsilon}\exp(-cn\epsilon^2)
			\end{aligned}
	\end{equation*}
	where the last inequality follows from \eqref{eq:dd}. 
	\subsubsection{Uniform sub-optimality of the AO cost away from the  AO equivalent solution. (Proof of \eqref{eq:resultat2})}
	The goal in this section is to show that the AO cost  is uniformly sub-optimum away from any ball containing the AO equivalent solution  with high probability. For that,  we  define function $\mathcal{L}^{\circ}({\bf x})$ as:
	\begin{equation}
	\mathcal{L}_{\lambda,\rho}^{\circ}({\bf x}):=\left(\sqrt{\delta}\sqrt{\frac{\|{\bf x}\|^2}{n}+\rho}-\frac{1}{n}{\bf h}^{T}{\bf x}\right)_{+}^2 +\frac{\lambda_2}{n}\|{\bf x}\|^2+\frac{\lambda_1}{n}\|{\bf x}\|_1. \label{eq:L0}
	\end{equation}
	It is easy to check:
	$$
	\mathcal{L}_{\lambda,\rho}^{\circ}({\bf x})=\sup_{\beta\geq 0} \ell_{\lambda,\rho}^{\circ}({\bf x},\beta).
	$$
	Hence, due to \eqref{eq:uniform2}, we thus obtain with probability $1-C\exp(-cn\epsilon^2)$, 
	\begin{equation}
		\mathcal{L}_{\lambda,\rho}^{\circ}({\bf x})+K\epsilon\geq \mathcal{L}_{\lambda,\rho}({\bf x})\geq \mathcal{L}_{\lambda,\rho}^{\circ}({\bf x}) -K\epsilon,  \ \ \ \forall {\bf x}\in B(0,\sqrt{P}\sqrt{n}). \label{eq:reference}
	\end{equation}
Hence, with probability at least $1-\frac{C}{\epsilon}\exp(-cn\epsilon^2)$,
\begin{equation}
\min_{\|{\bf x}\|_{\infty}\leq \sqrt{P}}\mathcal{L}_{\lambda,\rho}^\circ({\bf x})\geq \min_{\|{\bf x}\|_{\infty}\leq \sqrt{P}}\mathcal{L}_{\lambda,\rho}({\bf x}) -K\epsilon\geq \psi(\tau^\star,\beta^\star)-(K+\gamma_l)\epsilon. \label{eq:L0_rev}
\end{equation}
	Hence, the following inequalities hold true:
	\begin{equation}
		\mathcal{L}_{\lambda,\rho}^{\circ}(\overline{\bf x}^{\rm AO})+K\epsilon\geq \mathcal{L}_{\lambda,\rho}(\overline{\bf x}^{\rm AO})\geq \mathcal{L}_{\lambda,\rho}^{\circ}(\overline{\bf x}^{\rm AO}) -K\epsilon,  \label{eq:lll} 
	\end{equation}
	and
	\begin{equation}
		\min_{\substack{{\bf x}\\ x_i^2\leq P}}\mathcal{L}_{\lambda,\rho}^{\circ}({\bf x})+K\epsilon \geq \min_{\substack{{\bf x}\\ x_i^2\leq P}} \mathcal{L}_{\lambda,\rho}({\bf x}) \geq \min_{\substack{{\bf x}\\ x_i^2\leq P}} \mathcal{L}_{\lambda,\rho}^{\circ}({\bf x})-K\epsilon. \label{eq:lll1}
	\end{equation}
	It follows from \eqref{eq:min_dif} that:
	$$
	\min_{\substack{{\bf x}\\ x_i^2\leq P}} \mathcal{L}_{\lambda,\rho}({\bf x})\geq \mathcal{L}_{\lambda,\rho}(\overline{\bf x}^{\rm AO})-2\max(\gamma_u,\gamma_l)\epsilon.
	$$
	This together with \eqref{eq:lll} and \eqref{eq:lll1} yields:
	\begin{equation}
		\min_{\substack{{\bf x}\\ x_i^2\leq P}}\mathcal{L}_{\lambda,\rho}^{\circ}({\bf x})+ (2K+2\max(\gamma_u,\gamma_l))\epsilon\geq \mathcal{L}_{\lambda,\rho}^{\circ}(\overline{\bf x}^{\rm AO}). \label{eq:rr}
	\end{equation}
To continue, we shall prove the following Lemma:
\begin{lemma}
	Let $a=\frac{\sqrt{\delta} \beta^\star\rho}{4(\rho+P)^{\frac{3}{2}}}$. There exists a constant $r>0$ such that for all $\tilde{\epsilon}\in(0,\min(\frac{r^2a}{8}, 2K+2\max(\gamma_u,\gamma_l)))$, we obtain with probability $1-C\exp(-cn\tilde{\epsilon}^2)$, 	
	\begin{align}
		\forall {\bf x} \ \text{such that} \  \|{\bf x}\|_{\infty}\leq \sqrt{P} \  \text{and } \|{\bf x}-\overline{\bf x}^{\rm AO}\|^2\geq \frac{8n\tilde{\epsilon}}{a} &\Longrightarrow \mathcal{L}_{\lambda,\rho}^{\circ}({\bf x}) \geq \min_{\substack{{\bf x}\\ x_i^2\leq P}} \mathcal{L}_{\lambda,\rho}^{\circ}({\bf x}) + \tilde{\epsilon} \label{eq:proved}\\
		&\Longrightarrow \mathcal{L}_{\lambda,\rho}({\bf x}) \geq \min_{\substack{{\bf x}\\ x_i^2\leq P}} \mathcal{L}_{\lambda,\rho}({\bf x}) + \frac{\max(\gamma_u,\gamma_l)\tilde{\epsilon}}{K+\max(\gamma_u,\gamma_l)}. \label{eq:proved_bis}
\end{align}
\label{eq:lem_ball}
\end{lemma}
\begin{proof}
Consider function $f$ defined as:$$
	f:{\bf x}\mapsto \sqrt{\delta}\sqrt{\frac{\|{\bf x}\|^2}{n}+\rho}-\frac{1}{n}{\bf h}^{T}{\bf x}.
	$$
	From Lemma F.14 in \cite{Miolane},  function $f$ is $\frac{\sqrt{\delta}\rho}{n(\rho+P)^{\frac{3}{2}}}$-strongly convex on the ball $\mathcal{B}(0,\sqrt{nP}):=\{{\bf x}, \|{\bf x}\|\leq \sqrt{nP}\}$. Moreover, one can easily check that it is $\frac{2+2\sqrt{\delta}}{\sqrt{n}}$-Lipschitz on the event $\{\frac{\|{\bf h}\|}{\sqrt{n}}\leq 2\}$ which occurs with probability $1-Ce^{-cn}$.
	Hence, with probability $1-Ce^{-cn}$, $\forall {\bf x}\in \mathcal{B}(0,\sqrt{nP})$,
	\begin{equation}
		|f({\bf x})-f(\overline{\bf x}^{\rm AO})|\leq \frac{2+2\sqrt{\delta}}{\sqrt{n}}\|{\bf x}-\overline{\bf x}^{\rm AO}\| .\label{eq:lipsc1}
	\end{equation}
	From the derivations of the previous section, we can easily check that: $$
	f(\overline{\bf x}^{\rm AO})\geq \frac{\beta^\star}{4}
	$$
	with probability $1-C\exp(-cn)$. Hence, using \eqref{eq:lipsc1}, there exists $r$ such that function $f$ satisfies:
	$$
	f({\bf x})\geq \frac{\beta^\star}{8}, \ \ \forall {\bf x}\in B(\overline{\bf x}^{\rm AO}, r\sqrt{n}).
	$$
	On the ball $B(\overline{\bf x}^{\rm AO},r\sqrt{n})$, function $\mathcal{L}_{\lambda,\rho}^{\circ}({\bf x})$ can be shown to be $\frac{a}{n}$-strongly convex. As a matter of fact, it can be written as:
	$$
	\mathcal{L}_{\lambda,\rho}^{\circ}({\bf x})=  (f({\bf x}))_{+}^2 +\frac{\lambda_2}{n}\|{\bf x}\|^2+\frac{\lambda_1}{n}\|{\bf x}\|_1.
	$$
	Hence,  computing the Hessian of $\mathcal{L}_{\lambda,\rho}^{\circ}({\bf x})$
	on $B(\overline{\bf x}^{\rm AO},r\sqrt{n})$, we obtain:
	$$
	\nabla^2 \mathcal{L}_{\lambda,\rho}^{\circ}\succeq 2 f({\bf x})\nabla^2 f({\bf x}) + 2\nabla f({\bf x})\nabla f({\bf x}) \succeq   \frac{\sqrt{\delta}\beta^\star\rho}{4n(\rho+P)^{\frac{3}{2}}}{\bf I}_n=\frac{a}{n}{\bf I}_n.
	$$
	Function $\mathcal{L}_{\lambda,\rho}^{c}$, which is ${\bf x}:\mapsto  \mathcal{L}_{\lambda,\rho}^{\circ} +\sup_{{\bf b}_{+}\geq 0} {\bf b}_{+}^{T}({\bf x}-\sqrt{P}{\bf 1}_n)+\sup_{{\bf b}_{-}\geq 0} {\bf b}_{-}^{T}(-{\bf x}-\sqrt{P}{\bf 1}_n)$ is also strongly convex. 
 It follows from \eqref{eq:rr} that for all ${\epsilon}\in(0, 1)$ with probability $1-C\exp(-cn{\epsilon}^2)$, 
$$
\mathcal{L}^{\circ}_{\lambda,\rho}(\overline{\bf x}^{\rm AO})\leq \min_{\|{\bf x}\|_{\infty}\leq \sqrt{P}} \mathcal{L}_{\lambda,\rho}^{\circ}({\bf x})+(2K+2\max(\gamma_u,\gamma_l))\epsilon.
$$
Let $\tilde{\epsilon}=(2K+2\max(\gamma_u,\gamma_l))\epsilon$.
Then, applying Lemma B.1 in \cite{Miolane} to function $\mathcal{L}_{\lambda,\rho}^{c}$, \eqref{eq:proved} directly follows. Finally \eqref{eq:proved_bis} directly follows from using \eqref{eq:lll1}.
	\end{proof}
 The result in Lemma \ref{eq:lem_ball} can be reformulated as follows. For $\epsilon\in(0,1]$ There exists $C, c$ and $\gamma$ positive constants such that with probability $1-\frac{C}{\epsilon}\exp(-cn\epsilon^2)$, for all feasible ${\bf x}$ such that:
	$$
	\frac{1}{n}\|{\bf x}-\overline{\bf x}^{\rm AO}\|^2\geq \epsilon  \ \ \ \Longrightarrow \mathcal{L}_{\lambda,\rho}({\bf x})\geq  \psi(\tau^\star,\beta^\star)+\gamma\epsilon.
	$$
	In other words, for all $\epsilon\in(0,1]$ 
	\begin{equation}
	\underset{{\bf s}\in \mathcal{S}^{\otimes m}}{\sup}	\mathbb{P}\left[\exists \ \ \text{feasible} \  {\bf x} \  \text{such that } \frac{1}{n}\|{\bf x}-\overline{\bf x}^{\rm AO}\|^2\geq \epsilon \ \ \text{and} \ \ \mathcal{L}_{\lambda,\rho}({\bf x})\leq \psi(\tau^\star,\beta^\star)+\gamma\epsilon \right]\leq \frac{C}{\epsilon}\exp(-cn\epsilon^2) .\label{eq:exi}
	\end{equation}
	For the sake of the  proofs, we may need later to control probabilities of events in the same form of \eqref{eq:exi} but with $\epsilon$ in the inequality $\frac{1}{n}\|{\bf x}-\overline{\bf x}^{\rm AO}\|^2\geq \epsilon$ being replaced by $\tilde{\lambda}\epsilon^2$ for some positive constant $\tilde{\lambda}\in(0,1]$ independent of $\epsilon$. By performing the change of variable $\epsilon \leftrightarrow \tilde{\lambda} \epsilon^2$ and absorbing the incurred constants into $C$ and $c$, we can also state the following result: 
	\begin{equation}
	\underset{{\bf s}\in \mathcal{S}^{\otimes m}}{\sup}	\mathbb{P}\left[\exists \ \ \text{feasible} \  {\bf x} \  \text{such that } \frac{1}{n}\|{\bf x}-\overline{\bf x}^{\rm AO}\|^2\geq \tilde{\lambda}\epsilon^2 \ \ \text{and} \ \ \mathcal{L}_{\lambda,\rho}({\bf x})\leq \psi(\tau^\star,\beta^\star) +\gamma\tilde{\lambda}\epsilon^2 \right]\leq \frac{C}{\epsilon^2}\exp(-cn\epsilon^4) \label{eq:exi1}
	\end{equation}
	which holds for any sufficiently small positive $\epsilon$.
 \subsection{From the AO to the PO: control of the functions of the precoder solution.}
	\label{app:function_f}
	Let $f:\mathbb{R}^{n}\to \mathbb{R}$ be a pseudo-Lipschitz function of order $r$. Consider the set $D_\epsilon$ defined as:
	$$
	D_{\epsilon}:=\left\{ {\bf x} \ \text{feasible}, \ \ | f(\frac{1}{\sqrt{n}}{\bf x})- \mathbb{E}[f(\frac{1}{\sqrt{n}}\overline{\bf x}^{\rm AO})]|\geq \epsilon\right\}.
	$$
	Our goal is to prove that with probability $1-\frac{C}{\epsilon^2}\exp(-cn\epsilon^4)$,  $\hat{\bf x}^{\rm PO}\notin D_{\epsilon}$. Obviously, this can be equivalently stated as:
	\begin{equation}
	\underset{{\bf s}\in \mathcal{S}^{\otimes m}}{\sup}	\mathbb{P}\left[\min_{{\bf x}\in D_\epsilon} \mathcal{C}_{\lambda,\rho}({\bf x})\leq  \min_{\substack{{\bf x}\\ x_i^2\leq P}} \mathcal{C}_{\lambda,\rho}({\bf x})+\gamma\alpha\epsilon^2 \right] \leq\frac{C}{\epsilon^2}\exp(-cn\epsilon^4). \label{eq:des}
	\end{equation}
	A key step in proving \eqref{eq:des} is to demonstrate that, with probability approaching one (specifically, $1-C\exp(-cn\epsilon^2)$), the following holds in the random space of the AO:
	\begin{equation}
	{\bf x}\in D_{\epsilon} \Longrightarrow \frac{1}{n}\|{\bf x}-\overline{\bf x}^{\rm AO}\|\geq \tilde{\lambda}\epsilon^2 \label{eq:es}
	\end{equation}
	for some $\tilde{\lambda}>0$. 
	Since
		\begin{align}
		&\underset{{\bf s}\in\mathcal{S}^{\otimes m}}{\sup}\mathbb{P}\left[\min_{{\bf x}\in D_\epsilon} \mathcal{C}_{\lambda,\rho}({\bf x})\leq  \min_{\|{\bf x}\|_{\infty}\leq \sqrt{P}} \mathcal{C}_{\lambda,\rho}({\bf x})+\gamma\frac{\tilde{\lambda}}{2}\epsilon^2 \right]\nonumber\\
		\leq&\underset{{\bf s}\in\mathcal{S}^{\otimes m}}{\sup}\mathbb{P}\left[\left\{\min_{{\bf x}\in D_\epsilon} \mathcal{C}_{\lambda,\rho}({\bf x})\leq  \min_{\|{\bf x}\|_{\infty}\leq\sqrt{P}} \mathcal{C}_{\lambda,\rho}({\bf x})+\gamma\frac{\tilde{\lambda}}{2}\epsilon^2\right\}\cap\left\{\min_{\|{\bf x}\|_{\infty}\leq \sqrt{P}}\mathcal{C}_{\lambda,\rho}({\bf x})\leq \psi(\tau^\star,\beta^\star)+\gamma\frac{\tilde{\lambda}}{2}\epsilon^2\right\} \right]\nonumber\\
		&+\mathbb{P}\Big[\min_{\|{\bf x}\|_{\infty}\leq \sqrt{P}}\mathcal{C}_{\lambda,\rho}({\bf x})\geq \psi(\tau^\star,\beta^\star)+\gamma\frac{\tilde{\lambda}}{2}\epsilon^2\Big] \label{eq:first_1}\\
		\leq &\underset{{\bf s}\in\mathcal{S}^{\otimes m}}{\sup}\mathbb{P}\left[\min_{{\bf x}\in D_\epsilon} \mathcal{C}_{\lambda,\rho}({\bf x})\leq  \psi(\tau^\star,\beta^\star)+\gamma\tilde{\lambda}\epsilon^2\right]+\mathbb{P}\Big[\min_{\|{\bf x}\|_{\infty}\leq \sqrt{P}}\mathcal{C}_{\lambda,\rho}({\bf x})\geq \psi(\tau^\star,\beta^\star)+\gamma\frac{\tilde{\lambda}}{2}\epsilon^2\Big]\nonumber\\
		\leq&\underset{{\bf s}\in\mathcal{S}^{\otimes m}}{\sup} 2\mathbb{P}\left[\min_{{\bf x}\in D_\epsilon} \mathcal{L}_{\lambda,\rho}({\bf x})\leq  \psi(\tau^\star,\beta^\star)+\gamma\tilde{\lambda}\epsilon^2\right]+2\mathbb{P}\Big[\min_{\|{\bf x}\|_{\infty}\leq \sqrt{P}}\mathcal{L}_{\lambda,\rho}({\bf x})\geq \psi(\tau^\star,\beta^\star)+\gamma\frac{\tilde{\lambda}}{2}\epsilon^2\Big]+C\exp(-cn)\label{eq:last_1}\\
		\leq&\underset{{\bf s}\in\mathcal{S}^{\otimes m}}{\sup} 2\mathbb{P}\left[\min_{{\bf x}\in D_\epsilon} \mathcal{L}_{\lambda,\rho}({\bf x})\leq  \psi(\tau^\star,\beta^\star)+\gamma\tilde{\lambda}\epsilon^2\right]+C\exp(-cn\epsilon^4), \label{eq:llas}
	\end{align}
	where to obtain \eqref{eq:last_1}, we used \eqref{eq:1} and \eqref{eq:31} and to obtain \eqref{eq:llas}, we used \eqref{eq:dd}.
	
	Now if \eqref{eq:es} holds true, then:
	\begin{align*}
	&\underset{{\bf s}\in\mathcal{S}^{\otimes m}}{\sup} 2\mathbb{P}\left[\min_{{\bf x}\in D_\epsilon} \mathcal{L}_{\lambda,\rho}({\bf x})\leq  \psi(\tau^\star,\beta^\star)+\gamma\tilde{\lambda}\epsilon^2\right]\\
	\leq& \underset{{\bf s}\in\mathcal{S}^{\otimes m}}{\sup} 2\mathbb{P}\left[\exists {\bf x } \ \text{such that } \frac{1}{n}\|{\bf x}-\overline{\bf x}^{\rm AO}\|^2\geq \tilde{\lambda}\epsilon^2  \ \text{and } \mathcal{L}_{\lambda,\rho}({\bf x})\leq  \psi(\tau^\star,\beta^\star)+\gamma\tilde{\lambda}\epsilon^2\right]+C\exp(-cn\epsilon^2)\\
	\leq &\frac{C}{\epsilon^2}\exp(-cn\epsilon^4),
	\end{align*}
	which proves \eqref{eq:des}.
	
	Hence, to complete the proof, it suffices thus to show \eqref{eq:es} on event 
	$$
	\mathcal{E}:=\left\{| f(\frac{1}{\sqrt{n}}\overline{\bf x}^{\rm AO})-\mathbb{E}\left[f(\frac{1}{\sqrt{n}}\overline{\bf x}^{\rm AO})\right]|\leq \frac{\epsilon}{2}\right\}
	$$ 
	and to demonstrate that $\mathcal{E}$ occurs   with probability $1-C\exp(-cn\epsilon^2)$. 
	
	 Denote by $D_{\epsilon_1}$ the set:
	\begin{align*}
		D_{\epsilon_1}&:=\left\{ {\bf x} \ \text{feasible}, \ \ | f(\frac{1}{\sqrt{n}}{\bf x})- f(\frac{1}{\sqrt{n}}\overline{\bf x}^{\rm AO})|\geq \frac{\epsilon}{2}\right\}.
	\end{align*} 
Using the triangular inequality:
$$
| f(\frac{1}{\sqrt{n}}{\bf x})- f(\frac{1}{\sqrt{n}}\overline{\bf x}^{\rm AO})|\geq | f(\frac{1}{\sqrt{n}}{\bf x})- \mathbb{E}[f(\frac{1}{\sqrt{n}}\overline{\bf x}^{\rm AO})]| -| f(\frac{1}{\sqrt{n}}\overline{\bf x}^{\rm AO})-\mathbb{E}\left[f(\frac{1}{\sqrt{n}}\overline{\bf x}^{\rm AO})\right]| ,
$$
we see that on the event $\mathcal{E}$, 
$$
{\bf x}\in D_\epsilon \Longrightarrow {\bf x}\in D_{\epsilon_1}.
$$
Function $f$ is pseudo-lipschitz with Lipschitz constant equal to $f_r$. Hence, for any ${\bf x}$ and ${\bf y}$ with $\|{\bf x}\|_{\infty}\leq \sqrt{P}$ and $\|{\bf y}\|_{\infty}\leq \sqrt{P}$, 
\begin{equation}
|f(\frac{1}{\sqrt{n}}{\bf x})-f(\frac{1}{\sqrt{n}}{\bf y})|\leq f_r(1+\|\frac{1}{\sqrt{n}}{\bf x}\|^{r-1}+\|\frac{1}{\sqrt{n}}{\bf y}\|^{r-1})\frac{1}{\sqrt{n}}\|{\bf x}-{\bf y}\|\leq f_r(1+2\sqrt{P})\frac{1}{\sqrt{n}}\|{\bf x}-{\bf y}\| .\label{eq:Lipschitz}
\end{equation}
	For all ${\bf x}\in D_{\epsilon_1}$, we thus obtain:
\begin{align*}
	\left|f(\frac{1}{\sqrt{n}}{\bf x})-f(\frac{1}{\sqrt{n}}\overline{\bf x}^{\rm AO})\right|	\geq \frac{\epsilon}{2} &\Rightarrow \frac{1}{\sqrt{n}}\|{\bf x}-\overline{\bf x}^{\rm AO}\|\geq \frac{\epsilon}{2f_r{(1+2\sqrt{P}^{r-1})}} \\
	& \Rightarrow \frac{1}{n} \|{\bf x}-\overline{\bf x}^{\rm AO}\|^2\geq \frac{\epsilon^2}{4f_r^2(1+2\sqrt{P}^{r-1})^{2}},
\end{align*}
which proves \eqref{eq:es} for $\tilde{\lambda}=\min(\frac{1}{4f_r^2(1+2\sqrt{P}^{r-1})^2},1)$. Finally, to complete the proof, we need to check that the event $\mathcal{E}$ occurs with probability $1-C\exp(-cn\epsilon^2)$. To see this, note that from \eqref{eq:Lipschitz}, the function ${\bf x}\mapsto f(\frac{1}{\sqrt{n}}{\bf x})$ is Lipschitz with Lipschitz constant equal to $\frac{f_r}{\sqrt{n}}(1+2\sqrt{P})$. Recall that $\overline{\bf x}^{\rm AO}$ is a $\tilde{\tau}^\star$ Lipschitz function of ${\bf h}$. Hence ${\bf h}\mapsto f(\frac{1}{\sqrt{n}}\overline{\bf x}^{\rm AO})$ is a $\frac{\tilde{\tau}^\star f_r}{\sqrt{n}}(1+2\sqrt{P})$ Lipschitz function of ${\bf h}$. Using Gaussian concentrations for Lipschtiz functions, we thus obtain:
$$
\mathbb{P}\left[\Big|f(\frac{1}{\sqrt{n}}\overline{\bf x}^{\rm AO})-\mathbb{E}\Big[f(\frac{1}{\sqrt{n}}\overline{\bf x}^{\rm AO})\Big]\Big|\geq \frac{\epsilon}{2}\right]\leq C\exp(-cn\epsilon^2).
$$

\subsection{From the AO to the PO: convergence of the empirical measure of the precoder solution.} 
\label{app:measure}
Consider the set $D_\epsilon$ defined as:
$$
D_\epsilon:=\left\{{\bf x} \ \text{feasible}, \ \ \mathcal{W}_r(\hat{\mu}({\bf x}),\nu^\star)^r\geq \epsilon\right\}.
$$
Similar to the previous analysis, to show the desired result, it suffices to check that there exists $\alpha>0$ such that
$$
\mathbb{P}\left[\min_{{\bf x}\in D_\epsilon} \mathcal{C}_{\lambda,\rho}({\bf x})\leq \min_{\substack{{\bf x}\\ \|{\bf x}\|_{\infty}\leq \sqrt{P}}} \mathcal{C}_{\lambda,\rho}({\bf x}) +\gamma\alpha\epsilon^2\right]\leq \frac{C}{\epsilon^2}\exp(-cn\epsilon^4).
$$
For that let $\tilde{\lambda}$ be some positive constant in $(0,1)$, and consider the event:
$$\mathcal{E}:=\left\{\mathcal{W}_r(\hat{\mu}(\overline{\bf x}^{\rm AO}),\nu^\star)^r\leq (1-\sqrt{\tilde{\lambda}})^r\epsilon\right\}.$$
Since the support of $\nu^\star$ is bounded, $\nu^\star$ satisfies the assumption of  Lemma \ref{lem:convergence_empirical_rate}. Hence, the event $\mathcal{E}$
occurs with probability $1-C\exp(-cn\epsilon^2)$. Moreover, on the event $\mathcal{E}$, 
we have:
\begin{equation}
	\forall {\bf x}\in D_\epsilon, \ \ \frac{1}{N}\|{\bf x}-\overline{\bf x}^{\rm AO}\|_r^r\geq \big(\mathcal{W}_r(\hat{\mu}({\bf x}),\hat{\mu}(\overline{\bf x}^{\rm AO}))\big)^r\geq \big(\mathcal{W}_r(\hat{\mu}({\bf x}),\nu^\star)-\mathcal{W}_r(\hat{\mu}(\overline{\bf x}^{\rm AO}),\nu^\star) \big)^r\geq \tilde{\lambda}^{\frac{r}{2}} \epsilon.\label{eq:last_form}
\end{equation}
If $r\in[1,2)$, using the relation:
$$
\|{\bf x}\|_r^r\leq \|{\bf x}\|_2^r N^{1-\frac{r}{2}}
$$
which holds for any vector ${\bf x}\in\mathbb{R}^{N}$, \eqref{eq:last_form} implies that:
$$
\forall {\bf x}\in D_\epsilon, \ \left(\frac{1}{\sqrt{N}}\|{\bf x}-\overline{\bf x}^{\rm AO}\|_2\right)^r\geq \tilde{\lambda}^{\frac{r}{2}}\epsilon
$$
or equivalently:
$$\forall {\bf x}\in D_\epsilon, \ \  \frac{1}{N}\|{\bf x}-\overline{\bf x}^{\rm AO}\|_2^2\geq \tilde{\lambda}\epsilon^{\frac{2}{r}}\geq \tilde{\lambda} \epsilon^2. 
$$
Similarly, if $r\in[2,\infty)$, using the fact that $\|{\bf x}\|_2^r\geq \|{\bf x}\|_r^{r}$, we get also, on the event $\mathcal{E}$,
$$
\forall {\bf x}\in D_\epsilon, \  \frac{1}{N}\|{\bf x}-\overline{\bf x}^{\rm AO}\|_2^2\geq \tilde{\lambda}\epsilon^2.
$$
This shows that with probability greater than $1-C\exp(-cn\epsilon^2)$, for all $r\geq 1$,
\begin{equation}
\forall {\bf x}\in D_\epsilon, \  \frac{1}{N}\|{\bf x}-\overline{\bf x}^{\rm AO}\|_2^2\geq \tilde{\lambda}\epsilon^2. \label{eq:aoo}
\end{equation}
Take $\alpha=\frac{\tilde{\lambda}}{2}$. Then, based on the same calculations as in \eqref{eq:first_1}-\eqref{eq:last_1}, we obtain
\begin{align}
	&\underset{{\bf s}\in\mathcal{S}^{\otimes m}}{\sup}\mathbb{P}\left[\min_{{\bf x}\in D_\epsilon} \mathcal{C}_{\lambda,\rho}({\bf x})\leq \min_{\substack{{\bf x}\\ \|{\bf x}\|_{\infty}\leq \sqrt{P}}} \mathcal{C}_{\lambda,\rho}({\bf x}) +\gamma\alpha\epsilon^2\right]\nonumber\\
	\leq &2\underset{{\bf s}\in\mathcal{S}^{\otimes m}}{\sup}\mathbb{P}\Big[\min_{{\bf x}\in D_\epsilon} \mathcal{L}_{\lambda,\rho}({\bf x})\leq \psi(\tau^\star,\beta^\star)+\gamma\tilde{\lambda}\epsilon^2\Big]+\frac{C}{\epsilon^2}\exp(-cn\epsilon^4) \nonumber\\
	\leq &2\underset{{\bf s}\in\mathcal{S}^{\otimes m}}{\sup}\mathbb{P}\Big[\exists \ {\bf x} \ \text{such that} \ \frac{1}{N}\|{\bf x}-\overline{\bf x}^{\rm AO}\|^2\geq \tilde{\lambda}\epsilon^2 \ \text{and} \ \mathcal{L}_{\lambda,\rho}({\bf x})\leq \psi(\tau^\star,\beta^\star)+\gamma\tilde{\lambda}\epsilon^2\Big]+\frac{C}{\epsilon^2}\exp(-cn\epsilon^4)\label{eq:alpha_1}\\\leq& \frac{C}{\epsilon^2}\exp(-cn\epsilon^4)\label{eq:alpha_2}
\end{align}
where  \eqref{eq:alpha_1} follows from \eqref{eq:aoo} and \eqref{eq:alpha_2} from \eqref{eq:exi1}. 

\subsection{From the AO to the PO: number of elements in an interval.}
\subsubsection{Thresholding parameter}\label{ix_a}
For a vector ${\bf x}\in\mathbb{R}^{n}$ and $t>0$, we define 
$$
\pi_{(t)}(|{\bf x}|)=\{i=1,\cdots,n, |[{\bf x}]_i|\leq  t \}
$$
and 
$$
\pi_{(t)}^c(|{\bf x}|)=\{i=1,\cdots,n, |[{\bf x}]_i|>  t \}.
$$The goal of this section is to prove the following theorem:
\begin{theorem}
	Under assumption \ref{ass:regime}, \ref{ass:statistic} and \ref{ass:regime_lambda}, there exists constants $C$, $c$ and $\tilde{\lambda}>0$ such that:
	\begin{align}
		\mathbb{P}\Big[\exists \  {\bf x}, \  |\frac{\#\pi_{(t)}^c(|{\bf x}|)}{n}-2Q(\frac{t}{\tilde{\tau}^\star}+\frac{\lambda_1}{\beta^\star})|\geq \epsilon,  \ \mathcal{C}_{\lambda,\rho}({\bf x})\leq \min_{\|{\bf x}\|_{\infty}\leq \sqrt{P}}\mathcal{C}_{\lambda,\rho}({\bf x})+\gamma\tilde{\lambda}\epsilon^3\Big]\leq \frac{C}{\epsilon^3}\exp(-cn\epsilon^6). \label{eq:threshold}
	\end{align}
	\label{th:app_threshold}
	\end{theorem}
Before delving into the proof, it is worth noting that  \eqref{eq:threshold} can be implied by:
\begin{align*}
	&\mathbb{P}\big[\frac{\#\pi_{(t)}^c(|\hat{\bf x}_{\ell_1}|)}{n}\leq 2Q(\frac{t}{\tilde{\tau}^\star}+\frac{\lambda_1}{\beta^\star})-\epsilon\big]\leq \frac{C}{\epsilon^3}\exp(-cn\epsilon^6), \\
	& \mathbb{P}\big[\frac{\#\pi_{(t)}^c(|\hat{\bf x}_{\ell_1}|)}{n}\geq 2Q(\frac{t}{\tilde{\tau}^\star}+\frac{\lambda_1}{\beta^\star})+\epsilon\big]\leq \frac{C}{\epsilon^3}\exp(-cn\epsilon^6),
\end{align*}
which establishes Theorem \ref{th:threshold}. 
\begin{proof}
To begin with, for $r>0$, we define the following random quantities:
\begin{align*}
	\pi_{(t+r)}^c(|\overline{\bf x}^{\rm AO}|)=\{i=1,\cdots,n, |[\overline{\bf x}^{\rm AO}]_i|\geq  t+r \},\\
	\pi_{(t-r)}^c(|\overline{\bf x}^{\rm AO}|)=\{i=1,\cdots,n, |[\overline{\bf x}^{\rm AO}]_i|\geq  t-r \},\end{align*}
and compute the following expectations:
\begin{align*}\mathbb{E}\big[\frac{\#\pi_{(t+r)}^c(|\overline{\bf x}^{\rm AO}|)}{n}\big]&=2Q(\frac{t+r}{\tilde{\tau}^\star}+\frac{\lambda_1}{\beta^\star}) \geq 2Q(\frac{t_x}{\tilde{\tau}^\star}+\frac{\lambda_1}{\beta^\star})-\frac{r}{\tilde{\tau}^\star},\\
	\mathbb{E}\big[	\frac{\#\pi_{(t-r)}^c(|\overline{\bf x}^{\rm AO}|)}{n}\big]&=2Q(\frac{t-r}{\tilde{\tau}^\star}+\frac{\lambda_1}{\beta^\star})\leq 2Q(\frac{t_x}{\tilde{\tau}^\star}+\frac{\lambda_1}{\beta^\star})+\frac{r}{\tilde{\tau}^\star}.
\end{align*}
It follows from Hoeffding inequality that:
\begin{align*}\mathbb{P}\left[\frac{\#\pi_{(t+r)}^c(|\overline{\bf x}^{\rm AO}|)}{n}\geq \frac{\mathbb{E}\big[\#\pi_{(t+r)}^c(\overline{\bf x}^{\rm AO})\big]}{n}-\frac{\epsilon}{4}\right]\geq 1- \exp(-\frac{1}{8}n\epsilon^2),\\
	\mathbb{P}\left[\frac{\#\pi_{(t-r)}^c(|\overline{\bf x}^{\rm AO}|)}{n}\geq \frac{\mathbb{E}\big[\#\pi_{(t-r)}^c(|\overline{\bf x}^{\rm AO}|)\big]}{n}-\frac{\epsilon}{4}\right]\geq 1- \exp(-\frac{1}{8}n\epsilon^2).
\end{align*}
By choosing $r=\frac{\epsilon\tilde{\tau}^{\star}}{4}$, we obtain that the following events
\begin{align*}
	\mathcal{E}_1&=\left\{\frac{\#\pi_{(t+r)}^c(|\overline{\bf x}^{\rm AO}|)}{n}\geq 2Q(\frac{t_x}{\tilde{\tau}^\star}+\frac{\lambda_1}{\beta^\star})-\frac{\epsilon}{2}\right\},\\
	\mathcal{E}_2&= \left\{\frac{\#\pi_{(t-r)}^c(|\overline{\bf x}^{\rm AO}|)}{n}\leq 2Q(\frac{t_x}{\tilde{\tau}^\star}+\frac{\lambda_1}{\beta^\star})+\frac{\epsilon}{2}\right\},
\end{align*}
occur each with probability at least $1-\exp(-\frac{1}{8}n\epsilon^2)$.

Now, define the following sets:
\begin{align*}
	D_\epsilon^{-}&=\left\{{\bf x} \ \text{feasible},  \frac{\#\pi_{(t)}^c(|{\bf x}|)}{n}\leq 2Q(\frac{t_x}{\tilde{\tau}^\star}+\frac{\lambda_1}{\beta^\star})-\epsilon  \right\},\\
	D_\epsilon^{+}&=\left\{{\bf x} \ \text{feasible}, \frac{\#\pi_{(t)}^c(|{\bf x}|)}{n} \geq 2Q(\frac{t_x}{\tilde{\tau}^\star}+\frac{\lambda_1}{\beta^\star})+\epsilon  \right\}.
\end{align*}
The basic idea is to show that with probability $1-\frac{C}{\epsilon^3}\exp(-cn\epsilon^6)$, $\hat{\bf x}_{\ell_1}\notin D_\epsilon^{-}\cup D_\epsilon^{+}$. This can be reduced to showing that:
\begin{align}
	\mathbb{P}\left[\hat{\bf x}_{\ell_1}\in D_\epsilon^{-}\right]&\leq \frac{C}{\epsilon^3}\exp(-cn\epsilon^6),\nonumber\\
	\mathbb{P}\left[\hat{\bf x}_{\ell_1}\in D_\epsilon^{+}\right]&\leq \frac{C}{\epsilon^3}\exp(-cn\epsilon^6). \label{eq:second_ineq}
\end{align}
To begin with, we consider the AO problem and prove that with probability $1-C\exp(-cn\epsilon^2)$, 
$$
\forall {\bf x}\in D_\epsilon^{-}, \ \  \|{\bf x}-\overline{\bf x}^{\rm AO}\|^2\geq \frac{n\epsilon^3\tilde{\tau}^2}{32}.
$$
Indeed, on the event $\mathcal{E}_1$, vectors ${\bf x}\in D_\epsilon^{-}$ and $\overline{\bf x}^{\rm AO}$ are apart from each other by $r=\frac{\tilde{\tau}^\star\epsilon}{4}$ in absolute value in at least $n\frac{\epsilon}{2}$ positions. Hence, with probability $1-\exp(-\frac{1}{8}n\epsilon^2)$,
$$
\forall {\bf x}\in D_\epsilon^{-}, \ \  \|{\bf x}-\overline{\bf x}^{\rm AO}\|^2\geq \frac{n\epsilon^3\tilde{\tau}^2}{32}.
$$
Set $\tilde{\lambda}=\min(\frac{\tilde{\tau}^2}{32},1)$
 By following the same calculations as in \eqref{eq:first_1}-\eqref{eq:last_1}, we have 
\begin{align*}
	\mathbb{P}\Big[\min_{{\bf x}\in D_\epsilon^{-}}\mathcal{C}_{\lambda,\rho}({\bf x})\leq \min_{\|{\bf x}\|_{\infty}\leq \sqrt{P} }\mathcal{C}_{\lambda,\rho}({\bf x}) +\gamma\frac{\tilde{\lambda}}{2}\epsilon^3\Big]\leq 	2\mathbb{P}\Big[\min_{{\bf x}\in D_\epsilon^{-}}\mathcal{L}_{\lambda,\rho}({\bf x})\leq \psi(\tau^\star,\beta^\star) +\gamma\tilde{\lambda}\epsilon^3 \Big]+\frac{C}{\epsilon^3}\exp(-cn\epsilon^6).
\end{align*}
Using \eqref{eq:exi1}, we thus obtain
$$
\mathbb{P}\Big[\min_{{\bf x}\in D_\epsilon^{-}}\mathcal{L}_{\lambda,\rho}({\bf x})\leq \psi(\tau^\star,\beta^\star) +\gamma\tilde{\lambda}\epsilon^3 \Big]\leq \frac{C}{\epsilon^3}\exp(-cn\epsilon^6).
$$
A similar argument applies to show \eqref{eq:second_ineq}. For the sake of brevity, we omit further details. 
\end{proof}
\subsubsection{Number of elements in an interval}
For a vector ${\bf x}$ and $0<a<b<\sqrt{P}$, we define
$$
\pi_{(a,b)}({\bf x}):=\#\{i=1,\cdots,n [{\bf x}]_i\in(a,b)\}
$$
In this section, we consider proving the following theorem:
\begin{theorem}
	Under assumption \ref{ass:regime}, \ref{ass:statistic} and \ref{ass:regime_lambda}, there exists constants $C$, $c$ and $\tilde{\lambda}>0$ such that:
	\begin{align*}
		\mathbb{P}\Big[\exists \  {\bf x}, \  |\frac{\pi_{(a,b)}(\hat{\bf x}_{\ell_1})}{n}- (Q(\frac{a}{\tilde{\tau}^\star}+\frac{\lambda_1}{\beta^\star})-Q(\frac{b}{\tilde{\tau}^\star}+\frac{\lambda_1}{\beta^\star}))|\geq \epsilon,  \ \mathcal{C}_{\lambda,\rho}({\bf x})\leq \psi(\tau^\star,\beta^\star)+\gamma\tilde{\lambda}\epsilon^3\Big]\leq \frac{C}{\epsilon^3}\exp(-cn\epsilon^6). 
	\end{align*}
	\label{th:interval}
\end{theorem}
\begin{proof}
To begin with, for $r>0$, we define the following random quantities:
\begin{align*}
\pi_{(a+r,b-r)}(\overline{\bf x}^{\rm AO})	&:=\{i=1,\cdots,n [\overline{\bf x}^{\rm AO}]_i\in(a+r,b-r)\},\\
\pi_{(a-r,b+r)}(\overline{\bf x}^{\rm AO})&:=\{i=1,\cdots,n [\overline{\bf x}^{\rm AO}]_i\in(a-r,b+r)\},
\end{align*}
and compute the following expectations:
\begin{align*}
	\mathbb{E}\Big[\frac{\#\pi_{(a+r,b-r)}(\overline{\bf x}^{\rm AO})}{n}\Big]&=Q(\frac{a+r}{\tilde{\tau}^\star}+\frac{\lambda_1}{\beta^\star})-Q(\frac{b-r}{\tilde{\tau^\star}}+\frac{\lambda_1}{\beta^\star})\geq Q(\frac{a}{\tilde{\tau}^\star}+\frac{\lambda_1}{\beta^\star})-Q(\frac{b}{\tilde{\tau^\star}}+\frac{\lambda_1}{\beta^\star})-\frac{r}{\tilde{\tau}^\star},\\
	\mathbb{E}\Big[\frac{\#\pi_{(a-r,b+r)}(\overline{\bf x}^{\rm AO})}{n}\Big]&=Q(\frac{a-r}{\tilde{\tau}^\star}+\frac{\lambda_1}{\beta^\star})-Q(\frac{b+r}{\tilde{\tau^\star}}+\frac{\lambda_1}{\beta^\star})\leq Q(\frac{a}{\tilde{\tau}^\star}+\frac{\lambda_1}{\beta^\star})-Q(\frac{b}{\tilde{\tau^\star}}+\frac{\lambda_1}{\beta^\star})+\frac{r}{\tilde{\tau}^\star} .
\end{align*}
It follows from Hoeffding's inequality that:
\begin{align*}
	&\mathbb{P}\Big[\frac{\#\pi_{(a+r,b-r)}(\overline{\bf x}^{\rm AO})}{n}\geq \mathbb{E}\Big[\frac{\#\pi_{(a+r,b-r)}(\overline{\bf x}^{\rm AO})}{n}\Big]-\frac{\epsilon}{4} \Big]\geq 1-\exp(-\frac{1}{8}n\epsilon^2)\\
	&\mathbb{P}\Big[\frac{\#\pi_{(a-r,b+r)}(\overline{\bf x}^{\rm AO})}{n}\leq \mathbb{E}\Big[\frac{\#\pi_{(a-r,b+r)}(\overline{\bf x}^{\rm AO})}{n}\Big]+\frac{\epsilon}{4} \Big]\geq 1-\exp(-\frac{1}{8}n\epsilon^2)
\end{align*}
By choosing $r=\frac{\tilde{\tau}^\star\epsilon}{4}$, we obtain that the following events
\begin{align*}
	\mathcal{E}_1&=\left\{\frac{\#\pi_{(a+r,b-r)}(\overline{\bf x}^{\rm AO})}{n}\geq Q(\frac{a}{\tilde{\tau}^\star}+\frac{\lambda_1}{\beta^\star})-Q(\frac{b}{\tilde{\tau^\star}}+\frac{\lambda_1}{\beta^\star})-\frac{\epsilon}{2}\right\},\\
	\mathcal{E}_2&= \left\{\frac{\#\pi_{(a-r,b+r)}(\overline{\bf x}^{\rm AO})}{n}\leq Q(\frac{a}{\tilde{\tau}^\star}+\frac{\lambda_1}{\beta^\star})-Q(\frac{b}{\tilde{\tau^\star}}+\frac{\lambda_1}{\beta^\star})+\frac{\epsilon}{2}\right\},
\end{align*}
occur with probability at least $1-\frac{C}{\epsilon^3}\exp(-cn\epsilon^6)$. 
Now, define the following sets:
\begin{align*}
	D_{\epsilon}^{-}&=\left\{{\bf x} \ | \ \|{\bf x}\|_{\infty}\leq \sqrt{P}, \frac{\#\pi_{(a,b)}({\bf x})}{n}\leq Q(\frac{a}{\tilde{\tau}^\star}+\frac{\lambda_1}{\beta^\star})-Q(\frac{b}{\tilde{\tau^\star}}+\frac{\lambda_1}{\beta^\star})-\epsilon \right\},\\
	D_{\epsilon}^{+}&=\left\{{\bf x} \ | \ \|{\bf x}\|_{\infty}\leq \sqrt{P}, \frac{\#\pi_{(a,b)}({\bf x})}{n}\geq Q(\frac{a}{\tilde{\tau}^\star}+\frac{\lambda_1}{\beta^\star})-Q(\frac{b}{\tilde{\tau^\star}}+\frac{\lambda_1}{\beta^\star})+\epsilon \right\}.
\end{align*}

 By using \eqref{eq:1}, we have for any $\tilde{\lambda}\in(0,1)$, 
\begin{align}
	\mathbb{P}\Big[\min_{{\bf x}\in D_\epsilon^{-}}\mathcal{C}_{\lambda,\rho}({\bf x})\leq \psi(\tau^\star,\beta^\star) +\gamma\tilde{\lambda}\epsilon^3\Big]\leq 	2\mathbb{P}\Big[\min_{{\bf x}\in D_\epsilon^{-}}\mathcal{L}_{\lambda,\rho}({\bf x})\leq \psi(\tau^\star,\beta^\star) +\gamma\tilde{\lambda}\epsilon^3 \Big]+C\exp(-cn).\label{e102}
\end{align}
On the event $\mathcal{E}_1$, any vector ${\bf x}\in D_\epsilon^{-}$ and $\overline{\bf x}^{\rm AO}$ are apart from each other by $r$ in absolute value in at least $\frac{n\epsilon}{2}$ positions. Hence, with probability $1-\exp(-\frac{1}{8}n\epsilon^2)$,
$$
\forall {\bf x}\in D_\epsilon^{-}, \ \ \|{\bf x}-\overline{\bf x}^{\rm AO}\|^2\geq \frac{n\epsilon^3(\tilde{\tau}^\star)^2}{32}.
$$ Set $\tilde{\lambda}=\min(\frac{\tilde{\tau}^2}{32},1)$.
Using \eqref{eq:exi1}, we thus obtain
$$
\mathbb{P}\Big[\min_{{\bf x}\in D_\epsilon^{-}}\mathcal{L}_{\lambda,\rho}({\bf x})\leq \psi(\tau^\star,\beta^\star) +\gamma\tilde{\lambda}\epsilon^3 \Big]\leq \frac{C}{\epsilon^3}\exp(-cn\epsilon^6),
$$proving \eqref{e102}.
Similarly, along the same lines we can show that:
\begin{align*}
	\mathbb{P}\Big[\min_{{\bf x}\in D_\epsilon^{+}}\mathcal{C}_{\lambda,\rho}({\bf x})\leq \psi(\tau^\star,\beta^\star) +\gamma\tilde{\lambda}\epsilon^3\Big]\leq \frac{C}{\epsilon^3}\exp(-cn\epsilon^6).
\end{align*}
Details are omitted for brevity. 
\end{proof}
\subsection{From the AO to the PO: convergence of functional of the thresholded precoding vector.}
\label{app:th_function}
To simplify the notation in this section, we adopt the following conventions.  For ${\bf x}\in \mathbb{R}^n$, we let $\pi_t$ denote the subset of $\{1,\cdots,n\}$ indexing the elements of ${\bf x}$ with magnitude less than $t$. The complement of $\pi_t({\bf x})$ is denoted by $\pi_t^{c}({\bf x})$, defined as  $\pi_t^{c}({\bf x})=\{1,\cdots,n\}\backslash \pi_t$. For $\pi\subset \{1,\cdots,n\}$, we denote by ${\bf x}_{\pi}$ the vector ${\bf x}$ in which all elements not in $\pi$ are replaced by zero. With these notations, we state the following theorem. 
\begin{theorem}
	Let $f$ be a continuous differentiable positive function on the interval $[0,\sqrt{P}]$ such that $f(0)=0$. Then, under Assumption \ref{ass:regime}, \ref{ass:statistic} and \ref{ass:regime_lambda}, there exist $C$, $c$, $\tilde{c}_1$, $\tilde{c}_2$ and $\gamma$ positive constants that depend only on $\mathcal{D}$ and $f$ such that:
	\begin{align}
	&\mathbb{P}\Big[\exists \ {\bf x} \ \  \textnormal{feasible} \ \ |\frac{1}{n}\sum_{i=1}^n f([{\bf x}_{\pi_{t}({\bf x})}]_i)-\frac{1}{n}\sum_{i=1}^n \mathbb{E}[f([\overline{\bf x}_{\pi_{t}(\overline{\bf x}^{\rm AO})}^{\rm AO}]_i)]|\geq \tilde{c}_1\sqrt{\epsilon}, \ \text{and } \mathcal{C}_{\lambda,\rho}({\bf x})\leq \psi(\tau^\star,\beta^\star)+\gamma \epsilon^2 \Big]\nonumber\\
	&\leq \frac{C}{\epsilon^2}\exp(-cn\epsilon^4) \label{eq:1th}
	\end{align}
	and 
	\begin{align}
		&\mathbb{P}\Big[\exists \ {\bf x} \ \  \textnormal{feasible} \ \ |\frac{1}{n}\sum_{i=1}^n f([{\bf x}_{\pi_{t}^c({\bf x})}]_i)-\frac{1}{n}\sum_{i=1}^n \mathbb{E}[f([\overline{\bf x}_{\pi_{t}^c(\overline{\bf x}^{\rm AO})}^{\rm AO}]_i)]|\geq \tilde{c}_2\sqrt{\epsilon}, \ \text{and } \mathcal{C}_{\lambda,\rho}({\bf x})\leq \psi(\tau^\star,\beta^\star)+\gamma \epsilon^2 \Big]\nonumber\\
		&\leq \frac{C}{\epsilon^2}\exp(-cn\epsilon^4). \label{eq:2th}
	\end{align}
	\label{th:function}
	\end{theorem}
\begin{proof}
	To begin with, we shall need the following technical lemma.
	\begin{lemma}
		Let $f$ be a function satisfying the conditions of Theorem \ref{th:function} and $t\in (0,\sqrt{P})$.  There exists a constant $c_1$ and $c_2$ such that for all $\epsilon>0$ sufficiently small, with probability at least $1-C\exp(-cn\epsilon)$, the following inequalities hold true:
		\begin{align}
			&\sup_{{\bf x}\in \mathcal{B}_{\overline{\bf x}^{\rm AO}}}|\frac{1}{n}\sum_{i=1}^n f([{\bf x}]_i){\bf 1}_{\{|[{\bf x}]_i|\leq t\}} - \mathbb{E}[f(x(\tilde{\tau}^\star,\beta^\star)){\bf 1}_{\{|x(\tilde{\tau}^\star,\beta^\star)|\leq t\}}]|\leq c_1\epsilon^{\frac{1}{2}} ,\label{eq:result_2b} \\
			&\sup_{{\bf x}\in \mathcal{B}_{\overline{\bf x}^{\rm AO}}}|\frac{1}{n}\sum_{i=1}^n f([{\bf x}]_i){\bf 1}_{\{|[{\bf x}]_i|\geq t\}} - \mathbb{E}[f(x(\tilde{\tau}^\star,\beta^\star)){\bf 1}_{|x(\tilde{\tau}^\star,\beta^\star)|\geq t}|]\leq c_2\epsilon^{\frac{1}{2}},\label{eq:result_1b}\end{align}
		where $\mathcal{B}_{\overline{\bf x}^{\rm AO}}:=\{{\bf x}, \ \text{feasible } \ \frac{1}{n}\|{\bf x}-\overline{\bf x}^{\rm AO}\|^2\leq \epsilon^2\}$ and $\overline{\bf x}^{\rm AO}$ is the auxiliary random equivalent defined in \eqref{eq:xoverlineAO}.
		\label{lem:usefulb}
	\end{lemma}
	\begin{proof}
		We provide details only  for the proof \eqref{eq:result_1b}. The proof of \eqref{eq:result_2b} follows along the same lines.  For $\eta>0$, we define $g_{+}$ and $g_{-}$ the following functions:
		$$
		\begin{array}{ll}g_{+}(x)=\left\{
			\begin{array}{ll}
				0 & \text{if} \  x\leq t-\eta \\
				\frac{1}{\eta}x+1-\frac{1}{\eta}t \ &\text{if} \  x\in(t-\eta,t)\\
				1& \text{if} \  x\geq t
			\end{array}\right.
			& g_{-}(x)=\left\{	\begin{array}{ll}
				0 & \text{if} \  x\leq t \\
				\frac{1}{\eta}x-\frac{1}{\eta}t \ &\text{if} \  x\in(t,t+\eta)\\
				1& \text{if} \  x\geq t+\eta
			\end{array}.\right.
		\end{array}
		$$	
		We can write the following upper-bound
		\begin{align}
			&\frac{1}{n}\sum_{i=1}^n f([{\bf x}]_i){\bf 1}_{\{|[{\bf x}]_i|\geq t\}}- \mathbb{E}[f(x(\tilde{\tau}^\star,\beta^\star)){\bf 1}_{\{|x(\tilde{\tau}^\star,\beta^\star)|\geq t\}}]\nonumber \\
			\leq& 	\frac{1}{n}\sum_{i=1}^n f([{\bf x}]_i)g_{+}(|[{\bf x}]_i|)- \mathbb{E}[f(x(\tilde{\tau}^\star,\beta^\star)){\bf 1}_{\{|x(\tilde{\tau}^\star,\beta^\star)|\geq t\}}]\nonumber \\
			=&\frac{1}{n}\sum_{i=1}^n f([{\bf x}]_i)(g_{+}(|[{\bf x}]_i|)-g_{+}(|[\overline{\bf x}^{\rm AO}]_i|)+\frac{1}{n}\sum_{i=1}^n f([{\bf x}]_i) g_{+}(|[\overline{\bf x}^{\rm AO}]_i|)- \mathbb{E}[f(x(\tilde{\tau}^\star,\beta^\star)){\bf 1}_{\{|x(\tilde{\tau}^\star,\beta^\star)|\geq t\}}]\nonumber\\
			=&\frac{1}{n}\sum_{i=1}^n f([{\bf x}]_i)(g_{+}(|[{\bf x}]_i|)-g_{+}(|[\overline{\bf x}^{\rm AO}]_i|)+\frac{1}{n}\sum_{i=1}^n (f([{\bf x}]_i) g_{+}(|[\overline{\bf x}^{\rm AO}]_i|)-f([\overline{\bf x}^{\rm AO}]_i) g_{+}(|[\overline{\bf x}^{\rm AO}]_i|))\nonumber \\
			&+ \frac{1}{n}\sum_{i=1}^n f([\overline{\bf x}^{\rm AO}]_i) g_{+}(|[\overline{\bf x}^{\rm AO}]_i|)- \mathbb{E}[f(x(\tilde{\tau}^\star,\beta^\star)){\bf 1}_{\{|x(\tilde{\tau}^\star,\beta^\star)|\geq t\}}] .\label{eq:controlled}
		\end{align}
		To prove the desired, we need to control each of the three  terms in \eqref{eq:controlled} uniformly in ${\bf x}\in \mathcal{B}_{\overline{\bf x}^{\rm AO}}$. We start with the first term. It can be controlled by noting that the function   $g_{+}$ is $\frac{1}{\eta}$- Lipschitz. Let $f_{\infty}=\sup_{-\sqrt{P}\leq x\leq \sqrt{P}}|f(x)|$. Hence, 
		\begin{equation*}
		\frac{1}{n}\sum_{i=1}^n f([{\bf x}]_i)(g_{+}(|[{\bf x}]_i|)-g_{+}(|[\overline{\bf x}^{\rm AO}]_i|)\leq \frac{1}{\eta}f_{\infty}\frac{1}{\sqrt{n}}\|{\bf x}-\overline{\bf x}^{\rm AO}\| . 
		\end{equation*}
		Hence,
		\begin{equation}
			\sup_{{\bf x}\in\mathcal{B}_{\overline{\bf x}^{\rm AO}}} |\frac{1}{n}\sum_{i=1}^n f([{\bf x}]_i)(g_{+}(|[{\bf x}]_i|)-g_{+}(|[\overline{\bf x}^{\rm AO}]_i|)|\leq \frac{1}{\eta}f_{\infty}\sup_{{\bf x}\in\mathcal{B}_{\overline{\bf x}^{\rm AO}}} \frac{1}{\sqrt{n}}\|{\bf x}-\overline{\bf x}^{\rm AO}\|\leq \frac{\epsilon f_{\infty}}{\eta} .\label{eq:f11}
		\end{equation}
	To control the second term, we use the fact that $f$ is a continuously differentiable function, and hence Lipschitz with a Lipschitz constant denoted by $L_f$. Then, we have:  
		\begin{equation}
			\sup_{{\bf x}\in\mathcal{B}_{\overline{\bf x}^{\rm AO}}}	|\frac{1}{n}\sum_{i=1}^n (f([{\bf x}]_i)-f([\overline{\bf x}^{\rm AO}]_i)) g_{+}(|[\overline{\bf x}^{\rm AO}]_i|)|\leq L_f \frac{1}{\sqrt{n}}\|{\bf x}-\overline{\bf x}^{\rm AO}\|\leq L_f\epsilon. \label{eq:f21}
		\end{equation}
		Next,  to control the last term, we use Hoeffding's inequality to obtain:
		\begin{equation}
		\mathbb{P}\Big[\left|\frac{1}{n}\sum_{i=1}^n f([\overline{\bf x}^{\rm AO}]_i)g_{+}(|[\overline{\bf x}^{\rm AO}]_i|)-  \mathbb{E}[f(x(\tilde{\tau}^\star,\beta^\star))g_{+}(|x(\tilde{\tau}^\star,\beta^\star)|)]\right|>\sqrt{\epsilon}\Big]\leq2\exp(-\frac{n\epsilon}{4f_{\infty}^2}).  \label{eq:hoeffding}
		\end{equation}
		We continue by noting that:
		\begin{align}
			\mathbb{E}[f(x(\tilde{\tau}^\star,\beta^\star))g_{+}(|x(\tilde{\tau}^\star,\beta^\star)|)]-\mathbb{E}[f(x(\tilde{\tau}^\star,\beta^\star)){\bf 1}_{\{|x(\tilde{\tau}^\star,\beta^\star)|\geq t\}}]&\leq 2\int_{t-\eta}^{t} f(x(\tilde{\tau}^\star,\beta^\star))\nu^\star(x)dx\nonumber\\
			&\leq 2f_{\infty}\eta\sup_{|x|\leq \sqrt{P}}\nu^\star(x). \label{eq:to_be_used}
		\end{align}
	By combining \eqref{eq:hoeffding} and \eqref{eq:to_be_used}, we find that with probability $1-C\exp(-cn\epsilon)$,
		\begin{equation}
		|\frac{1}{n}\sum_{i=1}^n f([\overline{\bf x}^{\rm AO}]_i)g_{+}(|[\overline{\bf x}^{\rm AO}]_i|) -\sum_{i=1}^n \mathbb{E}[f(x(\tilde{\tau}^\star,\beta^\star)){\bf 1}_{\{|x(\tilde{\tau}^\star,\beta^\star)|\geq t\}}]|\leq 2f_{\infty}\eta \sup_{|x|\leq \sqrt{P}}\nu^\star(x)+\sqrt{\epsilon}. \label{eq:f23}
		\end{equation}
		Now choosing $\eta=\sqrt{\epsilon}$, and using \eqref{eq:f11}, \eqref{eq:f21} and \eqref{eq:f23}, we obtain with probability at least $1-C\exp(-cn\epsilon)$,
		$$
			\sup_{{\bf x}\in \mathcal{B}_{\overline{\bf x}^{\rm AO}}}\frac{1}{n}\sum_{i=1}^n f([{\bf x}]_i){\bf 1}_{\{|[{\bf x}]_i|\geq t\}} - \mathbb{E}[f(x(\tilde{\tau}^\star,\beta^\star)){\bf 1}_{\{|x(\tilde{\tau}^\star,\beta^\star)|\geq t\}}]\leq c_2\sqrt{{\epsilon}}
		$$
		with $c_2=L_f+2f_{\infty} \sup_{|x|\leq \sqrt{P}}\nu^\star(x)+1+f_{\infty}$. 
 
 The proof that:
		$$
		\sup_{{\bf x}\in \mathcal{B}_{\overline{\bf x}^{\rm AO}}}\frac{1}{n}\sum_{i=1}^n f([{\bf x}]_i){\bf 1}_{\{|[{\bf x}]_i|\geq t\}} -\sum_{i=1}^n \mathbb{E}[f(x(\tilde{\tau}^\star,\beta^\star)){\bf 1}_{\{|x(\tilde{\tau}^\star,\beta^\star)|\geq t\}}]\geq -c_2\sqrt{\epsilon}
		$$
		follows along the same lines, by working with function $g_{-}$ instead of $g_{+}$. Details are thus omitted.
	\end{proof}
	With Lemma \ref{lem:usefulb} at hand, we are now in position to show Theorem \ref{th:function}. For that, we let $D_{\epsilon}$ the following set
	\begin{align*}
		D_{\epsilon}:=\left\{{\bf x} \ \text{feasible}, \ |\frac{1}{n}\sum_{i=1}^n f([{\bf x}_{\pi_{t}({\bf x})}]_i)-\frac{1}{n}\sum_{i=1}^n \mathbb{E}[f([\overline{\bf x}_{\pi_{t}(\overline{\bf x}^{\rm AO})}^{\rm AO}]_i)]|\geq 2{c}_1\sqrt{\epsilon}\right\}.
	\end{align*}
	Using \eqref{eq:1}, we obtain:
	\begin{align*}
		\mathbb{P}\Big[\min_{{\bf x}\in D_{\epsilon}} \mathcal{C}_{\lambda,\rho}({\bf x})\leq \psi(\tau^\star,\beta^\star)+\gamma\epsilon^2\Big]\leq 2\mathbb{P}\Big[\min_{{\bf x}\in D_{\epsilon}} \mathcal{L}_{\lambda,\rho}({\bf x})\leq \psi(\tau^\star,\beta^\star)+\gamma\epsilon^2\Big]+C\exp(-cn).
	\end{align*}
	It follows from Lemma \ref{lem:usefulb}, that with probability $1-C\exp(-cn\epsilon)$ for all ${\bf x}$ such that $\frac{1}{n}\|{\bf x}-\overline{\bf x}^{\rm AO}\|^2\leq \epsilon^2$,
	$$
	|\frac{1}{n}\sum_{i=1}^n f([{\bf x}_{\pi_{t}({\bf x})}]_i)-\frac{1}{n}\sum_{i=1}^n \mathbb{E}[f([\overline{\bf x}_{\pi_{t}(\overline{\bf x}^{\rm AO})}^{\rm AO}]_i)]|\leq {c}_1\sqrt{\epsilon}.
	$$
	Hence, with probability $1-C\exp(-cn\epsilon)$, $ D_{\epsilon}\subset \mathcal{B}^{c}:=\left\{{\bf x} \ \text{feasible}, \ \frac{1}{n}\|{\bf x}-\overline{\bf x}^{\rm AO}\|^2\geq \epsilon^2\right\}$. This implies that:
	$$
	\mathbb{P}\Big[\min_{{\bf x}\in D_{\epsilon}} \mathcal{L}_{\lambda,\rho}({\bf x})\leq \psi(\tau^\star,\beta^\star)+\gamma\epsilon^2\Big]\leq \mathbb{P}\Big[\min_{{\bf x}\in \mathcal{B}^{c}} \mathcal{L}_{\lambda,\rho}({\bf x})\leq \psi(\tau^\star,\beta^\star)+\gamma\epsilon^2\Big]+C\exp(-cn\epsilon).
	$$
	Using \eqref{eq:sub_opt}, we obtain:
	$$
	\mathbb{P}\Big[\min_{{\bf x}\in \mathcal{B}^{c}} \mathcal{L}_{\lambda,\rho}({\bf x})\leq \psi(\tau^\star,\beta^\star)+\gamma\epsilon^2\Big]\leq \frac{C}{\epsilon^2}\exp(-cn\epsilon^4)
	$$
	which shows the desired inequality \eqref{eq:1th} for $\tilde{c}_1=2c_1$. The proof of \eqref{eq:2th} follows along the same lines and is omitted for brevity. 
	\end{proof}
\section{Study of the residual $\hat{\bf u}^{\rm PO}$ via the AO problem $\max_{{\bf u}\in\mathcal{S}_{\bf u}} \mathcal{F}_{\lambda,\rho}({\bf u})$}
\subsection{Study of the AO problem}
\noindent{\bf Organization of the proof.} Similar to the previous section, we organize the proof into three major parts. Note that we will make use of the results of the previous part, and hence, some probability inequalities will be proven directly for the PO without having to invoke the AO. Particularly, in the first step, we make use directly of the previous results to show that
with probability $1-\frac{C}{\epsilon}e^{-cn\epsilon^2}$
\begin{equation}
	\left|\max_{{\bf u}\in\mathcal{S}_{{\bf u}}} \mathcal{V}_{\lambda,\rho}({\bf u})- \psi(\tau^\star,\beta^\star)\right|\leq \gamma\epsilon \label{eq:upper1}
\end{equation}
for some $C$ and $c$ positive constants and all $\epsilon\in(0,1]$. In the second part, we define 
function ${\bf u}\mapsto\tilde{\mathcal{F}}_{\lambda,\rho}({\bf u})$ as:
$$
\tilde{\mathcal{F}}_{\lambda,\rho}({\bf u}):=\frac{1}{n}\|\overline{\bf x}^{\rm AO}\|_2{\bf g}^{T}{\bf u}- \frac{1}{n}\|{\bf u}\|{\bf h}^{T}\overline{\bf x}^{\rm AO}-\frac{\sqrt{\rho}{\bf u}^{T}{\bf s}}{\sqrt{n}} -\frac{\|{\bf u}\|^2}{4} + \frac{\lambda_2\|\overline{\bf x}^{\rm AO}\|^2}{n}+\frac{\lambda_1\|\overline{\bf x}^{\rm AO}\|_1}{n}.
$$
We also let $\overline{\bf u}^{\rm AO}$ be:
$$
\overline{\bf u}^{\rm AO}= \beta^\star \frac{\sqrt{(\tau^\star)^2\delta-\rho}\frac{\bf g}{\sqrt{n}}-\sqrt{\frac{\rho}{n}}{\bf s}}{\tau^\star\delta}.
$$
Then, we prove that there exists positive constant $\overline{\gamma}$ such that for all $\epsilon\in(0,1]$ with probability $1-C\exp(-cn\epsilon^2)$, 
\begin{equation}
	\left|\tilde{\mathcal{F}}_{\lambda,\rho}(\overline{\bf u}^{\rm AO})-\psi(\tau^\star,\beta^\star)\right|\leq \overline{\gamma}\epsilon. \label{eq:AO_equivalent}
\end{equation}
In the same way as for the previous section, we refer to $\overline{\bf u}^{\rm AO}$ as the AO equivalent solution. Finally, as a final step, we note that  $\tilde{\mathcal{F}}_{\lambda,\rho}({\bf u})$ is strongly concave in ${\bf u}$ and exploit this property to show that it is with high probability uniformly sub-optimum outside any ball containing $\overline{\bf u}^{\rm AO}$. More specifically, we prove that there exists a positive constant $\gamma$ such that with probability $1-C\exp(-cn\epsilon^2)$, 
\begin{equation*}
	\forall {\bf u} \ \ \text{such that} \ \ \|{\bf u}-\overline{\bf u}^{\rm AO}\|^2\geq \epsilon \ \ \Longrightarrow \ \ \tilde{\mathcal{F}}_{\lambda,\rho}({\bf u})\leq \psi(\tau^\star,\beta^\star)-\gamma\epsilon . 
\end{equation*}
This can be equivalently stated as:
\begin{equation}
	\mathbb{P}\Big[\exists {\bf u} \ \text{such that }\|{\bf u}-\overline{\bf u}^{\rm AO}\|^2\geq \epsilon \  \text{and } \tilde{\mathcal{F}}_{\lambda,\rho}({\bf u})\geq \psi(\tau^\star,\beta^\star)-\gamma\epsilon\Big]\leq C\exp(-cn\epsilon^2).\label{eq:sub_optimality_u}
\end{equation}
Similar to the preceding analysis of the AO problem, this property proves instrumental in deducing asymptotic characterizations of the PO's solution in the variable ${\bf u}$ through a deviation argument. The methodology parallels that employed in the earlier investigation of the AO problem concerning the optimization variable ${\bf x}$. However, a key distinction arises: we can leverage the function $\tilde{\mathcal{F}}_{\lambda,\rho}$ as elucidated below as it surpasses $\mathcal{F}_{\lambda,\rho}$ for each ${\bf u}$.
More formally,  consider a deterministic compact set $\mathcal{S}_{u}\in \mathbb{R}^{m}$ for which with a probability approaching one, the following holds:
\begin{equation}
	\forall {\bf u}\in \mathcal{S}_u \Longrightarrow \|{\bf u}-\overline{\bf u}^{\rm AO}\|^2\geq \epsilon.  \label{eq:34}
\end{equation}
Based on \eqref{eq:2}, we have:
\begin{align*}
	\mathbb{P}\left[\max_{{\bf u}\in\mathcal{S}_{u}} \mathcal{V}_{\lambda,\rho}({\bf u})\geq \psi(\tau^\star,\beta^\star)-\gamma\epsilon\right]&\leq  2\mathbb{P}\left[\max_{{\bf u}\in\mathcal{S}_{u}} \mathcal{F}_{\lambda,\rho}({\bf u})\geq \psi(\tau^\star,\beta^\star)-\gamma\epsilon\right]\\
	&\leq 2\mathbb{P}\left[\max_{{\bf u}\in\mathcal{S}_{u}} \tilde{\mathcal{F}}_{\lambda,\rho}({\bf u})\geq \psi(\tau^\star,\beta^\star)-\gamma\epsilon\right]
\end{align*}
where the last inequality follows since $\tilde{\mathcal{F}}_{\lambda,\rho}({\bf u})\geq \mathcal{F}_{\lambda,\rho}({\bf u})$. Combining \eqref{eq:34} and \eqref{eq:sub_optimality_u}, we can thus deduce that with probability approaching one $\hat{\bf u}^{\rm PO}\notin \mathcal{S}_u$.  
\subsubsection{Convergence of the PO cost (Proof of \eqref{eq:upper1})}
The convergence result in \eqref{eq:upper1} is a direct by-product of the results in the previous part. Indeed, since the objective in the PO is convex in ${\bf x}$ and concave in ${\bf u}$:
\begin{equation}
	\max_{{\bf u}\in\mathcal{S}_{{\bf u}}} \mathcal{V}_{\lambda,\rho}({\bf u})= \min_{\|{\bf x}\|_{\infty}\leq \sqrt{P}} \mathcal{C}_{\lambda,\rho}({\bf x}). \label{eq:PO_e}
\end{equation}
In the previous section, we proved that the optimal cost of the AO problem $\min_{\|{\bf x}\|_{\infty}\leq \sqrt{P}} \mathcal{L}_{\lambda,\rho}({\bf x})$ satisfies with a probability at least $1-\frac{C}{\epsilon}e^{-cn\epsilon^2}$,
$$
\psi(\tau^\star,\beta^\star)-\gamma\epsilon\leq \min_{\|{\bf x}\|_{\infty}\leq \sqrt{P}} \mathcal{L}_{\lambda,\rho}({\bf x})\leq \psi(\tau^\star,\beta^\star)+\gamma\epsilon.
$$
A direct application of the cGMT allowed  us to prove in section \ref{sec:proof_optimal_cost} that the PO optimal cost satisfies with a probability at least $1-\frac{C}{\epsilon}e^{-cn\epsilon^2}$,
\begin{equation*}
	\psi(\tau^\star,\beta^\star)-\gamma\epsilon\leq \min_{\|{\bf x}\|_{\infty}\leq \sqrt{P}} \mathcal{C}_{\lambda,\rho}({\bf x})\leq \psi(\tau^\star,\beta^\star)+\gamma\epsilon , 
\end{equation*}
which using \eqref{eq:PO_e} implies also that:
\begin{equation*}
	\psi(\tau^\star,\beta^\star)-\gamma\epsilon\leq \max_{{\bf u}\in\mathcal{S}_{{\bf u}}} \mathcal{V}_{\lambda,\rho}({\bf u})\leq \psi(\tau^\star,\beta^\star)+\gamma\epsilon . 
\end{equation*} 
\subsubsection{Asymptotic properties of the AO equivalent solution. Proof of \eqref{eq:AO_equivalent} }
Let $\epsilon>0$. For all ${\bf s}\in\mathcal{S}^{\otimes m}$, the event 
$$
\tilde{\mathcal{A}}:=\left\{\left|\frac{1}{m}{\bf g}^{T}{\bf g}-1\right|\leq \epsilon\right\} \cap \left\{\left|\frac{\|{\bf g}\|}{\sqrt{m}}-1\right|\leq \epsilon\right\} \cap \left\{\left|\frac{{\bf g}^{T}{\bf s}}{m}\right|\leq \epsilon\right\}
$$
occurs with probability $1-C\exp(-cn\epsilon^2)-C\exp(-c\epsilon n)$ where $C$ and $c$ are independent of ${\bf s}$.
Optimizing $\tilde{\mathcal{F}}_{\lambda,\rho}$ with respect to ${\bf u}$, we obtain:
$$
\max_{{\bf u}\in\mathcal{S}_{\bf u}}\tilde{\mathcal{F}}_{\lambda,\rho}({\bf u})=\left(\|\frac{1}{n}\|\overline{\bf x}^{\rm AO}\|{\bf g}-\frac{\sqrt{\rho}}{\sqrt{n}}{\bf s}\|-\frac{1}{n}{\bf h}^{T}\overline{\bf x}^{\rm AO}\right)_{+}^2 + \frac{\lambda_1}{n}\|\overline{\bf x}^{\rm AO}\|_1+\frac{\lambda_2}{n}\|\overline{\bf x}^{\rm AO}\|^2.
$$
Recall the event $\mathcal{A}$ defined in \eqref{eq:A} which occurs with probability at least $1-C\exp(-cn\epsilon^2)-C\exp(-cn\epsilon)$.
It is easy to check that there exists constants $\overline{\gamma}$ such that on the event $\tilde{\mathcal{A}}\cap \mathcal{A}$ occurring with probability $1-C\exp(-cn\epsilon^2)-C\exp(-cn\epsilon)$, 
\begin{equation}
	\psi(\tau^\star,\beta^\star)-\overline{\gamma}\epsilon\leq \max_{{\bf u}\in\mathcal{S}_{\bf u}}\tilde{\mathcal{F}}_{\lambda,\rho}({\bf u})\leq \psi(\tau^\star,\beta^\star)+\overline{\gamma}\epsilon. \label{eq:ref}
\end{equation}

To continue, consider the following events:
\begin{align*}
	\mathcal{K}_1&:=\Big\{\big|\frac{{\bf g}^{T}\overline{\bf u}^{\rm AO}}{\sqrt{n}}-\frac{\beta^\star\sqrt{(\tau^\star)^2\delta-\rho}}{\tau^\star}\big|\leq \epsilon\Big\} \cap \Big\{\big|\frac{{\bf s}^{T}\overline{\bf u}^{\rm AO}}{\sqrt{n}}-\frac{\beta^\star\sqrt{\rho}}{\tau^\star} \big|\leq \epsilon\Big\},\\
	\mathcal{K}_2&:=\Big\{\left|\|\overline{\bf u}^{\rm AO}\|-\beta^\star\right|\leq\epsilon\Big\}\cap\Big\{\left|\|\overline{\bf u}^{\rm AO}\|^2-(\beta^\star)^2\right|\leq\epsilon\}.
\end{align*}
Based on Gaussian concentration results of Lipschitz functions, we can  check that  $\mathcal{K}_1\cap\mathcal{K}_2$ occur with probability $1-C\exp(-cn\epsilon^2)-C\exp(-cn\epsilon)$.

To continue,
we note also that on the event $\mathcal{A}$,
$$
\left|\frac{1}{\sqrt{n}}\|\overline{\bf x}^{{\rm AO}}\|-\sqrt{\mathbb{E}[\frac{1}{n} \|\overline{\bf x}^{\rm AO}\|^2]}\right|=\frac{\left|\frac{1}{n}\|\overline{\bf x}^{{\rm AO}}\|^2-\mathbb{E}[\frac{1}{n} \|\overline{\bf x}^{\rm AO}\|^2]\right|}{\frac{1}{\sqrt{n}}\|\overline{\bf x}^{{\rm AO}}\|+\sqrt{\mathbb{E}[\frac{1}{n} \|\overline{\bf x}^{\rm AO}\|^2]}}\leq \frac{\epsilon}{\sqrt{(\tau^\star)^2\delta-\rho}}\leq \frac{\epsilon}{\sqrt{\tau_{\rm min}}},
$$
where $\tau_{\rm min}$ is defined in Theorem \ref{th:control_bounds}. 
Using this, we may check that there exists constant $\tilde{\gamma}_l$, $C$ and $c$ such that for all $\epsilon\geq 0$ on the event $\mathcal{A}\cap\mathcal{K}_1\cap\mathcal{K}_2$, 
\begin{equation*}
\tilde{\mathcal{F}}(\overline{\bf u}^{\rm AO})\geq  \psi(\tau^\star,\beta^\star)-\tilde{\gamma}_l\epsilon, 
\end{equation*}
or equivalently stated:
$$
\mathbb{P}\Big[\tilde{\mathcal{F}}(\overline{\bf u}^{\rm AO})\geq  \psi(\tau^\star,\beta^\star)-\tilde{\gamma}_l\epsilon\Big]\geq 1-C\exp(-cn\epsilon^2)-C\exp(-cn\epsilon).
$$
Since the above inequality holds for any $\epsilon>0$, we may  perform the change of variable $\epsilon\leftrightarrow \epsilon\frac{\tilde{\gamma_l}}{\overline{\gamma}}$ to find:
\begin{equation}
\mathbb{P}\Big[\tilde{\mathcal{F}}(\overline{\bf u}^{\rm AO})\geq  \psi(\tau^\star,\beta^\star)-\overline{\gamma}_l\epsilon\Big]\geq 1-C\exp(-cn\epsilon^2)-C\exp(-cn\epsilon).\label{eq:equation1}
\end{equation}
Next, since $\tilde{\mathcal{F}}(\overline{\bf u}^{\rm AO})\leq \max_{{\bf u}} \tilde{\mathcal{F}}_{\lambda,\rho}({\bf u})$, we also have with probability  $1-C\exp(-cn\epsilon^2)-C\exp(-cn\epsilon)$, 
\begin{equation}
\tilde{\mathcal{F}}(\overline{\bf u}^{\rm AO})\leq \psi(\tau^\star,\beta^\star)+\bar{\gamma}_l\epsilon.\label{eq:equation2}
\end{equation}
By combining \eqref{eq:equation1} and \eqref{eq:equation2}, we obtain \eqref{eq:AO_equivalent}.

\subsubsection{Uniform sub-optimality of the AO cost away from the AO equivalent solution (Proof of \eqref{eq:sub_optimality_u}).}
To prove \eqref{eq:sub_optimality_u}, we exploit the fact that ${\bf u}\mapsto \tilde{\mathcal{F}}_{\lambda,\rho}({\bf u})$ is $\frac{1}{2}$-strongly concave. Thus, letting $\breve{\bf u}:=\displaystyle{\arg\max_{{\bf u}}} \  \tilde{\mathcal{F}}_{\lambda,\rho}({\bf u})$, we have:
$$
\tilde{\mathcal{F}}_{\lambda,\rho}(\breve{\bf u})-\tilde{\mathcal{F}}_{\lambda,\rho}({\bf u}) \geq \frac{1}{4}\|{\bf u}-\breve{\bf u}\|^2.
$$
Hence, for any $\epsilon>0$, 
	\begin{align*}
		\forall \ \ {\bf u} \ \text{such that} \ \ \|{\bf u}-\breve{\bf u}\|^2\geq 8\overline{\gamma}\epsilon &\Longrightarrow \tilde{\mathcal{F}}_{\lambda,\rho}({\bf u})\leq \tilde{\mathcal{F}}_{\lambda,\rho}(\breve{\bf u})-2\overline{\gamma}\epsilon 
	\end{align*} 
and hence using \eqref{eq:ref}, we obtain with probability at least $1-C\exp(-cn\epsilon^2)$, 
\begin{equation}
	\forall \ \ {\bf u} \ \text{such that} \ \ \|{\bf u}-\breve{\bf u}\|^2\geq 8\overline{\gamma}\epsilon \Longrightarrow\tilde{\mathcal{F}}_{\lambda,\rho}({\bf u})\leq \psi(\tau^\star,\beta^\star)-\overline{\gamma}\epsilon .\label{eq:t}
\end{equation}
Now, since from \eqref{eq:equation1}, $\tilde{\mathcal{F}}_{\lambda,\rho}(\overline{\bf u}^{\rm AO})\geq \psi(\tau^\star,\beta^\star)-\overline{\gamma}\epsilon$, we thus obtain:
\begin{equation*}
	\|\overline{\bf u}^{\rm AO}-\breve{\bf u}\|^2\leq 8\overline{\gamma}\epsilon. 
\end{equation*}
Take ${\bf u}$ such that $\|{\bf u}-\overline{\bf u}^{\rm AO}\|\geq 3\sqrt{\overline{\gamma}\epsilon}$. Hence, $\|{\bf u}-\breve{\bf u}\|\geq \|{\bf u}-\overline{\bf u}^{\rm AO}\|-\|\overline{\bf u}^{\rm AO}-\breve{\bf u}\|\geq (3-2\sqrt{2})\sqrt{\overline{\gamma}\epsilon}$. Using \eqref{eq:t}, we thus obtain:
$$
\forall \ {\bf u}\ \text{such that }\|{\bf u}-\overline{\bf u}^{\rm AO}\|\geq 3\sqrt{\overline{\gamma}\epsilon} \Longrightarrow \tilde{\mathcal{F}}_{\lambda,\rho}({\bf u})\leq \psi(\tau^\star,\beta^\star)-\frac{(3-2\sqrt{2})^2\overline{\gamma}\epsilon}{8}
$$
and since $\tilde{\mathcal{F}}_{\lambda,\rho}({\bf u})\geq \mathcal{F}_{\lambda,\rho}(\bf u)$, we also have:
$$
{\mathcal{F}}_{\lambda,\rho}({\bf u})\leq \psi(\tau^\star,\beta^\star)-\frac{(3-2\sqrt{2})^2\overline{\gamma}\epsilon}{8}.
$$

Putting all this together, we have just proved that there exists constants $C$, $c$ and $\gamma$ such that for all $\epsilon\in(0,1]$
\begin{equation}
	\mathbb{P}\left[{\bf u} \ \text{such that} \ \  \|{\bf u}-\overline{\bf u}^{\rm AO}\|^2\geq \epsilon,  \ \text{and} \ {\mathcal{F}}_{\lambda,\rho}({\bf u})\geq \psi(\tau^\star,\beta^\star)-\gamma \epsilon\right]\leq C\exp(-cn\epsilon^2).\label{eq:u}
\end{equation}
As in the previous analysis, we may need to control probabilities of events in the same form of \eqref{eq:u} but with $\epsilon$ being replaced by $\tilde{\lambda} \epsilon^2$ for some positive constant $\lambda\in(0,1)$ independent of $\epsilon$. By performing the change of variable $\epsilon\leftrightarrow \tilde{\lambda} \epsilon^2$ and absorbing the incurred constants into $C$ and $c$, we obtain the following result:
\begin{equation}
	\mathbb{P}\left[{\bf u} \ \text{such that} \ \  \|{\bf u}-\overline{\bf u}^{\rm AO}\|^2\geq \tilde{\lambda}\epsilon^2,  \ \text{and} \ {\mathcal{F}}_{\lambda,\rho}({\bf u})\geq \psi(\tau^\star,\beta^\star)-\gamma \tilde{\lambda} \epsilon^2\right]\leq C\exp(-cn\epsilon^4).\label{eq:uniform_u}
\end{equation}

\subsection{From the AO to the PO via exploitation of the cGMT: Convergence of the empirical measure $\hat{\mu}(\frac{\sqrt{n}}{2}\hat{\bf u}^{\rm PO}+\sqrt{\rho}{\bf s},{\bf s})$}\label{sec:conv_emp_measure}{
The distortion vector $\hat{\bf e}_{\ell_1}$  is related to the  solution $\hat{\bf u}^{\rm PO}$ of \eqref{eq:PO2} through the following equation:
$$\hat{\bf e}_{\ell_1}=\frac{\sqrt{n}}{2}\hat{\bf u}^{\rm PO}+\sqrt{\rho}{\bf s}.$$ Using the notation introduced in this appendix, we analyze the joint empirical distribution of $(\frac{\sqrt{n}}{2}\hat{\bf u}^{\rm PO}+\sqrt{\rho}{\bf s})
 $ denoted as $\hat{\mu}(\frac{\sqrt{n}}{2}\hat{\bf u}^{\rm PO}+\sqrt{\rho}{\bf s},{\bf s})$}

Consider the set $D_\epsilon$ defined as:{
$$
D_\epsilon:=\left\{{\bf u}\in\mathcal{S}_{{\bf u}}, \mathcal{W}_2\left(\hat{\mu}(\frac{\sqrt{n}}{2}{\bf u}+\sqrt{\rho}{\bf s},{\bf s}),\mu^\star\right)\geq \epsilon\right\}.
$$}
To prove the desired, it suffices to check that there exists $\alpha>0$ such that
\begin{equation}
\mathbb{P}\left[\max_{{\bf u}\in D_\epsilon}\mathcal{V}_{\lambda,\rho}({\bf u})\geq \max_{{\bf u}\in\mathcal{S}_{\bf u}}\mathcal{V}_{\lambda,\rho}({\bf u})-\alpha\gamma\epsilon^2\right]\leq \frac{C}{\epsilon^2}\exp(-cn\epsilon^4).
\label{eq:desired_u}
\end{equation}
The proof relies heavily on cGMT as well as previous concentration inequalities. First using Lemma \ref{lem:convergence_empirical_rate}, we prove the following Lemma.
\begin{lemma}
	Let $\hat{\mu}(\frac{\sqrt{n}}{2}\overline{\bf u}^{\rm AO}+{\sqrt{\rho}{\bf s}},{\bf s})$ be the empirical distribution defined as:
	$$
	\hat{\mu}(\frac{\sqrt{n}}{2}\overline{\bf u}^{\rm AO}+{\sqrt{\rho}{\bf s}},{\bf s})=\frac{1}{m}\sum_{i=1}^m \delta_{([\frac{\sqrt{n}}{2}\overline{\bf u}^{\rm AO}]_i{+\sqrt{\rho}s_i},s_i)},
	$$
	then for any $\epsilon>0$,
	\begin{equation*}
	\mathbb{P}\left[\mathcal{W}_2\big(\hat{\mu}(\frac{\sqrt{n}}{2}\overline{\bf u}^{\rm AO}{+\sqrt{\rho}{\bf s}},{\bf s})),\mu^\star\big)\geq \frac{\epsilon}{2}\right]\leq C\exp(-cn{\epsilon^4}). 
\end{equation*}
\label{lem:conv_u}
\end{lemma}

\begin{proof}
	Let $\hat{S}$ and $\hat{G}$ two random variables drawn from the empirical distributions $\hat{\mu}({\bf s}):=\frac{1}{m}\sum_{i=1}^m \delta_{s_i}$ and $\hat{\mu}({\bf g}):=\frac{1}{m}\sum_{i=1}^m \delta_{[{\bf g}]_i}$. Let $S$  be drawn from the {Rademacher distribution $\mathcal{R}$ such that $S=1,-1$ with probability $\frac{1}{2},\frac{1}{2}$} respectively and ${G}$ from the normal distribution such that $(\hat{S},S)$ and $(\hat{G},G)$ denote the couplings that achieve the Wassertein distances, that is:
	\begin{align*}
		\mathbb{E}[|\hat{G}-G|^2]&=(\mathcal{W}_2(\hat{\mu}({\bf g}),\mathcal{N}(0,1)))^2,\\
		\mathbb{E}[|\hat{S}-S|^2]&=(\mathcal{W}_2(\hat{\mu}({\bf s}),\mathcal{R}))^2.
	\end{align*} 
	 The variable {$(\frac{\beta^\star}{2\tau^\star\delta}\sqrt{(\tau^\star)^2\delta-\rho})\hat{G} +(1-\frac{\beta^\star}{2\tau^\star \delta})S,S)$} has distribution $\mu^\star$. The squared Wassertein distance  {$(\mathcal{W}_2(\hat{\mu}(\frac{\sqrt{n}}{2}\overline{\bf u}^{\rm AO},{\bf s}+\sqrt{\rho}{\bf s}),\mu^\star))^2$} can be thus upper-bounded by:
	 {\begin{align*}
		(\mathcal{W}_2(\hat{\mu}(\frac{\sqrt{n}}{2}\overline{\bf u}^{\rm AO}+\sqrt{\rho}{\bf s},{\bf s}),\mu^\star))^2&\leq \frac{(\beta^\star)^2((\tau^\star)^2\delta-\rho)}{2(\tau^\star)^2\delta^2} \mathbb{E}|\hat{G}-G|^2+(2(1-\frac{\beta^\star}{2\tau^\star\delta})^2+1)\mathbb{E}|\hat{S}-S|^2.
	\end{align*}}
	It follows from Lemma \ref{lem:convergence_empirical_rate} that with probability $1-C\exp(-cn\epsilon^4)$,
	 {$$
	\mathbb{E}[|\hat{G}-G|^2]=({\mathcal{W}}_2(\hat{\mu}({\bf g}),\mathcal{N}(0,1)))^2\leq \frac{{(\tau^\star)^2\delta^2}}{4(\beta^\star)^2((\tau^\star)^2\delta-\rho)}\epsilon^2
	$$}
	and
	 {$$
		\mathbb{E}[|\hat{S}-S|^2]=(\mathcal{W}_2(\hat{\mu}({\bf s}),\mathcal{R})\leq \frac{1}{8((2(1-\frac{\beta^\star}{2\tau^\star\delta})^2+1))})^2\epsilon^2.
	$$}
	Hence, with  probability $1-C\exp(-cn\epsilon^4)$, {
	$$
		(\mathcal{W}_2(\hat{\mu}(\frac{\sqrt{n}}{2}\overline{\bf u}^{\rm AO}+\sqrt{\rho}{\bf s},{\bf s}),\mu^\star))^2\leq \frac{\epsilon^2}{4}.
	$$}
\end{proof}
On the following event {
\begin{equation}
\{\big(\mathcal{W}_2(\hat{\mu}(\frac{\sqrt{n}}{2}\overline{\bf u}^{\rm AO}+\sqrt{\rho}{\bf s},{\bf s})),\mu^\star\big)\leq \frac{\epsilon}{2}\} \label{eq:evu}
\end{equation}}
occurring with probability at least $1-C\exp(-cn\epsilon^4)$ according to Lemma \ref{lem:conv_u}, we have for all ${\bf u}\in D_{\epsilon}$, 
{\begin{align*}
\frac{n}{4m}\|{\bf u}-\overline{\bf u}^{\rm AO}\|^2&\geq (\mathcal{W}_2(\hat{\mu}(\frac{\sqrt{n}}{2}{\bf u}+\sqrt{\rho}{\bf s},{\bf s}),\hat{\mu}(\frac{\sqrt{n}}{2}\overline{\bf u}^{\rm AO},{\bf s}))^2\\
&\geq  (\mathcal{W}_2(\hat{\mu}(\frac{\sqrt{n}}{2}{\bf u}+\sqrt{\rho}{\bf s},{\bf s}),\mu^\star)-\mathcal{W}_2(\hat{\mu}(\frac{\sqrt{n}}{2}\overline{\bf u}^{\rm AO}+\sqrt{\rho}{\bf s},{\bf s}),\mu^\star))^2\\
&\geq \frac{\epsilon^2}{4},
\end{align*}
}
or
$$
\|{\bf u}-\overline{\bf u}^{\rm AO}\|^2\geq \delta \epsilon^2.
$$
Let $\tilde{\lambda}=\min(\delta,1)$,
with this, we are now ready to show \eqref{eq:desired_u} using the following probability inequalities:
\begin{align*}
	&\mathbb{P}\left[\max_{{\bf u}\in D_\epsilon}\mathcal{V}_{\lambda,\rho}({\bf u})\geq \max_{{\bf u}}\mathcal{V}_{\lambda,\rho}({\bf u})-\frac{\tilde{\lambda}}{2}\gamma\epsilon^2\right] \\
	\leq& \mathbb{P}\left[\left\{\max_{{\bf u}\in D_\epsilon} \mathcal{V}_{\lambda,\rho}({\bf u})\geq \max_{{\bf u}\in \mathcal{S}_{{\bf u}}} \mathcal{V}_{\lambda,\rho}({\bf u})-\frac{\tilde{\lambda}}{2}\gamma\epsilon^2\right\}\cap \left\{\max_{{\bf u}\in \mathcal{S}_{{\bf u}}} \mathcal{V}_{\lambda,\rho}({\bf u})\geq \psi(\tau^\star,\beta^\star)-\gamma\frac{\tilde{\lambda}}{2}\epsilon^2\right\} \right] \\
	& + \mathbb{P}\left[\max_{{\bf u}\in \mathcal{S}_{{\bf u}}} \mathcal{V}_{\lambda,\rho}({\bf u})\leq \psi(\tau^\star,\beta^\star)-\gamma\frac{\tilde{\lambda}}{2}\epsilon^2\right] \\
	\leq& \mathbb{P}\left[\max_{{\bf u}\in D_\epsilon} \mathcal{V}_{\lambda,\rho}({\bf u})\geq  \psi(\tau^\star,\beta^\star)-\gamma\tilde{\lambda}\epsilon^2 \right]+ \mathbb{P}\left[\max_{{\bf u}\in \mathcal{S}_{{\bf u}}} \mathcal{V}_{\lambda,\rho}({\bf u})\leq \psi(\tau^\star,\beta^\star)-\gamma\frac{\tilde{\lambda}}{2}\epsilon^2\right].
\end{align*}
Next, we use \eqref{eq:2} and \eqref{eq:4} to obtain:
\begin{align*}
	&\mathbb{P}\left[\max_{{\bf u}\in D_\epsilon}\mathcal{V}_{\lambda,\rho}({\bf u})\geq \max_{{\bf u}}\mathcal{V}_{\lambda,\rho}({\bf u})-\alpha\gamma\epsilon^2\right]\\
	\leq& 2\mathbb{P}\left[\max_{{\bf u}\in D_\epsilon}\mathcal{F}_{\lambda,\rho}({\bf u})\geq \psi(\tau^\star,\beta^\star)-\tilde{\lambda}\gamma\epsilon^2\right] + 2\mathbb{P}\left[\max_{{\bf u}\in \mathcal{S}_{{\bf u}}} \mathcal{F}_{\lambda,\rho}({\bf u})\leq \psi(\tau^\star,\beta^\star)-\gamma\frac{\tilde{\lambda}}{2}\epsilon^2\right].
\end{align*}
The second term in the above right-hand side can be upper-bounded by $\frac{C}{\epsilon^2}\exp(-cn\epsilon^4)$. To handle the first term, we use the fact that on the event \eqref{eq:evu} occurring with $1-C\exp(-cn\epsilon^4)$, 
$$
\|{\bf u}-\overline{\bf u}^{\rm AO}\|^2\geq \tilde{\lambda}\epsilon^2,
$$
hence, 
\begin{equation*}
	\begin{aligned}
&\mathbb{P}\left[\max_{{\bf u}\in D_\epsilon}\mathcal{F}_{\lambda,\rho}({\bf u})\geq \psi(\tau^\star,\beta^\star)-\tilde{\lambda}\gamma\epsilon^2\right]\\\leq& \mathbb{P}\Big[\exists {\bf u} \ \text{such that } \|{\bf u}-\overline{\bf u}^{\rm AO}\|^2\geq \tilde{\lambda}\epsilon^2 \ \text{and} \  \mathcal{F}_{\lambda,\rho}({\bf u})\geq \psi(\tau^\star,\beta^\star)-\tilde{\lambda}\gamma\epsilon^2\Big]\\
\leq &C\exp(-cn\epsilon^4),
\end{aligned}
\end{equation*}
where in the last inequality we used \eqref{eq:uniform_u}.

\section{Study of the sub-gradient $\hat{\bf t}^{\rm PO}$ via the AO problem $\max_{\|{\bf t}\|_\infty\leq 1}\max_{{\bf u}\in\mathcal{S}_{\bf u}}\mathcal{S}_{\lambda,\rho}({\bf t,u})$}
The goal of this section is to study the properties of the subgradient of the $\ell_1$-norm, denoted as $\hat{\mathbf{t}}^{\rm PO}$, at the vector $\hat{\mathbf{x}}^{\rm PO}$,   the solution to the PO problem in \eqref{eq:PO1}. A crucial observation is that $\hat{\mathbf{t}}^{\rm PO}$ also serves as a solution in $\mathbf{t}$ to the PO problem in \eqref{eq:PO3}. Utilizing this insight, we can apply cGMT to study the asymptotic behavior of $\hat{\mathbf{t}}^{\rm PO}$. 

Our methodology aligns with the spirit of previous sections by associating the PO problem \eqref{eq:PO3} with the AO problem in \eqref{eq:AO3}. However, convergence of the optimal cost of \eqref{eq:PO3} can already be obtained from the previous studies. Leveraging prior results from the analysis of the PO problem \eqref{eq:PO1}  we can deduce that the optimal cost of \eqref{eq:PO3} concentrates around $\psi(\tau^\star,\beta^\star)$. Specifically, using \eqref{eq:equality}, we can directly assert that
\begin{equation}
\sup_{{\bf s}}\mathbb{P}\Big[|\max_{\|{\bf t}\|_{\infty}\leq 1 }\max_{{\bf u}\in \mathcal{S}_{\bf u}}\mathcal{T}_{\lambda,\rho}({\bf t},{\bf u})-\psi(\tau^\star,\beta^\star)|\geq \gamma\epsilon\Big]\leq \frac{C}{\epsilon}\exp(-cn\epsilon^2).\label{eq:subgradient}
\end{equation}
To infer the properties of the subgradient solution of the PO in \eqref{eq:PO3}, we consider the AO objective function:
$$
\mathcal{Z}_{\lambda,\rho}({\bf t}):=\max_{{\bf u}\in\mathcal{S}_{\bf u}}\mathcal{S}_{\lambda,\rho}({\bf t},{\bf u}),
$$
our objective is to demonstrate that the subgradient vector $\hat{\mathbf{t}}^{\rm PO}$ has the same asymptotic distribution as its AO equivalent solution $\overline{\mathbf{t}}^{\rm AO}$, which depends on the random variables of the AO. The elements of $\overline{\mathbf{t}}^{\rm AO}$ are defined as follows:
\begin{equation}
	[\overline{\bf t}^{\rm AO}]_i:=\left\{
	\begin{array}{ll}
		-\lambda_1^{-1}\beta^\star(\frac{[\overline{\bf x}^{\rm AO}]_i}{\tau^\star}-[{\bf h}]_i)-2\lambda_2\lambda_1^{-1}[\overline{\bf x}^{\rm AO}]_i & \text{if } |[\overline{\bf x}^{\rm AO}]_i|<\sqrt{P}\\
		{\rm sign}([\overline{\bf x}^{\rm AO}]_i) &\text{if } [\overline{\bf x}^{\rm AO}]_i\in\{-\sqrt{P},\sqrt{P}\}
	\end{array}
	\right.
	\label{eq:tAO1}
\end{equation}
where $\overline{\bf x}^{\rm AO}:={\rm prox}(\tilde{\tau}^\star{\bf h};\frac{\lambda_1\tilde{\tau}^\star}{\beta})$.
The core of the proof revolves around showing that the optimal solution of the AO objective $\mathcal{Z}({\bf t})$ lies within the ball centered at $\overline{\mathbf{t}}^{\rm AO}$ with high probability. More formally, we prove that there exists $\gamma_t$, $a_t$ such that for all $\epsilon$ sufficiently small: 
\begin{equation}
\sup_{{\bf s}}\mathbb{P}\Big[\exists \ {\bf t}\in \mathcal{B}_{\infty}, \frac{1}{n}\|{\bf t}-\overline{\bf t}^{\rm AO}\|_1\geq a_t\sqrt{\epsilon}, \ \ \mathcal{Z}_{\lambda,\rho}({\bf t})\geq \psi(\tau^\star,\beta^\star)-\gamma_t\epsilon\Big]\leq \frac{C}{\epsilon}\exp(-cn\epsilon^2)\label{eq:Zt}
\end{equation}
where $\mathcal{B}_{\infty}:=\{{\bf t}\in \mathbb{R}^n, \|{\bf t}\|_{\infty}\leq 1\}$. 
 As will be demonstrated in Section \ref{sec:exp}, these results can be exploited to infer properties of the solution $\hat{\mathbf{t}}^{\rm PO}$. Readers interested in this aspect can skip the proof of \eqref{eq:Zt} and refer to \ref{sec:exp} to see how it directly leads to insights regarding the asymptotic behavior of the PO solution $\hat{\mathbf{t}}^{\rm PO}$.
\subsection{Study of the AO problem}
\label{sec:AOt}
\subsubsection{Upper-bound on the AO objective}
Recalling the expression of $\mathcal{S}_{\lambda,\rho}({\bf t},{\bf u})$, $\mathcal{Z}_{\lambda,\rho}({\bf t})$ writes as:
\begin{align*}
	\mathcal{Z}_{\lambda,\rho}({\bf t})=\max_{{\bf u}\in \mathcal{S}_{{\bf u}}}\min_{\|{\bf x}\|_{\infty}\leq \sqrt{P}}\frac{1}{n}\|{\bf x}\|_2{\bf g}^{T}{\bf u}-\frac{1}{n}\|{\bf u}\|{\bf h}^{T}{\bf x}-\sqrt{\rho}\frac{{\bf u}^{T}{\bf s}}{\sqrt{n}}-\frac{\|{\bf u}\|^2}{4}+\frac{\lambda_1}{n}{\bf t}^{T}{\bf x}+\frac{\lambda_2}{n}\|{\bf x}\|^2.
\end{align*}
By switching the order of the max-min, we can build the following upper-bound for $\mathcal{Z}({\bf t})$, 
\begin{align}
	\tilde{\mathcal{Z}}_{\lambda,\rho}({\bf t}):=\min_{\|{\bf x}\|_{\infty}\leq \sqrt{P}}\max_{{\bf u}\in \mathcal{S}_{{\bf u}}}\frac{1}{n}\|{\bf x}\|_2{\bf g}^{T}{\bf u}-\frac{1}{n}\|{\bf u}\|{\bf h}^{T}{\bf x}-\sqrt{\rho}\frac{{\bf u}^{T}{\bf s}}{\sqrt{n}}-\frac{\|{\bf u}\|^2}{4}+\frac{\lambda_1}{n}{\bf t}^{T}{\bf x}+\frac{\lambda_2}{n}\|{\bf x}\|^2 \label{eq:exp}
\end{align}
which obviously satisfies:
$$
\forall {\bf t}\in \mathcal{B}_{\infty}, \ \ \mathcal{Z}_{\lambda,\rho}({\bf t})\leq \tilde{\mathcal{Z}}_{\lambda,\rho}({\bf t}),
$$
and hence,
\begin{align*}
&\sup_{{\bf s}}\mathbb{P}\Big[\exists \ {\bf t}\in \mathcal{B}_{\infty}, \frac{1}{n}\|{\bf t}-\overline{\bf t}^{\rm AO}\|_1\geq a_t\sqrt{\epsilon}, \ \ \mathcal{Z}_{\lambda,\rho}({\bf t})\geq \psi(\tau^\star,\beta^\star)-\gamma_t\epsilon\Big]  \\
\leq& \sup_{{\bf s}}\mathbb{P}\Big[\exists \ {\bf t}\in \mathcal{B}_{\infty}, \frac{1}{n}\|{\bf t}-\overline{\bf t}^{\rm AO}\|_1\geq a_t\sqrt{\epsilon}, \ \ \tilde{\mathcal{Z}}_{\lambda,\rho}({\bf t})\geq \psi(\tau^\star,\beta^\star)-\gamma_t\epsilon\Big] . 
\end{align*}
Hence, to prove \eqref{eq:Zt}, it suffices to show that:
\begin{equation}
	\sup_{{\bf s}}\mathbb{P}\Big[\exists \ {\bf t}\in \mathcal{B}_{\infty}, \frac{1}{n}\|{\bf t}-\overline{\bf t}^{\rm AO}\|_1\geq a_t\sqrt{\epsilon}, \ \ \tilde{\mathcal{Z}}_{\lambda,\rho}({\bf t})\geq \psi(\tau^\star,\beta^\star)-\gamma_t\epsilon\Big]\leq \frac{C}{\epsilon}\exp(-cn\epsilon^2).\label{eq:Z_tildet}
\end{equation}
\subsubsection{Asymptotic equivalence for $\tilde{\mathcal{Z}}_{\lambda,\rho}({\bf t})$}
 By optimizing the objective function in \eqref{eq:exp} with respect to the variable ${\bf u}$, we obtain: 
$$
\tilde{\mathcal{Z}}_{\lambda,\rho}({\bf t})=\min_{\|{\bf x}\|_{\infty}\leq \sqrt{P}}\left(\|\frac{\|{\bf x}\|{\bf g}}{n}-\frac{\sqrt{\rho}{\bf s}}{\sqrt{n}}\|-\frac{1}{n}{\bf h}^{T}{\bf x}\right)_{+}^2+\frac{\lambda_1}{n}{\bf t}^{T}{\bf x}+\frac{\lambda_2}{n}\|{\bf x}\|^2.
$$
Using the same calculations that led to \eqref{eq:reference}, we can prove that with probability $1-C\exp(-cn\epsilon^2)$ for $\epsilon$ sufficiently small,
$$
|\tilde{\mathcal{Z}}_{\lambda,\rho}({\bf t})-\tilde{\mathcal{Z}}_{\lambda,\rho}^\circ({\bf t})|\leq K\epsilon,
$$
where 
\begin{align*}
	\tilde{\mathcal{Z}}_{\lambda,\rho}^\circ({\bf t})=\min_{\|{\bf x}\|_{\infty}\leq \sqrt{P}}\left(\sqrt{\delta}\sqrt{\frac{\|{\bf x}\|^2}{n}+\rho}-\frac{1}{n}{\bf h}^{T}{\bf x}\right)_{+}^2+\frac{\lambda_1}{n}{\bf t}^{T}{\bf x}+\frac{\lambda_2}{n}\|{\bf x}\|^2 . 
\end{align*}
Hence, proving the following probability inequality:
\begin{equation}
	\sup_{{\bf s}}\mathbb{P}\Big[\exists \ {\bf t}\in \mathcal{B}_{\infty}, \frac{1}{n}\|{\bf t}-\overline{\bf t}^{\rm AO}\|_1\geq a_t\sqrt{\epsilon}, \ \ \tilde{\mathcal{Z}}_{\lambda,\rho}^\circ({\bf t})\geq \psi(\tau^\star,\beta^\star)-(\gamma_t+K)\epsilon\Big]\leq \frac{C}{\epsilon}\exp(-cn\epsilon^2)\label{eq:Z_circ}
\end{equation}
directly leads to \eqref{eq:Z_tildet}. In the sequel, we will thus focus on showing \eqref{eq:Z_circ}. 
\subsubsection{Asymptotic equivalence for $\overline{\bf t}^{\rm AO}$}
 Define $\tilde{\bf x}^{\star}$ a solution to the following optimization problem:
$$
\tilde{\bf x}^\star\in\arg\min_{\|{\bf x}\|_{\infty}\leq \sqrt{P}} \left(\sqrt{\delta}\sqrt{\frac{\|{\bf x}\|^2}{n}+\rho}-\frac{1}{n}{\bf h}^{T}{\bf x}\right)_{+}^2+\frac{\lambda_1}{n}\|{\bf x}\|_1+\frac{\lambda_2}{n}\|{\bf x}\|^2
$$
and let  $\tilde{\bf t}^{\star}=\argmax_{\|\tilde{\bf t}\|_{\infty}\leq 1}\tilde{\mathcal{Z}}_{\lambda,\rho}^\circ({\bf t})$.

Then from the first order optimality conditions, $\tilde{\bf t}^\star$ is a subgradient of the $\ell_1$-norm at ${\bf x}=\tilde{\bf x}^\star$, formally stated as:
$$
\tilde{\bf t}^\star\in \partial \|\tilde{\bf x}^\star\|_1
$$
with
$$
[\tilde{\bf t}^{\star}]_i=\left\{\begin{array}{ll} -\lambda_1^{-1}(\hat{\beta}^\star)_{+}(\frac{[\tilde{\bf x}^\star]_i}{\hat{\tau}^\star}-[{\bf h}]_i)-2\lambda_2\lambda_1^{-1}[\tilde{\bf x}^\star]_i & \text{if } |[\tilde{\bf x}^\star]_i|<\sqrt{P}\\{\rm sign}([\tilde{\bf x}^\star]_i) &\text{if } [\tilde{\bf x}^\star]_i\in\{-\sqrt{P},\sqrt{P}\}
\end{array}\right.
$$
where 
\begin{align*}
\hat{\beta}^\star&=2\big(\sqrt{\delta}\sqrt{\frac{1}{n}\|\tilde{\bf x}^\star\|^2+\rho}-\frac{1}{n}{\bf h}^{T}\tilde{\bf x}^\star\big),\\
	\hat{\tau}^\star&=\frac{1}{\sqrt{\delta}}\sqrt{\frac{\|\tilde{\bf x}^\star\|^2}{n}+\rho}.
\end{align*}
Next we show that for any $a_t>0$, we can choose  $\epsilon$ sufficiently small,  such that with probability at least $1-\frac{C}{\epsilon}\exp(-cn\epsilon^2)$, 
\begin{equation}
\frac{1}{n}\|\overline{\bf t}^{\rm AO}-\tilde{\bf t}^\star\|_1\leq \frac{a_t\sqrt{\epsilon}}{2}. \label{eq:t_ineq}
\end{equation}
It follows from Lemma \ref{eq:lem_ball}, that with probability $1-\frac{C}{{\epsilon}}\exp(-cn\epsilon^2)$ for $\epsilon$ sufficiently small, 
\begin{equation}
\frac{1}{n}\|\tilde{\bf x}^\star-\overline{\bf x}^{\rm AO}\|^2\leq {\epsilon}. \label{eq:x}
\end{equation}
 By plugging the expression of $\overline{\bf x}^{\rm AO}$ into that of \eqref{eq:tAO1}, we can easily check that for all $i=1,\cdots,n$
$$
[\overline{\bf t}^{\rm AO}]_i={\rm sign}([{\bf x}^{\rm AO}]_i), \ \  \ \text{if } |[{\bf x}^{\rm AO}]_i|= \sqrt{P}.
$$
Let $i=1,\cdots,n$. We can thus obtain the following upper-bound:
\begin{align*}
	\left|[\overline{\bf t}^{\rm AO}]_i-[\tilde{\bf t}^\star]_i\right|\leq& 2{\bf 1}_{\{[\tilde{\bf x}^\star]_i\leq 0,[\overline{\bf x}^{\rm AO}]_i=\sqrt{P}\}}+2{\bf 1}_{\{[\tilde{\bf x}^\star]_i\geq 0,[\overline{\bf x}^{\rm AO}]_i=-\sqrt{P}\}}\\
	&+\lambda_1^{-1}|(\hat{\beta}^\star)_{+}(\frac{[\tilde{\bf x}^\star]_i}{\hat{\tau}^\star}-[{\bf h}]_i)-\beta^\star(\frac{[\overline{\bf x}^{\rm AO}]_i}{{\tau}^\star}-[{\bf h}]_i)|{\bf 1}_{\{|[\tilde{\bf x}^\star]_i|<\sqrt{P},|[\overline{\bf x}^{\rm AO}]_i|< \sqrt{P}\}} \\
	&+2\lambda_2\lambda_1^{-1}|[\overline{\bf x}^{\rm AO}]_i-[\tilde{\bf x}^\star]_i|{\bf 1}_{\{|[\tilde{\bf x}^\star]_i|<\sqrt{P},|[\overline{\bf x}^{\rm AO}]_i|< \sqrt{P}\}} .
\end{align*}
It follows from \eqref{eq:x} that:
$$
{\epsilon}\geq \frac{1}{n}\|\tilde{\bf x}^\star-\overline{\bf x}^{\rm AO}\|^2\geq \frac{1}{n}\sum_{i=1}^n |[\tilde{\bf x}^\star]_i-[\overline{\bf x}^{\rm AO}]_i|^2 {\bf 1}_{\{\{[\tilde{\bf x}^\star]_i\leq 0,[\overline{\bf x}^{\rm AO}]_i=\sqrt{P}\}}\geq P \frac{1}{n}\sum_{i=1}^n{\bf 1}_{\{\{[\tilde{\bf x}^\star]_i\leq 0,[\overline{\bf x}^{\rm AO}]_i=\sqrt{P}\}},
$$
hence, 
\begin{equation}
\frac{1}{n}\sum_{i=1}^n{\bf 1}_{\{\{[\tilde{\bf x}^\star]_i\leq 0,[\overline{\bf x}^{\rm AO}]_i=\sqrt{P}\}}\leq \frac{{\epsilon}}{P} .\label{eq:xstar_1}
\end{equation}
In the same way, we can prove that:
\begin{equation}
\frac{1}{n}\sum_{i=1}^n{\bf 1}_{\{\{[\tilde{\bf x}^\star]_i\geq 0,[\overline{\bf x}^{\rm AO}]_i=-\sqrt{P}\}}\leq \frac{{\epsilon}}{P}. \label{eq:xstar_2}
\end{equation}
By combining \eqref{eq:xstar_1}, \eqref{eq:xstar_2} along with \eqref{eq:x}, we 
 can obtain \eqref{eq:t_ineq}.
 Then using the following triangular inequality:
$$
\frac{1}{n}\|{\bf t}-\tilde{\bf t}^\star\|_1\geq \frac{1}{{n}}\|{\bf t}-\overline{\bf t}^{\rm AO}\|_1-\frac{1}{{n}}\|\tilde{\bf t}^\star-\overline{\bf t}^{\rm AO}\|_1,
$$
we deduce that with probability $1-\frac{C}{\epsilon}\exp(-cn\epsilon^2)$, for any ${\bf t}\in \mathbb{R}^n$,
$$
\frac{1}{n}\|{\bf t}-\overline{\bf t}^{\rm AO}\|_1\geq {a_t\sqrt{\epsilon}} \Longrightarrow \frac{1}{n}\|{\bf t}-\tilde{\bf t}^{\star}\|_1\geq \frac{a_t\sqrt{\epsilon}}{2}.
$$
Hence, the proof of \eqref{eq:Z_circ} amounts to showing:
\begin{equation}
	\sup_{{\bf s}}\mathbb{P}\Big[\exists \ {\bf t}\in \mathcal{B}_{\infty}, \frac{1}{n}\|{\bf t}-\tilde{\bf t}^{\star}\|_1\geq \frac{a_t}{2}\sqrt{\epsilon}, \ \ \tilde{\mathcal{Z}}_{\lambda,\rho}^\circ({\bf t})\geq \psi(\tau^\star,\beta^\star)-(\gamma_t+K)\epsilon\Big]\leq \frac{C}{\epsilon}\exp(-cn\epsilon^2).\label{eq:Z_circ_a}
\end{equation}
\subsubsection{Proof of \eqref{eq:Z_circ_a}}
Let $\pi_{\tilde{\bf x}^\star}$ be the subset in $\{1,\cdots,n\}$ indexing the elements of $\tilde{\bf x}$ whose magnitude is in $[0,\sqrt{P}-\frac{2\lambda_1}{\tilde{c}})$, where  $\tilde{c}$ is a sufficiently large constant  and define $\pi_{\tilde{\bf x}^\star}^c=\{1,\cdots,n\}\backslash \pi_{\tilde{\bf x}^\star}$. For ${\bf a}\in\mathbb{R}^n$, define ${\bf a}_{\pi_{\tilde{\bf x}^\star}}$ to be equal to ${\bf a}$ at elements indexed by  $\pi_{\tilde{\bf x}^\star}$ and zero otherwise and by ${\bf a}_{\pi_{\tilde{\bf x}^\star}^c}$ to be equal to ${\bf a}$ at the elements indexed by the elements in $\pi_{\tilde{\bf x}^\star}^c$ and zero otherwise. Note that for all ${\bf t}\in \mathcal{B}_{\infty}$,
$$
\frac{1}{n}\|{\bf t}-\tilde{\bf t}^{\star}\|_1\geq \frac{a_t\sqrt{\epsilon}}{2} \Longrightarrow  \frac{1}{n}\|{\bf t}_{\pi_{\tilde{\bf x}^\star}}-\tilde{\bf t}_{\pi_{\tilde{\bf x}^\star}}^{\star}\|_1\geq \frac{a_t\sqrt{\epsilon}}{4} \ \ \text{or } \frac{1}{n}\|{\bf t}_{\pi_{\tilde{\bf x}^\star}^c}-\tilde{\bf t}_{\pi_{\tilde{\bf x}^\star}^c}^{\star}\|_1\geq \frac{a_t\sqrt{\epsilon}}{4},
$$
hence, 
\begin{equation*}
	\begin{aligned}
	&\sup_{{\bf s}}\mathbb{P}\Big[\exists \ {\bf t}\in \mathcal{B}_{\infty}, \frac{1}{n}\|{\bf t}-\tilde{\bf t}^{\star}\|_1\geq \frac{a_t}{2}\sqrt{\epsilon}, \ \ \tilde{\mathcal{Z}}_{\lambda,\rho}^\circ({\bf t})\geq \psi(\tau^\star,\beta^\star)-(\gamma_t+K)\epsilon\Big]\\
	\leq& \sup_{{\bf s}}\mathbb{P}\Big[\exists \ {\bf t}\in \mathcal{B}_{\infty}, \frac{1}{n}\|{\bf t}_{\pi_{\tilde{\bf x}^\star}}-\tilde{\bf t}_{\pi_{\tilde{\bf x}^\star}}^{\star}\|_1\geq \frac{a_t}{4}\sqrt{\epsilon}, \ \ \tilde{\mathcal{Z}}_{\lambda,\rho}^\circ({\bf t})\geq \psi(\tau^\star,\beta^\star)-(\gamma_t+K)\epsilon\Big]\\
	&+\sup_{{\bf s}}\mathbb{P}\Big[\exists \ {\bf t}\in \mathcal{B}_{\infty}, \frac{1}{n}\|{\bf t}_{\pi_{\tilde{\bf x}^\star}^c}-\tilde{\bf t}_{\pi_{\tilde{\bf x}^\star}^c}^{\star}\|_1\geq \frac{a_t}{4}\sqrt{\epsilon}, \ \ \tilde{\mathcal{Z}}_{\lambda,\rho}^\circ({\bf t})\geq \psi(\tau^\star,\beta^\star)-(\gamma_t+K)\epsilon\Big].
	\end{aligned}
\end{equation*}

The proof of \eqref{eq:Z_circ_a} reduces to thus showing the following inequalities:
\begin{align}
	\sup_{{\bf s}}\mathbb{P}\Big[\exists \ {\bf t}\in \mathcal{B}_{\infty}, \frac{1}{n}\|{\bf t}_{\pi_{\tilde{\bf x}^\star}}-\tilde{\bf t}_{\pi_{\tilde{\bf x}^\star}}^{\star}\|_1\geq \frac{a_t}{4}\sqrt{\epsilon}, \ \ \tilde{\mathcal{Z}}_{\lambda,\rho}^\circ({\bf t})\geq \psi(\tau^\star,\beta^\star)-(\gamma_t+K)\epsilon\Big]\leq \frac{C}{\epsilon}\exp(-cn\epsilon^2),\label{eq:Z1}\\
	\sup_{{\bf s}}\mathbb{P}\Big[\exists \ {\bf t}\in \mathcal{B}_{\infty}, \frac{1}{n}\|{\bf t}_{\pi_{\tilde{\bf x}^\star}^c}-\tilde{\bf t}_{\pi_{\tilde{\bf x}^\star}^c}^{\star}\|_1\geq \frac{a_t}{4}\sqrt{\epsilon}, \ \ \tilde{\mathcal{Z}}_{\lambda,\rho}^\circ({\bf t})\geq \psi(\tau^\star,\beta^\star)-(\gamma_t+K)\epsilon\Big]\leq \frac{C}{\epsilon}\exp(-cn\epsilon^2).\label{eq:Z_2}
\end{align}
Now to prove the probability inequalities in \eqref{eq:Z1} and \eqref{eq:Z_2}, we express $\tilde{\mathcal{Z}}_{\lambda,\rho}^\circ({\bf t})$ as:
\begin{equation*}
	\begin{aligned}
\tilde{\mathcal{Z}}_{\lambda,\rho}^\circ({\bf t})=&\min_{\|{\bf x}\|_{\infty}\leq \sqrt{P}} \left(\sqrt{\delta}\sqrt{\frac{\|{\bf x}_{\pi_{\tilde{\bf x}^\star}}\|^2}{n}+\frac{\|{\bf x}_{\pi_{\tilde{\bf x}^\star}^c}\|^2}{n}+\rho}-\frac{1}{n}{\bf h}_{\pi_{\tilde{\bf x}^\star}}^{T}{\bf x}_{\pi_{\tilde{\bf x}^\star}}-\frac{1}{n}{\bf h}_{\pi_{\tilde{\bf x}^\star}^c}^{T}{\bf x}_{\pi_{\tilde{\bf x}^\star}^c}\right)_{+}^2\\
&+\frac{\lambda_1}{n}{\bf t}_{\pi_{\tilde{\bf x}^\star}}^{T}{\bf x}_{\pi_{\tilde{\bf x}^\star}}+\frac{\lambda_1}{n}{\bf t}_{\pi_{\tilde{\bf x}^\star}^c}^{T}{\bf x}_{\pi_{\tilde{\bf x}^\star}^c}+\frac{\lambda_2}{n}\|{\bf x}_{\pi_{\tilde{\bf x}^\star}}\|^2+\frac{\lambda_2}{n}\|{\bf x}_{\pi_{\tilde{\bf x}^\star}^c}\|^2
\end{aligned}
\end{equation*}
and consider upper-bounding it by:
\begin{equation}
\begin{aligned}
\tilde{\mathcal{Z}}_{\lambda,\rho}^\circ({\bf t})\leq \tilde{\mathcal{Z}_1}^\circ({\bf t}_{\pi_{\tilde{\bf x}^\star}}):=\min_{\|{\bf x}_{\pi_{\tilde{\bf x}^\star}}\|_{\infty}\leq \sqrt{P}}&\left(\sqrt{\delta}\sqrt{\frac{\|{\bf x}_{\pi_{\tilde{\bf x}^\star}}\|^2}{n}+\frac{\|{\bf \tilde{\bf x}^\star}_{\pi_{\tilde{\bf x}^\star}^c}\|^2}{n}+\rho}-\frac{1}{n}{\bf h}_{\pi_{\tilde{\bf x}^\star}}^{T}{\bf x}_{\pi_{\tilde{\bf x}^\star}}-\frac{1}{n}|{\bf h}_{\pi_{\tilde{\bf x}^\star}^c}^{T}||{\bf \tilde{\bf x}^\star}_{\pi_{\tilde{\bf x}^\star}^c}|\right)_{+}^2\\
&+\frac{\lambda_1}{n}{\bf t}_{\pi_{\tilde{\bf x}^\star}}^{T}{\bf x}_{\pi_{\tilde{\bf x}^\star}}+\frac{\lambda_1}{n}\|\tilde{\bf x}_{\pi_{\tilde{\bf x}^\star}^c}^\star\|_1+\frac{\lambda_2}{n}\|{\bf x}_{\pi_{\tilde{\bf x}^\star}}\|^2+\frac{\lambda_2}{n}\|{\bf \tilde{\bf x}^\star}_{\pi_{\tilde{\bf x}^\star}^c}\|^2
\end{aligned}
\label{eq:strong}
\end{equation}
which is obtained by substituting ${\bf x}_{\pi_{\tilde{\bf x}^\star}^c}$ with  $\tilde{{\bf x}}^\star_{\pi_{\tilde{\bf x}^\star}^c}$
or alternatively by:
\begin{equation}
\tilde{\mathcal{Z}}_{\lambda,\rho}^\circ({\bf t})\leq \tilde{\mathcal{Z}_2}^\circ({\bf t}_{\pi_{\tilde{\bf x}^\star}^c}):=\mathcal{L}_{\lambda,\rho}^\circ(\tilde{\bf x}^\star)-\frac{\lambda_1}{n}\|{\bf x}_{\pi_{\tilde{\bf x}^\star}^c}\|_1+\frac{\lambda_1}{n}{\bf t}_{\pi_{\tilde{\bf x}^\star}^c}^{T}\tilde{\bf x}_{\pi_{\tilde{\bf x}^\star}^c}^\star \label{eq:strong2}
\end{equation}
which is obtained by replacing ${\bf x}$ with $\tilde{\bf x}^\star$ and optimizing  only over ${\bf t}_{\pi_{\tilde{\bf x}}^\star}$. 
Next, we will use \eqref{eq:strong} to show \eqref{eq:Z1} and \eqref{eq:strong2} to show  \eqref{eq:Z_2}.

\noindent{\underline{Proof of \eqref{eq:Z1}}.} For ${\bf x}\in \mathbb{R}^n$, define function ${g}$ as:
$$
g({\bf x})=\sqrt{\delta}\sqrt{\frac{\|{\bf x}_{\pi_{\tilde{\bf x}^\star}}\|^2}{n}+\frac{\|\tilde{\bf x}_{\pi_{\tilde{\bf x}^\star}^c}^\star\|^2}{n}+\rho}-\frac{1}{n}{\bf h}_{\pi_{\tilde{\bf x}^\star}}^{T}{\bf x}_{\pi_{\tilde{\bf x}^\star}}-\frac{1}{n}|{\bf h}_{\pi_{\tilde{\bf x}^\star}^c}^{T}||{\tilde{\bf x}^\star}_{\pi_{\tilde{\bf x}^\star}^c}|.
$$
The gradient and hessian of ${\bf g}$ are given by:
$$
\nabla {\bf g}({\bf x})=\sqrt{\delta}\frac{{\bf x}_{\pi_{\tilde{\bf x}^\star}}}{n\sqrt{\frac{\|{\bf x}_{\pi_{\tilde{\bf x}^\star}}\|^2}{n}+\frac{\|\tilde{\bf x}_{\pi_{\tilde{\bf x}^\star}^c}^\star\|^2}{n}+\rho}}-\frac{1}{n}{\bf h}_{\pi_{\tilde{\bf x}^\star}},
$$
\begin{align*}
\nabla^2{\bf g}({\bf x})&=\frac{\sqrt{\delta}}{n\sqrt{\frac{\|{\bf x}_{\pi_{\tilde{\bf x}^\star}}\|^2}{n}+\frac{\|\tilde{\bf x}_{\pi_{\tilde{\bf x}^\star}^c}^\star\|^2}{n}+\rho}}\left({\rm diag}({\bf 1}_{\pi_{\tilde{\bf x}^\star}^c})-\frac{{\bf x}_{\pi_{\tilde{\bf x}^\star}}{\bf x}_{\pi_{\tilde{\bf x}^\star}}^{T}}{n\sqrt{\frac{\|{\bf x}_{\pi_{\tilde{\bf x}^\star}}\|^2}{n}+\frac{\|\tilde{\bf x}_{\pi_{\tilde{\bf x}^\star}^c}^\star\|^2}{n}+\rho}}\right)\\
&\preceq \frac{\sqrt{\delta}}{n\sqrt{\frac{\|{\bf x}_{\pi_{\tilde{\bf x}^\star}}\|^2}{n}+\frac{\|\tilde{\bf x}_{\pi_{\tilde{\bf x}^\star}^c}^\star\|^2}{n}+\rho}}{\bf I}_n.
\end{align*}
As a consequence, 
$$
\|\nabla {\bf g}({\bf x})\|^2\leq \frac{2\delta}{n}+\frac{2}{n^2}\|{\bf h}\|^2
$$
and 
$$
g({\bf x})\|\nabla^2{\bf g}({\bf x})\|\leq \frac{\delta}{n}+\frac{1}{n}\sqrt{\delta}\frac{\|{\bf h}\|}{\sqrt{n}}.
$$
The Hessian of $h({\bf x}):=({g}({\bf x}))_{+}^2$ is $2(\nabla {\bf g}({\bf x})\nabla {\bf g}({\bf x})^T+g({\bf x})\nabla^2{\bf g}({\bf x})){\bf 1}_{{g}({\bf x})\geq 0}$. Hence, on the event $\frac{\|{\bf h}\|}{\sqrt{n}}\leq \sqrt{2}$ occurring with probability $1-C\exp(-cn)$, there exists a constant $c_h$ such that:
$$
\|\nabla^2h({\bf x})\|\leq \frac{c_h}{n}.
$$
Hence, function $h$ is $\frac{c_h}{n}$ strongly smooth \cite[Definiton G.1]{Miolane}. 
Now let function $k({\bf x}):=h({\bf x})+\frac{\lambda_1}{n}\|\tilde{\bf x}_{\pi_{\tilde{\bf x}^\star}^c}^\star\|_1+\frac{\lambda_2}{n}\|\tilde{\bf x}_{\pi_{\tilde{\bf x}^\star}^c}^\star\|^2+\frac{\lambda_2}{n}\|{\bf x}_{\pi_{\tilde{\bf x}^\star}}\|^2$. Obviously, function $k$ is also $\frac{c_k}{n}:=\frac{c_h+\lambda_2}{n}$ strongly smooth. Therefore, for any ${\bf x}$, ${\bf y}$ in $\mathbb{R}^n$, 
$$
{ k}({\bf x})\leq { k}({\bf y})+(\nabla k({\bf y}))^{T}({\bf x}-{\bf y})+\frac{c_k}{2n}\|{\bf y}-{\bf x}\|^2\leq { k}({\bf y})+(\nabla k({\bf y}))^{T}({\bf x}-{\bf y})+\frac{\tilde{c}}{2n}\|{\bf y}-{\bf x}\|^2
$$
where $\tilde{c}$ is chosen larger than $c_k$ and such that $\sqrt{P}-\frac{2\lambda_1}{\tilde{c}}>0$.
To continue, we express  $\tilde{\mathcal{Z}_1}^\circ({\bf t}_{\pi_{\tilde{\bf x}^\star}})$ as:
\begin{align*}
	\tilde{\mathcal{Z}_1}^\circ({\bf t}_{\pi_{\tilde{\bf x}^\star}})&=\min_{\|\tilde{\bf x}_{\pi_{\tilde{\bf x}^\star}}\|_{\infty}\leq \sqrt{P}}\frac{\lambda_1}{n}{\bf t}_{\pi_{\tilde{\bf x}^\star}}^{T}\tilde{\bf x}_{\pi_{\tilde{\bf x}^\star}}+k(\tilde{\bf x}_{\pi_{\tilde{\bf x}}^\star})\\
	&\leq \min_{\|\tilde{\bf x}_{\pi_{\tilde{\bf x}^\star}}\|_{\infty}\leq \sqrt{P}}\frac{\lambda_1}{n}{\bf t}_{\pi_{\tilde{\bf x}^\star}}^{T}\tilde{\bf x}_{\pi_{\tilde{\bf x}^\star}}+k(\tilde{\bf x}^\star_{\pi_{\tilde{\bf x}^\star}})+\nabla k(\tilde{\bf x}_{\pi_{\tilde{\bf x}^\star}}^\star)^{T}(\tilde{\bf x}_{\pi_{\tilde{\bf x}^\star}}-\tilde{\bf x}_{\pi_{\tilde{\bf x}^\star}}^\star)+\frac{\tilde{c}}{2n}\|\tilde{\bf x}_{\pi_{\tilde{\bf x}^\star}}^\star-\tilde{\bf x}_{\pi_{\tilde{\bf x}^\star}}\|^2.\end{align*}
Let ${\bf u}=\tilde{\bf x}_{\pi_{\tilde{\bf x}^\star}}-\tilde{\bf x}_{\pi_{\tilde{\bf x}^\star}}^\star$. Then, we obtain:
$$
\tilde{\mathcal{Z}_1}^\circ({\bf t}_{\pi_{\tilde{\bf x}^\star}})\leq \min_{\|{\bf u}+\tilde{\bf x}_{\pi_{\tilde{\bf x}^\star}}^\star\|_{\infty}\leq \sqrt{P}}\frac{\lambda_1}{n}{\bf t}_{\pi_{\tilde{\bf x}^\star}}^{T}{\bf u}+k(\tilde{\bf x}_{\pi_{\tilde{\bf x}^\star}}^\star)+\nabla k(\tilde{\bf x}_{\pi_{\tilde{\bf x}^\star}}^\star)^{T}{\bf u}+\frac{\tilde{c}}{2n}\|{\bf u}\|^2+\frac{\lambda_1}{n}{\bf t}_{\pi_{\tilde{\bf x}^\star}}^{T}\tilde{\bf x}_{\pi_{\tilde{\bf x}^\star}}^\star.
$$
Now, from the definition of function $k$, and from the optimality conditions for $\tilde{\bf x}^\star$, we have:
$$
\nabla k(\tilde{\bf x}_{\pi_{\tilde{\bf x}^\star}}^\star)+\frac{\lambda_1}{n} \tilde{\bf t}_{\pi_{\tilde{\bf x}^\star}}^\star=0.
$$
Hence,
$$
\tilde{\mathcal{Z}_1}^\circ({\bf t}_{\pi_{\tilde{\bf x}^\star}})\leq \min_{\|{\bf u}+\tilde{\bf x}_{\pi_{\tilde{\bf x}^\star}}^\star\|_{\infty}\leq \sqrt{P}}\frac{\lambda_1}{n}({\bf t}_{\pi_{\tilde{\bf x}^\star}}-\tilde{\bf t}_{\pi_{\tilde{\bf x}^\star}}^{\star})^{T}{\bf u}+k(\tilde{\bf x}_{\pi_{\tilde{\bf x}^\star}}^\star)+\frac{\tilde{c}}{2n}\|{\bf u}\|^2+\frac{\lambda_1}{n}{\bf t}_{\pi_{\tilde{\bf x}^\star}}^{T}\tilde{\bf x}_{\pi_{\tilde{\bf x}^\star}}^\star.
$$
Function ${\bf u}:\mapsto \frac{\lambda_1}{n}({\bf t}_{\pi_{\tilde{\bf x}^\star}}-\tilde{\bf t}_{\pi_{\tilde{\bf x}^\star}}^{\star})^{T}{\bf u}+k(\tilde{\bf x}_{\pi_{\tilde{\bf x}^\star}}^\star)+\frac{\tilde{c}}{2n}\|{\bf u}\|^2$ is minimized when ${\bf u}={\bf u}^\star({\bf t})$ with ${\bf u}^\star({\bf t})$ given by:
$$
{\bf u}^\star({\bf t}):=-\frac{\lambda_1({\bf t}_{\pi_{\tilde{\bf x}^\star}}-\tilde{\bf t}_{\pi_{\tilde{\bf x}^\star}}^\star)}{\tilde{c}},
$$
hence, it remains to check that $\|{\bf u}^\star({\bf t})+\tilde{\bf x}_{\pi_{\tilde{\bf x}^\star}}^\star\|_{\infty}\leq \sqrt{P}$. This follows easily by using  the definition of $\pi_{\tilde{\bf x}^\star}$. Replacing ${\bf u}^\star({\bf t})$ by its value, we obtain:
$$
\tilde{\mathcal{Z}_1}^\circ({\bf t}_{\pi_{\tilde{\bf x}^\star}})\leq k(\tilde{\bf x}_{\pi_{\tilde{\bf x}^\star}}^\star)-\frac{\lambda_1^2}{2\tilde{ c}n}\|{\bf t}_{\pi_{\tilde{\bf x}^\star}}-\tilde{\bf t}_{\pi_{\tilde{\bf x}^\star}}^\star\|^2 +\frac{\lambda_1}{n}\|\tilde{\bf x}^\star_{\pi_{\tilde{\bf x}^\star}}\|_1.
$$
It takes no much effort to check that:
$$
k(\tilde{\bf x}_{\pi_{\tilde{\bf x}^\star}}^\star)+\frac{\lambda_1}{n}\|\tilde{\bf x}^\star_{\pi_{\tilde{\bf x}^\star}}\|_1=\mathcal{L}_{\lambda,\rho}^\circ(\tilde{\bf x}^\star).
$$
Using the fact that with probability $1-\frac{C}{\epsilon}\exp(-cn\epsilon^2)$, 
$$
\mathcal{L}_{\lambda,\rho}^\circ(\tilde{\bf x}^\star)\leq \psi(\tau^\star,\beta^\star)+(\gamma_l+K)\epsilon,
$$
we obtain:
$$
\tilde{\mathcal{Z}_1}^\circ({\bf t}_{\pi_{\tilde{\bf x}^\star}})\leq \psi(\tau^\star,\beta^\star)+(\gamma_l+K)\epsilon-\frac{\lambda_1^2}{2\tilde{ c}n}\|{\bf t}_{\pi_{\tilde{\bf x}^\star}}-\tilde{\bf t}_{\pi_{\tilde{\bf x}^\star}}^\star\|^2.
$$
Using the fact
$$
\frac{1}{n}\|{\bf t}_{\pi_{\tilde{\bf x}}^\star}-\tilde{\bf t}^\star_{\pi_{\tilde{\bf x}}^\star}\|_1\geq \frac{a_t\sqrt{\epsilon}}{4}\Longrightarrow \frac{1}{n}\|{\bf t}_{\pi_{\tilde{\bf x}^\star}}-\tilde{\bf t}_{\pi_{\tilde{\bf x}^\star}}^\star\|^2\geq \frac{a_t^2\epsilon}{16},
$$
we observe that by setting $\gamma_t:=\max(\gamma_l,\gamma_u)$ and ${a}_t$ a constant such that: $\frac{\lambda_1^2a_t^2}{32\tilde{c}}\geq 2(\gamma_t+K)$, we get that with probability at least $1-\frac{C}{\epsilon}\exp(-cn\epsilon^2)$, 
$$
\frac{1}{n}\|{\bf t}_{\pi_{\tilde{\bf x}}^\star}-\tilde{\bf t}_{\pi_{\tilde{\bf x}}^\star}^\star\|_1\geq \frac{a_t\sqrt{\epsilon}}{4}\Longrightarrow \tilde{\mathcal{Z}_1}^\circ({\bf t}_{\pi_{\tilde{\bf x}^\star}})\leq \psi(\tau^\star,\beta^\star)-(\gamma_t+K)\epsilon\Longrightarrow \tilde{\mathcal{Z}}^\circ_{\lambda,\rho}({\bf t})\leq \psi(\tau^\star,\beta^\star)-(\gamma_t+K)\epsilon.
$$
\noindent \underline{ Proof of \eqref{eq:Z_2}} We start by expressing $\tilde{\mathcal{Z}_2}^\circ({\bf t}_{\pi_{\tilde{\bf x}^\star}^c})$ as:
$$
\tilde{\mathcal{Z}_2}^\circ({\bf t}_{\pi_{\tilde{\bf x}^\star}^c})=\mathcal{L}_{\lambda,\rho}^\circ(\tilde{\bf x}^\star)+\frac{\lambda_1}{n}({\bf t}_{\pi_{\tilde{\bf x}^\star}^c}-\tilde{\bf t}_{\pi_{\tilde{\bf x}^\star}^c}^\star)^{T}\tilde{\bf x}_{\pi_{\tilde{\bf x}^\star}^c}^\star.
$$
Now let ${\bf t}\in \mathcal{B}_{\infty}$ such that:
\begin{equation}
\frac{1}{n}\|{\bf t}_{\pi_{\tilde{\bf x}^\star}^c}-\tilde{\bf t}_{\pi_{\tilde{\bf x}^\star}^c}^\star\|_1\geq\frac{a_t}{4}\sqrt{\epsilon}. \label{eq:relation_t}
\end{equation}
Since the non-zero elements of $\tilde{\bf t}_{\pi_{\tilde{\bf x}^\star}^c}^\star$ are equal to either $1$ or $-1$, any ${\bf t}\in \mathcal{B}_{\infty}$ satisfying \eqref{eq:relation_t} can be written as:
$$
{\bf t}=\tilde{\bf t}^\star\circ (\boldsymbol{1}-\Delta)
$$
where $\Delta_{\pi_{\tilde{\bf x}^\star}^c}$ has non-negative elements and satisfies  $\frac{\|\Delta_{\pi_{\tilde{\bf x}^\star}^c}\|_1}{n}\geq \frac{a_t}{4}\sqrt{\epsilon}$. We thus obtain for any ${\bf t}\in \mathcal{B}_{\infty}$ satisfying \eqref{eq:relation_t},
$$
\tilde{\mathcal{Z}_2}^\circ({\bf t}_{\pi_{\tilde{\bf x}^\star}^c})\leq \mathcal{L}_{\lambda,\rho}^\circ(\tilde{\bf x}^\star)-\min_{\substack{\Delta\in \mathbb{R}^n \\ \Delta \geq 0 \\\frac{1}{n}\|\Delta_{\pi_{\tilde{\bf x}^\star}^c}\|_1\geq \frac{a_t}{4}\sqrt{\epsilon} }}\frac{\lambda_1}{n}(\sqrt{P}-\frac{2\lambda_1}{\tilde{c}})\Delta^{T}{\bf 1}_{\pi_{\tilde{\bf x}^\star}^c}=\mathcal{L}_{\lambda,\rho}^\circ(\tilde{\bf x}^\star)-\lambda_1(\sqrt{P}-\frac{2\lambda_1}{\tilde{c}})\frac{a_t}{4}\sqrt{\epsilon}.
$$
Using the fact that with probability $1-\frac{C}{\epsilon}\exp(-cn\epsilon^2)$, 
$$
\mathcal{L}_{\lambda,\rho}^\circ(\tilde{\bf x}^\star)\leq \psi(\tau^\star,\beta^\star)+(\gamma_t+K)\epsilon,
$$
we note that by choosing $\gamma_l=\gamma_t$ and $a_t$ such that:
$$
\lambda_1(\sqrt{P}-\frac{2\lambda_1}{\tilde{c}})\frac{a_t}{4}\geq 2(\gamma_t+K),
$$
we obtain that
$$
\tilde{\mathcal{Z}_2}^\circ({\bf t}_{\pi_{\tilde{\bf x}^\star}^c}) \leq \psi(\tau^\star,\beta^\star)-(\gamma_t+K)\epsilon,
$$
which proves \eqref{eq:Z_2}. 
\subsection{From the AO to the PO via exploitation of the cGMT.  }

Let $\nu_{t}^\star$ the law of the following random variable:
$$
T=\left\{\begin{array}{ll}
	\lambda_1^{-1}\tilde{\beta}^\star H & \text{if } |\lambda_1^{-1}\tilde{\beta}^\star H|\leq 1\\
	{\rm sign}(H) &\text{otherwise}
\end{array}\right.
$$
where $H$ is a standard Gaussian random variable. 

The goal of this section is to prove the following theorem.
\begin{theorem}Consider the set $D_{\epsilon}$ defined as:
	$$
	D_{\epsilon}:=\{{\bf t}\in \mathcal{B}_{\infty}, \  \mathcal{W}_1(\hat{\mu}({\bf t}), \nu_t^\star)\geq 2a_t\sqrt{\epsilon}\}.
	$$
	Under Assumption \ref{ass:regime}, \ref{ass:statistic} and \ref{ass:regime_lambda}, there exists constants $C$, $c$ and $\gamma_t$ such that for all $\epsilon$ sufficiently small:
	$$
	\mathbb{P}\Big[\exists {\bf t}\in D_\epsilon \ \text{such that } \max_{{\bf u}\in\mathcal{S}_{\bf u}}\mathcal{T}_{\lambda,\rho}({\bf u},{\bf t})\geq \max_{{\bf t}\in \mathcal{B}_{\infty}}\max_{{\bf u}\in\mathcal{S}_{\bf u}}\mathcal{T}_{\lambda,\rho}({\bf u},{\bf t})-\frac{\gamma_t}{2}\epsilon\Big]\leq \frac{C}{\epsilon}\exp(-cn\epsilon^2).
	$$
\end{theorem}
\begin{proof}
  The proof follows a similar structure to previous analyses and relies heavily on \eqref{eq:3t} and \eqref{eq:5}, which are related to the variable ${\bf t}$. To begin with, we start by writing the following probability inequalities:
\begin{align}
	&\sup_{{\bf s}}\mathbb{P}\Big[\max_{{\bf t}\in D_\epsilon}\max_{{\bf u}\in\mathcal{S}_{\bf u}} \mathcal{T}_{\lambda,\rho}({\bf t},{\bf u})\geq \max_{{\bf t}\in \mathcal{B}_{\infty}}\max_{{\bf u}\in\mathcal{S}_{\bf u}} \mathcal{T}_{\lambda,\rho}({\bf t},{\bf u}) - \frac{\gamma_t}{2} \epsilon\Big] \label{eq:t1}\\
	\leq& \sup_{{\bf s}}\mathbb{P}\Big[\{\max_{{\bf t}\in D_\epsilon}\max_{{\bf u}\in\mathcal{S}_{\bf u}} \mathcal{T}_{\lambda,\rho}({\bf t},{\bf u})\geq \max_{{\bf t}\in \mathcal{B}_{\infty}}\max_{{\bf u}\in \mathcal{S}_{{\bf u}}} \mathcal{T}_{\lambda,\rho}({\bf t},{\bf u}) - \frac{\gamma_t}{2} \epsilon\}\cap\{\max_{{\bf t}\in \mathcal{B}_{\infty}}\max_{{\bf u}\in\mathcal{S}_{\bf u}} \mathcal{T}_{\lambda,\rho}({\bf t},{\bf u})\geq \psi(\tau^\star,\beta^\star)-\frac{\gamma_t}{2}\epsilon\}\Big]\nonumber\\
	&+\sup_{{\bf s}}\mathbb{P}\Big[\max_{{\bf t}\in \mathcal{B}_{\infty}}\max_{{\bf u}\in\mathcal{S}_{\bf u}} \mathcal{T}_{\lambda,\rho}({\bf t},{\bf u})\leq \psi(\tau^\star,\beta^\star)-\frac{\gamma_t}{2}\epsilon\Big]\nonumber\\
	\leq & \sup_{{\bf s}}\mathbb{P}\Big[\max_{{\bf t}\in D_\epsilon}\max_{{\bf u}\in\mathcal{S}_{\bf u}} \mathcal{T}_{\lambda,\rho}({\bf t},{\bf u})\geq \psi(\tau^\star,\beta^\star)-\gamma_t\epsilon\Big]+\frac{C}{\epsilon}\exp(-cn\epsilon^2)\label{eq:t3}
	\end{align}
	where the last inequality follows from \eqref{eq:subgradient}. Using \eqref{eq:3t}, we obtain:
	\begin{equation}
	\sup_{{\bf s}}\mathbb{P}\Big[\max_{{\bf t}\in D_\epsilon}\max_{{\bf u}\in\mathcal{S}_{\bf u}} \mathcal{T}_{\lambda,\rho}({\bf t},{\bf u})\geq \psi(\tau^\star,\beta^\star)-\gamma_t\epsilon\Big]\leq 2\sup_{{\bf s}}\mathbb{P}\Big[\max_{{\bf t}\in D_\epsilon} \mathcal{Z}_{\lambda,\rho}({\bf t})\geq \psi(\tau^\star,\beta^\star)-\gamma_t\epsilon\Big]+C\exp(-cn).\label{eq:right}
	\end{equation}
	Next, we exploit the previous analysis of the AO problem in section \ref{sec:AOt} to control the first term of the right-hand side of \eqref{eq:right}. Using the triangular inequality, 
	$$
	 \mathcal{W}_1(\hat{\mu}({\bf t}),\nu_t^\star)\leq \mathcal{W}_1(\hat{\mu}({\bf t}),\hat{\mu}(\overline{\bf t}^{\rm AO}))+\mathcal{W}_1(\hat{\mu}(\overline{\bf t}^{\rm AO}),\nu_t^\star)
	$$
	we conclude that for all ${\bf t}\in D_\epsilon$, 
	$$
	 \mathcal{W}_1(\hat{\mu}({\bf t}),\hat{\mu}(\overline{\bf t}^{\rm AO}))\geq 2a_t\sqrt{\epsilon}- \mathcal{W}_1(\hat{\mu}(\overline{\bf t}^{\rm AO}),\nu_t^\star).
	$$
	It follows from Lemma \ref{lem:convergence_empirical_rate} that with probability $1-C\exp(-cn\epsilon)$, 
	$$
	\mathcal{W}_1(\hat{\mu}(\overline{\bf t}^{\rm AO}),\nu_t^\star)\leq a_t\sqrt{\epsilon},
	$$
	hence, with probability $1-C\exp(-cn\epsilon)$, for all ${\bf t}\in D_\epsilon$, 
	$$
	\frac{1}{n}\|{\bf t}-\overline{\bf t}^{\rm AO}\|_1\geq \mathcal{W}_1(\hat{\mu}({\bf t}),\hat{\mu}(\overline{\bf t}^{\rm AO}))\geq a_t\sqrt{\epsilon}.
	$$
	Using \eqref{eq:Zt}, we obtain that 
	$$
	\sup_{{\bf s}}\mathbb{P}\Big[\max_{{\bf t}\in D_\epsilon} \mathcal{Z}_{\lambda,\rho}({\bf t})\geq \psi(\tau^\star,\beta^\star)-\gamma_t\epsilon\Big]\leq \frac{C}{\epsilon}\exp(-cn\epsilon^2)
	$$
	which proves the desired. \end{proof}

\label{sec:exp}

\section{Sparsity estimation of the precoder solution $\hat{\bf x}_{\ell_1}$}
\label{sec:sparsity_estimation}
\subsection{High probability lower-bound on the sparsity of the solution.}
For $\epsilon>0$, define the set $D_\epsilon$ as:
$$
D_{\epsilon}=\{{\bf x}\in \mathbb{R}^n, \ \ \frac{1}{n}\|{\bf x}\|_0\leq \ell_0^\star-\epsilon\}
$$
where 
$$
\ell_0^\star=2Q(\frac{\lambda_1}{\beta^\star}).
$$
For $r> 0$ and ${\bf x}\in \mathbb{R}^n$, 
define $\#\pi_{(r)}^c(|{\bf x}|)$ as:
$$
\pi_{(r)}^c(|{\bf x}|)=\{i=1,\cdots,n, |[{\bf x}]_i|\geq r\},
$$
the goal of this section is to prove the following theorem.
\begin{theorem}
Under Assumption \ref{ass:regime}, \ref{ass:statistic} and \ref{ass:regime_lambda}, there exists constants $C$, $c$ and $\tilde{\lambda}>0$ such that:
$$
\mathbb{P}\Big[\min_{{\bf x}\in D_\epsilon} \mathcal{C}_{\lambda,\rho}({\bf x})\leq \min_{\|{\bf x}\|_{\infty}\leq \sqrt{P}} \mathcal{C}_{\lambda,\rho}({\bf x})+\gamma\tilde{\lambda}\epsilon^3\Big] \leq \frac{C}{\epsilon^3}\exp(-cn\epsilon^6).
$$
\label{th:lower_bound}
\end{theorem}
Before delving into the proof, it is worth mentioning that the result of theorem \ref{th:lower_bound} directly implies that:
$$
\mathbb{P}\Big[\frac{1}{n}\|\hat{\bf x}_{\ell_1}\|_0\leq \ell_0^\star-\epsilon\Big]\leq \frac{C}{\epsilon^3}\exp(-cn\epsilon^6).
$$
\begin{proof}
Let
$$
{D}_{r}^\circ=\{{\bf x}\in \mathbb{R}^n, \ \ \frac{\pi_{r}^c(|{\bf x}|)}{n}\leq \ell_0^\star-\epsilon\},
$$
then obviously ${D}_\epsilon\subset D_r^\circ$. Hence, 
$$
\mathbb{P}\Big[\min_{{\bf x}\in D_\epsilon} \mathcal{C}_{\lambda,\rho}({\bf x})\leq \min_{\|{\bf x}\|_{\infty}\leq \sqrt{P}} \mathcal{C}_{\lambda,\rho}({\bf x})+\gamma\tilde{\lambda}\epsilon^3\Big]\leq \mathbb{P}\Big[\min_{{\bf x}\in D_r^\circ} \mathcal{C}_{\lambda,\rho}({\bf x})\leq \min_{\|{\bf x}\|_{\infty}\leq \sqrt{P}} \mathcal{C}_{\lambda,\rho}({\bf x})+\gamma\tilde{\lambda}\epsilon^3\Big].
$$
Using the relation
$$
2Q(\frac{r}{\tilde{\tau}^\star}+\frac{\lambda_1}{\beta^\star})\geq 2Q(\frac{\lambda_1}{\beta^\star})-\frac{r}{\tilde{\tau}^\star},
$$
we can easily see by taking $r=\frac{\tilde{\tau}^\star{\epsilon}}{2}$, that:
$$
{\bf x}\in D_r^\circ \Longrightarrow {\bf x}\in \overline{D}
$$
with 
$$
\overline{D}:=\{{\bf x}\in \mathbb{R}^n, \ \frac{\pi_{\frac{\tilde{\tau}^\star{\epsilon}}{2}}^c(|{\bf x}|)}{n}\leq 2Q(\frac{\epsilon}{2}+\frac{\lambda_1}{\beta^\star})-\frac{\epsilon}{2}\}.
$$
Hence, for any $\tilde{\lambda}>0$,
$$
\mathbb{P}\Big[\min_{{\bf x}\in D_r^\circ} \mathcal{C}_{\lambda,\rho}({\bf x})\leq \min_{\|{\bf x}\|_{\infty}\leq \sqrt{P}} \mathcal{C}_{\lambda,\rho}({\bf x})+\gamma\tilde{\lambda}\epsilon^3\Big]\leq \mathbb{P}\Big[\min_{{\bf x}\in \overline{D}} \mathcal{C}_{\lambda,\rho}({\bf x})\leq \min_{\|{\bf x}\|_{\infty}\leq \sqrt{P}} \mathcal{C}_{\lambda,\rho}({\bf x})+\gamma\tilde{\lambda}\epsilon^3\Big].
$$
It follows from Theorem \ref{th:app_threshold} that there exists $\tilde{\lambda}>0$ such that
$$
\mathbb{P}\Big[\min_{{\bf x}\in \overline{D}} \mathcal{C}_{\lambda,\rho}({\bf x})\leq  \min_{\|{\bf x}\|_{\infty}\leq \sqrt{P}} \mathcal{C}_{\lambda,\rho}({\bf x})+\gamma\tilde{\lambda}\epsilon^3\Big]\leq \frac{C}{\epsilon^3}\exp(-cn\epsilon^6)
$$
which shows that
$$
\mathbb{P}\Big[\min_{{\bf x}\in D_\epsilon} \mathcal{C}_{\lambda,\rho}({\bf x})\leq \min_{\|{\bf x}\|_{\infty}\leq \sqrt{P}} \mathcal{C}_{\lambda,\rho}({\bf x})+\gamma\tilde{\lambda}\epsilon^3\Big] \leq \frac{C}{\epsilon^3}\exp(-cn\epsilon^6).
$$
\end{proof}
\subsection{High probability upper-bound on the sparsity of the solution.}

The goal of this section is to prove the following theorem.
\begin{theorem}
	Under Assumption \ref{ass:regime}, \ref{ass:statistic} and \ref{ass:regime_lambda}, there exists constants $C$, $c$ and $b_t$ such that:
	$$
	\mathbb{P}\Big[\frac{1}{n}\|\hat{\bf x}_{\ell_1}\|_0\geq \ell_0^\star+b_t\epsilon\Big]\leq \frac{C}{\epsilon^4}\exp(-cn\epsilon^8).
	$$
	\label{th:upper_bound}
	\end{theorem}
	\begin{proof}
For $a_t>0$, 	Define the set $D_\epsilon$ as:
	$$
	D_\epsilon=\{{\bf x}\in \mathbb{R}^n, \frac{1}{n}\|{\bf x}\|_0\geq \ell_0^\star+2(\frac{\lambda_1}{\beta^\star}\epsilon+a_t\epsilon)\}.
	$$
For ${\bf t}\in \mathbb{R}^n$, define:
$$
\pi_{(1-\epsilon)}^c(|{\bf t}|)={\#\{i=1,\cdots, n \ \  |[{\bf t}]_i|\geq 1-\epsilon\}}.
$$
Then, noting that:
$$
\frac{1}{n}\|\hat{\bf x}^{\rm PO}\|_0\leq \frac{1}{n}\pi_{(1-\epsilon)}^c(|\hat{\bf t}^{\rm PO}|),
$$
we thus get:
$$
\hat{\bf x}^{\rm PO}\in D_{\epsilon}\Longrightarrow \hat{\bf t}^{\rm PO}\in D_{\epsilon}^t
$$
where 
$$
D_\epsilon^t=\{{\bf t}\in \mathbb{R}^n,\frac{1}{n}\pi_{(1-\epsilon)}^c(|{\bf t}|)\geq \ell_0^\star+2(a_t+\frac{\lambda_1}{\beta^\star})\epsilon\}.
$$
This can be equivalently stated by:
$$
\mathbb{P}[\hat{\bf x}^{\rm PO}\in D_\epsilon]\leq \mathbb{P}[\hat{\bf t}^{\rm PO}\in D_\epsilon^t].
$$
We will use the PO problem concerning the subgradient to control the above probability. 
Using the same display in \eqref{eq:t1}-\eqref{eq:t3}, we obtain
\begin{align*}
	\mathbb{P}\Big[\hat{\bf t}^{\rm PO}\in D_{\epsilon}^t\Big]&\leq \frac{C}{\epsilon^4}\exp(-cn\epsilon^8)+2	\mathbb{P}\Big[\max_{{\bf t}\in D_\epsilon^{t}} \mathcal{Z}_{\lambda,\rho}({\bf t})\geq \psi(\tau^\star,\beta^\star)-\gamma\epsilon^4\Big].
\end{align*}
Using Hoeffding inequality, the following inequality holds:
$$
\mathbb{P}\Big[ \frac{1}{n}\pi_{(1-2\epsilon)}^c(|\overline{\bf t}^{\rm AO}|)\leq \mathbb{E}[\frac{1}{n}\pi_{(1-2\epsilon)}^c(|\overline{\bf t}^{\rm AO}|)]+a_t\epsilon\Big]\geq 1-C\exp(-cn\epsilon^2)
$$
where
$$
\mathbb{E}[\frac{1}{n}\pi_{(1-2\epsilon)}^c(|\overline{\bf t}^{\rm AO}|)]=2Q(\frac{\lambda_1}{\beta^\star}(1-2\epsilon))\leq 2Q(\frac{\lambda_1}{\beta^\star})+2\frac{\lambda_1}{\beta^\star}\epsilon=\ell_0^\star+2\frac{\lambda_1}{\beta^\star}\epsilon.
$$
Hence, with probability $1-C\exp(-cn\epsilon^2)$, 
$$
\frac{1}{n}\pi_{(1-2\epsilon)}^c(|\overline{\bf t}^{\rm AO}|)\leq \ell_0^\star+2\frac{\lambda_1}{\beta^\star}\epsilon+a_t\epsilon.
$$
Hence, for all ${\bf t}\in D_{\epsilon}^t$, vectors ${\bf t}$ and $\overline{\bf t}^{\rm AO}$ are apart from each other in absolute value by $\epsilon$ in at least $na_t\epsilon$ positions. Hence,
$$
\forall {\bf t}\in D_{\epsilon}^{t},  \ \ \frac{1}{n}\|{\bf t}-\overline{\bf t}^{\rm AO}\|_1\geq a_t\epsilon^2,
$$
then according to \eqref{eq:Zt}, we obtain
$$
\mathbb{P}\Big[\max_{{\bf t}\in D_\epsilon^{t}} \mathcal{Z}_{\lambda,\rho}({\bf t})\geq \psi(\tau^\star,\beta^\star)-\gamma\epsilon^4\Big]\leq \mathbb{P}\Big[\exists {\bf t}\in D_{\epsilon}^t \ \  \mathcal{Z}_{\lambda,\rho}({\bf t})\geq \psi(\tau^\star,\beta^\star)-\gamma\epsilon^4\Big]\leq \frac{C}{\epsilon^4}\exp(-cn\epsilon^8), 
$$which proves the desired.
\end{proof}

\section{Study of the thresholded $\ell_1$-norm precoder}\label{app_core}
 {We first define some notations that will be useful in this section: For $t\geq 0$ and vector ${\bf x} \in \mathbb{R}^n$,  we denote
by $\pi_{(t)}(|{\bf x}|)$ the subset in $\{1, \cdots , n\}$, formed by the positions at which the vector ${\bf x}$ has magnitude less than $t$. For $(a, b)$ an
interval in $\mathbb{R}$, we denote by $\pi_{(a,b)}(|{\bf x}|)$ the subset in $\{1, \cdots , n\}$ in which the vector ${\bf x}$ has elements with magnitude in $(a, b)$. For $\pi$ a subset of $\{1, \cdots , n\}$, we define $\pi^{c} = \{1, \cdots , n\}\backslash \pi$, and ${\bf S}_\pi$ the diagonal matrix taking $1$ at the positions given by $\pi$ and zero otherwise.
The number of elements of $\pi$ is denoted by $\#\pi$. If $\pi$ is the empty set, then $\#\pi=0$. To shorten notations, we simplify $t_x$ as $t$.}
\subsection{ {Preliminaries and organization of the proof}}
\noindent{ {\bf Objective.} } {Let $\hat{\bf e}_{\ell_1}^{t}:={\bf H}{\bf S}_{\pi_{(t)}^c(|\hat{\bf x}_{\ell_1}|)}\hat{\bf x}_{\ell_1}$. The goal in this section is to show that  for $\epsilon$ sufficiently small
\begin{equation*}
\mathbb{P}\Big[\hat{\bf e}_{\ell_1}^{t}\in \mathcal{S}_e^{\circ}\Big]\to 0  
\end{equation*}
where $\mathcal{S}_e^{\circ}$ is defined in Theorem \ref{th:main_theorem} and depends on the variable $\epsilon$. For ${\bf x}\in \mathbb{R}^n$, define ${\bf e}({\bf x})={\bf H}{\bf S}_{\pi_{(t)}^c(|{\bf x}|)}{\bf x}$. The proof thus boils down to showing that for all ${\bf x}$ such that $\|{\bf x}\|_{\infty}\leq \sqrt{P}$, and ${\bf e}({\bf x})\in \mathcal{S}_e^{\circ}$, the value of $\mathcal{C}_{\lambda,\rho}({\bf x})$ is larger than the optimal cost $\mathcal{C}_{\lambda,\rho}(\hat{\bf x}_{\ell_1})$ or equivalently:
\begin{equation}
\mathbb{P}\Big[\min_{\substack{\|{\bf x}\|_{\infty}\leq \sqrt{P}\\ e({\bf x})\in \mathcal{S}_e^{\circ}}} \mathcal{C}_{\lambda,\rho}({\bf x})\leq \min_{\|{\bf x}\|_{\infty}\leq \sqrt{P}} \mathcal{C}_{\lambda,\rho}({\bf x}) \Big]\to 0. \label{eq:desired_equation_rev}
\end{equation}
\noindent{\bf Preliminaries.} The proof relies on a novel Gaussian min max inequality that enables to handle probability inequalities between Gaussian processes. However, to pave the way towards proving this theorem, some preliminary work is needed. First, we make use of our previous result in Theorem \ref{th:optimal_cost} to obtain:
$$
\mathbb{P}\Big[\min_{\|{\bf x}\|_{\infty}\leq \sqrt{P}} \mathcal{C}_{\lambda,\rho}({\bf x})\geq \psi(\tau^\star,\beta^\star)+\gamma\epsilon^{\frac{3}{2}}\Big]\leq \frac{C}{\epsilon^{\frac{3}{2}}}\exp(-cn\epsilon^3). 
$$
and hence, 
$$
\mathbb{P}\Big[\min_{\substack{\|{\bf x}\|_{\infty}\leq \sqrt{P}\\ e({\bf x})\in \mathcal{S}_e^{\circ}}} \mathcal{C}_{\lambda,\rho}({\bf x})\leq \min_{\|{\bf x}\|_{\infty}\leq \sqrt{P}} \mathcal{C}_{\lambda,\rho}({\bf x})\Big]\leq \mathbb{P}\Big[\min_{\substack{\|{\bf x}\|_{\infty}\leq \sqrt{P}\\ e({\bf x})\in \mathcal{S}_e^{\circ}}} \mathcal{C}_{\lambda,\rho}({\bf x})\leq \psi(\tau^\star,\beta^\star)+\gamma\epsilon^{\frac{3}{2}}\Big]+\frac{C}{\epsilon^{\frac{3}{2}}}\exp(-cn\epsilon^3). 
$$
Hence, to establish \eqref{eq:desired_equation_rev}, it suffices  to  show that:
\begin{equation}
\mathbb{P}\Big[\min_{\substack{\|{\bf x}\|_{\infty}\leq \sqrt{P}\\ e({\bf x})\in \mathcal{S}_e^{\circ}}} \mathcal{C}_{\lambda,\rho}({\bf x})\leq \psi(\tau^\star,\beta^\star)+\gamma\epsilon^2\Big]\to 0.  \label{eq:desired_equation_rev1}
\end{equation}
For that, we set throughout this section the variable $r=\sqrt{\epsilon}$ and  exploit the classical cGMT to show that there exists a set $\mathcal{S}_x^{\circ}$ depending on $r$ such that the optimization outside this set leads to a higher optimization cost. More specifically, this set is defined as:
$$
\mathcal{S}_x^{\circ}=\mathcal{S}_1(\epsilon,c_1) \cap \mathcal{S}_{2}(\epsilon,c_2)
$$
where
\begin{align*}
\mathcal{S}_1(\epsilon,c_1)&=\{{\bf x}, \left|\frac{\|{\bf S}_{\pi_{(t)}(|{\bf x}|)}{\bf x}\|}{\|{\bf S}_{\pi_{(t)}^c(|{\bf x}|)}{\bf x}\|}-\eta^\star\right|\leq c_{1}\sqrt{\epsilon}\}\\
\mathcal{S}_2(\epsilon,c_2)&=\{{\bf x}, C_l\sqrt{\epsilon}\leq \frac{\pi_{t(1-r),t(1+r)}(|{\bf x}|)}{n}\leq C_u\sqrt{\epsilon}\}
\end{align*}
where
$$
\eta^\star=\frac{\sqrt{\mathbb{E}[x(\tilde{\tau}^\star,\beta^\star)^2{\bf 1}_{\{x(\tilde{\tau}^\star,\beta^\star)\leq t\}]}}}{\sqrt{\mathbb{E}[x(\tilde{\tau}^\star,\beta^\star)^2{\bf 1}_{\{x(\tilde{\tau}^\star,\beta^\star)\geq t}}}
$$
and $c_1, C_l$ and $C_u$ are some positive constants and are chosen carefully such that: 
\begin{equation}
\mathbb{P}\Big[\min_{\substack{\|{\bf x}\|_{\infty}\leq \sqrt{P}\\ {\bf x}\notin \mathcal{S}_x^{\circ}}} \mathcal{C}_{\lambda,\rho}({\bf x})\leq \psi(\tau^\star,\beta^\star)+\gamma\epsilon^{\frac{3}{2}}\Big]\leq \frac{C}{\epsilon^{\frac{3}{2}}}\exp(-cn\epsilon^3). \label{eq:proof_sol_rev}
\end{equation}
Finally let $\mathcal{S}_{3}(\epsilon,C_0)$ be the set defined as:
\begin{equation}
\mathcal{S}_3(\epsilon,C_0)=\{{\bf x}, \frac{\pi_{(t_1,t_2)}(|{\bf x}|)}{n}\geq C_0\} \label{eq:C_3_rev}
\end{equation}
where $t_1=\frac{t}{2}$ and $t_2=\frac{t}{1+\epsilon^{\frac{1}{4}}}$. We can show that for sufficiently small $\epsilon$ and any $\gamma_{\rm th}>0$ there exists $C_0(\gamma_{\rm th})$ such that:
\begin{equation}
\mathbb{P}\Big[\min_{\substack{\|{\bf x}\|_{\infty}\leq \sqrt{P}\\ {\bf x}\notin \mathcal{S}_3(\epsilon,C_0(\gamma_{\rm th}))}} \mathcal{C}_{\lambda,\rho}({\bf x})\leq \psi(\tau^\star,\beta^\star)+\gamma_{\rm th}\epsilon^{\frac{3}{2}}\Big]\leq \frac{C}{\epsilon^{\frac{3}{2}}}\exp(-cn\epsilon^3). \label{eq:proof_sol_rev2}
\end{equation}
This particularly implies that:
\begin{equation}
\mathbb{P}\Big[\exists {\bf x} \ \text{such that }{\bf x}\notin \mathcal{S}_3(\epsilon,C_0(\gamma_{\rm th})) \ \text{and } \mathcal{C}_{\lambda,\rho}({\bf x})\leq \psi(\tau^\star,\beta^\star)+\gamma_{\rm th}\epsilon^{\frac{3}{2}}\Big]\leq \frac{C}{\epsilon^{\frac{3}{2}}}\exp(-cn\epsilon^3) .\label{eq:C_0_rev1}
\end{equation}
Using \eqref{eq:proof_sol_rev}, we obtain:
\begin{equation*}
\mathbb{P}\Big[\min_{\substack{\|{\bf x}\|_{\infty}\leq \sqrt{P}\\ {\bf e}({\bf x})\in \mathcal{S}_e^{\circ}}} \mathcal{C}_{\lambda,\rho}({\bf x})\leq \psi(\tau^\star,\beta^\star)+\gamma\epsilon^{\frac{3}{2}}\Big]\leq \mathbb{P}\Big[\min_{\substack{\|{\bf x}\|_{\infty}\leq \sqrt{P}\\ {\bf x}\in \mathcal{S}_x^\circ\\ {\bf e}({\bf x})\in \mathcal{S}_e^{\circ}}} \mathcal{C}_{\lambda,\rho}({\bf x})\leq \psi(\tau^\star,\beta^\star)+\gamma\epsilon^{\frac{3}{2}}\Big]+\frac{C}{\epsilon^{\frac{3}{2}}}\exp(-cn\epsilon^{3}).  
\end{equation*}
To prove \eqref{eq:desired_equation_rev1}, it thus suffices to show 
\begin{equation}
\mathbb{P}\Big[\min_{\substack{\|{\bf x}\|_{\infty}\leq \sqrt{P}\\ {\bf x}\in \mathcal{S}_x^\circ\\ {\bf e}({\bf x})\in \mathcal{S}_e^{\circ}}} \mathcal{C}_{\lambda,\rho}({\bf x})\leq \psi(\tau^\star,\beta^\star)+\gamma\epsilon^{\frac{3}{2}}\Big]\to 0.\label{eq:desired_rev3}
\end{equation}
The core of the proof lies in showing \eqref{eq:desired_rev3}, whereas the proofs of \eqref{eq:proof_sol_rev} and \eqref{eq:proof_sol_rev2} relies on leveraging the results presented in Theorem \ref{th:function} and Theorem \ref{th:interval}, which are established using the standard Gaussian-min max theorem.} {To show \eqref{eq:proof_sol_rev}, we first exploit the fact that  the complementary of $\mathcal{S}_1(\epsilon,c_1)$ is included in the union of the sets:
\begin{align*}
\mathcal{S}_{11}(\epsilon,c_1)&:=\{{\bf x}, |\frac{\|\frac{1}{\sqrt{n}}{\bf S}_{\pi_{(t)}(|{\bf x}|)}{\bf x}\|}{t}-\frac{\sqrt{\mathbb{E}[x(\tilde{\tau}^\star,\beta^\star)^2{\bf 1}_{\{|x(\tilde{\tau}^\star,\beta^\star)|\leq t\}}]}}{t}|\geq \frac{c_1\sqrt{\epsilon}}{2}\},\\
\mathcal{S}_{12}(\epsilon,c_1)&:=\{{\bf x}, |\|\frac{1}{\sqrt{n}}{\bf S}_{\pi_{(t)}^c(|{\bf x}|)}{\bf x}\|-{\sqrt{\mathbb{E}[x(\tilde{\tau}^\star,\beta^\star)^2{\bf 1}_{\{|x(\tilde{\tau}^\star,\beta^\star)|\geq t\}}]}}|\geq \frac{t c_1\sqrt{\epsilon}{\sqrt{\mathbb{E}[x(\tilde{\tau}^\star,\beta^\star)^2{\bf 1}_{\{|x(\tilde{\tau}^\star,\beta^\star)|\geq t\}}]}}}{2\sqrt{\mathbb{E}[x(\tilde{\tau}^\star,\beta^\star)^2{\bf 1}_{\{|x(\tilde{\tau}^\star,\beta^\star)|\geq t\}}]}}\}.
\end{align*}
We can then use Theorem \ref{th:function} to show that for appropriately chosen $c_1$, we obtain: 
\begin{align*}
&\mathbb{P}\Big[\min_{\substack{\|{\bf x}\|_{\infty}\leq \sqrt{P}\\ {\bf x}\in \mathcal{S}_{11}(\epsilon,c_1)}} \mathcal{C}_{\lambda,\rho}({\bf x})\leq \psi(\tau^\star,\beta^\star)+\gamma\epsilon^{\frac{3}{2}}\Big]\leq \frac{C}{\epsilon^{\frac{3}{2}}}\exp(-cn\epsilon^3),\\
&\mathbb{P}\Big[\min_{\substack{\|{\bf x}\|_{\infty}\leq \sqrt{P}\\ {\bf x}\in \mathcal{S}_{12}(\epsilon,c_1)}} \mathcal{C}_{\lambda,\rho}({\bf x})\leq \psi(\tau^\star,\beta^\star)+\gamma\epsilon^{\frac{3}{2}}\Big]\leq \frac{C}{\epsilon^{\frac{3}{2}}}\exp(-cn\epsilon^3).
\end{align*}
This directly implies that:
$$
\mathbb{P}\Big[\min_{\substack{\|{\bf x}\|_{\infty}\leq \sqrt{P}\\ {\bf x}\notin \mathcal{S}_{1}(\epsilon,c_1)}} \mathcal{C}_{\lambda,\rho}({\bf x})\leq \psi(\tau^\star,\beta^\star)+\gamma\epsilon^{\frac{3}{2}}\Big]\leq \frac{C}{\epsilon^{\frac{3}{2}}}\exp(-cn\epsilon^3).
$$
To complete the proof of \eqref{eq:proof_sol_rev}, we also need to show that:
\begin{equation}
\mathbb{P}\Big[\min_{\substack{\|{\bf x}\|_{\infty}\leq \sqrt{P}\\ {\bf x}\notin \mathcal{S}_{2}(\epsilon,c_1)}} \mathcal{C}_{\lambda,\rho}({\bf x})\leq \psi(\tau^\star,\beta^\star)+\gamma\epsilon^{\frac{3}{2}}\Big]\leq \frac{C}{\epsilon^{\frac{3}{2}}}\exp(-cn\epsilon^3). \label{eq:rev_155}
\end{equation}
This follows directly from applying Theorem \ref{th:interval} to the sets
\begin{align*}
\mathcal{S}_{21}(\epsilon,c_{21}):=\{{\bf x}\ \text{with }\|{\bf x}\|_{\infty}\leq \sqrt{P} \ \text{and } |\frac{\pi_{(t(1-r),t(1+r))}({\bf x})}{n}-(Q(\frac{t(1-r)}{\tilde{\tau}^\star}+\frac{\lambda_1}{\beta^\star})-Q(\frac{t(1+r)}{\tilde{\tau}^\star}+\frac{\lambda_1}{\beta^\star})|\geq c_{21}\epsilon\},\\
\mathcal{S}_{22}(\epsilon,c_{22}):=\{{\bf x}\ \text{with }\|{\bf x}\|_{\infty}\leq \sqrt{P} \  \text{and } |\frac{\pi_{(-t-tr,-t+tr)}({\bf x})}{n}-(Q(\frac{-t-tr}{\tilde{\tau}^\star}+\frac{\lambda_1}{\beta^\star})-Q(\frac{-t+tr}{\tilde{\tau}^\star}+\frac{\lambda_1}{\beta^\star})|\geq c_{22}\epsilon\},
\end{align*}
where $c_{21}$ and $c_{22}$ are chosen so that:
\begin{align*}
&\mathbb{P}\Big[\min_{\substack{\|{\bf x}\|_{\infty}\leq \sqrt{P}\\ {\bf x}\in \mathcal{S}_{21}(\epsilon,c_{21})}} \mathcal{C}_{\lambda,\rho}({\bf x})\leq \psi(\tau^\star,\beta^\star)+\gamma\epsilon^{\frac{3}{2}}\Big]\leq \frac{C}{\epsilon^{\frac{3}{2}}}\exp(-cn\epsilon^3), \\
&\mathbb{P}\Big[\min_{\substack{\|{\bf x}\|_{\infty}\leq \sqrt{P}\\ {\bf x}\in \mathcal{S}_{22}(\epsilon,c_{22})}} \mathcal{C}_{\lambda,\rho}({\bf x})\leq \psi(\tau^\star,\beta^\star)+\gamma\epsilon^{\frac{3}{2}}\Big]\leq \frac{C}{\epsilon^{\frac{3}{2}}}\exp(-cn\epsilon^3) .
\end{align*}
To complete the proof of \eqref{eq:rev_155}, we note that the absolute value of the derivatives of functions $g_1:x\mapsto Q(\frac{t(1+x)}{\tilde{\tau}^\star}+\frac{\lambda_1}{\beta^\star})$ and $g_2:x\mapsto Q(\frac{-t+tx}{\tilde{\tau}^\star}+\frac{\lambda_1}{\beta^\star})$ are bounded and bounded from below zero when $x\in(-1,1)$, that is there exists constants $m_1$ and $M_1$ such that $m_1<\sup_{|x|\leq 1}|g_1'(x)|<M_1$ and  $m_2<\sup_{|x|\leq 1}|g_2'(x)|<M_2$, thereby implies that $m_1\sqrt{\epsilon}<g_1(-r)-g_1(r)<M_1\sqrt{\epsilon}$ and $m_2\sqrt{\epsilon}<g_2(-r)-g_2(r)<M_2\sqrt{\epsilon}$.} 

 {Similarly, the proof of \eqref{eq:proof_sol_rev2} follows by applying Theorem~\ref{th:interval}, and is omitted for brevity.
}

 {In the sequel, we will focus on showing \eqref{eq:desired_rev3}. The proof relies on a new Gaussian min-max theorem, which will be introduced in the following section. However, its application is not immediate, as it requires several non-standard and technical steps that will be detailed subsequently. We structure the remainder of this section as follows: we begin by presenting the new Gaussian min-max theorem as a standalone result, as it may have broader applications. We then outline the key technical steps required for its application, and finally, demonstrate how the theorem is employed in the proof.}
\subsection{ {New Gaussian min-max Theorem}}
 {In this section, we provide a key generalization of the Gordon's Gaussian min-max theorem that will be useful to study the thresholded $\ell_1$-norm  precoder. }
\begin{theorem}
	For $m,n\in\mathbb{N}^{*}$,  let  $\mathcal{S}_{u}$,   $\mathcal{S}_b$ be two compact sets of $\mathbb{R}^{m}$. Let $\mathcal{S}_x$ be a compact set in $\mathbb{R}^{n}$. Let $\mathcal{S}_\gamma$ be either a compact set of $ \mathbb{R}$ or the entire real line $\mathbb{R}$. Let ${\bf G}\in\mathbb{R}^{m\times n}$, ${\bf g},\tilde{\bf g}\in\mathbb{R}^m$ and ${\bf h}\in\mathbb{R}^n$ have independent and identically distributed entries following $\mathcal{N}(0,1)$ and are independent of each other. 
	
	Assume that $\forall {\bf x}\in\mathcal{S}_x$, we have
	\begin{subequations}
		\begin{align}
			&\frac{\|{\bf S}_{\pi_{(t)}(|{\bf x}|)}{\bf x}\|}{\|{\bf S}_{\pi_{(t)}^c(|{\bf x}|)}{\bf x}\|}=\eta, \label{req1}\\
			&\#\pi_{(t)}^c(|{\bf x}|)=k, \ \ \ \label{req2} \\
			&\#\pi_{ {(t(1-r),t(1+r))}}({\bf x})=\#\pi_{(-t-tr,-t+tr)}({\bf x})=0,  \label{req3}
		\end{align}
	\end{subequations}where $\eta$ is a fixed real, $k$ is a fixed integer and $r$ is some positive real in $(0,1)$. Let $\tilde{\eta}=\frac{(1-r)}{\sqrt{\eta^2+(1-r)^2}}$ and $\nu=\sqrt{1-\tilde{\eta}^2}$. Consider the following two random processes: 
	\begin{align*}
		X_{t}({\bf x},{\bf b},{\bf u},\gamma)=&{\bf u}^{T}{\bf G}({\bf x}-r{\bf S}_{\pi_{(t)}^c(|{\bf x}|)}{\bf x})+\gamma{\bf u}^{T}{\bf G}{\bf S}_{\pi_{(t)}^c(|{\bf x}|)}{\bf x}+\psi({\bf x},{\bf b},{\bf u},\gamma),\\Y_{t}({\bf x},{\bf b},{\bf u},\gamma)=&{\bf g}^{T}{\bf u}\|{\bf x}-r{\bf S}_{\pi_{(t)}^c(|{\bf x}|)}{\bf x}\|-\|{\bf u}\|{\bf h}^T({\bf x}-r{\bf S}_{\pi_{(t)}^c(|{\bf x}|)}{\bf x})\\&+\gamma {\bf u}^{T}(\tilde{\eta}{\bf g}+\nu \tilde{\bf g})\|{\bf S}_{\pi_{(t)}^c(|{\bf x}|)}{\bf x}\|-|\gamma|{\bf h}^{T}{\bf S}_{\pi_{(t)}^{c}(|{\bf x}|)}{\bf x}\|{\bf u}\|+\psi({\bf x},{\bf b},{\bf u},\gamma),
	\end{align*}
	where  $({\bf x},{\bf b},{\bf u},\gamma):\mapsto{\psi}({\bf x},{\bf b},{\bf u},\gamma)$  is continuous and  does not depend on ${\bf G,g},\tilde{\bf g}$ and ${\bf h}$.  Then, for any $c\in \mathbb{R}$, we have:
	\begin{align*}
		&\mathbb{P}\big[\min_{\substack{{\bf x}\in\mathcal{S}_x\\{\bf b}\in\mathcal{S}_b}}\max_{\substack{{\bf u}\in\mathcal{S}_u\\\gamma\in\mathcal{S}_\gamma}} X_t({\bf x},{\bf b},{\bf u},\gamma)  \leq c\big]\leq 4\mathbb{P}\big[\min_{\substack{{\bf x}\in\mathcal{S}_x\\{\bf b}\in\mathcal{S}_b}}\max_{\substack{{\bf u}\in\mathcal{S}_u\\\gamma\in\mathcal{S}_\gamma}} Y_t({\bf x},{\bf b},{\bf u},\gamma)\leq c\big].
	\end{align*}
	\label{th:new_cgmt}
\end{theorem}
\begin{proof}
The proof is deferred to Appendix \ref{isolate_proof}. 
\end{proof}
\subsection{ {Preliminary steps toward proving \eqref{eq:desired_rev3} via Theorem \ref{th:new_cgmt}}}
 {While Theorem \ref{th:new_cgmt} is instrumental for analyzing the convergence behavior in \eqref{eq:desired_rev3}, its direct application is essentially impeded by the fact that the optimization variable ${\bf x}$ does not satisfy  the theorem’s assumptions. To circumvent this obstacle, we demonstrate that the probability appearing in \eqref{eq:desired_rev3} can be dominated by that of a properly defined Gaussian process, which aligns with the conditions required by Theorem \ref{th:new_cgmt}. }
\subsubsection{ {Reformulation using the technique of Lagrange multipliers}}
 {Using the technique of Lagrange multipliers, we may write $\mathcal{C}_{\lambda,\rho}({\bf x})$ as:
\begin{equation*}
\mathcal{C}_{\lambda,\rho}({\bf x})=\min_{{\bf e}\in \mathbb{R}^m}
\mathcal{M}_{\lambda,\rho}({\bf x},{\bf e}) 
\end{equation*}
where
$$
\mathcal{M}_{\lambda,\rho}({\bf x},{\bf e}):=\frac{1}{n}\|{\bf H}({\bf x}-r{\bf S}_{\pi_{(t)}^c(|{\bf x}|)}{\bf x})+r{\bf e}-\sqrt{\rho}{\bf s}\|^2+\frac{\lambda_1}{n}\|{\bf x}\|_1+\frac{\lambda_2}{n}\|{\bf x}\|^2+\sup_{\boldsymbol{\lambda}\in \mathbb{R}^{m}}\frac{\boldsymbol{\lambda}^{T}}{\sqrt{n}}({\bf H}{\bf S}_{\pi_{(t)}^c(|{\bf x}|)}{\bf x}-{\bf e})
$$
where we recall that $r=\sqrt{\epsilon}$. Note here that for any ${\bf x}$, the solution in ${\bf e}$ coincides with ${\bf e}({\bf x})={\bf H}{\bf S}_{\pi_{(t)}^c(|{\bf x}|)}{\bf x}$. Hence, for any set $\mathcal{S}_x\subset\mathbb{R}^n$ and $\mathcal{S}_e\subset \mathbb{R}^m$,
$$
\min_{\substack{{\bf x}\in \mathcal{S}_x\\ {\bf e}({\bf x})\in \mathcal{S}_e}} \mathcal{C}_{\lambda,\rho}({\bf x}) = \min_{{\bf x}\in \mathcal{S}_x}\min_{{\bf e}\in \mathcal{S}_e}\mathcal{M}_{\lambda,\rho}({\bf x},{\bf e}).
$$
The proof of \eqref{eq:desired_rev3} boils down to showing that:
$$
\mathbb{P}\Big[\min_{\substack{{\bf x}\in \mathcal{S}_x^\circ\\\|{\bf x}\|_\infty\leq\sqrt{P}}}\min_{{\bf e}\in \mathcal{S}_e^\circ}\mathcal{M}_{\lambda,\rho}({\bf x},{\bf e}) \leq \psi(\tau^\star,\beta^\star)+\gamma\epsilon^2\Big]\to 0.
$$
To continue, we use the relation $\|{\bf x}\|^2=\max_{{\bf u}} {\bf u}^{T}{\bf x}-\frac{\|{\bf u}\|^2}{4}$ to rewrite $\mathcal{M}_{\lambda,\rho}({\bf x},{\bf e})$ as:
\begin{align}
\mathcal{M}_{\lambda,\rho}({\bf x},{\bf e})&=\max_{{\bf u}} \frac{1}{\sqrt{n}} {\bf u}^{T}({\bf H}({\bf x}-r{\bf S}_{\pi_{(t)}^c(|{\bf x}|)}{\bf x}))+\frac{r}{\sqrt{n}}{\bf u}^{T}{\bf e}-\frac{\sqrt{\rho}}{\sqrt{n}}{\bf u}^{T}{\bf s}\nonumber \\
&+\frac{\lambda_1}{n}\|{\bf x}\|_1+\frac{\lambda_2}{n}\|{\bf x}\|^2+\sup_{\boldsymbol{\lambda}\in \mathbb{R}^{m}}\frac{\boldsymbol{\lambda}^{T}}{\sqrt{n}}({\bf H}{\bf S}_{\pi_{(t)}^c(|{\bf x}|)}{\bf x}-{\bf e}).\label{eq:Mlambda_revised}
\end{align} 
With probability $1-C\exp(-cn)$, $\|{\bf H}\|\leq 3\max(1,\sqrt{\delta})$. Hence, for any ${\bf x}$ and ${\bf e}$, there exists a constant $K_u$ such that the optimum in ${\bf u}$ in \eqref{eq:Mlambda_revised} is less than $K_u$. Let the compact set ${\mathcal{S}}_{\bf u}:=\{{\bf u}, \|{\bf u}\|\leq K_u\}$ Hence, with probability $1-C\exp(-cn)$, for all ${\bf x}\in \mathcal{S}_{x}^{\circ}$ such that $\|{\bf x}\|_\infty\leq\sqrt{P}$  and ${\bf e}\in \mathcal{S}_e^{\circ}$, 
\begin{align*}
\mathcal{M}_{\lambda,\rho}({\bf x},{\bf e})&=\max_{{\bf u}\in \mathcal{S}_{{\bf u}}} \frac{1}{\sqrt{n}} {\bf u}^{T}({\bf H}({\bf x}-r{\bf S}_{\pi_{(t)}^c(|{\bf x}|)}{\bf x}))+\frac{r}{\sqrt{n}}{\bf u}^{T}{\bf e}-\frac{\sqrt{\rho}}{\sqrt{n}}{\bf u}^{T}{\bf s} \\
&+\frac{\lambda_1}{n}\|{\bf x}\|_1+\frac{\lambda_2}{n}\|{\bf x}\|^2+\sup_{\boldsymbol{\lambda}\in \mathbb{R}^{m}}\frac{\boldsymbol{\lambda}^{T}}{\sqrt{n}}({\bf H}{\bf S}_{\pi_{(t)}^c(|{\bf x}|)}{\bf x}-{\bf e}). 
\end{align*} 
The main objective is to apply the Gaussian min-max inequality stated in Theorem~\ref{th:new_cgmt} to the Gaussian process appearing in~\eqref{eq:Mlambda_revised}, in order to establish~\eqref{eq:desired_rev3}. However, a direct application of Theorem~\ref{th:new_cgmt} is hindered by two main difficulties. First, the optimization set \( \mathcal{S}_x^{\circ} \) does not satisfy the structural requirements of the theorem. Second, the variable \( \boldsymbol{\lambda} \) is not aligned with the Gaussian vector \( \mathbf{u} \). As we will see, the second issue is relatively straightforward to address. In contrast, the first issue—namely, the incompatibility of \( \mathcal{S}_x^\circ \) with the assumptions of Theorem~\ref{th:new_cgmt}—is more involved and requires a set of non-standard technical arguments, which will be developed in the sequel.}
\subsubsection{ {Variable transformation}}  {Theorem~\ref{th:new_cgmt} imposes three key conditions on the set to which the Gaussian min-max theorem can be applied. For clarity, we recall these conditions below:
\begin{align*}
\frac{\|{\bf S}_{\pi_{(t)}(|{\bf x}|)}{\bf x}\|}{\|{\bf S}_{\pi_{(t)}^c(|{\bf x}|)}{\bf x}\|}=\eta,\\
\#\pi_{(t)}^{c}(|{\bf x}|)=k,\\
\#\pi_{t(1-r),t(1+r)}(|{\bf x}|)=0.
\end{align*}
The first requirement is approximately satisfied by any ${\bf x}\in \mathcal{S}_x^{\circ}$ when $\eta=\eta^\star$ since for all ${\bf x}\in \mathcal{S}_x^\circ$, 
$$
\Big|\frac{\|{\bf S}_{\pi_{(t)}(|{\bf x}|)}{\bf x}\|}{\|{\bf S}_{\pi_{(t)}^c(|{\bf x}|)}{\bf x}\|}-\eta^\star\Big|\leq c_1\sqrt{\epsilon}.
$$
The second requirement can be handled by using a union bound argument, as will be shown next. However, the third requirement could not be deemed as asymptotically satisfied by elements in $\mathcal{S}_x^{\circ}$. Indeed, from the definition of $\mathcal{S}_x^{\circ}$, any element of ${\bf x}\in \mathcal{S}_x^{\circ}$ contains more than $C_ln r$ elements with magnitude in $(t(1-r),t(1+r))$. This makes the handling of the third requirement more challenging. The key idea to handle this requirement is to prove that for any ${\bf x}\in \mathcal{S}_x^{\circ}$, there exists $\tilde{\bf x}$ close to ${\bf x}$ such that $\tilde{\bf x}$ satisfies the third requirement. The closeness of ${\bf x}$ to $\tilde{\bf x}$ will be defined so as to enable the subsequent steps of the proof. To be specific, we prove the following result:}
 {
\begin{lemma}
There exists a constant $s$ and constant $C_e$ such that with probability $1-C\exp(-cn)$, for any $\epsilon>0$, and ${\bf x}\in \mathcal{S}_x^{\circ}$ with $\|{\bf x}\|_{\infty}\leq \sqrt{P}$, there exists $\tilde{\bf x}$ with $\|\tilde{\bf x}\|_{\infty}\leq \sqrt{P}$, such that:
\begin{align}
&|\mathcal{C}_{\lambda,\rho}({\bf x})-\mathcal{C}_{\lambda,\rho}(\tilde{\bf x})|\leq s\max(\epsilon^{\frac{3}{2}},\frac{1}{\sqrt{n}}), \label{eq:per_rev}\\
&\frac{1}{n}\|{\bf x}-\tilde{\bf x}\|^2\leq C_u t^2\epsilon^{\frac{3}{2}}, \nonumber\\
&\frac{1}{n}\|{\bf S}_{\pi_{(t)}(|{\bf x}|)}{\bf x}-{\bf S}_{\pi_{(t)}(|\tilde{\bf x}|)}\tilde{\bf x}\|^2\leq 4PC_u\sqrt{\epsilon} \ \text{ and  }\frac{1}{n}\|{\bf S}_{\pi_{(t)}^c(|{\bf x}|)}{\bf x}-{\bf S}_{\pi_{(t)}^c(|\tilde{\bf x}|)}\tilde{\bf x}\|^2\leq 4PC_u\sqrt{\epsilon},\label{eq:norm_x_thresh_rev}\\
&\frac{1}{m}\|{\bf e}({\bf x})-{\bf e}(\tilde{\bf x})\|^2\leq C_e\sqrt{\epsilon}, \nonumber\\
&\#\pi_{(t(1-r),t(1+r))}(|\tilde{\bf x}|)=0 .\label{eq:no_point_rev}
\end{align}
\label{lem:revised_inter}
\end{lemma}
\begin{proof}
Let ${\bf a}_1$, ${\bf a}_2$ and ${\bf a}_3$ be defined as:
\begin{align*}
{\bf a}_1&={\bf S}_{\pi_{(t(1-r),t(1+r))}(|{\bf x}|)} {\bf H}^{T}({\bf H}{\bf x}-\sqrt{\rho}{\bf s}),\\
{\bf a}_2&={\bf S}_{\pi_{(t(1-r),t(1+r))}(|{\bf x}|)}{\bf x},\\
{\bf a}_3&={\bf S}_{\pi_{(t(1-r),t(1+r))}(|{\bf x}|)}{\bf 1},\\
{\bf a}_4&={\bf S}_{\pi_{(-t-tr,-t+tr)}(|{\bf x}|)}{\bf 1},
\end{align*}
where ${\bf 1}$ is the vector of all ones. Let $\mathcal{A}=\{{\bf a}_1,{\bf a}_2,{\bf a}_3,{\bf a}_4\}$. We can prove that there exists a constant $C_a$ such that With probability $1-C\exp(-cn)$, 
$$
\forall {\bf a}\in \mathcal{A}, \|{\bf a}\|\leq C_a\sqrt{n}.
$$
Indeed, it is easy to see that $\|{\bf a}_2\|\leq \sqrt{P}\sqrt{n}$ while $\|{\bf 
 a}_3\|\leq \sqrt{n}$ and $\|{\bf a}_4\|\leq\sqrt{n}$. Furthermore with probability $1-C\exp(-cn)$, $\|{\bf H}\|\leq 3\max(1,\sqrt{\delta})$, hence, 
 $$
 \|{\bf a}_1\|\leq 9\max(1,\delta)\sqrt{P}\sqrt{n}+3\max(1,\sqrt{\delta})\sqrt{\rho}\sqrt{m}.
 $$
 Next, given $\boldsymbol{\sigma}=[\sigma_1,\cdots,\sigma_n]^{T}$ where $\sigma_1,\cdots,\sigma_n$ independent and identically distributed Rademacher random variables, define $\tilde{\bf x}(\boldsymbol{\sigma})$ as:
 \begin{equation}
 \tilde{\bf x}(\boldsymbol{\sigma})= {\bf x}\odot {\bf 1}_{\{|{\bf x}|\notin (t(1-r),t(1+r))}+({\bf x}+2\boldsymbol{\sigma} tr)\odot{\bf 1}_{\{|{\bf x}|\in (t(1-r),t(1+r))} \label{eq:xsigma_revised}
 \end{equation}
which is given by adding or subtracting $2tr$ to all elements of ${\bf x}$ whose magnitude are in $(t(1-r),t(1+r))$. Here, the notation ${\bf a}\odot {\bf b}$ refers to the Hadamard product between ${\bf a}$ and ${\bf b}$ while ${\bf 1}_{\{|{\bf x}|\notin (t(1-r),t(1+r))}$ and ${\bf 1}_{\{|{\bf x}|\in (t(1-r),t(1+r))}$ represent vectors in which the indicator function is applied element-wise to each entry. Clearly, for any rademacher sequence $\boldsymbol{\sigma}$, $\#\pi_{t(1-r),t(1+r)}(|\tilde{\bf x}(\boldsymbol{\sigma})|)=0.$ Moreover, we note that:
\begin{align*}
&\mathbb{E}_{\boldsymbol{\sigma}}[\max(\frac{1}{n}|(\tilde{\bf x}(\boldsymbol{\sigma})-{\bf x})^{T}{\bf x}|,\frac{1}{n}|(\tilde{\bf x}(\boldsymbol{\sigma})-{\bf x})^{T}{\bf H}^{T}({\bf H}{\bf x}-\sqrt{\rho}{\bf s})|,\frac{1}{n}|(\tilde{\bf x}(\boldsymbol{\sigma})-{\bf x})^{T}{\bf a}_3,\frac{1}{n}|(\tilde{\bf x}(\boldsymbol{\sigma})-{\bf x})^{T}{\bf a}_4|)]\\
&=\mathbb{E}_{\boldsymbol{\sigma}}[\max(\frac{2tr}{n}|\sum_{i\in \mathcal{I}}\boldsymbol{\sigma}_i[{\bf a}_1]_i|,\frac{2tr}{n}|\sum_{i\in \mathcal{I}}\boldsymbol{\sigma}_i[{\bf a}_2]_i|,\frac{2tr}{n}|\sum_{i\in \mathcal{I}}\boldsymbol{\sigma}_i[{\bf a}_3]_i|,\frac{2tr}{n}|\sum_{i\in \mathcal{I}}\boldsymbol{\sigma}_i[{\bf a}_4]_i|]
\end{align*}
where $\mathcal{I}$ is the set of indexes corresponding to positions of ${\bf x}$ with magnitude in $(t(1-r),t(1+r))$. Using Lemma \ref{lem:rademacher} in Appendix \ref{app:technical_lemmas_revised}, we thus obtain
\begin{equation}
\mathbb{E}_{\boldsymbol{\sigma}}[\max(\frac{1}{n}|(\tilde{\bf x}(\boldsymbol{\sigma})-{\bf x})^{T}{\bf x}|,\frac{1}{n}|(\tilde{\bf x}(\boldsymbol{\sigma})-{\bf x})^{T}{\bf H}^{T}({\bf H}{\bf x}-\sqrt{\rho}{\bf s})|,\frac{1}{n}|(\tilde{\bf x}(\boldsymbol{\sigma})-{\bf x})^{T}{\bf a}_3|,\frac{1}{n}|(\tilde{\bf x}(\boldsymbol{\sigma})-{\bf x})^{T}{\bf a}_4|)]\leq \frac{2t\sqrt{\epsilon}C_a\log 4}{\sqrt{n}}. \label{eq:upper_bound_rev}
\end{equation}
Define $\mathcal{S}_{\boldsymbol{\sigma}}$ the set of all possible values for the random vector $\boldsymbol{\sigma}$. Obviously:
\begin{align*}
&\mathbb{E}_{\boldsymbol{\sigma}}[\max(\frac{1}{n}|(\tilde{\bf x}(\boldsymbol{\sigma})-{\bf x})^{T}{\bf x}|,\frac{1}{n}|(\tilde{\bf x}(\boldsymbol{\sigma})-{\bf x})^{T}{\bf H}^{T}({\bf H}{\bf x}-\sqrt{\rho}{\bf s})|,\frac{1}{n}|(\tilde{\bf x}(\boldsymbol{\sigma})-{\bf x})^{T}{\bf a}_3|,\frac{1}{n}|(\tilde{\bf x}(\boldsymbol{\sigma})-{\bf x})^{T}{\bf a}_4|)]  \\
&\geq \min_{\boldsymbol{\sigma}} \max(\frac{1}{n}|(\tilde{\bf x}(\boldsymbol{\sigma})-{\bf x})^{T}{\bf x}|,\frac{1}{n}|(\tilde{\bf x}(\boldsymbol{\sigma})-{\bf x})^{T}{\bf H}^{T}({\bf H}{\bf x}-\sqrt{\rho}{\bf s})|,\frac{1}{n}|(\tilde{\bf x}(\boldsymbol{\sigma})-{\bf x})^{T}{\bf a_3}|,\frac{1}{n}|(\tilde{\bf x}(\boldsymbol{\sigma})-{\bf x})^{T}{\bf a}_4|).
\end{align*}
Using \eqref{eq:upper_bound_rev}, we conclude that for any ${\bf x}\in \mathcal{S}_x^{\circ}$, there exists $\boldsymbol{\sigma}^\star({\bf x})$ such that:
$$
\min_{\boldsymbol{\sigma}} \max(\frac{1}{n}|(\tilde{\bf x}(\boldsymbol{\sigma}^{\star}({\bf x}))-{\bf x})^{T}{\bf x}|,\frac{1}{n}|(\tilde{\bf x}(\boldsymbol{\sigma}^\star({\bf x}))-{\bf x})^{T}{\bf H}^{T}({\bf H}{\bf x}-\sqrt{\rho}{\bf s})|,\frac{1}{n}|(\tilde{\bf x}(\boldsymbol{\sigma}^\star({\bf x}))-{\bf x})^{T}{\bf a}_3|,\frac{1}{n}|(\tilde{\bf x}(\boldsymbol{\sigma}^\star({\bf x}))-{\bf x})^{T}{\bf a}_4|)\leq \frac{2t\sqrt{\epsilon}C_a\log 4}{\sqrt{n}} .
$$
Having this, we shall prove that  $\tilde{\bf x}:=\tilde{\bf x}(\boldsymbol{\sigma}^{\star}({\bf x}))$ satisfies the properties \eqref{eq:per_rev}-\eqref{eq:no_point_rev}.  To do so, we use the property $\|{\bf a}\|^2-\|{\bf b}\|^2=({\bf a}-{\bf b})^{T}({\bf a}+{\bf b})$, to obtain
\begin{align}
|\mathcal{C}_{\lambda,\rho}({\bf x})-\mathcal{C}_{\lambda,\rho}(\tilde{\bf x})|&\leq \frac{1}{n}|(\tilde{\bf x}-{\bf x})^{T}{\bf H}^{T}({\bf H}{\bf x}+{\bf H}\tilde{\bf x}-2\sqrt{\rho}{\bf s})|+ \frac{\lambda_2}{n}({\bf x}-\tilde{\bf x})^{T}({\bf x}+\tilde{\bf x})\nonumber\\
&+\frac{\lambda_1}{n} |{\bf a}_3^{T}({\bf x}-\tilde{\bf x})|+ \frac{\lambda_1}{n} |{\bf a}_4^{T}({\bf x}-\tilde{\bf x})| \nonumber\\
&\leq \frac{1}{n}\|{\bf H}({\bf x}-\tilde{\bf x})\|^2+\frac{2}{n}|(\tilde{\bf x}-{\bf x})^{T}{\bf H}^{T}({\bf H}{\bf x}-\sqrt{\rho}{\bf s})|+\frac{\lambda_2}{n} \|{\bf x}-\tilde{\bf x}\|^2+\frac{2\lambda_2}{n}({\bf x}-\tilde{\bf x})^{T}{\bf x}\nonumber\\
&+\frac{\lambda_1}{n} |{\bf a}_3^{T}({\bf x}-\tilde{\bf x})|+ \frac{\lambda_1}{n} |{\bf a}_4^{T}({\bf x}-\tilde{\bf x})|. \label{eq:last_rev}
\end{align}
It follows from \eqref{eq:xsigma_revised} that:
$$
\frac{1}{n}\|{\bf x}-\tilde{\bf x}\|^2\leq 4t^2C_u\epsilon^{\frac{3}{2}}.
$$
Hence, with probability $1-C\exp(-cn)$, 
\begin{equation}
\frac{1}{n}\|{\bf H}({\bf x}-\tilde{\bf x})\|^2\leq 36\max(1,\delta)t^2C_u\epsilon^{\frac{3}{2}} \label{eq:term_1_rev}
\end{equation}
and 
\begin{equation}
\frac{\lambda_2}{n} \|{\bf x}-\tilde{\bf x}\|^2\leq 4\lambda_2 C_u t^2\epsilon^{\frac{3}{2}}. \label{eq:term_2_rev}
\end{equation}
By combining \eqref{eq:term_1_rev} and \eqref{eq:term_2_rev} with \eqref{eq:last_rev}, we obtain for some appropriate constants $\tilde{C}_1$ and $\tilde{C}_2$,
$$
|\mathcal{C}_{\lambda,\rho}({\bf x})-\mathcal{C}_{\lambda,\rho}(\tilde{\bf x})|\leq \tilde{C}_1\epsilon^{\frac{3}{2}}+\frac{\tilde{C}_2}{\sqrt{n}}.
$$
We thus prove that $\tilde{\bf x}(\boldsymbol{\sigma})$ satisfies \eqref{eq:per_rev} for $s\geq 2\max(\tilde{C}_1,\tilde{C}_2)$. Next, using the fact that:
$$
\|{\bf S}_{\pi_{(t)}^c(|{\bf x}|)}{\bf x}-{\bf S}_{\pi_{(t)}^c(|\tilde{\bf x}|)}\tilde{\bf x}\|_{\infty}\leq 2\sqrt{P}
$$
we thus obtain:
\begin{equation}
\frac{1}{n}\|{\bf S}_{\pi_{(t)}^c(|{\bf x}|)}{\bf x}-{\bf S}_{\pi_{(t)}^c(|\tilde{\bf x}|)}\tilde{\bf x}\|^2\leq 4PC_u\sqrt{\epsilon} \label{eq:rev_s}
\end{equation}
and similarly
$$
\frac{1}{n}\|{\bf S}_{\pi_{(t)}(|{\bf x}|)}{\bf x}-{\bf S}_{\pi_{(t)}(|\tilde{\bf x}|)}\tilde{\bf x}\|^2\leq 4PC_u\sqrt{\epsilon}
$$
thereby proving \eqref{eq:norm_x_thresh_rev}. Furthermore, using \eqref{eq:rev_s}, we conclude that with probability $1-C\exp(-cn)$
$$
\frac{1}{m}\|{\bf e}({\bf x})-{\bf e}(\tilde{\bf x})\|^2\leq 36P\max(1,\delta^{-1})C_u\sqrt{\epsilon}.
$$
\end{proof}}
 {
\begin{corollary}
\label{cor:f_revised}
There exists a constant $s$ such that with probability $1-C\exp(-cn)$, for any ${\bf x}\in \mathcal{S}_x^{\circ}\cap\{{\bf e}({\bf x})\in \mathcal{S}_e^{\circ}\}$ and $\|{\bf x}\|_{\infty}\leq \sqrt{P}$. there exists $\tilde{\bf x}\in \tilde{\mathcal{S}}_{{x}}^{\circ}\cap\{{\bf x} | {\bf e}({\bf x})\in \tilde{\mathcal{S}}_e^{\circ}\}$ such that:
$$
|\mathcal{C}_{\lambda,\rho}({\bf x})-\mathcal{C}_{\lambda,\rho}(\tilde{\bf x})|\leq s\max(\epsilon^{\frac{3}{2}},\frac{1}{\sqrt{n}}) 
$$
where $\tilde{\mathcal{S}}_{x}^{\circ}=\tilde{\mathcal{S}}_{x1}^{\epsilon}\cap\tilde{\mathcal{S}}_{x2}^{\epsilon} $ with:
\begin{align}
\tilde{\mathcal{S}}_{x1}^{\circ}&:=\{{\bf x}, |\frac{\|{\bf S}_{\pi_{(t)}(|{\bf x}|)}{\bf x}\|}{\|{\bf S}_{\pi_{(t)}^c(|{\bf x}|)}{\bf x}\|}-\eta^\star|\leq \tilde{c}_1\sqrt{\epsilon}\}, \nonumber\\
\tilde{\mathcal{S}}_{x2}^{\circ}&:= \{{\bf x}, \|{\bf x}\|_{\infty}\leq \sqrt{P} \ \text{and }\#\pi_{t(1-r),t(1+r)}(|{\bf x}|)=0\},\label{eq:Sx2_rev}
\end{align}
and $\tilde{\mathcal{S}}_e^{\circ}$ is defined as:
$$
\tilde{\mathcal{S}}_e^{\circ}:=\{{\bf e}\in \mathbb{R}^m, (\mathcal{W}_2(\hat{\mu}({\bf e},{\bf s}),\mu_{t}^{\star}))^2\geq (\sqrt{c}_e-\sqrt{C}_e)^2\epsilon^{\frac{1}{2}}\}.
$$
\end{corollary}}
 {
\begin{proof}
It follows from Lemma \ref{lem:revised_inter} that with probability $1-C\exp(-cn)$, for any ${\bf x}\in \mathcal{S}_{x}^{\circ}$, there exists ${\tilde{\bf x}}$ that satisfies \eqref{eq:per_rev}-\eqref{eq:no_point_rev}. Obviously, for some appropriately chosen $\tilde{c}_1$ $\tilde{\bf x}\in \tilde{\mathcal{S}}_{x}^{\circ}.$ It remains now to check that for such $\tilde{\bf x}$, we have:
$$
{\bf e}({\bf x})\in \mathcal{S}_e^{\circ}  \ \ \Longrightarrow \  \ {\bf e}(\tilde{\bf x})\in \tilde{\mathcal{S}}_e^{\circ}.
$$
To see this, we use the fact that:
\begin{align*}
&\mathcal{W}_2(\hat{\mu}({\bf e}(\tilde{\bf x}),{\bf s}),\mu_{t}^{\star})\geq \mathcal{W}_2(\hat{\mu}({\bf e}({\bf x}),{\bf s}),\mu_{t}^{\star})- \mathcal{W}_2(\hat{\mu}({\bf e}({\bf x}),{\bf s}),\hat{\mu}({\bf e}(\tilde{\bf x}),{\bf s}))\\
&\geq \mathcal{W}_2(\hat{\mu}({\bf e}({\bf x}),{\bf s}),\mu_{t}^{\star})-\frac{1}{\sqrt{m}}\|{\bf e}({\bf x})-{\bf e}(\tilde{\bf x})\|\geq (\sqrt{c}_e-\sqrt{C}_e)\epsilon^{\frac{1}{4}}
\end{align*}
which shows that ${\bf e}(\tilde{\bf x})\in \tilde{\mathcal{S}}_e^{\circ}$.
\end{proof}}
\subsubsection{ {Recasting the convergence in \eqref{eq:desired_rev3} within the scope of Theorem \ref{th:new_cgmt}}}  {As previously noted, Theorem \ref{th:new_cgmt} cannot be directly applied to establish \eqref{eq:desired_rev3}. The goal of this section is  to reformulate \eqref{eq:desired_rev3} as a probability inequality involving a Gaussian process that satisfies the conditions of Theorem \ref{th:new_cgmt}. First, using Corollary \ref{cor:f_revised}, we deduce that with probability $1-C\exp(-cn)$
$$
\min_{\substack{\|{\bf x}\|_{\infty}\leq \sqrt{P}\\ {\bf x}\in \mathcal{S}_x^{\circ}\\ {\bf e}({\bf x})\in \mathcal{S}_e^{\circ}}}\mathcal{C}_{\lambda,\rho}({\bf x})\geq \min_{\substack{ \tilde{\bf x}\in \tilde{\mathcal{S}}_x^{\circ}\\ {\bf e}(\tilde{\bf x})\in \tilde{\mathcal{S}}_e^{\circ}}}\mathcal{C}_{\lambda,\rho}(\tilde{\bf x})-s\max(\epsilon^{\frac{3}{2}},\frac{1}{\sqrt{n}}).
$$}
 {Therefore:
$$
\mathbb{P}\Big[\min_{\substack{\|{\bf x}\|_{\infty}\leq \sqrt{P}\\ {\bf x}\in \mathcal{S}_x^\circ\\ {\bf e}({\bf x})\in \mathcal{S}_e^{\circ}}} \mathcal{C}_{\lambda,\rho}({\bf x})\leq \psi(\tau^\star,\beta^\star)+\gamma\epsilon^2\Big]\leq \mathbb{P}\Big[\min_{\substack{\tilde{\bf x}\in \tilde{\mathcal{S}}_x^\circ\\ {\bf e}(\tilde{\bf x})\in \tilde{\mathcal{S}}_e^{\circ}}} \mathcal{C}_{\lambda,\rho}(\tilde{\bf x})\leq \psi(\tau^\star,\beta^\star)+\gamma\epsilon^{\frac{3}{2}}+s\max(\epsilon^{\frac{3}{2}},\frac{1}{\sqrt{n}})\Big]+C\exp(-cn).
$$
It suffices thus to show that:
\begin{equation}
\mathbb{P}\Big[\min_{\substack{\tilde{\bf x}\in \tilde{\mathcal{S}}_x^\circ\\ {\bf e}(\tilde{\bf x})\in \tilde{\mathcal{S}}_e^{\circ}}} \mathcal{C}_{\lambda,\rho}(\tilde{\bf x})\leq \psi(\tau^\star,\beta^\star)+\gamma\epsilon^{\frac{3}{2}}+s\max(\epsilon^{\frac{3}{2}},\frac{1}{\sqrt{n}})\Big]\to 0.\label{eq:prob_ineq_rev_new}
\end{equation}
For ${\bf x}\in \mathbb{R}^n$, define $\eta({\bf x}):=\frac{\|{\bf S}_{\pi_{(t)}(|{\bf x}|)}{\bf x}\|}{\|{\bf S}_{\pi_{(t)}^c(|{\bf x}|)}{\bf x}\|}$. As will be shown later, Theorem \ref{th:new_cgmt} can  be used  to handle probability inequalities involving  stochastic optimization problems of the form 
\begin{equation}
\min_{\substack{\tilde{\bf x}\in \tilde{\mathcal{S}}_{x}^{\circ} \\ {\bf e}(\tilde{\bf x})\in \tilde{\mathcal{S}}_e^{\circ}\\ \eta(\tilde{\bf x})=\eta\\ \#\{\pi_{(t)}(|\tilde{\bf x}|)=k\}}} \mathcal{C}_{\lambda,\rho}(\tilde{\bf x}) \label{eq:new_form}
\end{equation}
where $\eta\in \mathcal{I}_\eta:= [\eta^\star-\tilde{c}_1\sqrt{\epsilon}, \eta^\star+\tilde{c}_1\sqrt{\epsilon}]$. In order to relate \eqref{eq:prob_ineq_rev_new} to that of probability inequalities involving \eqref{eq:new_form}, we define for $\eta\in\mathcal{I}_\eta$ and $k\in \{1,\cdots,n\}$
$$
\Upsilon_k(\eta):=\min_{\substack{\tilde{\bf x}\in \tilde{\mathcal{S}}_{x}^{\circ} \\ {\bf e}(\tilde{\bf x})\in \tilde{\mathcal{S}}_e^{\circ}\\ \eta(\tilde{\bf x})=\eta\\ \#\pi_{(t)}(|\tilde{\bf x}|)=k}} \mathcal{C}_{\lambda,\rho}(\tilde{\bf x}).
$$
Noting that:
$$
\min_{\substack{\tilde{\bf x}\in \tilde{\mathcal{S}}_{x}^{\circ} \\ {\bf e}(\tilde{\bf x})\in \tilde{\mathcal{S}}_e^{\circ}}} \mathcal{C}_{\lambda,\rho}(\tilde{\bf x}) =\min_{k\in\{1,\cdots,n\}}\min_{\eta\in \mathcal{I}_\eta} \Upsilon_k(\eta),
$$
the proof of \eqref{eq:prob_ineq_rev_new} boils down to showing
\begin{equation}
\mathbb{P}\Big[\min_{\substack{\eta\in\mathcal{I}_\eta\\ k=1,\cdots,n}} \Upsilon_k(\eta)\leq \psi(\tau^\star,\beta^\star)+\gamma \epsilon^{\frac{3}{2}}+s\max(\epsilon^{\frac{3}{2}},\frac{1}{\sqrt{n}})\Big]\to 0. \label{eq:Upsilon_revised}
\end{equation}
Let $\hat{\eta}_k:=\argmin_{\eta\in \mathcal{I}_\eta} \Upsilon_k(\eta)$, and consider the event:
\begin{equation}
\mathcal{A}:=\left\{\min_{\substack{\eta\in\mathcal{I}_\eta}} \Upsilon_k(\eta)\leq \psi(\tau^\star,\beta^\star)+\gamma \epsilon^{\frac{3}{2}}+s\max(\epsilon^{\frac{3}{2}},\frac{1}{\sqrt{n}})\right\}. \label{eq:A_rev}
\end{equation}
We shall prove that on the event $\mathcal{A}$, there exists a constant $C_\eta$ such that with probability $1-\frac{C}{\epsilon^{\frac{3}{2}}}\exp(-cn\epsilon^3)$, 
\begin{equation}
\forall \eta\in \mathcal{I}_\eta, \ \ \ |\Upsilon_k(\eta)-\Upsilon_k(\hat{\eta}_k)|\leq C_\eta|\hat{\eta}_k-\eta| .\label{eq:lipschitz_revised}
\end{equation}
Prior to showing \eqref{eq:lipschitz_revised}, let us show how it reduces \eqref{eq:Upsilon_revised} to that of controlling a stochastic optimization problem complying with the requirement of Theorem \ref{th:new_cgmt}. Consider $\mathcal{R}_\eta:=\{\eta_i\}_{i=1}^{\lceil\frac{2\tilde{c}_1}{\epsilon}\rceil+1}$ be a grid of the interval $\mathcal{I}_\eta$, such that for any $\eta\in \mathcal{I}_\eta$, there exists $\eta_i\in \mathcal{R}$ such that $|\eta-\eta_i|\leq \epsilon^{\frac{3}{2}}$. Hence, taking \eqref{eq:lipschitz_revised} as granted, we obtain:
\begin{equation}
\Upsilon_k(\hat{\eta})\geq \min_{\eta\in \mathcal{R}_\eta}\Upsilon_k(\eta) -C_\eta\epsilon^{\frac{3}{2}}. \label{eq:Upsilon_eta_revised}
\end{equation}
Using \eqref{eq:Upsilon_eta_revised} along with a union bound argument, we obtain:
\begin{align*}
&\mathbb{P}\Big[\min_{k=1,\cdots,n}\min_{\eta\in\mathcal{I}_\eta }\Upsilon_k(\eta)\leq \psi(\tau^\star,\beta^\star)+\gamma \epsilon^{\frac{3}{2}}+s\max(\epsilon^{\frac{3}{2}},\frac{1}{\sqrt{n}})\Big] \\
&\leq \sum_{k=1}^n(\lceil2\tilde{c}_1\epsilon^{-1}\rceil+1) \max_{\eta\in\mathcal{I}_\eta} \mathbb{P}\Big[\Upsilon_k(\eta)\leq \psi(\tau^\star,\beta^\star)+C_\eta\epsilon^{\frac{3}{2}}+\gamma \epsilon^{\frac{3}{2}}+s\max(\epsilon^{\frac{3}{2}},\frac{1}{\sqrt{n}})\Big].
\end{align*}
Hence, to prove \eqref{eq:Upsilon_revised}, it suffices to show that for any $\eta\in \mathcal{I}_\eta$, and $n$ sufficiently large:
\begin{equation}
\mathbb{P}\Big[\Upsilon_k(\eta)\leq \psi(\tau^\star,\beta^\star)+C_\eta\epsilon^{\frac{3}{2}}+\gamma \epsilon^{\frac{3}{2}}+s\max(\epsilon^{\frac{3}{2}},\frac{1}{\sqrt{n}})\Big]\leq \frac{C}{\epsilon^{\frac{3}{2}}}\exp(-cn\epsilon^3). \label{eq:rev_new}
\end{equation}
We note that the optimization variable $\tilde{\bf x}$ in $\Upsilon_k(\eta)$ complies  with the criteria of Theorem \ref{th:new_cgmt}. It remains thus to establish \eqref{eq:lipschitz_revised}, which is crucial to reduce the probability inequality in   \eqref{eq:Upsilon_revised} to \eqref{eq:rev_new}. For that, we need the following intermediate result:
\begin{lemma}
  For $\epsilon$ chosen sufficiently small, let $\eta_1$ and $\eta_2$ in $\mathcal{I}_\eta$ and $C_0$ a certain constant. Then, there is a positive constant $g$ such that if there exists $\tilde{{\bf x}}_1\in \mathcal{S}_3(\epsilon,C_0)\cap\tilde{\mathcal{S}}_x^{\circ}$ ($\mathcal{S}_3(\epsilon,C_0)$ being defined in \eqref{eq:C_3_rev}) with  $\eta(\tilde{\bf x}_1)=\eta_1$, then there exists $\tilde{\bf x}_2(\tilde{\bf x}_1)\in \tilde{\mathcal{S}}_x^{\circ}$ with $\|\tilde{\bf x}_2\|_{\infty}\leq \sqrt{P}$ and $\eta(\tilde{\bf x}_2)=\eta_2$
  such that
  \begin{align}
  &\eta({\tilde{\bf x}}_2(\tilde{\bf x}_1))=\eta_2,  \label{eq:req_rev1} \\
  &{\bf S}_{\pi_{(t)}^c(|\tilde{\bf x}_1|)}\tilde{\bf x}_1={\bf S}_{\pi_{(t)}^c(|\tilde{\bf x}_2(\tilde{\bf x}_1)|)}\tilde{\bf x}_2(\tilde{\bf x}_1),  \label{eq:req_rev2}\\
   & \frac{1}{\sqrt{n}}\|\tilde{\bf x}_1-\tilde{\bf x}_2(\tilde{\bf x}_1)\|\leq g |\eta_1-\eta_2|.   \label{eq:req_rev3}
  \end{align}
  \label{lem:lips_rev}
    \end{lemma}}
  {\begin{proof}
  The objective here is to prove the existence of $\tilde{\bf x}_2(\tilde{\bf x}_1)$ satisfying the requirements in \eqref{eq:req_rev1} to \eqref{eq:req_rev3}. For that, we construct $\tilde{\bf x}_2(\tilde{\bf x}_1,y)$ as:
$$
\tilde{\bf x}_2(\tilde{\bf x}_1,y)=(1+y){\bf S}_{\pi_{(t_1,t_2)}(|\tilde{\bf x}_1|)}\tilde{\bf x}_1+{\bf S}_{\pi_{(t_1,t_2)}^c(|\tilde{\bf x}_1|)}\tilde{\bf x}_1
$$
and prove that there exists $y$ with $|y|\leq \frac{g}{\sqrt{P}}|\eta_2-\eta_1|$ and $(1+y)t_2<t$,  for  some constant $g$ independent of $\eta$ such that $\tilde{\bf x}_2(\tilde{\bf x}_1,y)$   satisfy the requirements in Lemma \ref{lem:lips_rev}.
For that, we note the following equivalent equations:
\begin{align}
&	\frac{\|{\bf S}_{\pi_{(t)}^c(|\tilde{\bf x}_1|)}\tilde{\bf x}_1\|}{\|(1+y){\bf S}_{\pi_{(t_1,t_2)}(|\tilde{\bf x}_1|)}\tilde{\bf x}_1+{\bf S}_{\pi_{(t)}(|\tilde{\bf x}_1|)}{\bf S}_{\pi_{(t_1,t_2)}^c(|\tilde{\bf x}_1|)}\tilde{\bf x}_1\|}=\frac{1}{\eta_2}\nonumber\\
	\Longleftrightarrow& \frac{\|{\bf S}_{\pi_{(t)}^c(|\tilde{\bf x}_1|)}\tilde{\bf x}_1\|}{\|{\bf S}_{\pi_{(t)}(|\tilde{\bf x}_1|)}\tilde{\bf x}_1+y {\bf S}_{\pi_{(t_1,t_2)}(|\tilde{\bf x}_1|)}\tilde{\bf x}_1\|}=\frac{1}{\eta_2}\nonumber\\
	\Longleftrightarrow &\frac{1}{\sqrt{\eta_1^2+y^2\frac{\|{\bf S}_{\pi_{(t_1,t_2)}(|\tilde{\bf x}_1|)}\tilde{\bf x}_1\|^2}{\|{\bf S}_{\pi_{(t)}^c(|\tilde{\bf x}_1|)}\tilde{\bf x}_1\|^2}+2y\frac{\|{\bf S}_{\pi_{(t_1,t_2)}(|\tilde{\bf x}_1|)}\tilde{\bf x}_1\|^2}{\|{\bf S}_{\pi_{(t)}^c(|\tilde{\bf x}_1|)}\tilde{\bf x}_1\|^2}}}=\frac{1}{\eta_2}\nonumber\\
\Longleftrightarrow& y^2\frac{1}{\eta_2^2}\frac{\|{\bf S}_{\pi_{(t_1,t_2)}(|\tilde{\bf x}_1|)}\tilde{\bf x}_1\|^2}{\|{\bf S}_{\pi_{(t)}^c(|\tilde{\bf x}_1|)}\tilde{\bf x}_1\|^2}+2y\frac{1}{\eta_2^2}\frac{\|{\bf S}_{\pi_{(t_1,t_2)}(|\tilde{\bf x}_1|)}\tilde{\bf x}_1\|^2}{\|{\bf S}_{\pi_{(t)}^c(|\tilde{\bf x}_1|)}\tilde{\bf x}_1\|^2}+\frac{\eta_1^2}{\eta_2^2}-1=0.\nonumber\\
\Longleftrightarrow &y^2+2y+(\frac{\eta_1^2}{\eta_2^2}-1)\frac{\eta_2^2\|{\bf S}_{\pi_{(t)}^c(|\tilde{\bf x}_1|)}\tilde{\bf x}_1\|^2}{\|{\bf S}_{\pi_{(t_1,t_2)}(|\tilde{\bf x}_1|)}\tilde{\bf x}_1\|^2}=0. \label{eq:y_eq_rev}
\end{align}
For $\tilde{\bf x}_1\in \mathcal{S}_3(\epsilon,C_0)\cap \tilde{\mathcal{S}}_x^{\circ}$, it can be readily checked that there exists constants $c_l $ and $c_u$ such that:  
$$
c_l\leq\frac{\|{\bf S}_{\pi_{(t)}^c(|\tilde{\bf x}_1|)}\tilde{\bf x}_1\|^2}{\|{\bf S}_{\pi_{(t_1,t_2)}(|\tilde{\bf x}_1|)}\tilde{\bf x}_1\|^2}\leq c_u.
$$
Consequently, for any $\eta_2,\eta_1\in \mathcal{I}_\eta$, there exists two constants $g_l$ and $g_u$ such that:
$$
\forall  \tilde{\bf x}\in \mathcal{S}_3(\epsilon,C_0)\cap \tilde{\mathcal{S}}_x^{\circ}, \   \ \ g_l|\eta_2-\eta_1|\leq |\frac{\eta_1^2}{\eta_2^2}-1|\frac{\eta_2^2\|{\bf S}_{\pi_{(t)}^c(|\tilde{\bf x}|)}\tilde{\bf x}\|^2}{\|{\bf S}_{\pi_{(t_1,t_2)}(|\tilde{\bf x}|)}\tilde{\bf x}\|^2}\leq g_u|\eta_2-\eta_1| \leq 2g_u\tilde{c}_1\sqrt{\epsilon}. 
$$
{For sufficiently small $\epsilon$, the  solutions to \eqref{eq:y_eq_rev} exist, and are given by:} 
$$
y_1^\star=-1+\sqrt{1-(\frac{\eta_1^2}{\eta_2^2}-1)\frac{\eta_2^2\|{\bf S}_{\pi_{(t)}^c(|\tilde{\bf x}_1|)}\tilde{\bf x}_1\|^2}{\|{\bf S}_{\pi_{(t_1,t_2)}(|\tilde{\bf x}_1|)}\tilde{\bf x}_1\|^2}}  \ \text{and} \ y_2^\star=-1-\sqrt{1-(\frac{\eta_1^2}{\eta_2^2}-1)\frac{\eta_2^2\|{\bf S}_{\pi_{(t)}^c(|\tilde{\bf x}_1|)}\tilde{\bf x}_1\|^2}{\|{\bf S}_{\pi_{(t_1,t_2)}(|\tilde{\bf x}_1|)}\tilde{\bf x}_1\|^2}}.
$$
Consider the solution $y_1^\star$. We will prove that whether $\eta_2\geq \eta_1$ or $\eta_2\leq \eta_1$, $|y_1^\star|\leq \frac{g}{\sqrt{P}}|\eta_2-\eta_1|$ for some constant $g$.  {The proof can then be completed by combining this bound with the fact that $\frac{\|\tilde{\bf x}_1-\tilde{\bf x}_2(\tilde{\bf x}_1)\|}{\sqrt n}=\frac{\|y^\star_1{\bf S}_{\pi_{(t_1,t_2)}(|\tilde{\bf x}_1|)}\tilde{\bf x}_1\|}{\sqrt n}\leq\sqrt{P}|y_1^\star|$}. 
Indeed, if $\eta_1\geq \eta_2$, then using the inequality $\sqrt{1-x}\geq 1-2x$ which holds for all $x\in(0,1-\frac{1}{16})$, we obtain
$$
y_1^\star\geq -2(\frac{\eta_1^2}{\eta_2^2}-1)\frac{\eta_2^2\|{\bf S}_{\pi_{(t)}^c(|\tilde{\bf x}_1|)}\tilde{\bf x}_1\|^2}{\|{\bf S}_{\pi_{(t_1,t_2)}(|\tilde{\bf x}_1|)}\tilde{\bf x}_1\|^2}.
$$
Since $y_1^\star\leq 0$, we thus have:
$$
|y_1^\star|\leq 2g_u|\eta_2-\eta_1|.  
$$
On the opposite case, when $\eta_2\geq \eta_1$, using the fact that $\sqrt{1+x}\leq 1+\frac{1}{2}x$ which holds for all positive $x$, we obtain:
$$
0<y_1^\star<\frac{1}{2}g_u|\eta_2-\eta_1|.
$$
 To complete the proof, we observe that by choosing \( \epsilon \) sufficiently small, the inequality \( (1 + y)t_2 < t \) holds. As a result, we obtain  
$$
\mathbf{S}_{\pi_{t}^c(|\tilde{\mathbf{x}}_2(\tilde{\mathbf{x}}_1, y)|)} \tilde{\mathbf{x}}_2(\tilde{\mathbf{x}}_1, y) = \mathbf{S}_{\pi_{t}^c(|\tilde{\bf x}|)}\tilde{\bf x}$$
which concludes the proof.
  \end{proof}}
 {\noindent{\underline{Proof of \eqref{eq:lipschitz_revised}}.}  For $\eta\in \mathcal{I}_\eta$ and $k=1,\cdots,n$, let $\hat{\bf x}_k$ be defined as:
$$
\hat{\bf x}_k:=\argmin_{\substack{\tilde{\bf x}\in \tilde{\mathcal{S}}_{x}^{\circ} \\ {\bf e}(\tilde{\bf x})\in \mathcal{S}_e^{\circ}\\ \eta(\tilde{\bf x})=\hat{\eta}_k\\ \#\pi_{(t)}(|{\bf x}|)=k}} \mathcal{C}_{\lambda,\rho}(\tilde{\bf x}).
$$
Then,  for $n\geq \epsilon^3$, it follows from \eqref{eq:C_0_rev1} that on the event $\mathcal{A}$ defined in \eqref{eq:A_rev}, there exists $C_0$ such that $\hat{\bf x}_k\in \mathcal{S}_{3}(\epsilon,C_0)$  with probability $1-\frac{C}{\epsilon^{\frac{3}{2}}}\exp(-cn\epsilon^3)$. Using Lemma \ref{lem:lips_rev}, we deduce that  there exists $\tilde{\bf x}$ such that $\eta(\tilde{\bf x})=\eta$ and $\frac{1}{\sqrt{n}}\|\hat{\bf x}-\tilde{\bf x}\|\leq g|\hat{\eta}_k-\eta|$, Hence,  
$$
0\leq \Upsilon_k(\eta)-\Upsilon_k(\hat{\eta}_k)\leq \mathcal{C}_{\lambda,\rho}(\tilde{\bf x})-\mathcal{C}_{\lambda,\rho}(\hat{\bf x}_k).
$$
We can easily check that with probability $1-C\exp(-cn)$, ${\bf x}\mapsto \mathcal{C}_{\lambda,\rho}$ is $\frac{C_L}{\sqrt{n}}$ Lipschitz for some constant $C_L$. Hence, 
$$
\mathcal{C}_{\lambda,\rho}(\tilde{\bf x})-\mathcal{C}_{\lambda,\rho}(\hat{\bf x}_k)\leq C_L\frac{1}{\sqrt{n}}\|\hat{\bf x}_k-\tilde{\bf x}\|\leq C_Lg|\hat{\eta}_k-\eta|.
$$
Hence \eqref{eq:lipschitz_revised} follows by setting $C_\eta=C_Lg$.
\subsubsection{Final steps to enable the application of Theorem \ref{th:new_cgmt}.} We recall that our goal is apply  Theorem \ref{th:new_cgmt} to show \eqref{eq:rev_new} for any $\eta\in \mathcal{I}_\eta$, or equivalently:
$$
\mathbb{P}\Big[\min_{\substack{\tilde{\bf x}\in \tilde{\mathcal{S}}_{x}^{\circ} \\  \eta(\tilde{\bf x})=\eta\\ \# \pi_{(t)}(|\tilde{\bf x}|) =k}} \min_{{\bf e}\in \tilde{\mathcal{S}}_e^{\circ}}\mathcal{M}_{\lambda,\rho}(\tilde{\bf x},{\bf e})\leq \psi(\tau^\star,\beta^\star)+(\gamma+C_\eta) \epsilon^{\frac{3}{2}}+s\max(\epsilon^{\frac{3}{2}},\frac{1}{\sqrt{n}})\Big].
$$
With probability $1-C\exp(-cn)$, $\|{\bf H}\|\leq 3\max(1,\sqrt{\delta})$. Hence, there exist $K_e$ such that for any $\tilde{\bf x}$ with $\|\tilde{\bf x}\|_{\infty}\leq \sqrt{P}$ $\|{\bf e}(\tilde{\bf x})\|\leq K_e\sqrt{m}$. Hence 
\begin{align*}
&\mathbb{P}\Big[\min_{\substack{\tilde{\bf x}\in \tilde{\mathcal{S}}_{x}^{\circ} \\  \eta(\tilde{\bf x})=\eta\\ \#\pi_{(t)}(|\tilde{\bf x}|)=k}} \min_{{\bf e}\in \tilde{\mathcal{S}}_e^{\circ}}\mathcal{M}_{\lambda,\rho}(\tilde{\bf x},{\bf e})\leq \psi(\tau^\star,\beta^\star)+(\gamma+C_\eta) \epsilon^{\frac{3}{2}}+s\max(\epsilon^{\frac{3}{2}},\frac{1}{\sqrt{n}})\Big] \\
&\leq \mathbb{P}\Big[\min_{\substack{\tilde{\bf x}\in \tilde{\mathcal{S}}_{x}^{\circ} \\  \eta(\tilde{\bf x})=\eta\\ \#\pi_{(t)}(|\tilde{\bf x}|)=k}} \min_{\substack{{\bf e}\in \tilde{\mathcal{S}}_e^{\circ}\\ \|{\bf e}\|\leq K_e\sqrt{m}}}\mathcal{M}_{\lambda,\rho}(\tilde{\bf x},{\bf e})\leq \psi(\tau^\star,\beta^\star)+(\gamma+C_\eta) \epsilon^{\frac{3}{2}}+s\max(\epsilon^{\frac{3}{2}},\frac{1}{\sqrt{n}})\Big]+C\exp(-cn).
\end{align*}
Note that we are not ready yet to apply Theorem \ref{th:new_cgmt}, since the optimization variable $\boldsymbol{\lambda}$ arising in the expression of $\mathcal{M}_{\lambda,\rho}({\bf x},{\bf e})$ is not aligned with ${\bf u}$. To solve this issue, we note that $\mathcal{M}_{\lambda,\rho}({\bf x},{\bf e})$ can be lower-bounded as:
\begin{align*}
\mathcal{M}_{\lambda,\rho}({\bf x},{\bf e})\geq \tilde{\mathcal{M}}_{\lambda,\rho}({\bf x},{\bf e})&:=\max_{{\bf u}\in \mathcal{S}_{{\bf u}}} \frac{1}{\sqrt{n}} {\bf u}^{T}({\bf H}({\bf x}-r{\bf S}_{\pi_{(t)}^c(|{\bf x}|)}{\bf x}))+\frac{r}{\sqrt{n}}{\bf u}^{T}{\bf e}-\frac{\sqrt{\rho}}{\sqrt{n}}{\bf u}^{T}{\bf s} \\
&+\frac{\lambda_1}{n}\|{\bf x}\|_1+\frac{\lambda_2}{n}\|{\bf x}\|^2+\sup_{\ell \in \mathbb{R}}\frac{\ell {\bf u}^{T}}{\sqrt{n}}({\bf H}{\bf S}_{\pi_{(t)}^c(|{\bf x}|)}{\bf x}-{\bf e})
\end{align*}
and hence:
\begin{align*}
&\mathbb{P}\Big[\min_{\substack{\tilde{\bf x}\in \tilde{\mathcal{S}}_{x}^{\circ} \\  \eta(\tilde{\bf x})=\eta\\ \#\pi_{(t)}(|\tilde{\bf x}|)=k}} \min_{\substack{{\bf e}\in \tilde{\mathcal{S}}_e^{\circ}\\ \|{\bf e}\|\leq K_e\sqrt{m}}}\mathcal{M}_{\lambda,\rho}(\tilde{\bf x},{\bf e})\leq \psi(\tau^\star,\beta^\star)+(\gamma+C_\eta) \epsilon^{\frac{3}{2}}+s\max(\epsilon^{\frac{3}{2}},\frac{1}{\sqrt{n}})\Big]  \\
&\leq \mathbb{P}\Big[\min_{\substack{\tilde{\bf x}\in \tilde{\mathcal{S}}_{x}^{\circ} \\  \eta(\tilde{\bf x})=\eta\\ \#\pi_{(t)}(|\tilde{\bf x}|)=k}} \min_{\substack{{\bf e}\in \tilde{\mathcal{S}}_e^{\circ}\\ \|{\bf e}\|\leq K_e\sqrt{m}}}\tilde{\mathcal{M}}_{\lambda,\rho}(\tilde{\bf x},{\bf e})\leq \psi(\tau^\star,\beta^\star)+(\gamma+C_\eta) \epsilon^{\frac{3}{2}}+s\max(\epsilon^{\frac{3}{2}},\frac{1}{\sqrt{n}})\Big].
\end{align*}
With this, we reduce the proof of \eqref{eq:rev_new} to that of showing:
\begin{equation}
\mathbb{P}\Big[\min_{\substack{\tilde{\bf x}\in \tilde{\mathcal{S}}_{x}^{\circ} \\  \eta(\tilde{\bf x})=\eta\\ \#\pi_{(t)}(|\tilde{\bf x}|)=k}} \min_{\substack{{\bf e}\in \tilde{\mathcal{S}}_e^{\circ}\\ \|{\bf e}\|\leq K_e\sqrt{m}}}\tilde{\mathcal{M}}_{\lambda,\rho}(\tilde{\bf x},{\bf e})\leq \psi(\tau^\star,\beta^\star)+(\gamma+C_\eta) \epsilon^{\frac{3}{2}}+s\max(\epsilon^{\frac{3}{2}},\frac{1}{\sqrt{n}})\Big]\leq \frac{C}{\epsilon^{\frac{3}{2}}}\exp(-cn\epsilon^3). \label{eq:cgmt_form_rev}
\end{equation}}
\subsection{Proof of \eqref{eq:cgmt_form_rev} via Theorem \ref{th:new_cgmt}} {
\subsubsection{Preliminaries}
The Gaussian process arising in the left-hand side of \eqref{eq:cgmt_form_rev} satisfies the requirements of Theorem \ref{th:new_cgmt}. By applying Theorem \ref{th:new_cgmt}, we can upper-bound the probability term in \eqref{eq:cgmt_form_rev} as:
\begin{align*}
&\mathbb{P}\Big[\min_{\substack{\tilde{\bf x}\in \tilde{\mathcal{S}}_{x}^{\circ} \\  \eta(\tilde{\bf x})=\eta\\ \#\pi_{(t)}(|\tilde{\bf x}|)=k}} \min_{\substack{{\bf e}\in \tilde{\mathcal{S}}_e^{\circ}\\ \|{\bf e}\|\leq K_e\sqrt{m}}}\tilde{\mathcal{M}}_{\lambda,\rho}(\tilde{\bf x},{\bf e})\leq \psi(\tau^\star,\beta^\star)+(\gamma+C_\eta) \epsilon^{\frac{3}{2}}+s\max(\epsilon^{\frac{3}{2}},\frac{1}{\sqrt{n}})\Big]\\
&\leq 4\mathbb{P}\Big[\min_{\substack{\tilde{\bf x}\in \tilde{\mathcal{S}}_{x}^{\circ} \\  \eta(\tilde{\bf x})=\eta\\ \#\pi_{(t)}(|\tilde{\bf x}|)=k}} \min_{\substack{{\bf e}\in \tilde{\mathcal{S}}_e^{\circ}\\ \|{\bf e}\|\leq K_e\sqrt{m}}} \tilde{\mathcal{H}}_{\lambda,\rho}({\tilde{\bf x}},{\bf e}) \leq \psi(\tau^\star,\beta^\star)+(\gamma+C_\eta) \epsilon^{\frac{3}{2}}+s\max(\epsilon^{\frac{3}{2}},\frac{1}{\sqrt{n}}) \Big]
\end{align*}
where 
\begin{align*}
\tilde{\mathcal{H}}_{\lambda,\rho}({\tilde{\bf x}},{\bf e})&:=\max_{{\bf u}\in \mathcal{S}_{\bf u}} \frac{1}{n}{\bf g}^{T}{\bf u}\|\tilde{\bf x}-r{\bf S}_{\pi_{(t)}^{c}(\tilde{\bf x})}\tilde{\bf x}\|-\frac{\|{\bf u}\|}{n}{\bf h}^{T}(\tilde{\bf x}-r{\bf S}_{\pi_{(t)}{c}(\tilde{\bf x})}\tilde{\bf x})+\frac{r}{\sqrt{n}}{\bf u}^{T}{\bf e}-\frac{\sqrt{\rho}}{\sqrt{n}}{\bf u}^{T}{\bf s} -\frac{\|{\bf u}\|^2}{4} \\
&+\sup_{\ell\in \mathbb{R}} \ell (\frac{1}{n}{\bf u}^{T}(\tilde{\eta}_r{\bf g}+\nu_r\tilde{\bf g})\|{\bf S}_{\pi_{(t)}^c(|\tilde{\bf x}|)}\tilde{\bf x}\|) - \ell\frac{1}{\sqrt{n}}{\bf u}^{T}{\bf e} -\frac{|\ell|}{n}\|{\bf u}\|{\bf h}^{T}{\bf S}_{\pi_{(t)}^c(|\tilde{\bf x}|)}\tilde{\bf x}+\frac{\lambda_1}{n}\|\tilde{\bf x}\|_1+\frac{\lambda_2}{n}\|\tilde{\bf x}\|_2^2
\end{align*}
where ${\bf g}, \tilde{\bf g}\in \mathbb{R}^m$ and ${\bf h}\in \mathbb{R}^n$  are independent standard Gaussian vectors, whereas $\tilde{\eta}_r=\frac{1-r}{\sqrt{\eta^2+(1-r)^2}}$ and $\nu_r=\sqrt{1-\tilde{\eta}_r^2}$. To prove \eqref{eq:cgmt_form_rev}, it suffices that to show that 
\begin{equation}
\mathbb{P}\Big[\min_{\substack{\tilde{\bf x}\in \tilde{\mathcal{S}}_{x}^{\circ} \\  \eta(\tilde{\bf x})=\eta\\ \#\pi_{(t)}(|\tilde{\bf x}|)=k}} \min_{\substack{{\bf e}\in \tilde{\mathcal{S}}_e^{\circ}\\ \|{\bf e}\|\leq K_e\sqrt{m}}} \tilde{\mathcal{H}}_{\lambda,\rho}({\tilde{\bf x}},{\bf e}) \leq \psi(\tau^\star,\beta^\star)+(\gamma+C_\eta) \epsilon^{\frac{3}{2}}+s\max(\epsilon^{\frac{3}{2}},\frac{1}{\sqrt{n}}) \Big]\leq \frac{C}{\epsilon^{\frac{3}{2}}}\exp(-cn\epsilon^3).  \label{eq:required_H_rev}
\end{equation}
Let $\overline{\bf e}_t^{\rm AO}$ be defined as:
$$
\overline{\bf e}_t^{\rm AO}=\sqrt{\rho}\theta^\star{\bf s}+(\tilde{\eta}^\star\tilde{\alpha}^\star+\theta^\star\sqrt{(\tau^\star)^2\delta-\rho}){\bf g} +\nu^\star \tilde{\alpha}^{\star}\tilde{\bf g}
$$
where $\tilde{\alpha}^\star$ and $\theta^\star$ are defined in Theorem \ref{th:main_theorem} and 
$$
\tilde{\eta}^\star=\frac{1}{\sqrt{(\eta^\star)^2+1}} \ \ \text{and } \nu^\star=\frac{\eta^\star}{\sqrt{(\eta^\star)^2+1}}.
$$
We can easily check that: $\tilde{\eta}^\star=\frac{\tilde{\alpha}^\star}{\sqrt{(\tau^\star)^2\delta-\rho}}$. Then, $\overline{\bf e}_t^{\rm AO}$ can be written as:
$$
\overline{\bf e}_t^{\rm AO}:=-\theta^\star \sqrt{\rho}{\bf s} +\sqrt{(\tilde{\alpha}^\star)^2+(\theta^\star)^2(\delta(\tau^\star)^2-\rho)+2(\tilde{\alpha}^\star)^2\theta^\star} \overline{\bf g}
$$
where $\overline{\bf g}$ is a standard Gaussian vector. 
Let $\hat{\mu}(\overline{\bf e}_{t}^{\rm AO}, {\bf s})$ be the joint empirical distribution of $\overline{\bf e}_{t}^{\rm AO}$ and ${\bf s}$. In the following Lemma, we prove that with overwhelming probability, the empirical distribution $\hat{\mu}(\overline{\bf e}_{t}^{\rm AO}, {\bf s})$ is close to $\mu_{t}^{\star}$:
\begin{lemma}
There exist constants $C$ and $c$ such that for any $\epsilon>0$:
$$
\mathbb{P}\Big[(\mathcal{W}_2(\hat{\mu}(\overline{\bf e}_t^{\rm AO},{\bf s})),\mu_{t}^\star)^2\geq \frac{(\sqrt{c_e}-\sqrt{C_e})^2\sqrt{\epsilon}}{4}\Big]\leq C\exp(-cn\epsilon).
$$
\label{lem:distance_revised}
\end{lemma}
\begin{proof}
To avoid disrupting the flow of the main proof, we defer the proof of this result to the Appendix \ref{tech_lem}
\end{proof}
Recall that the set $\tilde{\mathcal{S}}_e^{\circ}$ is defined as:
$$
\tilde{\mathcal{S}}_e^{\circ}:=\{{\bf e}\in \mathbb{R}^m, \mathcal{W}_2(\hat{\mu}(\overline{\bf e}_t^{\rm AO},{\bf s}),\mu_{t}^\star)\geq (\sqrt{c_e}-\sqrt{C_e})\epsilon^{\frac{1}{4}}\}.
$$
Using the triangular inequality:
$$
\mathcal{W}_2(\hat{\mu}(\overline{\bf e}_t^{\rm AO},{\bf s}),\hat{\mu}({\bf e},{\bf s}))\geq \mathcal{W}_2(\hat{\mu}({\bf e},{\bf s}),\mu_{t}^\star)-\mathcal{W}_2(\hat{\mu}(\overline{\bf e}_t^{\rm AO},{\bf s}),\mu_{t}^\star)
$$
in combination with the result of Lemma \ref{lem:distance_revised}, we conclude that with probability $1-C\exp(-cn\epsilon)$:
$$
{\bf e}\in \tilde{\mathcal{S}}_e^{\circ} \ \Longrightarrow  \  \mathcal{W}_2(\hat{\mu}(\overline{\bf e}_t^{\rm AO},{\bf s}),\hat{\mu}({\bf e},{\bf s}))\geq \frac{(\sqrt{c_e}-\sqrt{C_e})}{2}\epsilon^{\frac{1}{4}}\Longrightarrow \frac{1}{\sqrt{m}}\|{\bf e}-\overline{\bf e}_{t}^{\rm AO}\|\geq \frac{(\sqrt{c_e}-\sqrt{C_e})}{2}\epsilon^{\frac{1}{4}}. 
$$
Hence, with probability $1-C\exp(-cn\epsilon)$:
$$
\tilde{\mathcal{S}}_e^{\circ}\subset \hat{\mathcal{S}}_e^{\circ}:=\{{\bf e}, \frac{1}{\sqrt{m}}\|{\bf e}-\overline{\bf e}_t^{\rm AO}\|\geq \frac{(\sqrt{c_e}-\sqrt{C_e})}{2}\epsilon^{\frac{1}{4}}\}.
$$
This combined with  the fact that $\tilde{\mathcal{S}}_x^{\circ}\subset \tilde{\mathcal{S}}_{x2}^{\circ}$  defined in \eqref{eq:Sx2_rev} leads to:
\begin{align*}
&\mathbb{P}\Big[\min_{\substack{\tilde{\bf x}\in \tilde{\mathcal{S}}_{x}^{\circ} \\  \eta(\tilde{\bf x})=\eta\\ \#\pi_{(t)}(|\tilde{\bf x}|)=k}} \min_{\substack{{\bf e}\in \tilde{\mathcal{S}}_e^{\circ}\\ \|{\bf e}\|\leq K_e\sqrt{m}}} \tilde{\mathcal{H}}_{\lambda,\rho}({\tilde{\bf x}},{\bf e}) \leq \psi(\tau^\star,\beta^\star)+(\gamma+C_\eta) \epsilon^{\frac{3}{2}}+s\max(\epsilon^{\frac{3}{2}},\frac{1}{\sqrt{n}}) \Big]  \\
&\leq \mathbb{P}\Big[\min_{\substack{\tilde{\bf x}\in \tilde{\mathcal{S}}_{x2}^{\circ} \\  \eta(\tilde{\bf x})=\eta\\ \#\pi_{(t)}(|\tilde{\bf x}|)=k}} \min_{\substack{{\bf e}\in \hat{\mathcal{S}}_e^{\circ}\\ \|{\bf e}\|\leq K_e\sqrt{m}}} \tilde{\mathcal{H}}_{\lambda,\rho}({\tilde{\bf x}},{\bf e}) \leq \psi(\tau^\star,\beta^\star)+(\gamma+C_\eta) \epsilon^{\frac{3}{2}}+s\max(\epsilon^{\frac{3}{2}},\frac{1}{\sqrt{n}}) \Big].
\end{align*}
Based on this, we reduce the proof of \eqref{eq:required_H_rev} to that of showing:
\begin{equation}
\mathbb{P}\Big[\min_{\substack{\tilde{\bf x}\in \tilde{\mathcal{S}}_{x2}^{\circ} \\  \eta(\tilde{\bf x})=\eta\\ \#\pi_{(t)}(|\tilde{\bf x}|)=k}} \min_{\substack{{\bf e}\in \hat{\mathcal{S}}_e^{\circ}\\ \|{\bf e}\|\leq K_e\sqrt{m}}} \tilde{\mathcal{H}}_{\lambda,\rho}({\tilde{\bf x}},{\bf e}) \leq \psi(\tau^\star,\beta^\star)+(\gamma+C_\eta) \epsilon^{\frac{3}{2}}+s\max(\epsilon^{\frac{3}{2}},\frac{1}{\sqrt{n}}) \Big]\leq \frac{C}{\epsilon^{\frac{3}{2}}}\exp(-cn\epsilon^3) \label{eq:desired_revision_final}
\end{equation}
for $n> \frac{1}{\epsilon^3}$. 
\subsubsection{Methodology of the proof}
To prove \eqref{eq:desired_revision_final}, we proceed into the following steps. \\
\underline{Step 1: Minimization with respect to ${\bf e}$: optimal value and solution set.}
For ${\tilde{\bf x}}\in \mathbb{R}^n$, we define the set:
$$
\mathcal{S}_e(\tilde{\bf x}):=\{\tilde{\bf e}\in \mathbb{R}^m, \ \|\frac{1}{n}(\tilde{\eta}_r{\bf g}+\nu_r\tilde{\bf g})\|{\bf S}_{\pi_{(t)}(|\tilde{\bf x}|)}\tilde{\bf x}\|-\frac{1}{\sqrt{n}}\tilde{\bf e}\|\leq \frac{1}{n}{\bf h}^{T}{\bf S}_{\pi_{(t)}(|\tilde{\bf x}|)}\tilde{\bf x}\}.
$$
Then, by noting that:
\begin{equation}
{\bf e}\notin \mathcal{S}_e(\tilde{\bf x}) \Longrightarrow \tilde{\mathcal{H}}_{\lambda,\rho}(\tilde{\bf x})=\infty \label{eq:infinite_rev}
\end{equation}
we obtain:
$$
\min_{\substack{\tilde{\bf x}\in \tilde{\mathcal{S}}_{x2}^{\circ} \\  \eta(\tilde{\bf x})=\eta\\ \#\pi_{(t)}(|\tilde{\bf x}|)=k}} \min_{\substack{{\bf e}\in \hat{\mathcal{S}}_e^{\circ}\\ \|{\bf e}\|\leq K_e\sqrt{m}}} \tilde{\mathcal{H}}_{\lambda,\rho}({\tilde{\bf x}},{\bf e}) = \min_{\substack{\tilde{\bf x}\in \tilde{\mathcal{S}}_{x2}^{\circ} \\  \eta(\tilde{\bf x})=\eta\\ \#\pi_{(t)}(|\tilde{\bf x}|)=k}} \min_{\substack{{\bf e}\in \hat{\mathcal{S}}_e^{\circ}\\ \|{\bf e}\|\leq K_e\sqrt{m}\\ {\bf e}\in \mathcal{S}_e(\tilde{\bf x})}} \tilde{\mathcal{H}}_{\lambda,\rho}({\tilde{\bf x}},{\bf e}).
$$
Based on this, we show that ${\bf e}\mapsto \tilde{\mathcal{H}}_{\lambda,\rho}(\tilde{\bf x},{\bf e})$ is minimized at ${\bf e}_r(\tilde{\bf x})$ defined in \eqref{eq:er_rev}, and we check that there exists constant $C_r$ such that with probability $1-\frac{C}{\epsilon^{\frac{3}{2}}}\exp(-cn\epsilon^3)$
\begin{equation}
\min_{{\bf e}}\tilde{\mathcal{H}}_{\lambda,\rho}(\tilde{\bf x},{\bf e})\geq \mathcal{L}_{\lambda,\rho}^{\circ}(\tilde{\bf x})-C_a\epsilon^{\frac{3}{2}} \label{eq:C_r_revised}
\end{equation}
where $\mathcal{L}_{\lambda,\rho}^{\circ}(\tilde{\bf x})$ is defined in \eqref{eq:L0}.
Hence, using \eqref{eq:L0_rev}, there exists a constant $C_r$ such that with probability $1-\frac{C}{\epsilon^{\frac{3}{2}}}\exp(-cn\epsilon^3)$
$$
\min_{{\bf e}} \tilde{\mathcal{H}}_{\lambda,\rho}(\tilde{\bf x},{\bf e})\geq \psi(\tau^\star,\beta^\star)-(C_r+C_a)\epsilon^{\frac{3}{2}}.
$$
\underline{Step 2: Suboptimality outside a neighborhood of $\overline{\bf x}^{\rm AO}$.} Define $\overline{\bf x}^{\rm AO}$
$$
\overline{\bf x}^{\rm AO}={\rm prox}(\tilde{\tau}^{\star}{\bf h};\frac{\lambda_1\tilde{\tau}^{\star}}{\beta^\star}).
$$
By replacing ${\bf e}$ by ${\bf e}_r(\tilde{\bf x})$, we prove that there exists a constant $C_B$ such that with probability $1-\frac{C}{\epsilon^{\frac{3}{2}}}\exp(-cn\epsilon^3)$
\begin{equation}
\frac{1}{n}\|\tilde{\bf x}-\overline{\bf x}^{\rm AO}\|^2\geq C_B\epsilon^{\frac{3}{2}} \Longrightarrow \mathcal{H}_{\lambda,\rho}(\tilde{\bf x},{\bf e})\geq \psi(\tau^\star,\beta^\star)+(\gamma+C_\eta+s)\epsilon^{\frac{3}{2}}. \label{eq:step_2_revised}
\end{equation}
Let $\mathcal{B}_{\epsilon}(\overline{\bf x}_{\rm AO}):=\{\tilde{\bf x}, \ \|\tilde{\bf x}\|_{\infty}\leq \sqrt{P} \ \text{and }\frac{1}{n}\|\tilde{\bf x}-\overline{\bf x}^{\rm AO}\|^2\leq C_B\epsilon^{\frac{3}{2}} \}$. Hence, we obtain:
\begin{align*}
&\mathbb{P}\Big[\min_{\substack{   \eta(\tilde{\bf x})=\eta\\ \tilde{\bf x}\in \tilde{\mathcal{S}}_{x2}^{\circ}\\\#\pi_{(t)}(|\tilde{\bf x}|)=k}} \min_{\substack{{\bf e}\in \hat{\mathcal{S}}_e^{\circ}\\ \|{\bf e}\|\leq K_e\sqrt{m}}} \tilde{\mathcal{H}}_{\lambda,\rho}({\tilde{\bf x}},{\bf e}) \leq \psi(\tau^\star,\beta^\star)+(\gamma+C_\eta) \epsilon^{\frac{3}{2}}+s\max(\epsilon^{\frac{3}{2}},\frac{1}{\sqrt{n}}) \Big]\\
&\leq\mathbb{P}\Big[\min_{\substack{   \eta(\tilde{\bf x})=\eta\\ \tilde{\bf x}\in \tilde{\mathcal{S}}_{x2}^{\circ}\\\#\pi_{(t)}(|\tilde{\bf x}|)=k\\ \tilde{\bf x}\in \mathcal{B}_{\epsilon}(\tilde{\bf x}^{\rm AO})   }} \min_{\substack{{\bf e}\in \hat{\mathcal{S}}_e^{\circ}\\ \|{\bf e}\|\leq K_e\sqrt{m}}} \tilde{\mathcal{H}}_{\lambda,\rho}({\tilde{\bf x}},{\bf e}) \leq \psi(\tau^\star,\beta^\star)+(\gamma+C_\eta) \epsilon^{\frac{3}{2}}+s\max(\epsilon^{\frac{3}{2}},\frac{1}{\sqrt{n}}) \Big] +\frac{C}{\epsilon^{\frac{3}{2}}}\exp(-cn\epsilon^3)
\end{align*}
and thus we reduce the proof of \eqref{eq:desired_revision_final} to that of showing
\begin{equation}
\mathbb{P}\Big[\min_{\substack{   \eta(\tilde{\bf x})=\eta\\ \tilde{\bf x}\in \tilde{\mathcal{S}}_{x2}^{\circ}\\\#\pi_{(t)}(|\tilde{\bf x}|)=k\\ \tilde{\bf x}\in \mathcal{B}_{\epsilon}(\tilde{\bf x}^{\rm AO})   }} \min_{\substack{{\bf e}\in \hat{\mathcal{S}}_e^{\circ}\\ \|{\bf e}\|\leq K_e\sqrt{m}}} \tilde{\mathcal{H}}_{\lambda,\rho}({\tilde{\bf x}},{\bf e}) \leq \psi(\tau^\star,\beta^\star)+(\gamma+C_\eta) \epsilon^{\frac{3}{2}}+s\max(\epsilon^{\frac{3}{2}},\frac{1}{\sqrt{n}}) \Big]\leq \frac{C}{\epsilon^{\frac{3}{2}}}\exp(-cn\epsilon^3).\label{eq:desired_revision_final_step2}
\end{equation}
\underline{Step 3: Deviation argument.}  We prove that with probability at least \(1 - \frac{C}{\epsilon^{\frac{3}{2}}} \exp(-c n \epsilon^3)\), the function  
\[
\mathcal{G}_{\lambda,\rho}(\mathbf{e}) := \min_{\substack{   \eta(\tilde{\bf x})=\eta\\ \tilde{\bf x}\in \tilde{\mathcal{S}}_{x2}^{\circ}\\\#\pi_{(t)}(|\tilde{\bf x}|)=k\\ \tilde{\bf x}\in \mathcal{B}_{\epsilon}(\tilde{\bf x}^{\rm AO})   }} \tilde{\mathcal{H}}_{\lambda,\rho}(\tilde{\mathbf{x}}, \mathbf{e})
\]  
satisfies the following inequality for all \(\mathbf{e}\) such that \(\|\mathbf{e}\| \leq K_e \sqrt{m}\):
\begin{equation}
\mathcal{G}_{\lambda,\rho}(\mathbf{e}) \geq \psi(\tau^\star,\beta^\star) -(C_r+C_a)\epsilon^{\frac{3}{2}} + M_e \epsilon \cdot \frac{1}{m} \|\mathbf{e} - \mathbf{e}_r(\tilde{\mathbf{x}})\|^2
\label{eq:deviation_rev}
\end{equation}
where \(M_e\) is a certain positive constant, and \(\mathbf{e}_r(\tilde{\mathbf{x}})\) is defined in \eqref{eq:er_rev}. This property directly leads to the desired inequality \eqref{eq:desired_revision_final_step2}. To that end, note that there exists a constant \(D_e\) such that with probability at least \(1 - C \exp(-c n\epsilon)\), the following holds:
\begin{equation}
\sup_{\tilde{\mathbf{x}} \in \mathcal{B}_\epsilon(\tilde{\mathbf{x}}^{\rm AO})} \frac{1}{m} \|\mathbf{e}_r(\tilde{\mathbf{x}}) - \overline{\mathbf{e}}_t^{\rm AO}\|^2 \leq D_e \sqrt{\epsilon}. \label{eq:distance_er_et_rev}
\end{equation}
Combining this with \eqref{eq:deviation_rev}, we obtain:
\[
\mathcal{G}_{\lambda,\rho}(\mathbf{e}) \geq \psi(\tau^\star,\beta^\star)-(C_r+C_a)\epsilon^{\frac{3}{2}} + M_e \epsilon \left( \frac{1}{\sqrt{m}} \|\mathbf{e} - \overline{\mathbf{e}}_t^{\rm AO}\| - \sqrt{D_e} \epsilon^{1/4} \right)^2.
\]
Now, recalling from the definition of \(\hat{\mathcal{S}}_e^{\circ}\) that
\[\forall {\bf e}\in \hat{\mathcal{S}}_e^{\circ}\  \ 
\frac{1}{\sqrt{m}} \|\mathbf{e} - \overline{\mathbf{e}}_t^{\rm AO}\| \geq \frac{(\sqrt{c_e} - \sqrt{C_e}) \epsilon^{1/4}}{2},
\]
we deduce that:
\[
\forall \mathbf{e} \in \hat{\mathcal{S}}_e^{\circ} \ \text{such that } \|\mathbf{e}\| \leq K_e \sqrt{m}, \quad \mathcal{G}_{\lambda,\rho}(\mathbf{e}) \geq \psi(\tau^\star,\beta^\star)-(C_r+C_a)\epsilon^{\frac{3}{2}}+ M_e \epsilon^{3/2} (\frac{\sqrt{c_e} - \sqrt{C_e}}{2} - \sqrt{D_e})^2.
\]
Finally, by choosing \(c_e\) sufficiently large so that
\[
M_e (\frac{\sqrt{c_e} - \sqrt{C_e}}{2} - \sqrt{D_e})^2 \geq \gamma + C_\eta + s+C_r+C_a,
\]
we conclude the proof of the desired inequality \eqref{eq:desired_revision_final_step2}.
\subsubsection{Steps in detail} We provide here the necessary proofs for each step: the proofs of \eqref{eq:infinite_rev} and \eqref{eq:C_r_revised} for Step 1; the proof of \eqref{eq:step_2_revised} for Step 2; and the proofs of \eqref{eq:deviation_rev} and \eqref{eq:distance_er_et_rev} for Step 3. As explained above, by combining these results, we establish \eqref{eq:desired_revision_final}. }
\begin{enumerate}
\item  { \noindent{\underline{Step 1. Proof of \eqref{eq:infinite_rev} and \eqref{eq:C_r_revised}}  }
To show \eqref{eq:infinite_rev}, we write $\tilde{\mathcal{H}}_{\lambda,\rho}(\tilde{\bf x},{\bf e})$ as:
\begin{align*}
	\tilde{\mathcal{H}}_{\lambda,\rho}(\tilde{\bf x},{\bf e})=&\max_{{\bf u}} \frac{1}{n}{\bf g}^{T}{\bf u}\|\tilde{\bf x}-r{\bf S}_{\pi_{(t)}^c(|\tilde{\bf x}|)}\tilde{\bf x}\|-\frac{\|{\bf u}\|}{n}{\bf h}^{T}(\tilde{\bf x}-r{\bf S}_{\pi^c_{(t)}(|\tilde{\bf x}|)}\tilde{\bf x})+ \frac{r}{\sqrt{n}}{\bf u}^{T} {\bf e}-\frac{1}{\sqrt{n}}\sqrt{\rho}{\bf u}^{T}{\bf s}\\
	&+\sup_{\ell}|\ell|{\rm sign}(l)\left(\frac{1}{n}(\tilde{\eta}_r{\bf g}+\nu_r\tilde{\bf g})^{T}{\bf u}\|{\bf S}_{\pi_{(t)}^c(|\tilde{\bf x}|)}\tilde{\bf x}\|-\frac{1}{\sqrt{n}}{\bf u}^{T} {\bf e}\right)-|\ell|\frac{\|{\bf u}\|}{n}{\bf h}^{T}{\bf S}_{\pi_{(t)}^c(|\tilde{\bf x}|)}\tilde{\bf x}\nonumber\\
	&-\frac{\|{\bf u}\|^2}{4}+\frac{\lambda_1}{n}\|\tilde{\bf x}\|_1+\frac{\lambda_2}{n}\|\tilde{\bf x}\|^2.
\end{align*}
It is easy to see that if for $s\in\{1,-1\}$,  there exists ${\bf u}$ such that
$$
s \left(\frac{1}{n}(\tilde{\eta}_r{\bf g}+\nu_r\tilde{\bf g})^{T}{\bf u}\|{\bf S}_{\pi_{(t)}^c(|\tilde{\bf x}|)}\tilde{\bf x}\|-\frac{1}{\sqrt{n}}{\bf u}^{T} {\bf e}\right)- \frac{\|{\bf u}\|}{n}{\bf h}^{T}{\bf S}_{\pi_{(t)}^c(|\tilde{\bf x}|)}\tilde{\bf x}>0,
$$
then function $\tilde{\mathcal{H}}_{\lambda,\rho}(\tilde{\bf x},{\bf e})$ is infinite. Hence, necessarily the minimizer in ${\bf e}$ of function $\tilde{\mathcal{H}}_{\lambda,\rho}(\tilde{\bf x},{\bf e})$ belongs to the set:
$$
\{{\bf e}\in \mathbb{R}^m,  \forall {\bf u},  \ \frac{1}{n}(\tilde{\eta}_r{\bf g}+\nu_r\tilde{\bf g})^{T}{\bf u}\|{\bf S}_{\pi_{(t)}^c(|\tilde{\bf x}|)}\tilde{\bf x}\|-\frac{1}{\sqrt{n}}{\bf u}^{T}{\bf e}-\frac{\|{\bf u}\|}{n}{\bf h}^{T}{\bf S}_{\pi_{(t)}^c(|\tilde{\bf x}|)}\tilde{\bf x} \leq 0\}.
$$
Taking the maximum over ${\bf u}$, we can see that this set coincides with $\mathcal{S}_e(\tilde{\bf x})$. To continue, we optimize over ${\bf u}$ and $\ell$ to find:
\begin{equation}
\tilde{\mathcal{H}}_{\lambda,\rho}(\tilde{\bf x},{\bf e})=\Big(\|\frac{1}{n}{\bf g}\|\tilde{\bf x}-r{\bf S}_{\pi_{(t)}^c(|\tilde{\bf x}|)}\tilde{\bf x}\|+\frac{r}{\sqrt{n}}{\bf e}-\frac{\sqrt{\rho}{\bf s}}{\sqrt{n}}\|-\frac{1}{n}{\bf h}^{T}(\tilde{\bf x}-r{\bf S}_{\pi_{(t)}^{c}(|\tilde{\bf x}|)}\tilde{\bf x})\Big)_{+}^2+\frac{\lambda_1}{n}\|\tilde{\bf x}\|_1+\frac{\lambda_2}{n}\|\tilde{\bf x}\|^2. \label{eq:mathcal_H_revised}
\end{equation}
Next, to show \eqref{eq:C_r_revised}, define for $\tilde{\bf x}\in \mathbb{R}^n$:
\begin{align*}
{\bf c}_r(\tilde{\bf x})&=\sqrt{\rho}{\bf s}-\frac{1}{\sqrt{n}}{\bf g}\|\tilde{\bf x}-r{\bf S}_{\pi_{(t)}^c(|\tilde{\bf x}|)}\tilde{\bf x}\|,\\
{\bf d}_r(\tilde{\bf x})&=-\frac{1}{\sqrt{n}}(\tilde{\eta}_r{\bf g}+\nu_r\tilde{\bf g})\|{\bf S}_{\pi_{(t)}^c(|\tilde{\bf x}|)}\tilde{\bf x}\|.
\end{align*}
We can easily check that 
with probability $1-C\exp(-cn)$, for sufficiently small $\epsilon$,   
\begin{equation}
\forall \tilde{\bf x}\  \text{such that } \|\tilde{\bf x}\|_{\infty}\leq \sqrt{P} \ \ \frac{1}{\sqrt{n}}\|{\bf c}_r(\tilde{\bf x})-r{\bf d}_r(\tilde{\bf x})\|\geq \frac{r}{n}{\bf h}^{T}{\bf S}_{\pi_{(t)}^{c}(|\tilde{\bf x}|)}\tilde{\bf x}. \label{eq:s_revised}
\end{equation}
Using Lemma \ref{lem:KKT}, the minimum in  ${\bf e}\in \mathcal{S}_e(\tilde{\bf x})$ of $\tilde{\mathcal{H}}(\tilde{\bf x},{\bf e})$ is given by:
\begin{equation}
{\bf e}_r(\tilde{\bf x}):=-{\bf d}_r(\tilde{\bf x})+\frac{1}{\sqrt{n}} {\bf h}^{T}{\bf S}_{\pi_{(t)}^c(|\tilde{\bf x}|)}\tilde{\bf x}\frac{{\bf c}_r(\tilde{\bf x})-r{\bf d}_r(\tilde{\bf x})}{\|{\bf c}_r(\tilde{\bf x})-r{\bf d}_r(\tilde{\bf x})\|} .\label{eq:er_rev}
\end{equation}
Replacing ${\bf e}$ by ${\bf e}_r(\tilde{\bf x})$, we obtain:
$$
\min_{{\bf e}\in \mathbb{R}^m}\tilde{\mathcal{H}}_{\lambda,\rho}(\tilde{\bf x},{\bf e})=\Big(\frac{1}{\sqrt{n}}\|{\bf c}_r(\tilde{\bf x})-r{\bf d}_r(\tilde{\bf x}))\|+\frac{1}{n}{\bf h}^{T}\tilde{\bf x}\Big)_{+}^2+\frac{\lambda_1}{n}\|\tilde{\bf x}\|_1+\frac{\lambda_2}{n}\|\tilde{\bf x}\|^2.
$$
Considering the event $\mathcal{E}_t$
$$
\mathcal{E}_t=\{|\frac{\|{\bf g}\|^2}{n}-\delta|<\epsilon^{\frac{3}{2}}\}\cap \{|\frac{\|\tilde{\bf g}\|^2}{n}-\delta|<\epsilon^{\frac{3}{2}}\}\cap\{|\frac{{\bf g}^{T}{\bf s}}{n}|\leq \epsilon^3\}\cap\{|\frac{\tilde{\bf g}^{T}{\bf s}}{n}|\leq \epsilon^3\}
$$
which occurs with probability $1-C\exp(-cn\epsilon^3)$, we can approximate $\frac{1}{n}\|{\bf c}_r(\tilde{\bf x})-r{\bf d}_r(\tilde{\bf x}))\|^2$ as:
\begin{align}
\frac{1}{n}\|{\bf c}_r(\tilde{\bf x})-r{\bf d}_r(\tilde{\bf x}))\|^2&=\rho\delta +\frac{\delta}{n}\|\tilde{\bf x}-r{\bf S}_{\pi_{(t)}^c(|\tilde{\bf x}|)}\tilde{\bf x}\|^2+\frac{r^2\delta}{n}\|{\bf S}_{\pi_{(t)}^c(|\tilde{\bf x}|)}\tilde{\bf x}\|^2-\frac{2r}{n}\tilde{\eta}_r\delta\|\tilde{\bf x}-r{\bf S}_{\pi_{(t)}^c(|\tilde{\bf x}|)}\tilde{\bf x}\|\|{\bf S}_{\pi_{(t)}^{c}(|\tilde{\bf x}|)}\tilde{\bf x}\|+\varepsilon(\tilde{\bf x})\nonumber\\
&=\rho\delta+\frac{\delta}{n}\|\tilde{\bf x}\|^2+\varepsilon(\tilde{\bf x}) \label{eq:last_revised}
\end{align}
 where $\varepsilon(\tilde{\bf x})$ satisfies with probability $1-C\exp(-cn\epsilon^3)$
 $$
 \sup_{\|\tilde{\bf x}\|_{\infty}\leq \sqrt{P}} |\varepsilon(\tilde{\bf x})|\leq D\epsilon^{\frac{3}{2}}
 $$
 for some constant $D$. 
and in \eqref{eq:last_revised}, we used the fact that 
$$
\tilde{\eta}_r=\frac{(\tilde{\bf x}^{T}-r\tilde{\bf x}^{T}{\bf S}_{\pi_{(t)}^c(|\tilde{\bf x}|)}){\bf S}_{\pi_{(t)}(|\tilde{\bf x}|)}\tilde{\bf x}}{\|\tilde{\bf x}-r{\bf S}_{\pi_{(t)}^c(|\tilde{\bf x}|)}\tilde{\bf x}\|\|{\bf S}_{\pi_{(t)}^c(|\tilde{\bf x}|)}\tilde{\bf x}\|}.
$$
Since the function $x\mapsto (x)_{+}^2$ is Lipschitz over compact sets, there exists a constant $C_a$ such that with probability $1-C\exp(-cn\epsilon^3)$,
\begin{equation}
\forall \tilde{\bf x}  \ \text{such that} \ \|\tilde{\bf x}\|_{\infty}\leq \sqrt{P}, \ |(\sqrt{\rho\delta+\frac{\delta}{n}\|\tilde{\bf x}\|^2+\varepsilon(\tilde{\bf x})}-\frac{1}{n}{\bf h}^{T}\tilde{\bf x})_{+}^2-(\sqrt{\rho\delta+\frac{\delta}{n}\|\tilde{\bf x}\|^2}-\frac{1}{n}{\bf h}^{T}\tilde{\bf x})_{+}^2|\leq C_a\epsilon^{\frac{3}{2}} \label{eq:new_revised}
\end{equation}
and hence, 
\begin{equation}
\min_{{\bf e}} \tilde{\mathcal{H}}_{\lambda,\rho}(\tilde{\bf x},{\bf e})+C_a\epsilon^{\frac{3}
{2}}\geq \Big(\sqrt{\delta \rho+ \frac{\delta}{n}\|\tilde{\bf x}\|^2}-\frac{1}{n}{\bf h}^{T}\tilde{\bf x}\Big)_{+}^2 +\frac{\lambda_1}{n}\|\tilde{\bf x}\|_1+\frac{\lambda_2}{n}\|\tilde{\bf x}\|^2 .\label{eq:inequality_rev_H_lambda}
\end{equation}
The right-hand side term in the above inequality coincides with $\mathcal{L}^{\circ}(\tilde{\bf x})$ studied in \eqref{eq:L0}. Hence, there exists a constant $C_r$ such that with probability $1-\frac{C}{\epsilon^{\frac{3}{2}}}\exp(-cn\epsilon^3)$
$$
\min_{{\bf e}} \tilde{\mathcal{H}}_{\lambda,\rho}(\tilde{\bf x},{\bf e})\geq \psi(\tau^\star,\beta^\star)-(C_r+C_a)\epsilon^{\frac{3}{2}}.
$$}
\item {{\underline{Step 2: Proof of \eqref{eq:step_2_revised}}.} Using  Lemma \ref{eq:lem_ball}, establish that there exists a constant $C_B$ such that  with probability $1-\frac{C}{\epsilon^{\frac{3}{2}}}\exp(-cn\epsilon^3)$ the following holds:
$$
\frac{1}{n}\|\tilde{\bf x}-\overline{\bf x}^{\rm AO}\|^2\geq C_B\epsilon^{\frac{3}{2}} \Longrightarrow \Big(\sqrt{\delta + \frac{\delta}{n}\|\tilde{\bf x}\|^2}-\frac{1}{n}{\bf h}^{T}\tilde{\bf x}\Big)_{+}^2 +\frac{\lambda_1}{n}\|\tilde{\bf x}\|_1+\frac{\lambda_2}{n}\|\tilde{\bf x}\|^2\geq \psi(\tau^\star,\beta^\star)+(\gamma+C_\eta+s+C_a)\epsilon^{\frac{3}{2}}.
$$
Hence, \eqref{eq:step_2_revised} directly follows from \eqref{eq:inequality_rev_H_lambda}. }
\item {\underline{Step 3: Proof of \eqref{eq:deviation_rev} and \eqref{eq:distance_er_et_rev}}
Using Lemma \ref{lem:KKT}, for any ${\bf e}\in \mathcal{S}_e(\tilde{\bf x})$, we obtain:
\begin{align*}
\|\frac{1}{n}{\bf g}\|\tilde{\bf x}-r{\bf S}_{\pi_{(t)}^{c}(|\tilde{\bf x}|)}\tilde{\bf x}\|+\frac{r}{\sqrt{n}}{\bf e}-\frac{\sqrt{\rho}{\bf s}}{\sqrt{n}}\|&\geq \sqrt{(\frac{1}{\sqrt{n}}\|{\bf c}_r(\tilde{\bf x})-r{\bf d}_r(\tilde{\bf x})\|-\frac{r}{n}{\bf h}^{T}{\bf S}_{\pi_{(t)}^c(|\tilde{\bf x}|)}\tilde{\bf x})^2+\frac{r^2}{n}\|{\bf e}-{\bf e}_r(\tilde{\bf x})\|^2}\\
&=m_r(\tilde{\bf x})+\frac{\frac{r^2}{n}\|{\bf e}-{\bf e}_r(\tilde{\bf x})\|^2}{m_r(\tilde{\bf x})(1+\sqrt{1+\frac{r^2\|{\bf e}-{\bf e}_r(\tilde{\bf x})\|^2}{m_r(\tilde{\bf x})}})}
\end{align*}
where 
$$
m_r(\tilde{\bf x}):=|\frac{1}{\sqrt{n}}\|{\bf c}_r(\tilde{\bf x})-r{\bf d}_r(\tilde{\bf x})\|-\frac{r}{n}{\bf h}^{T}{\bf S}_{\pi_{(t)}^c(|\tilde{\bf x}|)}\tilde{\bf x}|.
$$
It follows from \eqref{eq:s_revised}, that with probability $1-C\exp(-cn)$, 
$$
m_r(\tilde{\bf x})=\frac{1}{\sqrt{n}}\|{\bf c}_r(\tilde{\bf x})-r{\bf d}_r(\tilde{\bf x})\|-\frac{r}{n}{\bf h}^{T}{\bf S}_{\pi_{(t)}^c(|\tilde{\bf x}|)}\tilde{\bf x}.
$$
Starting from \eqref{eq:mathcal_H_revised}, and using \eqref{eq:last_revised}, we conclude that for any $\tilde{\bf x}\in \tilde{\mathcal{S}}_x^{\circ}\cap \mathcal{B}_{\epsilon}(\overline{\bf x}^{\rm AO})$, we obtain:
$$
\tilde{\mathcal{H}}_{\lambda,\rho}(\tilde{\bf x},{\bf e})\geq \big(\sqrt{\rho\delta+\frac{\delta}{n}\|\tilde{\bf x}\|^2+\varepsilon(\tilde{\bf x})}-\frac{1}{n}{\bf h}^{T}\tilde{\bf x} + \frac{\frac{r^2}{n}\|{\bf e}-{\bf e}_r(\tilde{\bf x})\|^2}{m_r(\tilde{\bf x})(1+\sqrt{1+\frac{r^2\|{\bf e}-{\bf e}_r(\tilde{\bf x})\|^2}{n m_r(\tilde{\bf x})}})}\big)_{+}^2+\frac{\lambda_1}{n}\|\tilde{\bf x}\|_1+\frac{\lambda_2}{n}\|\tilde{\bf x}\|^2
$$
where $\sup_{\tilde{\bf x}\in \mathcal{B}_{\epsilon}(\overline{\bf x}^{\rm AO})}|\varepsilon(\tilde{\bf x})|\leq D \epsilon^{\frac{3}{2}}$. To continue, we use the concentration inequalities in \eqref{eq:concen1} and \eqref{eq:concen2} to see that $\sqrt{\rho\delta+\frac{\delta}{n}\|\overline{\bf x}^{\rm AO}\|^2}-\frac{1}{n}{\bf h}^{T}\overline{\bf x}^{\rm AO}$ concentrates around $\frac{\beta^{\star}}{2}$, that is there exists constants $C$ and $c$ such that for all $t>0$
$$
\mathbb{P}\Big[|\sqrt{\rho\delta+\frac{\delta}{n}\|\overline{\bf x}^{\rm AO}\|^2}-\frac{1}{n}{\bf h}^{T}\overline{\bf x}^{\rm AO}-\frac{\beta^\star}{2}|\geq t\Big]\leq C\exp(-cnt^2).
$$
Based on this and the fact that $|\varepsilon(\tilde{\bf x})|\leq D\epsilon^{\frac{3}{2}}$, we deduce that with probability $1-C\exp(-cn\epsilon^{\frac{3}{2}})$
$$
\forall \tilde{\bf x}\in \mathcal{B}_{\epsilon}(\overline{\bf x}^{\rm AO}),  \ \  \sqrt{\rho\delta+\frac{\delta}{n}\|\tilde{\bf x}\|^2+\varepsilon(\tilde{\bf x})}-\frac{1}{n}{\bf h}^{T}\tilde{\bf x}\geq \frac{\beta^\star}{4}.
$$
Consequently, we obtain:
\begin{align*}
\forall \tilde{\bf x}\in \mathcal{B}_{\epsilon}(\overline{\bf x}^{\rm AO}), \ \ \ \tilde{\mathcal{H}}_{\lambda,\rho}(\tilde{\bf x},{\bf e})&\geq \big(\sqrt{\rho\delta+\frac{\delta}{n}\|\tilde{\bf x}\|^2+\varepsilon(\tilde{\bf x})}-\frac{1}{n}{\bf h}^{T}\tilde{\bf x}\big)_{+}^2+\frac{\lambda_1}{n}\|\tilde{\bf x}\|_1+\frac{\lambda_2}{n}\|\tilde{\bf x}\|^2+\frac{\beta^\star\frac{r^2}{n}\|{\bf e}-{\bf e}_r(\tilde{\bf x})\|^2}{4m_r(\tilde{\bf x})(1+\sqrt{1+\frac{r^2\|{\bf e}-{\bf e}_r(\tilde{\bf x})\|^2}{n(m_r(\tilde{\bf x}))^2}})}.
\end{align*}
Using \eqref{eq:new_revised}, we obtain:
\begin{align*}
\mathcal{G}_{\lambda,\rho}({\bf e})&\geq \psi(\tau^\star,\beta^\star)+\frac{\beta^\star\frac{r^2}{n}\|{\bf e}-{\bf e}_r(\tilde{\bf x})\|^2}{4m_r(\tilde{\bf x})(1+\sqrt{1+\frac{r^2\|{\bf e}-{\bf e}_r(\tilde{\bf x})\|^2}{n(m_r(\tilde{\bf x}))^2}})}-(C_r+C_a)\epsilon^{\frac{3}{2}}
\end{align*}
where the last inequality follows from \eqref{eq:new_revised} and \eqref{eq:inequality_rev_H_lambda}. To finish the proof of \eqref{eq:deviation_rev}, we use the fact that with probability $1-C\exp(-cn)$, $\sup_{\|\tilde{\bf x}\|_{\infty}\leq \sqrt{P}} |m_r(\tilde{\bf x})|\leq 2\sqrt{\rho\delta+\delta P}$, to obtain for sufficiently small $\epsilon$:
$$
\forall {\bf e}\ \text{with }\|{\bf e}\| \leq K_e\sqrt{m} \ \ 
\mathcal{G}_{\lambda,\rho}({\bf e})\geq \psi(\tau^\star,\beta^\star)+\frac{\beta^\star\frac{r^2}{n}\|{\bf e}-{\bf e}_r(\tilde{\bf x})\|^2}{16\sqrt{\rho\delta+\delta P}}-(C_r+C_a)\epsilon^{\frac{3}{2}}
$$
which shows \eqref{eq:deviation_rev}.
Next, it remains to prove \eqref{eq:distance_er_et_rev}. For that, we need the following Lemma. }
 {\begin{lemma}
There exists a constant $C_x$ such that for $\epsilon$ sufficiently small the following holds true:
$$
\tilde{\bf x}\in \mathcal{B}_{\epsilon}(\overline{\bf x}^{\rm AO})\cap\tilde{\mathcal{S}}_{x2}^{\circ} \Longrightarrow \frac{1}{n}\|{\bf S}_{\pi_{(t)}^c(|\tilde{\bf x}|)}\tilde{\bf x}-{\bf S}_{\pi_{(t)}^c(|\overline{\bf x}^{\rm AO}|)}\overline{\bf x}^{\rm AO}\|^2\leq C_x\sqrt{\epsilon} \ \text{and}\frac{1}{n}\|{\bf S}_{\pi_{(t)}(|\tilde{\bf x}|)}\tilde{\bf x}-{\bf S}_{\pi_{(t)}(|\overline{\bf x}^{\rm AO}|)}\overline{\bf x}^{\rm AO}\|^2\leq C_x\sqrt{\epsilon}.
$$
\label{lem:approx_revised}
\end{lemma}
\begin{proof}
To begin with, we use the relation:
$$
\frac{1}{n}\|\tilde{\bf x}-\overline{\bf x}^{\rm AO}\|^2\geq \frac{1}{n}\sum_{i=1}^n |[\tilde{\bf x}]_i-[\overline{\bf x}^{\rm AO}]_i|^2{\bf 1}_{\{|[\overline{\bf x}^{\rm AO}]_i|\geq t\}}{\bf 1}_{\{|[\tilde{\bf x}]_i|\leq t-\sqrt{\epsilon\}}}\geq \epsilon\frac{1}{n}\sum_{i=1}^n{\bf 1}_{\{|[\overline{\bf x}^{\rm AO}]_i|\geq t\}}{\bf 1}_{\{|[\tilde{\bf x}]_i|\leq t-\sqrt{\epsilon}\}}.
$$
As such, we obtain:  
\begin{equation}
\forall \tilde{\bf x}\in \mathcal{B}_{\epsilon}(\overline{\bf x}^{\rm AO})\cap\tilde{\mathcal{S}}_{x2}^{\circ}, \ \ \ \ \frac{1}{n}\sum_{i=1}^n{\bf 1}_{\{|[\overline{\bf x}^{\rm AO}]_i|\geq t\}}{\bf 1}_{\{|[\tilde{\bf x}]_i|\leq t-\sqrt{\epsilon}\}}\leq C_B\sqrt{\epsilon}. \label{eq:fund_1_revised}
\end{equation}
Similarly, by noting that:
$$
\frac{1}{n}\|\tilde{\bf x}-\overline{\bf x}^{\rm AO}\|^2\geq \frac{1}{n}\sum_{i=1}^n |[\tilde{\bf x}]_i-[\overline{\bf x}^{\rm AO}]_i|^2{\bf 1}_{\{|[\overline{\bf x}^{\rm AO}]_i|\leq t\}}{\bf 1}_{\{|[\tilde{\bf x}]_i|\geq t+\sqrt{\epsilon}\}}\geq \epsilon\frac{1}{n}\sum_{i=1}^n{\bf 1}_{\{|[\overline{\bf x}^{\rm AO}]_i|\leq t\}}{\bf 1}_{\{|[\tilde{\bf x}]_i|\geq t+\sqrt{\epsilon}\}},
$$
we prove 
\begin{equation}
\forall \tilde{\bf x}\in \mathcal{B}_{\epsilon}(\overline{\bf x}^{\rm AO})\cap\tilde{\mathcal{S}}_{x2}^{\circ}, \ \ \ \frac{1}{n}\sum_{i=1}^n{\bf 1}_{\{|[\overline{\bf x}^{\rm AO}]_i|\leq t\}}{\bf 1}_{\{|[\tilde{\bf x}]_i|\geq t+\sqrt{\epsilon}\}}\leq C_B\sqrt{\epsilon}. \label{eq:fund_2_revised}
\end{equation}
To continue, we use the fact that:
\begin{align*}
\frac{1}{n}\|\tilde{\bf x}-\overline{\bf x}^{\rm AO}\|^2\geq &\frac{1}{n}\sum_{i=1}^n |[\tilde{\bf x}]_i-[\overline{\bf x}^{\rm AO}]_i|^2\geq \frac{1}{n}\sum_{i=1}^n |[\tilde{\bf x}]_i-[\overline{\bf x}^{\rm AO}]_i|^2{\bf 1}_{\{|[\overline{\bf x}^{\rm AO}]_i|\geq t\}}{\bf 1}_{\{|[\tilde{\bf x}]_i|\geq t+\sqrt{\epsilon}\}}\\
&=\frac{1}{n}\sum_{i=1}^n (|[\tilde{\bf x}]_i|^2+|[\overline{\bf x}^{\rm AO}]_i|^2){\bf 1}_{\{|[\overline{\bf x}^{\rm AO}]_i|\geq t\}}{\bf 1}_{\{|[\tilde{\bf x}]_i|\geq t+\sqrt{\epsilon}\}}-\frac{2}{n}\sum_{i=1}^n [\tilde{\bf x}]_i[\overline{\bf x}^{\rm AO}]_i{\bf 1}_{\{|[\overline{\bf x}^{\rm AO}]_i|\geq t\}}{\bf 1}_{\{|[\tilde{\bf x}]_i|\geq t+\sqrt{\epsilon}\}}
\end{align*}
and hence, for all $\tilde{\bf x}\in \mathcal{B}_\epsilon(\overline{\bf x}^{\rm AO})$:
\begin{equation}
-\frac{2}{n}\sum_{i=1}^n [\tilde{\bf x}]_i[\overline{\bf x}^{\rm AO}]_i{\bf 1}_{\{|[\overline{\bf x}^{\rm AO}]_i|\geq t\}}{\bf 1}_{\{|[\tilde{\bf x}]_i|\geq t+\sqrt{\epsilon}\}}\leq C_B\epsilon^{\frac{3}{2}}-\frac{1}{n}\sum_{i=1}^n (|[\tilde{\bf x}]_i|^2+|[\overline{\bf x}^{\rm AO}]_i|^2){\bf 1}_{\{|[\overline{\bf x}^{\rm AO}]_i|\geq t\}}{\bf 1}_{\{|[\tilde{\bf x}]_i|\geq t+\sqrt{\epsilon}\}}. \label{eq:revised}
\end{equation}
Using this, and by expanding the expression   $\frac{1}{n}\|{\bf S}_{\pi_{(t)}^c(|\tilde{\bf x}|)}\tilde{\bf x}-{\bf S}_{\pi_{(t)}^c(|\overline{\bf x}^{\rm AO}|)}\overline{\bf x}^{\rm AO}\|^2$ we obtain:
\begin{align*}
\frac{1}{n}\|{\bf S}_{\pi_{(t)}^c(|\tilde{\bf x}|)}\tilde{\bf x}-{\bf S}_{\pi_{(t)}^c(|\overline{\bf x}^{\rm AO}|)}\overline{\bf x}^{\rm AO}\|^2&=\frac{1}{n}\sum_{i=1}^n |[\tilde{\bf x}]_i|^2{\bf 1}_{\{|[\tilde{\bf x}]_i|\}\geq t+\sqrt{\epsilon}\}}+\frac{1}{n}\sum_{i=1}^n |\overline{\bf x}^{\rm AO}]_i|^2{\bf 1}_{\{|[\overline{\bf x}^{\rm AO}]_i|\}\geq t\}} \\
&-\frac{2}{n}\sum_{i=1}^n[\tilde{\bf x}]_i[\overline{\bf x}^{\rm AO}]_i{\bf 1}_{\{|[\overline{\bf x}^{\rm AO}]_i|\geq t\}}{\bf 1}_{\{|[\tilde{\bf x}]_i|\geq t+\sqrt{\epsilon}\}}\\
&\leq C_B\epsilon^{\frac{3}{2}}+\frac{1}{n}\sum_{i=1}^n |[\tilde{\bf x}]_i|^2{\bf 1}_{\{|[\tilde{\bf x}]_i|\geq t+\sqrt{\epsilon}\}}{\bf 1}_{\{|[\overline{\bf x}^{\rm AO}]_i|\leq t\}}  \\
&+\frac{1}{n}\sum_{i=1}^n |[\overline{\bf x}^{\rm AO}]_i|^2{\bf 1}_{\{|[\overline{\bf x}^{\rm AO}]_i|\geq t\}}{\bf 1}_{\{|[\tilde{\bf x}]_i|\leq t-\sqrt{\epsilon}\}}
\end{align*}
where in the last inequality we used \eqref{eq:revised} along with the fact that ${\tilde{\bf x}}\in \tilde{\mathcal{S}}_{x2}^{\circ}$ and hence there is no elements of $\tilde{\bf x}$ with magnitude in $(t-\sqrt{\epsilon},t+\sqrt{\epsilon})$. Consequently, using \eqref{eq:fund_1_revised} and \eqref{eq:fund_2_revised}, we obtain:
$$
\frac{1}{n}\|{\bf S}_{\pi_{(t)}^c(|\tilde{\bf x}|)}\tilde{\bf x}-{\bf S}_{\pi_{(t)}^c(|\overline{\bf x}^{\rm AO}|)}\overline{\bf x}^{\rm AO}\|^2\leq (C_B+2PC_B)\sqrt{\epsilon}.
$$
To prove the second inequality, we use the fact ${\bf S}_{\pi_{(t)}(|\tilde{\bf x}|)}\tilde{\bf x}=\tilde{\bf x}-{\bf S}_{\pi_{(t)}^c(|\tilde{\bf x}|)}\tilde{\bf x}$ and ${\bf S}_{\pi_{(t)}(|\overline{\bf x}^{\rm AO}|)}\overline{\bf x}^{\rm AO}=\overline{\bf x}^{\rm AO}-{\bf S}_{\pi_{(t)}^c(|\overline{\bf x}^{\rm AO}|)}\overline{\bf x}^{\rm AO}$ and exploit the triangular inequality. 
\end{proof}
With Lemma \ref{lem:approx_revised}, we are now in position to show \eqref{eq:distance_er_et_rev}. For that, we define 
$$
{\bf c}_0=\sqrt{\rho}{\bf s}-{\bf g}\sqrt{(\tau^\star)^2\delta-\rho} , \ \ {\bf d}_0=-\tilde{\alpha}^\star(\tilde{\eta}^\star {\bf g}+\nu^\star\tilde{\bf g}).
$$
We can easily see that  $\overline{\bf e}_t^{\rm AO}=-{\bf d}_0- \theta^\star{\bf c}_0$.
Using Lemma \ref{lem:approx_revised}, we conclude that there exists $C_\alpha$ and  and $C_\theta$ such that with probability $1-C\exp(-cn\epsilon)$, 
$$
\sup_{\tilde{\bf x}\in \mathcal{B}_{\epsilon}(\overline{\bf x}^{\rm AO})}|\frac{1}{\sqrt{n}}\|{\bf S}_{\pi_{(t)}^c(|\tilde{\bf x}|)}\tilde{\bf x}\|-\tilde{\alpha}^\star|\leq C_\alpha \sqrt{\epsilon} \ \text{and }\sup_{\tilde{\bf x}\in \mathcal{B}_{\epsilon}(\overline{\bf x}^{\rm AO})}|\frac{1}{n}{\bf h}^{T}{\bf S}_{\pi_{(t)}^{c}(|\tilde{\bf x}|)}\tilde{\bf x}+\theta^\star\tau^\star\delta|\leq C_\alpha \sqrt{\epsilon}.
$$
Moreover, since $\eta\in \mathcal{I}_\eta,$  there exists constants $C_\eta $ and $C_\nu $ such that  $|\tilde{\eta}-\tilde{\eta}_r|\leq C_\eta\sqrt{\epsilon}$ and $|\nu_r-\nu^\star|\leq C_\nu\sqrt{\epsilon}$ 
Hence, by combining all these results we obtain \eqref{eq:distance_er_et_rev}. }
\end{enumerate}
\subsection{Proof of some technical results.}
\subsubsection{Proof of Theorem \ref{th:new_cgmt} }
\label{isolate_proof}
The proof  {
involves two main  steps: First, we establish the Gordon inequality when the sets are discrete. Next, we employ a compactness argument to extend the result to continuous compact sets, and we also generalize it to the entire real axis for the variable $\gamma$. After the two steps, Theorem \ref{th:new_cgmt} is easily obtained from a conditional probability.}

\noindent {\underline{Step 1: A Gordon inequality for random processes indexed on discrete sets.} }
\begin{lemma}
\label{th:discrete_cgmt_without_constraint_s}
For $m,n\in\mathbb{N}^\star$, let $I_b$ and $I_u$ be finite sets of vectors in $\mathbb{R}^{m}$, $I_x$ a finite set of vectors in $\mathbb{R}^n$, $I_\gamma$ a finite set of reals. Let  $z,\tilde{z}\in\mathbb{R}$, ${\bf G}\in\mathbb{R}^{m\times n}$, ${\bf g},\tilde{\bf g}\in\mathbb{R}^m$ and $ {\bf h}\in\mathbb{R}^n$ have independent and identically distributed entries following $\mathcal{N}(0,1)$ and are independent of each other. 
We assume that all ${\bf x}$ in $I_x$ satisfies \eqref{req1}, \eqref{req2} and \eqref{req3} for some $k\in\{1,\cdots,n\}$, $r\in(0,1)$ and $\eta\in\mathbb{R}_{+}$. Consider the following two random processes:
\begin{align*}
	\tilde{X}_{t}({\bf x},{\bf u},\gamma)=&{\bf u}^{T}{\bf G}({\bf x}-r{\bf S}_{\pi_{(t)}^c(|{\bf x}|)}{\bf x})+\gamma{\bf u}^{T}{\bf G}{\bf S}_{\pi_{(t)}^c(|{\bf x}|)}{\bf x}+z\|{\bf x}-r{\bf S}_{\pi_{(t)}^c(|{\bf x}|)}{\bf x}\|\|{\bf u}\|\nonumber \\
		&+|\gamma|(\tilde{\eta}z+\nu \tilde{z})\|{\bf u}\|\|{\bf S}_{\pi_{(t)}^c(|{\bf x}|)}{\bf x}\|,\\
 		\tilde{Y}_t({\bf x},{\bf u},{\gamma})=&{\bf g}^{T}{\bf u}\|{\bf x}-r{\bf S}_{\pi_{(t)}^c(|{\bf x}|)}{\bf x}\|-\|{\bf u}\|{\bf h}^T({\bf x}-r{\bf S}_{\pi_{(t)}^c(|{\bf x}|)}{\bf x})+\gamma {\bf u}^{T}(\tilde{\eta}{\bf g}+\nu \tilde{\bf g})\|{\bf S}_{\pi_{(t)}^c(|{\bf x}|)}{\bf x}\|\nonumber \\
		&-|\gamma|{\bf h}^{T}{\bf S}_{\pi_{(t)}^{c}(|{\bf x}|)}{\bf x}\|{\bf u}\| .
	\end{align*}
	Let $\psi $ be a finite function defined in $I_x\times I_b \times I_u\times I_\gamma$. Then, for any $c\in \mathbb{R}$, 
	$$	\mathbb{P}\Big[\min_{\substack{{\bf b}\in I_b\\ {\bf x}\in I_{x}}} \max_{\substack{{\bf u}\in I_u\\ \gamma\in I_{\gamma}}} \tilde{X}_{t}({\bf x},{\bf u},\gamma) +\psi({\bf x},{\bf b},{\bf u},\gamma)\leq c\Big]\leq \mathbb{P}\Big[\min_{\substack{{\bf b}\in I_b\\ {\bf x}\in I_{x}}} \max_{\substack{{\bf u}\in I_u\\ {\gamma}\in I_{\gamma}}} \tilde{Y}_{t}({\bf x},{\bf u},\gamma) +\psi({\bf x},{\bf b},{\bf u},\gamma)\leq c\Big].
	$$
	\begin{proof}
	Note that for all ${\bf x}\in I_x$,
		$$
		\frac{({\bf x}^{T}-r{\bf x}^{T}{\bf S}_{\pi_{(t)}^c(|{\bf x}|)}){\bf S}_{\pi_{(t)}^c(|{\bf x}|)}{\bf x}}{\|{\bf x}-r{\bf S}_{\pi_{(t)}^c(|{\bf x}|)}{\bf x}\|\|{\bf S}_{\pi_{(t)}^c(|{\bf x}|)}{\bf x}\|}=\tilde{\eta}.
	$$
	We can show that both processes satisfy the conditions of Gordon's min-max inequality as stated in Theorem \ref{lem:gor}. First, it is easy to see that they are centered. Furthermore, for ${\bf x},{\bf x}^{'}\in I_x$, ${\bf b}, {\bf b}^{'}\in I_b$, ${\bf u},{\bf u}^{'}\in I_u$ and $\gamma,\gamma^{'}\in I_\gamma$, 
 		\begin{align}
 			&\mathbb{E}[\tilde{X}_{t}({\bf x},{\bf u},{\gamma})\tilde{X}_{t}({\bf x}^{'},{\bf u}^{'},{\gamma}^{'})]\nonumber\\=&{\bf u}^{T}{\bf u}^{'} ({\bf x}^{T}-r{\bf x}^{T}{\bf S}_{\pi_{(t)}^c(|{\bf x}|)})({\bf x}^{'}-r{\bf S}_{\pi_{(t)}^c(|{\bf x}^{'}|)}{\bf x}^{'}) 
 			+\|{\bf u}\|\|{\bf u}^{'}\| \|{\bf x}-r{\bf S}_{\pi_{(t)}^c(|{\bf x}|)}{\bf x}\|\|{\bf x}^{'}-r{\bf S}_{\pi_{(t)}^c(|{\bf x}^{'}|)}{\bf x}^{'}\|\nonumber\\
 			&+\gamma\gamma^{'} {\bf u}^{T}{\bf u}^{'}{\bf x}^{T}{\bf S}_{\pi_{(t)}^{c}(|{\bf x}|)}{\bf S}_{\pi_{(t)}^c(|{\bf x}^{'}|)}{\bf x}^{'}+|\gamma \gamma^{'}| \|{\bf u}\|\|{\bf u}^{'}\|\|{\bf S}_{\pi_{(t)}^c(|{\bf x}|)}{\bf x}\|\|{\bf S}_{\pi_{(t)}^c(|{\bf x}^{'}|)}{\bf x}^{'}\|\nonumber \\
 			&+\gamma^{'}{\bf u}^{T}{\bf u}^{'}({\bf x}^{T}-r{\bf x}^{T}{\bf S}_{\pi_{(t)}^c(|{\bf x}|)}){\bf S}_{\pi_{(t)}^c(|{\bf x}^{'}|)}{\bf x}^{'}+\tilde{\eta}|\gamma^{'}|\|{\bf u}\|\|{\bf u}^{'}\|\|{\bf x}-r{\bf S}_{\pi_{(t)}^c(|{\bf x}|)}{\bf x}\|\|{\bf S}_{\pi_{(t)}^c(|{\bf x}^{'}|)}{\bf x}^{'}\|\nonumber \\
 			&+\gamma {\bf u}^{T}{\bf u}^{'}({\bf x}^{'}-r{\bf S}_{\pi_{(t)}^c(|{\bf x}^{'}|)}{\bf x}^{'})^{T}{\bf S}_{\pi_{(t)}^c(|{\bf x}|)}{\bf x}+\tilde{\eta}|\gamma|\|{\bf u}\|\|{\bf u}^{'}\|\|{\bf x}^{'}-r{\bf S}_{\pi_{(t)}^c(|{\bf x}^{'}|)}{\bf x}^{'}\|\|{\bf S}_{\pi_{(t)}^c(|{\bf x}|)}{\bf x}\|, \nonumber
 		\end{align}
 		\begin{align}
			&\mathbb{E}[\tilde{Y}_{t}({\bf x},{\bf u},{\gamma})\tilde{Y}_{t}({\bf x}^{'},{\bf u}^{'},{\gamma}^{'})]\nonumber\\=&{\bf u}^{T}{\bf u}^{'}\|{\bf x}-r{\bf S}_{\pi_{(t)}^c(|{\bf x}|)}{\bf x}\|\|{\bf x}^{'}-r{\bf S}_{\pi_{(t)}^c(|{\bf x}^{'}|)}{\bf x}^{'}\|+
 			\gamma \gamma^{'}{\bf u}^{T}{\bf u}^{'}\|{\bf S}_{\pi_{(t)}^{c}(|{\bf x}|)}{\bf x}\|\|{\bf S}_{\pi_{(t)}^{c}(|{\bf x}^{'}|)}{\bf x}^{'}\|\nonumber\\
 			&+\gamma {\bf u}^{T}{\bf u}^{'}\tilde{\eta}\|{\bf S}_{\pi_{(t)}^{c}(|{\bf x}|)}{\bf x}\|\|{\bf x}^{'}-r{\bf S}_{\pi_{(t)}^{c}(|{\bf x}^{'}|)}{\bf x}^{'}\|+\gamma^{'} {\bf u}^{T}{\bf u}^{'}\|{\bf S}_{\pi_{(t)}^{c}(|{\bf x}^{'}|)}{\bf x}^{'}\|\|{\bf x}-r{\bf S}_{\pi_{(t)}^{c}(|{\bf x}|)}{\bf x}\|\tilde{\eta}\nonumber\\
 			&+({\bf x}-r{\bf x}^{T}{\bf S}_{\pi_{(t)}^c(|{\bf x}|)})({\bf x}^{'}-r{\bf S}_{\pi_{(t)}^c(|{\bf x}^{'}|)}{\bf x}^{'})\|{\bf u}\|\|{\bf u}^{'}\|
 			+|\gamma||\gamma^{'}|{\bf x}^{T}{\bf S}_{\pi_{(t)}^{c}(|{\bf x}|)}{\bf S}_{\pi_{(t)}^{c}(|{\bf x}^{'}|)}{\bf x}^{'}\|{\bf u}\|\|{\bf u}^{'}\|\nonumber \\
 			&+\|{\bf u}\|\|{\bf u}^{'}\||\gamma|{\bf x}^{T}{\bf S}_{\pi_{(t)}^{c}(|{\bf x}|)}({\bf x}^{'}-r{\bf S}_{\pi_{(t)}^{c}(|{\bf x}^{'}|)}{\bf x}^{'})+\|{\bf u}\|\|{\bf u}^{'}\||\gamma^{'}|({\bf x}^{'})^{T}{\bf S}_{\pi_{(t)}^{c}(|{\bf x}^{'}|)}({\bf x}-r{\bf S}_{\pi_{(t)}^{c}(|{\bf x}|)}{\bf x}).\nonumber
 		\end{align}
		We thus have:
		\begin{align}
			&		\mathbb{E}[\tilde{X}_{t}({\bf x},{\bf u},{\gamma})\tilde{X}_{t}({\bf x}^{'},{\bf u}^{'},{\gamma}^{'})]-\mathbb{E}[\tilde{Y}_{t}({\bf x},{\bf u},{\gamma})\tilde{Y}_{t}({\bf x}^{'},{\bf u}^{'},{\gamma}^{'})]\nonumber\\=
 			&{\bf u}^{T}{\bf u}^{'}\Big(({\bf x}^{T}-r{\bf x}^{T}{\bf S}_{\pi_{(t)}^c(|{\bf x}|)})({\bf x}^{'}-r{\bf S}_{\pi_{(t)}^c(|{\bf x}^{'}|)}{\bf x}^{'})-\|{\bf x}-r{\bf S}_{\pi_{(t)}^c(|{\bf x}|)}{\bf x}\|\|{\bf x}^{'}-r{\bf S}_{\pi_{(t)}^c(|{\bf x}^{'}|)}{\bf x}^{'}\|\Big)\nonumber \\
 			&+\|{\bf u}\|\|{\bf u}^{'}\|(\|{\bf x}-r{\bf S}_{\pi_{(t)}^c(|{\bf x}|)}{\bf x}\|\|{\bf x}^{'}-r{\bf S}_{\pi_{(t)}^c(|{\bf x}^{'}|)}{\bf x}^{'}\|-({\bf x}^{T}-r {\bf x}^{T}{\bf S}_{\pi_{(t)}^c(|{\bf x}|)})({\bf x}^{'}-r{\bf S}_{\pi_{(t)}^c(|{\bf x}^{'}|)}{\bf x}^{'}))\nonumber \\
			&+(|\gamma\gamma^{'}|\|{\bf u}\|\|{\bf u}^{'}\|-\gamma\gamma^{'}{\bf u}^{T}{\bf u}^{'})(\|{\bf S}_{\pi_{(t)}^{c}(|{\bf x}|)}{\bf x}\|\|{\bf S}_{\pi_{(t)}^{c}(|{\bf x}^{'}|)}{\bf x}^{'}\|-{\bf x}^T{\bf S}_{\pi_{(t)}^{c}(|{\bf x}|)}{\bf S}_{\pi_{(t)}^{c}(|{\bf x}^{'}|)}{\bf x}^{'})\nonumber\\
 			&+ {\bf u}^{T}{\bf u}^{'}(\gamma({\bf x}^{'}-r{\bf S}_{\pi_{(t)}^c(|{\bf x}^{'}|)}{\bf x}^{'})^{T}{\bf S}_{\pi_{(t)}^c(|{\bf x}|)}{\bf x}-\gamma\tilde{\eta}\|{\bf S}_{\pi_{(t)}^{c}(|{\bf x}|)}{\bf x}\|\|{\bf x}^{'}-r{\bf S}_{\pi_{(t)}^c(|{\bf x}^{'}|)}{\bf x}^{'}\|)\nonumber \\
			&+ {\bf u}^{T}{\bf u}^{'}(\gamma^{'}({\bf x}-r{\bf S}_{\pi_{(t)}^c(|{\bf x}|)}{\bf x})^{T}{\bf S}_{\pi_{(t)}^c(|{\bf x}^{'}|)}{\bf x}^{'}-\tilde{\eta}\gamma^{'}\|{\bf S}_{\pi_{(t)}^{c}(|{\bf x}^{'}|)}{\bf x}^{'}\|\|{\bf x}-r{\bf S}_{\pi_{(t)}^c(|{\bf x}|)}{\bf x}\|)\nonumber \\
 			&+\|{\bf u}\|\|{\bf u}^{'}\|(\tilde{\eta}|\gamma^{'}|\|{\bf S}_{\pi_{(t)}^{c}(|{\bf x}^{'}|)}{\bf x}^{'}\|\|{\bf x}-r{\bf S}_{\pi_{(t)}^c(|{\bf x}|)}{\bf x}\|-|\gamma^{'}|({\bf x}^{'})^{T}{\bf S}_{\pi_{(t)}^c(|{\bf x}^{'}|)}({\bf x}-r{\bf S}_{\pi_{(t)}^c(|{\bf x}|)}{\bf x}))\nonumber \\
 			&+\|{\bf u}\|\|{\bf u}^{'}\|(\tilde{\eta}|\gamma|\|{\bf S}_{\pi_{(t)}^{c}(|{\bf x}|)}{\bf x}\|\|{\bf x}^{'}-r{\bf S}_{\pi_{(t)}^c(|{\bf x}^{'}|)}{\bf x}^{'}\|-|\gamma|{\bf x}^{T}{\bf S}_{\pi_{(t)}^c(|{\bf x}|)}({\bf x}^{'}-r{\bf S}_{\pi_{(t)}^c(|{\bf x}^{'}|)}{\bf x}^{'}))\nonumber\\=
			&(\|{\bf u}\|\|{\bf u}^{'}\|-{\bf u}^{T}{\bf u}^{'})(\|{\bf x}-r{\bf S}_{\pi_{(t)(|{\bf x}|)}}{\bf x}\|\|{\bf x}^{'}-r{\bf S}_{\pi_{(t)}(|{\bf x}^{'}|)}{\bf x}^{'}\|-({\bf x}^{T}-r {\bf x}^{T}{\bf S}_{\pi_{(t)}^c(|{\bf x}|)})({\bf x}-r{\bf S}_{\pi_{(t)}^c(|{\bf x}^{'}|)}{\bf x}^{'}))\nonumber\\
			&+(|\gamma\gamma^{'}|\|{\bf u}\|\|{\bf u}^{'}\|-\gamma\gamma^{'}{\bf u}^{T}{\bf u}^{'})(\|{\bf S}_{\pi_{(t)}^{c}(|{\bf x}|)}{\bf x}\|\|{\bf S}_{\pi_{(t)}^{c}(|{\bf x}^{'}|)}{\bf x}^{'}\|-{\bf x}^T{\bf S}_{\pi_{(t)}^{c}(|{\bf x}|)}{\bf S}_{\pi_{(t)}^{c}(|{\bf x}^{'}|)}{\bf x}^{'})\|\nonumber \\
			&+(\|{\bf u}\|\|{\bf u}^{'}\gamma^{'}\|-{\bf u}^{T}{\bf u}^{'}\gamma^{'})(\tilde{\eta}\|{\bf S}_{\pi_{(t)}^{c}(|{\bf x}^{'}|)}{\bf x}^{'}\|\|{\bf x}-r{\bf S}_{\pi_{(t)}^c(|{\bf x}|)}{\bf x}\|-({\bf x}^{'})^{T}{\bf S}_{\pi_{(t)}^c(|{\bf x}^{'}|)}({\bf x}-r{\bf S}_{\pi_{(t)}^c(|{\bf x}|)}{\bf x}))\nonumber \\
 			&+(\|\gamma{\bf u}\|\|{\bf u}^{'}\|-\gamma{\bf u}^{T}{\bf u}^{'})(\tilde{\eta}\|{\bf S}_{\pi_{(t)}^{c}(|{\bf x}|)}{\bf x}\|\|{\bf x}^{'}-r{\bf S}_{\pi_{(t)}^c(|{\bf x}^{'}|)}{\bf x}^{'}\|-{\bf x}^{T}{\bf S}_{\pi_{(t)}^c(|{\bf x}|)}({\bf x}^{'}-r{\bf S}_{\pi_{(t)}^c(|{\bf x}^{'}|)}{\bf x}^{'})).\label{eq:last}
 		\end{align}
 		It is easy to check that when ${\bf x}= {\bf x}^{'}$, \begin{align*}\mathbb{E}[\tilde{X}_{t}({\bf x},{\bf u},{\gamma})\tilde{X}_{t}({\bf x}^{'},{\bf u}^{'},{\gamma}^{'})]-\mathbb{E}[\tilde{Y}_{t}({\bf x},{\bf u},{\gamma})\tilde{Y}_{t}({\bf x}^{'},{\bf u}^{'},{\gamma}^{'})]=0\leq 0.\end{align*}
		When ${\bf x}\neq {\bf x}^{'}$, we can prove that:
		\begin{align}
			\mathbb{E}[\tilde{X}_{t}({\bf x},{\bf u},{\gamma})\tilde{X}_{t}({\bf x}^{'},{\bf u}^{'},{\gamma}^{'})]-\mathbb{E}[\tilde{Y}_{t}({\bf x},{\bf u},{\gamma})\tilde{Y}_{t}({\bf x}^{'},{\bf u}^{'},{\gamma}^{'})]\geq 0.\label{second_ine}
 		\end{align}
 		Indeed, in view of \eqref{eq:last},  to show \eqref{second_ine}, it suffices to check that:
 		\begin{align}
			&\tilde{\eta}\|{\bf S}_{\pi_{(t)}^{c}(|{\bf x}|)}{\bf x}\|\|{\bf x}^{'}-r{\bf S}_{\pi_{(t)}^c(|{\bf x}^{'}|)}{\bf x}^{'}\|-{\bf x}^{T}{\bf S}_{\pi_{(t)}^{c}(|{\bf x}|)}({\bf x}^{'}-r{\bf S}_{\pi_{(t)}^c(|{\bf x}^{'}|)}{\bf x}^{'})\geq 0, \label{eq:g1}\\\
 			&  \tilde{\eta}\|{\bf S}_{\pi_{(t)}^{c}(|{\bf x}^{'}|)}{\bf x}^{'}\|\|{\bf x}-r{\bf S}_{\pi_{(t)}^c(|{\bf x}|)}{\bf x}\|-({\bf x}^{'})^{T}{\bf S}_{\pi_{(t)}^{c}(|{\bf x}^{'}|)}({\bf x}-r{\bf S}_{\pi_{(t)}^c(|{\bf x}|)}{\bf x})\geq 0.\label{eq:g2}
 		\end{align}
 		The proofs of \eqref{eq:g1} and \eqref{eq:g2} follow the same procedure. Therefore, we will provide detailed steps only for the proof of \eqref{eq:g1}. We start by writing:
 		\begin{align}
 			&\tilde{\eta}\|{\bf S}_{\pi_{(t)}^{c}(|{\bf x}|)}{\bf x}\|\|{\bf x}^{'}-r{\bf S}_{\pi_{(t)}^c(|{\bf x}^{'}|)}{\bf x}^{'}\|-{\bf x}^{T}{\bf S}_{\pi_{(t)}^c(|{\bf x}|)}({\bf x}^{'}-r{\bf S}_{\pi_{(t)}^c(|{\bf x}^{'}|)}{\bf x}^{'})\nonumber \\
 			=& \|{\bf S}_{\pi_{(t)}^{c}(|{\bf x}|)}{\bf x}\|(\tilde{\eta}\|{\bf x}^{'}-r{\bf S}_{\pi_{(t)}^c(|{\bf x}^{'}|)}{\bf x}^{'}\|-\frac{{\bf x}^{T}{\bf S}_{\pi_{(t)}^{c}(|{\bf x}|)}({\bf x}^{'}-r{\bf S}_{\pi_{(t)}^c(|{\bf x}^{'}|)}{\bf x}^{'})}{\|{\bf S}_{\pi_{(t)}^c(|{\bf x}|)}{\bf x}\|})\nonumber \\
 			=&\|{\bf S}_{\pi_{(t)}^{c}(|{\bf x}|)}{\bf x}\|(\frac{({\bf x}^{'})^{T}{\bf S}_{\pi_{(t)}^{c}(|{\bf x}^{'}|)}({\bf x}^{'}-r{\bf S}_{\pi_{(t)}^c(|{\bf x}^{'}|)}{\bf x}^{'})}{\|{\bf S}_{\pi_{(t)}^c(|{\bf x}^{'}|)}{\bf x}^{'}\|}-\frac{{\bf x}^{T}{\bf S}_{\pi_{(t)}^{c}(|{\bf x}|)}({\bf x}^{'}-r{\bf S}_{\pi_{(t)}^c(|{\bf x}^{'}|)}{\bf x}^{'})}{\|{\bf S}_{\pi_{(t)}^c(|{\bf x}|)}{\bf x}\|}).\label{eq:sec}
 		\end{align}
 		To get the desired, we use the fact that ${\bf S}_{\pi_{(t)}^{c}(|{\bf x}^{'}|)}({\bf x}^{'}-r{\bf S}_{\pi_{(t)}^c(|{\bf x}^{'}|)}{\bf x}^{'})$ and ${\bf S}_{\pi_{(t)}^{c}(|{\bf x}^{'}|)}{\bf x}^{'}$ are aligned,  hence
 		\begin{equation*}
			\frac{({\bf x}^{'})^{T}{\bf S}_{\pi_{(t)}^{c}(|{\bf x}^{'}|)}({\bf x}^{'}-r{\bf S}_{\pi_{(t)}^c(|{\bf x}^{'}|)}{\bf x}^{'})}{\|{\bf S}_{\pi_{(t)}^c(|{\bf x}^{'}|)}{\bf x}^{'}\|}=\|{\bf x}^{'}-r{\bf S}_{\pi_{(t)}^c(|{\bf x}^{'}|)}{\bf x}^{'})\| . 
 		\end{equation*}
 		Next, we use the fact that when the values of ${\bf x}$ are fixed up to a sign but not their positions, the selection matrix ${\bf S}_{\pi_{(t)}^{c}(|{\bf x}|)}$ that maximizes 
		$$
		\frac{{\bf x}^{T}{\bf S}_{\pi_{(t)}^{c}(|{\bf x}|)}({\bf x}^{'}-r{\bf S}_{\pi_{(t)}^c(|{\bf x}^{'}|)}{\bf x}^{'})}{\|{\bf S}_{\pi_{(t)}^c(|{\bf x}|)}{\bf x}\|}
 		$$
		 is when it selects the $k$  positions in ${\bf x}^{'}-r{\bf S}_{\pi_{(t)}^c(|{\bf x}^{'}|)}{\bf x}^{'}$ corresponding to the largest magnitude.  Since by assumption, for all ${\bf x}^{'}\in \mathcal{I}_x$, the magnitudes of elements of ${\bf x}^{'}$ do not belong to the interval  {$(t(1-r),t(1+r))$}, at optimum ${\bf S}_{\pi_{(t)}^{c}}(|{\bf x}|)$ should be equal to ${\bf S}_{\pi_{(t)}^{c}(|{\bf x}^{'}|)}$ to select the $k$ largest entries in magnitude of ${\bf x}^{'}-r{\bf S}_{\pi_{(t)}^c(|{\bf x}^{'}|)}{\bf x}^{'}$.The result then follows from an application of Cauchy Shwartz inequality. Similarly, \eqref{eq:g2} holds.
  
		For each ${\bf x}\in I_x$, ${\bf b}\in I_b$, ${\bf u}\in I_u$ and ${\gamma}\in I_\gamma$, let $\lambda_{{\bf x},{\bf b},{\bf u},\gamma}=c-\psi({\bf x},{\bf b},{\bf u},\gamma)$. We observe that:
 		$$
	\Big\{\min_{\substack{{\bf b}\in I_b\\ {\bf x}\in I_{x}}} \max_{\substack{{\bf u}\in I_u\\ \gamma\in I_{\gamma}}} \tilde{X}_{t}({\bf x},{\bf u},\gamma)\leq \lambda_{{\bf x},{\bf b},{\bf u},\gamma}\Big\}={\underset{\substack{{\bf x}\in I_x\\ {\bf b}\in I_b}}{\cap}\underset{\substack{{\bf u}\in I_u \\ \gamma\in I_\gamma}}{\cup} \Big\{\tilde{X}_t({\bf x},{\bf u},\gamma)\leq\lambda_{{\bf x},{\bf b},{\bf u},\gamma} \Big\}},
		$$
	and a similar expression holds for $\tilde{Y}_{t}({\bf x},{\bf u},\gamma)$. By noting that these processes satisfy the conditions of Theorem \ref{lem:gor}, the lemma is established.
	\end{proof}
\end{lemma}
  {
\underline{Step 2: Extension of Lemma \ref{th:discrete_cgmt_without_constraint_s}} }
 In the following, we first generalize the result of the above lemma which analyzes random processes defined on discrete sets, to compact sets. More specifically, adopting the definition of the processes in Lemma \ref{th:discrete_cgmt_without_constraint_s}, we prove that for any $c\in \mathbb{R}$, 
 \begin{equation}
 \mathbb{P}\big[\min_{\substack{{\bf x}\in\mathcal{S}_x\\{\bf b}\in\mathcal{S}_b}}\max_{\substack{{\bf u}\in\mathcal{S}_u\\\gamma\in\mathcal{S}_\gamma}} \tilde{X}_t({\bf x},{\bf u},\gamma)+\psi({\bf x},{\bf b},{\bf u},\gamma)  \leq c\big]\leq \mathbb{P}\big[\min_{\substack{{\bf x}\in\mathcal{S}_x\\{\bf b}\in\mathcal{S}_b}}\max_{\substack{{\bf u}\in\mathcal{S}_u\\\gamma\in\mathcal{S}_\gamma}} \tilde{Y}_t({\bf x},{\bf u},\gamma)+\psi({\bf x},{\bf b},{\bf u},\gamma)\leq c\big]. \label{eq:proof}
 \end{equation}
 where $\mathcal{S}_x,\mathcal{S}_b$, $\mathcal{S}_u$ and $\mathcal{S}_\gamma$ are compact sets.

 We let $R=\sup_{\gamma\in \mathcal{S}_\gamma} |\gamma|$. 
 Since $\psi$ is continuous with respect to its variables, then it is uniformly continuous on  $\mathcal{S}_x\times\mathcal{S}_b\times\mathcal{S}_u\times\mathcal{S}_\gamma$.
 For any $\epsilon>0$, $\exists \delta:=\delta(\epsilon)>0$ such that $\forall ({\bf x,b,u},\gamma),(\tilde{{\bf x}},\tilde{{\bf b}},\tilde{{\bf u}},\tilde{\gamma})\in\mathcal{S}_x\times\mathcal{S}_b\times\mathcal{S}_u\times\mathcal{S}_\gamma$ with \begin{align*}\max(\|{\bf x}-\tilde{{\bf x}}\|,\|{\bf b}-\tilde{{\bf b}}\|,\|{\bf u}-\tilde{{\bf u}}\|,\|\gamma-\tilde{\gamma}\|)\leq\delta,\end{align*}we have
 \begin{align*}|\psi({\bf x,b,u},\gamma)-\psi(\tilde{{\bf x}},\tilde{{\bf b}},\tilde{{\bf u}},\tilde{\gamma})|<\epsilon.\end{align*}

 {Similarly,  the processes $\tilde{X}_t( {\bf x},{\bf u},\gamma)$ and $\tilde{Y}_t( {\bf x},{\bf u},\gamma)$ are uniformly continuous on the set $  \mathcal{S}_x\times\mathcal{S}_b\times\mathcal{S}_u\times\mathcal{S}_\gamma$. (Recall that \eqref{req3} ensures that the functions are continuous in ${\bf x}\in\mathcal{S}_x$.) For any $\epsilon>0$, there exists thus $\delta_2(\epsilon)>0$ such that $\forall ({\bf x,u},\gamma),(\tilde{{\bf x}},\tilde{{\bf u}},\tilde{\gamma})\in\mathcal{S}_x\times\mathcal{S}_u\times\mathcal{S}_\gamma$ with \begin{align*}\max(\|{\bf x}-\tilde{{\bf x}}\|,\|{\bf u}-\tilde{{\bf u}}\|,\|\gamma-\tilde{\gamma}\|)\leq\delta_2,\end{align*}we have
 $$ |\tilde{X}_t( {\bf x},{\bf u},\gamma)-\tilde{X}_t( \tilde{\bf x},\tilde{\bf u},\tilde{\gamma})|\leq \epsilon \  \text{ and }\ |\tilde{Y}_t( {\bf x},{\bf u},\gamma)-\tilde{Y}_t( \tilde{\bf x},\tilde{\bf u},\tilde{\gamma})|\leq \epsilon.
 $$
 Also given $\epsilon>0$, there exists $K_\epsilon$ such that with probability $1-\epsilon$, 
 $$
 \max( {\|{\bf G}\|},|z|,|\tilde{z}|, {\|{\bf g}\|}, {\|\tilde{\bf g}\|},{\|{\bf h}\|})\leq K_\epsilon.
 $$
 Based on the above facts,
 we can always find a constant $d_\epsilon$ such that with probability $1-\epsilon$, for all $\tilde{\delta}\in(0,d_\epsilon]$, 
 \begin{align}
 &\max(\|{\bf x}-\tilde{\bf x}\|, \|{\bf b}-\tilde{\bf b}\|, \|{\bf u}-\tilde{\bf u}\|, |\gamma-\tilde{\gamma}|)\leq \tilde{\delta}\nonumber \\\Longrightarrow &|\tilde{Y}_t({\bf x},{\bf u},\gamma)+\psi({\bf x},{\bf b},{\bf u},\gamma)-\tilde{Y}_t(\tilde{\bf x},\tilde{\bf u},\tilde\gamma)-\psi(\tilde{\bf x},\tilde{\bf b},\tilde{\bf u},\tilde\gamma)|\leq 2\epsilon 
 \nonumber\\\text{ and } &|\tilde{X}_t({\bf x},{\bf u},\gamma)+\psi({\bf x},{\bf b},{\bf u},\gamma)-\tilde{X}_t(\tilde{\bf x},\tilde{\bf u},\tilde\gamma)-\psi(\tilde{\bf x},\tilde{\bf b},\tilde{\bf u},\tilde\gamma)|\leq 2\epsilon.  \label{e263}
 \end{align}
 Now for such a $\tilde{\delta}$, consider $\mathcal{S}_x^\delta, \mathcal{S}_b^\delta$, $\mathcal{S}_u^\delta$ and $\mathcal{S}_\gamma^\delta$, the $\tilde{\delta}$-nets of $\mathcal{S}_x$, $\mathcal{S}_b$, $\mathcal{S}_u$ and ${\mathcal{S}}_\gamma$. Then based on \eqref{e263}, we obtain that with probability $1-\epsilon$
$$
 \Big|\min_{\substack{{\bf b}\in \mathcal{S}_b\\ {\bf x}\in {\mathcal{S}}_x}}\max_{\substack{{\bf u}\in \mathcal{S}_u\\ \gamma\in \mathcal{S}_\gamma}}\tilde{Y}_t({\bf x},{\bf u},\gamma) +\psi({\bf x,b},{\bf u},\gamma)- \min_{\substack{{\bf b}\in \mathcal{S}_b^\delta\\ {\bf x}\in {\mathcal{S}}_{x}^\delta}}\max_{\substack{{\bf u}\in \mathcal{S}_u^\delta\\ \gamma\in \mathcal{S}_\gamma^\delta}}\tilde{Y}_t({\bf x},{\bf u},\gamma) +\psi({\bf x,b},{\bf u},\gamma)\Big|\leq 2\epsilon $$and $$\Big|\min_{\substack{{\bf b}\in \mathcal{S}_b\\ {\bf x}\in {\mathcal{S}}_x}}\max_{\substack{{\bf u}\in \mathcal{S}_u\\ \gamma\in \mathcal{S}_\gamma}}\tilde{X}_t({\bf x},{\bf u},\gamma) +\psi({\bf x,b},{\bf u},\gamma)- \min_{\substack{{\bf b}\in \mathcal{S}_b^\delta\\ {\bf x}\in {\mathcal{S}}_{x}^\delta}}\max_{\substack{{\bf u}\in \mathcal{S}_u^\delta\\ \gamma\in \mathcal{S}_\gamma^\delta}}\tilde{X}_t({\bf x},{\bf u},\gamma) +\psi({\bf x,b},{\bf u},\gamma)\Big|\leq 2\epsilon.
 $$
Furthermore, from Lemma \ref{th:discrete_cgmt_without_constraint_s},we have for any $t\in \mathbb{R}$
 $$
\mathbb{P}\Big[\min_{\substack{{\bf b}\in \mathcal{S}_b^\delta\\ {\bf x}\in {\mathcal{S}}_{x}^\delta}}\max_{\substack{{\bf u}\in \mathcal{S}_u^\delta\\ \gamma\in \mathcal{S}_\gamma^\delta}}\tilde{X}_t({\bf x},{\bf u},\gamma) +\psi({\bf x,b},{\bf u},\gamma)\leq t\big]\leq \mathbb{P}\Big[\min_{\substack{{\bf b}\in \mathcal{S}_b^\delta\\ {\bf x}\in {\mathcal{S}}_{x}^\delta}}\max_{\substack{{\bf u}\in \mathcal{S}_u^\delta\\ \gamma\in \mathcal{S}_\gamma^\delta}}\tilde{Y}_t({\bf x},{\bf u},\gamma) +\psi({\bf x,b},{\bf u},\gamma)\leq t\big].
 $$
 Hence,
 \begin{align*}
 &\mathbb{P}\Big[\min_{\substack{{\bf b}\in \mathcal{S}_b\\ {\bf x}\in {\mathcal{S}}_{x}}}\max_{\substack{{\bf u}\in \mathcal{S}_u\\ \gamma\in \mathcal{S}_\gamma}}\tilde{X}_t({\bf x},{\bf u},\gamma) +\psi({\bf x,b},{\bf u},\gamma)\leq t\Big]\\\leq& \mathbb{P}\Big[\min_{\substack{{\bf b}\in \mathcal{S}_b^\delta\\ {\bf x}\in {\mathcal{S}}_{x}^\delta}}\max_{\substack{{\bf u}\in \mathcal{S}_u^\delta\\ \gamma\in \mathcal{S}_\gamma^\delta}}\tilde{X}_t({\bf x},{\bf u},\gamma) +\psi({\bf x,b},{\bf u},\gamma)\leq t+2\epsilon\big]+\epsilon\\
 \leq& \mathbb{P}\Big[\min_{\substack{{\bf b}\in \mathcal{S}_b^\delta\\ {\bf x}\in {\mathcal{S}}_{x}^\delta}}\max_{\substack{{\bf u}\in \mathcal{S}_u^\delta\\ \gamma\in \mathcal{S}_\gamma^\delta}}\tilde{Y}_t({\bf x},{\bf u},\gamma) +\psi({\bf x,b},{\bf u},\gamma)\leq t+2\epsilon\big]+\epsilon\\
\leq &\mathbb{P}\Big[\min_{\substack{{\bf b}\in \mathcal{S}_b\\ {\bf x}\in {\mathcal{S}}_{x}}}\max_{\substack{{\bf u}\in \mathcal{S}_u\\ \gamma\in \mathcal{S}_\gamma}}\tilde{Y}_t({\bf x},{\bf u},\gamma) +\psi({\bf x,b},{\bf u},\gamma)\leq t+4\epsilon\big]+\epsilon.
 \end{align*}
 By taking $\epsilon$ to zero, we prove thus the desired which is
 \begin{align}
 \mathbb{P}\Big[\min_{\substack{{\bf b}\in \mathcal{S}_b\\ {\bf x}\in {\mathcal{S}}_{x}}}\max_{\substack{{\bf u}\in \mathcal{S}_u\\ \gamma\in \mathcal{S}_\gamma}}\tilde{X}_t({\bf x},{\bf u},\gamma) +\psi({\bf x,b},{\bf u},\gamma)\leq t\Big] 
 \leq \mathbb{P}\Big[\min_{\substack{{\bf b}\in \mathcal{S}_b\\ {\bf x}\in {\mathcal{S}}_{x}}}\max_{\substack{{\bf u}\in \mathcal{S}_u\\ \gamma\in \mathcal{S}_\gamma}}\tilde{Y}_t({\bf x},{\bf u},\gamma) +\psi({\bf x,b},{\bf u},\gamma)\leq t\big].\label{finite_R}
\end{align} }

  {Next, we demonstrate that the $\mathcal{S}_\gamma$ in \eqref{finite_R} can be extended to $\mathbb{R}$. Let $\mathcal{S}_\gamma^R=\{\gamma|\ |\gamma|<R\}$. We can write:
 $$
 \max_{\substack{{\bf u}\in \mathcal{S}_u\\ \gamma\in \mathbb{R}}} \tilde{X}_t({\bf x},{\bf u},\gamma) +\psi({\bf x,b},{\bf u},\gamma)=\sup_{R\geq 0} \max_{\substack{{\bf u}\in \mathcal{S}_u\\ \gamma\in \mathcal{S}_\gamma^R}}\tilde{X}_t({\bf x},{\bf u},\gamma) +\psi({\bf x,b},{\bf u},\gamma)=\lim_{R\to\infty} \max_{\substack{{\bf u}\in \mathcal{S}_u\\ \gamma\in \mathcal{S}_\gamma^R}}\tilde{X}_t({\bf x},{\bf u},\gamma) +\psi({\bf x,b},{\bf u},\gamma).
 $$
 Similarly, 
 $$
 \max_{\substack{{\bf u}\in \mathcal{S}_u\\ \gamma\in \mathbb{R}}} \tilde{Y}_t({\bf x},{\bf u},\gamma) +\psi({\bf x,b},{\bf u},\gamma)=\sup_{R\geq 0} \max_{\substack{{\bf u}\in \mathcal{S}_u\\ \gamma\in \mathcal{S}_\gamma^R}}\tilde{Y}_t({\bf x},{\bf u},\gamma) +\psi({\bf x,b},{\bf u},\gamma)=\lim_{R\to\infty} \max_{\substack{{\bf u}\in \mathcal{S}_u\\ \gamma\in \mathcal{S}_\gamma^R}}\tilde{Y}_t({\bf x},{\bf u},\gamma) +\psi({\bf x,b},{\bf u},\gamma).
 $$
 For ${\bf b}$ and ${\bf x}$ fixed, functions ${R}\mapsto \max_{\substack{{\bf u}\in \mathcal{S}_u\\ \gamma\in \mathcal{S}_\gamma^R}}\tilde{X}_t({\bf x},{\bf u},\gamma) +\psi({\bf x,b},{\bf u},\gamma)$ and ${R}\mapsto \max_{\substack{{\bf u}\in \mathcal{S}_u\\ \gamma\in \mathcal{S}_\gamma^R}}\tilde{Y}_t({\bf x},{\bf u},\gamma) +\psi({\bf x,b},{\bf u},\gamma)$  are non-decreasing. Using the fact that the minimum of the non-decreasing limit of a continuous function defined over compact sets is equal to the limit of minimum of this function, we obtain:
$$
 \min_{\substack{{\bf b}\in \mathcal{S}_b\\ {\bf x}\in {\mathcal{S}}_x}}\max_{\substack{{\bf u}\in \mathcal{S}_u\\ \gamma\in \mathbb{R}}}\tilde{X}_t({\bf x},{\bf u},\gamma) +\psi({\bf x,b},{\bf u},\gamma)=\lim_{R\to\infty} \min_{\substack{{\bf b}\in \mathcal{S}_b\\ {\bf x}\in {\mathcal{S}}_x}}\max_{\substack{{\bf u}\in \mathcal{S}_u\\ \gamma\in \mathcal{S}_\gamma^R}}\tilde{X}_t({\bf x},{\bf u},\gamma) +\psi({\bf x,b},{\bf u},\gamma)
 $$
 and
 $$
\min_{\substack{{\bf b}\in \mathcal{S}_b\\ {\bf x}\in {\mathcal{S}}_x}}\max_{\substack{{\bf u}\in \mathcal{S}_u\\ \gamma\in \mathbb{R}}}\tilde{Y}_t({\bf x},{\bf u},\gamma) +\psi({\bf x,b},{\bf u},\gamma)=\lim_{R\to\infty} \min_{\substack{{\bf b}\in \mathcal{S}_b\\ {\bf x}\in {\mathcal{S}}_x}}\max_{\substack{{\bf u}\in \mathcal{S}_u\\ \gamma\in \mathcal{S}_\gamma^R}}\tilde{Y}_t({\bf x},{\bf u},\gamma) +\psi({\bf x,b},{\bf u},\gamma),
 $$
 then \eqref{finite_R} also holds if  $\mathcal{S}_\gamma=\mathbb{R}$. }

 {\underline{Proving Theorem \ref{th:new_cgmt} by \eqref{finite_R}.}  } Theorem \ref{th:new_cgmt} follows as a direct by-product from \eqref{finite_R} where $\mathcal{S}_\gamma$ is either a compact set or the entire real axis.

 Now if $z$ and $\tilde{z}$ are negative:
 $$
 X_t({\bf x},{\bf b},{\bf u},\gamma)\geq \tilde{X}_t({\bf x},{\bf u},\gamma)+\psi({\bf x},{\bf b},{\bf u},\gamma).
$$
 Hence, for any $c\in \mathbb{R}$, 
 $$
 \mathbb{P}\Big[\min_{\substack{{\bf b}\in \mathcal{S}_b\\ {\bf x}\in \mathcal{S}_{x}}}\max_{\substack{{\bf u}\in\mathcal{S}_u\\ \gamma\in\mathcal{S}_\gamma}} {X}_t({\bf x},{\bf b},{\bf u},\gamma)\leq c\Big]\leq \mathbb{P}\Big[\min_{\substack{{\bf b}\in \mathcal{S}_b\\ {\bf x}\in \mathcal{S}_{x}}}\max_{\substack{{\bf u}\in\mathcal{S}_u\\ \gamma\in\mathcal{S}_\gamma}} \tilde{X}_t({\bf x},{\bf u},\gamma)+\psi({\bf x},{\bf b},{\bf u},\gamma)\leq c \ | \  z<0, \tilde{z}<0\Big].
$$
 Now, we use the fact that:
 $$
 \mathbb{P}\Big[\min_{\substack{{\bf b}\in \mathcal{S}_b\\ {\bf x}\in \mathcal{S}_{x}}}\max_{\substack{{\bf u}\in\mathcal{S}_u\\ \gamma\in\mathcal{S}_\gamma}} \tilde{X}_t({\bf x},{\bf u},\gamma)+\psi({\bf x},{\bf b},{\bf u},\gamma)\leq c \Big]\geq \frac{1}{4}\mathbb{P}\Big[\min_{\substack{{\bf b}\in \mathcal{S}_b\\ {\bf x}\in \mathcal{S}_{x}}}\max_{\substack{{\bf u}\in\mathcal{S}_u\\ \gamma\in\mathcal{S}_\gamma}} \tilde{X}_t({\bf x},{\bf u},\gamma)+\psi({\bf x},{\bf b},{\bf u},\gamma)\leq c \ | \ z<0, \tilde{z}<0 \Big]
 $$
to finally obtain:
 $$
 \mathbb{P}\Big[\min_{\substack{{\bf b}\in \mathcal{S}_b\\ {\bf x}\in \mathcal{S}_{x}}}\max_{\substack{{\bf u}\in\mathcal{S}_u\\ \gamma\in\mathcal{S}_\gamma}} {X}_t({\bf x},{\bf b},{\bf u},\gamma)\leq c\Big]\leq 4\mathbb{P}\Big[\min_{\substack{{\bf b}\in \mathcal{S}_b\\ {\bf x}\in \mathcal{S}_{x}}}\max_{\substack{{\bf u}\in\mathcal{S}_u\\ \gamma\in\mathcal{S}_\gamma}} \tilde{X}_t({\bf x},{\bf u},\gamma)+\psi({\bf x},{\bf b},{\bf u},\gamma)\leq c \Big].
$$
 Using the established inequality \eqref{eq:proof}, we thus get the desired. 

\subsubsection{Proof of Lemma \ref{lem:distance_revised} }\label{tech_lem}	Let $\hat{S}$ and $\hat{\bar{G}}$ be random variables drawn from the empirical distributions $\frac{1}{m}\sum_{i=1}^m \delta_{[{\bf s}]_i}$ and $\frac{1}{m}\sum_{i=1}^m \delta_{[\bar{\bf g}]_i}$. Denote the couplings that achieve the Wassertein distances, that is 
	\begin{align*} 
		\mathbb{E}(\hat{\bar{G}}-\bar{G})^2&=(\mathcal{W}_2(\frac{1}{m}\sum_{i=1}^m \boldsymbol{\delta}_{[\tilde{\bf g}]_i},\mathcal{N}(0,1)))^2,\\
		\mathbb{E}(\hat{S}-S)^2&=(\mathcal{W}_2(\frac{1}{m}\sum_{i=1}^m \boldsymbol{\delta}_{[{\bf s}]_i},\mathcal{R}))^2.
		\end{align*} 
        where $\mathcal{R}$ refers to the Rademacher distribution and  $S$ and $\overline{G}$ are two random variables following the Rademacher  and the standard normal distributions, respectively. 
  The Wassertein distance $(\mathcal{W}_2(\hat{\mu}(\overline{\bf e}_t^{\rm AO},{\bf s}), \mu_t^\star))^2$ can be thus upper-bounded by:
	\begin{align*}
	&(\mathcal{W}_2(\hat{\mu}(\overline{\bf e}_t^{\rm AO},{\bf s}), \mu_t^\star))^2\\\leq &\mathbb{E}[(1+2\rho(\theta^\star)^2) (\hat{S}-\bar{S})^2+2((\tilde{\alpha}^\star)^2+(\theta^\star)^2(\delta(\tau^\star)^2-\rho)+2(\tilde{\alpha}^\star)^2\theta^\star)(\hat{\bar{G}}-\bar{G})^2] \\ 
	\leq& C_\mathcal{W} ( (\mathcal{W}_2(\frac{1}{m}\sum_{i=1}^m \boldsymbol{\delta}_{[\bar{\bf g}]_i},\mathcal{N}(0,1)))^2+(\mathcal{W}_2(\frac{1}{m}\sum_{i=1}^m \boldsymbol{\delta}_{[{\bf s}]_i},\mathcal{R}))^2)
	\end{align*}
	where $C_\mathcal{W}=\max(1+2\rho(\theta^\star)^2,2((\tilde{\alpha}^\star)^2+(\theta^\star)^2(\delta(\tau^\star)^2-\rho)+2(\tilde{\alpha}^\star)^2\theta^\star))$. 
Using Lemma \ref{lem:convergence_empirical_rate}, we know that with probability $1-C\exp(-cn\epsilon)$, 
$$
\max( ( \mathcal{W}_2(\frac{1}{m}\sum_{i=1}^m \boldsymbol{\delta}_{[\bar{\bf g}]_i},\mathcal{N}(0,1))^2),(\mathcal{W}_2(\frac{1}{m}\sum_{i=1}^m \boldsymbol{\delta}_{[{\bf s}]_i},\mathcal{R}))^2)\leq \frac{(\sqrt{c_e}-\sqrt{C}_e)^2}{8\mathcal{C}_{W}}\sqrt{\epsilon},
$$
hence with  probability $1-C\exp(-cn\epsilon)$,
$$
\mathcal{W}_2^2(\hat{\mu}(\overline{\bf e}_t^{\rm AO},{\bf s}), \mu_t^\star)\leq \frac{(\sqrt{c}_e-\sqrt{C}_e)^2}{4}\sqrt{\epsilon}
$$
This completed the proof of Lemma \ref{lem:distance_revised}.
\section{Technical Lemmas}
\label{app:technical_lemmas_revised}
\begin{theorem}[Theorem 1.1 in \cite{Gordon85}]
	Let $X_{i,j}$ and $Y_{i,j}$, $i=1,\cdots, I$, $j=1,\cdots,J$ be centered Gaussian processes such that:
	$$
	\left\{ \begin{array}{ll}&\mathbb{E}X_{ij}^2=\mathbb{E}Y_{ij}^2, \ \forall i,j\\
		&\mathbb{E}X_{ij}X_{ik}\geq \mathbb{E}Y_{ij}Y_{ik}, \forall i,j,k\\
		&\mathbb{E}X_{ij}X_{lk}\leq \mathbb{E}Y_{ij}Y_{lk}. \  \forall i\neq l \ \text{and} \  j,k
	\end{array}\right. 
	$$ 
	Then, for all $\lambda_{ij}\in\mathbb{R}$,
	$$
\displaystyle\mathbb{P}\Big[\displaystyle{\cap_{i=1}^{I}}\cup_{j=1}^J \Big\{Y_{i,j}\geq \lambda_{ij}\Big\}\Big]\geq \displaystyle\mathbb{P}\Big[\cap_{i=1}^{I}\cup_{j=1}^J \Big\{X_{i,j}\geq \lambda_{ij}\Big\}\Big].
	$$
	\label{lem:gor}
\end{theorem}
\begin{lemma}[\cite{Fournier2015}]
	Let $d\geq 1$ and $\mathcal{P}(\mathbb{R}^{d})$ be the set of all probability measures on $\mathbb{R}^{d}$. For $\mu\in \mathcal{P}(\mathbb{R}^{d})$, we consider an i.i.d sequence $({\bf x}_k)_{k\geq 1}$ of $\mu$-distributed random variables and denote for $N\geq 1$, the empirical measure:
	$$
	\mu_N:=\frac{1}{N}\sum_{k=1}^N \boldsymbol{\delta}_{{\bf x}_k}.
	$$
	For $\alpha,\gamma>0$, denote by $\mathcal{E}_{\alpha,\gamma}$ the quantity 
	$$
\mathcal{E}_{\alpha,\gamma}:=\int_{\mathbb{R}^{d}}\exp(\gamma \|{\bf x}\|^\alpha)\mu(dx).
	$$
	Let $r\geq \frac{d}{2}$. Assume that there exists $\alpha >r$ and $\gamma>0$ such that  $\mathcal{E}_{\alpha,\gamma}<\infty$.   If $r>\frac{d}{2}$, then for any $0<\epsilon<1$, 
	$$
\mathbb{P}\Big[\big(\mathcal{W}_r(\mu,\mu_N)\big)^r\geq \epsilon\Big]\leq C\exp(-cN\epsilon^2)
	$$
	where $C$ and $c$ are constants that depend only on $r$ and $d$.
	\label{lem:convergence_empirical_rate}
\end{lemma}
\begin{lemma}
	Let ${\bf X}=[{\bf x}_1,\cdots,{\bf x}_n]^{T}$ be a $m\times n$ matrix with i.i.d standard Gaussian entries. Then,
	$$
	\mathbb{P}\Big[\|{\bf X}{\bf X}^{T}\|\leq 9{\rm max}(m,n) \Big] \geq 1-2\exp(-{\rm max}(m,n)/2)
	$$
	and for $m\geq 4$, 
	$$
	\mathbb{P}\Big[\max_{1\leq i\leq n}\|{\bf x}_i\|\leq 2\sqrt{m}\Big]\geq 1-2n\exp(-(\sqrt{m}-1)^2/2)\geq 1-2n\exp(-m/8).
	$$
	\label{lem:spectral_norm}
\end{lemma}
\begin{lemma}
	Let ${\bf c}$ and ${\bf d}$ two distinct vectors in $\mathbb{R}^{m}$. Let $b\in\mathbb{R}_{> 0}$. Consider the following convex problem:
	$$
	\begin{array}{ll}
		m= \displaystyle{	\min_{{\bf e}\in\mathbb{R}^m}} &\|{\bf c}+{\bf e}\|^2 ,\\
		{\rm s.t.}&\|{\bf e}+{\bf d}\|^2\leq b.
	\end{array}
	$$
	Then, the above problem admits a unique minimizer given by:
	$$
	{\bf e}^{\star}=\frac{-{\bf c}-\lambda^\star {\bf d} }{1+\lambda^\star}
	$$
	where $\lambda^\star=-1+\frac{\|{\bf d}-{\bf c}\|}{\sqrt{b}}$.
	Moreover, at optimum, the optimal cost is given by:
	$$
	m=(-\sqrt{b}+\|{\bf c}-{\bf d}\|)^2	.
	$$
	Moreover, for all ${\bf e}$ such that $\|{\bf e}+{\bf d}\|^2\leq b$,  
	$$
	\|{\bf c}+{\bf e}\|^2\geq m+\|{\bf e}-{\bf e}^\star\|^2.
	$$
	\label{lem:KKT}
\end{lemma}
\begin{definition}
	The Rademacher average of a bounded set $\mathscr{A}\subset \mathbb{R}^n$ is given by:
$$
R_n(\mathscr{A})=\mathbb{E}\Big[\sup_{{\bf a}\in \mathscr{A}}|\frac{1}{n}\sum_{i=1}^n \sigma_i[{\bf a}]_i|\Big]
$$
where the expectation is taken over the $n$ independent  Rademacher random variables  $\sigma_i, i=1,\cdots,n$. 
	\end{definition}
\begin{lemma}[\cite{bucheron}]
	If $\mathscr{A}=\{{\bf a}_1,\cdots,{\bf a}_N\}$ is a finite set with $\|{\bf a}_i\|\leq L$, $i=1,\cdots,N$, then:
	$$
	R_n(\mathscr{A})\leq \frac{2L \log N}{n}.
	$$ 
	\label{lem:rademacher}
\end{lemma}

\bibliographystyle{IEEEtran}
\bibliography{ref}
\end{document}